\documentclass{memo-l}
\pdfoutput=1

\usepackage{epsf,latexsym,amsmath,amssymb,amscd,datetime}
\usepackage{amsmath,amsthm,amssymb,enumerate,eucal,url}
\usepackage{pdfsync}
\usepackage{afterpage}
\usepackage{float}
\usepackage{longtable}
\usepackage{graphicx}
\usepackage{epsf}
\usepackage{latexsym}
\usepackage{amsmath}
\usepackage{amssymb}
\usepackage{amsthm}
\usepackage{amscd}
\usepackage[all]{xy}
\usepackage{graphics}
\usepackage{enumitem}
\usepackage[usenames,dvipsnames,svgnames,table]{xcolor}
\usepackage{tikz}
\usepackage{todonotes}
\usepackage{verbatim}
\usepackage{setspace}
\usepackage{}
\usepackage{graphicx}

\usepackage{color}
\newenvironment{jfnote}{ \bgroup \color{blue} }{\egroup}
\newcommand{\mynote}[1]{\begin{jfnote}#1\end{jfnote}}

\newcommand{\oddf}{!_{\rm odd}}

\newcommand{\nonsigma}{\xi}

\DeclareMathOperator{\Var}{Var}

\newcommand{\etafund}{{\eta_{\rm \,fund}}}

\newcommand{\TF}{{\rm TF}}
\newcommand{\MT}{{\rm ModHashTr}}
\newcommand{\SNB}{{\rm StrNonBack}}

\theoremstyle{plain}
\newtheorem{theorem}{Theorem}[section]
\newtheorem{lemma}[theorem]{Lemma}
\newtheorem{proposition}[theorem]{Proposition}
\newtheorem{corollary}[theorem]{Corollary}
\newtheorem{conjecture}[theorem]{Conjecture}

\theoremstyle{definition}
\newtheorem{definition}[theorem]{Definition}

\newtheorem{notation}[theorem]{Notation}

\newtheorem{remark}[theorem]{Remark}

\numberwithin{section}{chapter}
\numberwithin{equation}{chapter}

\newcommand{\cal}{\mathcal}

\newcommand{\trace}{{\rm Tr}}

\newcommand{\ignore}[1]{}

\newcommand{\cc}{{\cal C}}

\newcommand{\reals}{{\mathbb R}}

\newcommand{\integers}{{\mathbb Z}}

\newcommand{\complex}{{\mathbb C}}

\newcommand{\prob}[2]{{\Prob}_{#1}{\left[\; #2\; \right]}}

\newcommand{\expect}[2]{{\mathbb{E}}_{#1}{\left[\; #2\; \right]}}

\newcommand{\probb}[1]{\Prob_{G\in \cC_n(B)}\left[ #1 \right]}

\newcommand{\rhoroot}[1]{\rho^{1/2}(H_{#1})} 

\DeclareMathOperator{\HasTangle}{HasTangle}

\newcommand{\HasTanglemin}{\Tangle_{r,B}^{\min}}

\newcommand{\tanglefree}{{\mathrm{TF}(r,B)}}
\newcommand{\tanglefreeindicator}{\II_{\mathrm{TF}(r,B)}} 
\newcommand{\tanglefreeeps}{{\mathrm{TF}(r,B,\epsilon)}}
\newcommand{\tanglefreeindicatoreps}{\II_{\mathrm{TF}(r,B,\epsilon)}} 

\newcommand{\tohere}{\smallskip{\bgroup\color{Red} \bigskip\par \textbf{ ********************************* TO HERE *********************************}}\bigskip}

\DeclareMathOperator{\ord}{ord}

\DeclareMathOperator{\WalkSum}{WalkSum}
\DeclareMathOperator{\error}{error}
\DeclareMathOperator{\Tr}{Tr}
\DeclareMathOperator{\CertTr}{CertTr}
\DeclareMathOperator{\Prob}{Prob}
\DeclareMathOperator{\Spec}{Spec}
\DeclareMathOperator{\Cone}{Cone}
\DeclareMathOperator{\VLG}{VLG}
\DeclareMathOperator{\Minimal}{Minimal}
\DeclareMathOperator{\Tangle}{Tangle}
\DeclareMathOperator{\Types}{Types}
\DeclareMathOperator{\Line}{Line}
\DeclareMathOperator{\id}{id}
\DeclareMathOperator{\Graph}{Graph}
\DeclareMathOperator{\support}{support}
\DeclareMathOperator{\Occurs}{Occurs}

\DeclareMathOperator{\TypeGraph}{TypeGraph}

\DeclareMathOperator{\SNBC}{SNBC}
\DeclareMathOperator{\NB}{NB}
\DeclareMathOperator{\CertSNBC}{CertSNBC}
\DeclareMathOperator{\Term}{Term}
\DeclareMathOperator{\Aut}{Aut}
\DeclareMathOperator{\NumOccurs}{NumOccurs}

\newcommand{\from}{\colon}

\newcommand{\dmax}{d_{\max{}}}
\newcommand{\rhonew}{\rho^{\mathrm{new}}}
\newcommand{\specnew}{\Spec^{\mathrm{new}}}
\newcommand{\Specnew}{\Spec^{\mathrm{new}}}
\newcommand{\taufund}{\tau_{\mathrm{fund}}}

\newcommand{\Tanglemin}{\Tangle_{r,B}^{\min}}

\newcommand{\vecmk}{(\vec m, \vec k\,)}

\renewcommand{\emptyset}{\text{\O}}

\newcommand{\potwalk}{(w;\vec t\,)}
\newcommand{\potwalkclass}{[w;\vec t\,]_n}

\newcommand{\BigOn}[1]{O\left( #1 \right)}

\newcommand{\Edir}{E^{\mathrm{dir}}}

 \theoremstyle{plain}
 \newtheorem{thm}[theorem]{Theorem}
 \newtheorem{prop}[theorem]{Proposition}
 \newtheorem{lem}[theorem]{Lemma}
 \newtheorem{cor}[theorem]{Corollary}

 \newtheorem*{thm*}{Theorem}
 \newtheorem*{prop*}{Proposition}
 \newtheorem*{lem*}{Lemma}
 \newtheorem*{cor*}{Corollary}
 \newtheorem*{conj*}{Conjecture}
 
 \theoremstyle{definition}
 \newtheorem{defn}[theorem]{Definition}
 \newtheorem{ex}[theorem]{Example}

\newcommand\CC{\mathbb{C}}

\newcommand\EE{\mathbb{E}}

\newcommand\II{\mathbb{I}}

\newcommand\RR{\mathbb{R}}

\newcommand\ZZ{\mathbb{Z}}

\DeclareMathAlphabet{\mathcal}{OMS}{cmsy}{m}{n}

\newcommand\cC{\mathcal{C}}

\newcommand\cE{\mathcal{E}}
\newcommand\cF{\mathcal{F}}
\newcommand\cG{\mathcal{G}}
\newcommand\cH{\mathcal{H}}
\newcommand\cI{\mathcal{I}}
\newcommand\cJ{\mathcal{J}}
\newcommand\cK{\mathcal{K}}
\newcommand\cL{\mathcal{L}}

\newcommand\cS{\mathcal{S}}
\newcommand\cT{\mathcal{T}}

\newcommand\cW{\mathcal{W}}

\usepackage[order=letter]{glossaries}
\usepackage{glossary-tree}

\newcommand*{\glsgobblenumber}[1]{}

\makeatletter
\newcommand*{\glsaddnp}[2][]{%
  \glsdoifexists{#2}%
  {%
    \def\@glsnumberformat{glsgobblenumber}
    \edef\@gls@counter{\csname glo@#2@counter\endcsname}%
    \setkeys{glossadd}{#1}%
    \@gls@saveentrycounter
    \@do@wrglossary{#2}%
  }%
}
\newcommand{\glsaddallunused}[1][]{%
  \edef\@glo@type{\@glo@types}%
  \setkeys{glossadd}{#1}%
  \forallglsentries[\@glo@type]{\@glo@entry}{%
    \ifglsused{\@glo@entry}{}{%
     \glsaddnp[#1]{\@glo@entry}}}%
}
\makeatother

\newglossaryentry{asymptotic expansion}{
  sort={asymptotic expansion},
  name={$1/n$-asymptotic expansion},
  description={an asymptotic expansion of a function $f(k,n)$ in powers of
  $1/n$ with coefficients being functions of $k$} }

\newglossaryentry{B-Ramanujan-coef}{
  sort={B-Ramanujan-coef},
  name={$B$-Ramanujan},
  description={a $1/n$-asymptotic expansion whose coefficients are 
  $B$-Ramanujan functions},
  parent={asymptotic expansion} }

\newglossaryentry{coefficient}{
  sort={coefficient},
  name={coefficient},
  description={a function of $k$ that appears as the coefficient of a
  $1/n$ power in a $1/n$-asymptotic expansion},
  parent={asymptotic expansion} }

\newglossaryentry{beaded path}{
  sort={beaded path},
  name={beaded path},
  description={a walk in a graph, each of whose interior vertices have
  degree two; especially used in the type (graph) of a walk, where one
  deletes all or all but one vertices of degree two, which breaks the
  deleted vertices into interior vertices of beaded paths},
  see={type} }

\newglossaryentry{certified trace}{
  sort={certified trace},
  name={certified trace},
  description={the number of strictly non-backtracking closed walks, $w$, in a
  graph such that $\Graph(w)$ is of less than a given order and has its
  Hashimoto matrix spectral radius at most $\rhoroot B$ or,
  in Section~\ref{se:p2-fund-exp}, at most $\epsilon+\rhoroot B$ for
  a fixed value of $\epsilon>0$} }

\newglossaryentry{B-graph}{
  sort={B-graph},
  name={$B$-graph},
  description={a graph morphism to $B$, or, abusively,
  a graph with a given morphism to the graph $B$} }

\newglossaryentry{morphism of B-graphs}{
  sort={morphism of B-graphs},
  name={morphism of $B$-graphs},
  description={a morphism of the sources that respects the $B$ structure
  of the $B$-graphs},
  parent={B-graph}}

\newglossaryentry{B-eps-tangle}{
  sort={B-eps-tangle},
  name={$(B,\epsilon)$-tangle},
  description={a connected graph, $\psi\in\Occurs_B$, 
  for which $\rho(H_\psi)\ge
  \epsilon +\rhoroot B$} }

\newglossaryentry{B-Ramanujan function}{
  sort={B-Ramanujan function},
  name={$B$-Ramanujan function},
  description={a function with a polyexponential part in the eigenvalues
  $\mu_i(B)$ and an error term} } 

\newglossaryentry{principle part}{
  sort={principle part},
  name={principle part},
  description={the part of a $B$-Ramanujan function that is a
  polyexponential in the eigenvalues $\mu_i(B)$},
  parent={B-Ramanujan function} } 

\newglossaryentry{error term}{
  sort={error term},
  name={error term},
  description={the error term of a $B$-Ramanujan function},
  parent={B-Ramanujan function} } 

\newglossaryentry{B-tangle}{
  sort={B-tangle},
  name={$B$-tangle},
  description={a connected graph, $\psi\in\Occurs_B$, 
  for which $\rho(H_\psi)\ge
  \rhoroot B$},
  see=[same as]{tangle of B} } 

\newglossaryentry{bouquet half}{
  sort={bouquet half},
  name={bouquet of $d$ half-loops},
  description={the graph, $H_d$, which has one vertex and $d$
  half-loops} }

\newglossaryentry{bouquet whole}{
  sort={bouquet whole},
  name={bouquet of $d/2$ whole-loops},
  description={the graph, $W_{d/2}$, which has one vertex and $d/2$
  whole-loops (with $d$ even)} }

\newglossaryentry{Broder-Shamir model}{
  sort={Broder-Shamir model},
  name={Broder-Shamir model},
  description={our standard model, $\cC_n(B)$, of a covering map of degree $n$
  to a graph, $B$ } } 

\newglossaryentry{coincidence}{
  sort={coincidence},
  name={coincidence},
  description={a value, $i$, for which the head of the $i$-th edge in a 
  random walk
  was already visited in the walk, but the value of this $i$-th edge was
  not determined} }

\newglossaryentry{covering map}{
  sort={covering map},
  name={covering map},
  description={a morphism of graphs or directed graphs that is an
  isomorphism on heads neighbourhoods and tails neighbourhoods of
  each vertex in the domain with that of its image} }

\newglossaryentry{directed graph}{
  sort={directed graph},
  name={directed graph},
  description={a tuple $G=(V_G,\Edir_G,h_G,t_G)$ of a set of vertices,
  directed edges, and heads and tails maps} }

\newglossaryentry{morphism of directed graphs}{
  sort={morphism of directed graphs},
  name={morphism of directed graphs},
  description={a set theoretic map of vertices and edges from one graph to
  another that preserves the heads and tails relations},
  parent={directed graph} }

\newglossaryentry{directed line graph}{
  sort={directed line graph},
  name={directed line graph},
  description={the graph, $\Line(G)$, of a graph $G$ whose vertices are the
  directed edges of $G$, with an edge from $e_1$ to $e_2$ iff
  $h_G(e_1)=t_G(e_2)$ and $\iota_G(e_2)\ne\iota_G(e_1)$} }

\newglossaryentry{etafund}{
  sort={etafund},
  name={$\etafund(B)$},
  description={the order of the smallest strict tangle in $B$},
  see=[sometimes called the]{fundamental order}}

\newglossaryentry{fundamental order}{
  sort={fundamental order},
  name={fundamental order of $B$},
  description={the order of the smallest strict tangle in $B$},
  see=[denoted]{etafund} }

\newglossaryentry{etale map}{
  sort={etale map},
  name={\'etale map},
  description={a morphism of graphs or directed graphs that is an
  injection on heads neighbourhoods and tails neighbourhoods of
  each vertex in the domain with that of its image} }

\newglossaryentry{graph}{
  sort={graph},
  name={graph},
  description={a tuple $G=(V_G,\Edir_G,h_G,t_G,\iota_G)$ where
  $(V_G,\Edir_G,h_G,t_G)$ is a directed graph (the {\em underlying
  directed graph}), and $\iota_G$ is a
  heads/tails reversing involution (sometimes called the {\em opposite
  map})} }

\newglossaryentry{directed edge}{
  sort={directed edge},
  name={directed edge},
  description={the edge set of a directed graph, or, for a graph, that of
  its underlying directed graph}, parent={graph} }

\newglossaryentry{edge}{
  sort={edge},
  name={edge},
  description={an orbit of the graph involution in a graph, 
  i.e., a set of the form
  $\{e,\iota e\}$, where $e$ is a directed edge of the underlying 
  directed graph and where $\iota$ is the opposite map} 
  parent={graph} }

\newglossaryentry{involution}{
  sort={involution},
  name={involution},
  description={\nopostdesc},
  see=[same as]{opposite map}, parent={graph} }

\newglossaryentry{morphism of graphs}{
  sort={morphism of graphs},
  name={morphism of graphs},
  description={a morphism of underlying directed graphs that preserves
  the opposite map}, parent={graph} }

\newglossaryentry{opposite map}{
  sort={opposite map},
  name={opposite map},
  description={a heads/tail reversing involution that gives a directed
  graph the structure of a graph}, parent={graph} }

\newglossaryentry{orientation}{
  sort={orientation},
  name={orientation},
  description={the choice of one representative directed edge for an
  edge of a graph}, parent={graph} }

\newglossaryentry{oriented graph}{
  sort={oriented graph},
  name={oriented graph},
  description={a graph with an orientation for each of its edges},
  parent={graph} }

\newglossaryentry{undirected edge}{
  sort={undirected edge},
  name={undirected edge},
  description={\nopostdesc},
  see=[same as]{edge},
  parent={graph} }

\newglossaryentry{graph of a walk}{
  sort={graph of a walk},
  name={graph of a walk},
  description={the subgraph, $\Graph(w)$, traced out by a walk, $w$,
  in a graph} }

\newglossaryentry{half-loop}{
  sort={half-loop},
  name={half-loop},
  description={an edge in a graph which is paired (via the graph involution)
  with itself} }

\newglossaryentry{Hashimoto matrix}{
  sort={Hashimoto matrix},
  name={Hashimoto matrix},
  description={the adjacency matrix, $H_G$, of the directed line graph
  of $G$} }

\newglossaryentry{loop}{
  sort={loop},
  name={loop},
  description={a strictly non-backtracking closed walk each of whose
  vertices have degree two} }

\newglossaryentry{new function}{
  sort={new function},
  name={new function},
  description={a function on the vertices of a covering graph such that
  on all vertex fibres their sum is zero} }

\newglossaryentry{new spectrum}{
  sort={new spectrum},
  name={new spectrum},
  description={the part of the spectrum arising from new functions of a
  covering map} }

\newglossaryentry{occurs-B}{
  sort={occurs-B},
  name={$\Occurs_B$},
  description={a graph which is a subgraph of some element of $\cC_n(B)$}}

\newglossaryentry{old function}{
  sort={old function},
  name={old function},
  description={a function on the vertices of a covering graph such that
  on all vertex fibres they are constant} }

\newglossaryentry{old spectrum}{
  sort={old spectrum},
  name={old spectrum},
  description={the part of the spectrum arising from old functions of a
  covering map} }

\newglossaryentry{order}{
  sort={order},
  name={order},
  description={$\ord(G)=-\chi(G)$, minus the Euler characteristic of $G$,
  i.e., $|E_G|-|V_G|$} }

\newglossaryentry{oriented line graph}{
  sort={oriented line graph},
  name={oriented line graph},
  description={\nopostdesc},
  see=[same as]{directed line graph} }

\newglossaryentry{permutation assignment}{
  sort={permutation assignment},
  name={permutation assignment},
  description={a map $\Edir_B\to \cS_n$, used in defining the Broder-Shamir
  model, $\cC_n(B)$, of a random covering of degree $n$ of a graph, $B$} }

\newglossaryentry{polyexponential}{
  sort={polyexponential},
  name={polyexponential},
  description={a real or complex valued function on $\integers_{\ge 1}^m$
  given by a sum of a product of polynomials in the variables times
  exponential functions in the variables} }

\newglossaryentry{pruned}{
  sort={pruned},
  name={pruned},
  description={a graph each of whose vertices has degree at least two} }

\newglossaryentry{Ramanujan graph}{
  sort={Ramanujan graph},
  name={Ramanujan graph},
  plural={Ramanujan graphs},
  description={a $d$-regular graph, for some integer $d$, such that all
  its adjacency eigenvalues, aside from $d$ and possibly $-d$, are at most
  $2(d-1)^{1/2}$ } }

\newglossaryentry{self-loop}{
  sort={self-loop},
  name={self-loop},
  description={a directed edge or edge in a graph or directed graph whose
  head and tail are the same} }

\newglossaryentry{subgraphs occurring in a B covering}{
  sort={subgraphs occurring in a B covering},
  name={subgraphs occurring in a $B$ covering},
  description={the subgraphs occurring in $\cC_n(B)$ for some $n$, where
  $\cC_n(B)$ is the Broder-Shamir model or a related model} }

\newglossaryentry{tangle of B}{
  sort={tangle of B},
  name={tangle of $B$},
  see=[same as]{B-tangle},  
  description={a graph, $\psi\in\Occurs_B$ for which $\rho(H_\psi)\ge
  \rhoroot B$}}

\newglossaryentry{strict tangle of B}{
  sort={strict tangle of B},
  name={strict tangle of $B$},
  description={a graph, $\psi$, in $\Occurs_B$ for which $\rho(H_\psi)>
  \rhoroot B$},
  plural={strict tangles of $B$},
  parent={tangle of B} }

\newglossaryentry{tree}{
  sort={tree},
  name={tree},
  description={a connected graph of Euler characteristic $-1$, i.e., a
  tree in the usual sense, which includes the case of a graph with one
  vertex and no edges} }

\newglossaryentry{treeless}{
  sort={treeless},
  name={treeless},
  description={a graph with no connected components that are trees},
  parent={tree}}

\newglossaryentry{type}{
  sort={type},
  name={type},
  description={the data associated to any closed walk in a graph 
  which remembers the following information:
  the initial vertex, the vertices of length at least three, the 
  (beaded) paths
  of the walk between such vertices (yielding a graph called the
  {\em type graph} of the walk, which is an oriented graph by orienting
  the edges by the direction in which they are first traversed;
  the order in which each of these
  vertices and paths are first encountered; and the $B$-neighbourhood
  of each vertex (the {\em lettering}). 
  Alternatively, it is this data (an oriented graph, orderings
  of vertices and edges, and $B$-neighbourhood) as abstract data (not
  associated to a particular walk)} }

\newglossaryentry{type graph}{
  sort={type graph},
  name={type graph},
  description={The part of a type consisting only of the underlying 
  oriented graph
  of the type (without the ordering and $B$-neighbourhood date); the
  type graph of a walk, $w$, is denoted $\TypeGraph(w)$},
  parent={type} }

\newglossaryentry{undirected graph}{
  sort={undirected graph},
  name={undirected graph},
  description={\nopostdesc},
  see=[same as]{graph} }

\newglossaryentry{variable-length graph}{
  sort={variable-length graph},
  name={variable-length graph},
  description={a graph or directed graph whose edges or directed edges
  each have an associated length} }

\newglossaryentry{realization of a variable-length graph}{
  sort={realization of a variable-length graph},
  name={realization of a variable-length graph},
  description={the graph, $\VLG(G,\vec k)$, by taking a directed or undirected 
  variable-length graph and replacing each edge, $e$, by a path of length
  $k(e)$},
  parent={variable-length graph} }

\newglossaryentry{walk}{
  sort={walk},
  name={walk},
  description={an alternating sequence of vertices and 
  directed edges that ``follow in sequence,'' in
  a graph or directed graph} }

\newglossaryentry{closed}{
  sort={closed},
  name={closed},
  description={a walk whose first vertex equals its last},
  parent={walk} }

\newglossaryentry{interior vertex}{
  sort={interior vertex},
  name={interior vertex},
  description={a vertex in a walk which is not the first or last vertex},
  plural={interior vertices},
  parent={walk} }

\newglossaryentry{non-backtracking}{
  sort={non-backtracking},
  name={non-backtracking},
  description={a walk where any two consecutive edges are not opposites},
  parent={walk} }

\newglossaryentry{strictly non-backtracking closed}{
  sort={strictly non-backtracking closed},
  name={strictly non-backtracking closed},
  description={a closed, non-backtracking walk whose first and last edges
  are not opposites},
  parent={walk} }

\newglossaryentry{walk sum}{
  sort={walk sum},
  name={walk sum},
  description={a sum of expected values of walks, $w$, in graphs subject to
  certain restrictions on $\Graph(w)$, the length of $w$, and the manner
  in which $w$ traces out $\Graph(w)$ in its sequence of vertices and edges} }

\newglossaryentry{whole-loop}{
  sort={whole-loop},
  name={whole-loop},
  description={an edge in a graph which is paired (via the graph involution)
  with a different edge} }

\newglossaryentry{convolution}{
  sort={convolution},
  name={convolution},
  description={convolution in the additive sense, e.g., the sum of 
  $g_1(k_1)g_2(k_2)$ with $k_1+k_2$ fixed} }

\newglossaryentry{weighted convolution}{
  sort={weighted convolution},
  name={weighted convolution},
  description={weighted convolution in the additive sense, e.g., the sum of 
  $g_1(k_1)g_2(k_2)$ with 
  $m_1 k_1 + m_2 k_2$ fixed for fixed $m_1,m_2$ called the {\em weights}},
  parent={convolution} }

\newglossaryentry{growth}{
  sort={growth},
  name={growth},
  description={describes a function of $k$ bounded by $Ck^C \rho^k$ 
  for some given
  $C$ and (more importantly) $\rho$} }

\newglossaryentry{coefficient norm}{
  sort={coefficient norm},
  name={coefficient norm},
  description={the norm that takes a (real or complex) polynomial and 
  returns the absolute value of its largest coefficient} }

\newglossaryentry{potential walk}{
  sort={potential walk},
  name={potential walk},
  description={a pair $(w; \vec t\,)$, consisting of a walk, $w$, of
  length $k$ in the base graph, $B$, and a trajectory of values, $t$,
  which assigns to every vertex of $w$ an integer from $1$ to $n$ (for
  the model $\cC_n(B)$; intuitively a potential walk is a random event 
  that give rise to a walk in the random graphs in $\cC_n(B)$} }
  
\newglossaryentry{form}{
  sort={form},
  name={form},
  description={the data of all the information about a $G\in\cC_n(B)$
  that a potential walk determines; formally, it is
  a tuple $(F,\cE)$, where $F$ is a variable-length graph,
  and $\cE$ assigns to each edge of $F$ a walk in $B$; each potential
  walk has a unique form associated to it, and a forms are organized
  by their type} }

\newglossaryentry{reverse walk}{
  sort={reverse walk},
  name={reverse walk},
  description={the reverse walk of a walk, $w$, in a graph is the walk
  where the order of the vertices and edges are reversed, and each edge
  is replaced by its opposite},
  parent={walk} }

\newglossaryentry{shift operator in k}{
  sort={shift operator in k},
  name={shift operator in $k$},
  description={the operator taking $f(n,k)$ and returning $f(n,k+1)$}, }

\newglossaryentry{potential graph specialization}{  
  sort={potential graph specialization},  
  name={potential graph specialization},  
  description={an event in $\cC_n(B)$ which would give rise to the 
  inclusion of a $B$-graph to a graph $G\in\cC_n(B)$},
  see=[such an event is analogous to a]{potential walk} }

\newglossaryentry{omega-type}{
  sort={omega-type},
  name={$\Omega$-type},
  description={a structure that combines a potential graph specialization
  of a graph, $\Omega$, into a $G\in\cC_n(B)$, along with a potential
  walk, where we keep all of $\Omega$ in the $\Omega$-type but, as usual,
  discard those vertices of the potential walk of degree two that are
  interior vertices of the walk (not in $\Omega$)} }

\newglossaryentry{omega-form}{
  sort={omega-form},
  name={$\Omega$-form},
  description={the analogue of a form for an $\Omega$-type; i.e., all the
  data in the graph determined by a potential walk plus a graph inclusion} }

\newglossaryentry{abstract partial trace}{
  sort={abstract partial trace},
  name={abstract partial trace},
  description={a general setup to which we can apply side-stepping methods,
  which includes applications to $1/n$-asymptotic expansions arising in
  the certified trace} }

\newglossaryentry{Euler characteristic}{
  sort={Euler characteristic},
  name={Euler characteristic},
  description={$\chi(G)=|V_G|-|\Edir_G|/2$} }

\makeglossaries

\usepackage[pdftex]{hyperref} 

\begin{document}
\frontmatter

\title{The Relativized Second Eigenvalue Conjecture of Alon}


\author{Joel Friedman}
\address{Department of Computer Science, 
        University of British Columbia, Vancouver, BC\ \ V6T 1Z4, CANADA,
        and Department of Mathematics, University of British Columbia,
        Vancouver, BC\ \ V6T 1Z2, CANADA. }
\curraddr{}
\email{{\tt jf@cs.ubc.ca} or {\tt jf@math.ubc.ca}}
\thanks{Research supported in part by an NSERC grant.}
\author{David-Emmanuel Kohler}
\address{Department of Mathematics, University of British Columbia,
        Vancouver, BC\ \ V6T 1Z2, CANADA. }
\curraddr{}
\email{{\tt david.emmanuel.kohler@gmail.com}}
\thanks{Research supported in part by an NSERC grant.}

\date{March 14, 2014}

\subjclass[2010]{Primary: 68R10, 05C50; Secondary: 05C80, 15B52}

\keywords{random, cover, eigenvalues, graph, relativized Alon conjecture,
lifts, new eigenvalues}

\begin{abstract}
We prove a relativization of the Alon Second Eigenvalue Conjecture
for all $d$-regular base graphs, $B$, with $d\ge 3$:
for
any $\epsilon>0$, we show that a
random covering map of degree $n$ to $B$ has a new eigenvalue
greater than $2\sqrt{d-1}+\epsilon$ in absolute value
with probability $O(1/n)$.
Furthermore, if $B$ is a Ramanujan graph, we show that this probability
is proportional to $n^{-\etafund(B)}$, where $\etafund(B)$ is an
integer depending on $B$, which can be computed by a finite algorithm
for any fixed $B$.
For any $d$-regular graph, $B$, $\etafund(B)$ is greater than
$\sqrt{d-1}$.

Our proof introduces a number of ideas that simplify
and strengthen
the methods of Friedman's proof
of the original conjecture of Alon.
The most significant new idea is that of a ``certified trace,'' which
is not only greatly simplifies our trace methods, but is the reason
we can obtain the $n^{-\etafund(B)}$ estimate above.  This estimate
represents an improvement over Friedman's results of the original
Alon conjecture for random $d$-regular
graphs, for certain values of $d$.
\end{abstract}

\maketitle
\setcounter{tocdepth}{2}


\setcounter{page}{4}

\tableofcontents


%
\def\listofsymbols{
\begin{tabbing}
%
\( \parbox{1.1in}{\(A_G\)\hfil} \)~~~~~~~~~~~~~~~~~~~~~~~~~~~~~\=\parbox{3.6in}{the adjacency matrix of a graph, $G$\dotfill \pageref{symbol:AG}}\\
\addsymbol \lambda_i(G) : {$i$-th eigenvalue of \(A_G\), in decreasing order}{symbol:lambda}
\addsymbol \Specnew_B(A_G) : {the new adjacency spectrum of $G$ with respect to $B$}{symbol:SpecnewAG}
\addsymbol \rhonew_B(A_G) : {new adjacency spectral radius}{symbol:rhonewAG}
\addsymbol \Tr(\,\cdot\,) : {the trace of a matrix}{symbol:Tr}
\addsymbol \cC_n(B) : {the Broder-Shamir model}{symbol:CnB}
\addsymbol \eta_{\rm fund}(B) : {the fundamental power of B}{symbol:etafund}
\addsymbol \chi(G) : {Euler characteristic of \(G\)}{symbol:EulerChi}
\addsymbol \Spec(\;\cdot\;) : {the spectrum of a Hilbert space operator}{symbol:Spec}
\addsymbol \rho(\;\cdot\;) : {the spectral radius of a Hilbert space operator}{symbol:rho}
\addsymbol \widehat B : {the universal cover of a graph, \(B\)}{symbol:widehatB}
\addsymbol W_{d/2} : {the bouquet of $d/2$ whole-loops}{symbol:Wd2}
\addsymbol H_d : {the bouquet of $d$ half-loops}{symbol:Hd}
\addsymbol H_G : {Hashimoto matrix of the graph \(G\)}{symbol:HG}
\addsymbol \mu_1(G) : {the Perron-Frobenius eigenvalue of $H_G$}{symbol:mu1}
\addsymbol \mu_i(G),\ i>1 : {the other eigenvalues of $H_G$, in no particular order}{symbol:mui}
\addsymbol \dmax : {maximum degree of a graph }{symbol:dmax}
\addsymbol V_G : {the vertex set of a graph or digraph $G$}{symbol:VG}
\addsymbol \Edir_G : {the edge set of a directed graph, or the directed edge set of a graph, $G$}{symbol:EdirG}
\addsymbol E_G : {the edge set of a graph, $G$}{symbol:EG}
\addsymbol h_G : {the head map $\Edir_G\to V_G$ of a (di)graph, $G$}{symbol:hG}
\addsymbol t_G : {the tail map $\Edir_G\to V_G$ of a (di)graph, $G$}{symbol:tG}
\addsymbol \iota_G : {the edge involution $\Edir_G\to \Edir_G$ of a graph, $G$}{symbol:iotaG}
\addsymbol D_G : {the degree matrix of a graph, $G$}{symbol:DG}
\addsymbol \Line(G) : {directed line graph of \(G\)}{symbol:LineG}
\addsymbol \Specnew_B(H_G) : {new Hashimoto spectrum }{symbol:SpecnewHG}
\addsymbol \rhonew_B(H_G) : {new Hashimoto spectral radius}{symbol:rhonewHG}
\addsymbol \VLG(G,\vec k) : {realization of a variable-length graph}{symbol:vlg}
\addsymbol \ord{\psi} : {the order of a graph, $-\chi(\psi)$}{symbol:order_psi}
\addsymbol \rho^{1/2}(H_B) : {the square root of $\rho(H_B)$}{symbol:rho_root}
\addsymbol \Occurs_B : {connected graphs that can occur in \( \cC_n(B) \)}{symbol:OccursB}
\addsymbol \Tangle_B : {the set of \glspl{B-tangle}}{symbol:Tangle-B}
\addsymbol \Tangle_{<r,B} : {the set of $B$-tangles of order $<r$}{symbol:Tangle-r-B}
\addsymbol \Tangle_{B,\epsilon} : {the set of $(B,\epsilon)$-tangles}{symbol:Tangle-B-eps}
\addsymbol \Tangle_{<r,B,\epsilon} : {the set of $(B,\epsilon)$-tangles of order less than $r$}{symbol:Tangle-r-B-eps}
\addsymbol \Graph(w) : {the graph of a walk, $w$}{symbol:graph_of_walk}
\addsymbol \rho^k(H_B) : {shorthand for $(\rho(H_B))^k$}{symbol:rho_k_H_B}
\addsymbol \CertTr_{<r}(G,k) : {$r$-th certified trace of length $k$}{symbol:certified_r_k}
\addsymbol \TypeGraph(w) : {the underlying graph of the \gls{type} associated to $w$}{symbol:typegraph}
\addsymbol \HasTangle_{r,B} : {the set of graphs with a $B$-tangle of order less than $r$}{symbol:first_has_tangle}
\addsymbol \II_{\HasTangle_{r,B}} : {the indicator function of $\HasTangle_{r,B}$}{symbol:first_has_tangle_indicator}
\addsymbol {\rm TF}_{r,B} : {the set of graph free of $B$-tangles of order less than $r$}{symbol:first_tangle_free}
\addsymbol \II_{{\rm TF}_{r,B}} : {the indicator function of $\TF_{r,B}$}{symbol:first_tangle_free_indicator}
\addsymbol g_1 \ast g_2 : {the convolution of $g_1$ and $g_2$}{symbol:convolution}
\addsymbol \potwalk : {a potential walk}{symbol:potwalk}
\addsymbol \Graph\potwalk : {the graph traced out by a potential walk}{symbol:graphpotwalk}
\addsymbol \ord\potwalk : {the order of a potential walk}{symbol:orderpotwalk}
\addsymbol [\vec t\,]_{n,w} : {the \( (n,w) \)-equivalence symmetry class of \( \vec t \)}{symbol:vectnw}
\addsymbol \potwalkclass : {equivalence class of potential walks}{symbol:potwalkclass}
\addsymbol \cW : {a walk collection}{symbol:walkcollection}
\addsymbol \cW(k,n) : {\((k,n)\)-walks in \( \cW \)}{symbol:Wkn}
\addsymbol \cW_n(G,k) : {\((k,n)\)-modified trace of \(\cW\)}{symbol:WnGk}
\addsymbol \WalkSum(\cW,k,n) : {\((k,n)\) walk sum of \( \cW\)}{symbol:WalkSumWkn}
\addsymbol \SNBC : {strictly non-backtracking closed walks collection}{symbol:SNBC}
\addsymbol \cT : {a type}{symbol:type}
\addsymbol \cF : {a form}{symbol:form}
\addsymbol \cT\potwalk : {type of the potential walk \( \potwalk \)}{symbol:typepotwalk}
\addsymbol \cF\potwalk : {form of the potential walk \( \potwalk \)}{symbol:formpotwalk}
\addsymbol \Types[r] : {types of order less than \(r\)}{symbol:Typesr}
\addsymbol \mathrm{NB}(e_1,e_2,k) : {\# of non-backtracking \(k\)-walks from \(e_1\) to  \(e_2\) }{symbol:NBeek}
\addsymbol W_\cF(\cW,k,\vec m\,) : {consistent walks in \( \cF \) of multiplicity \( \vec m \)}{symbol:WFWkm}
\addsymbol W_\cT(\vec m, \vec k(\cF)) : {consistent \((\vec m,\vec k\,)\)-walks in \(\cT\)}{symbol:WTmkF}
\addsymbol \Tanglemin : {set of minimal tangles in \( \Tangle_{r,B} \)}{symbol:TangleminBr}
\addsymbol \CertSNBC : {certified SNBC walk collection}{symbol:CertSNBCr}
\addsymbol \CertTr : {certified trace}{symbol:CertTr}
\addsymbol \CertTr_{<r} : {truncated certified trace}{symbol:CertTr_r}
\end{tabbing}

\clearpage
}
\def\addsymbol #1: #2#3{$#1$ \> \parbox{3.6in}{#2 \dotfill \pageref{#3}}\\} 
\def\newnot#1{\label{#1}} 

\chapter*{List of Symbols\hfill} 
\listofsymbols

\newcommand{\jfignore}[1]{}   


\mainmatter
\setcounter{chapter}{-1}

\chapter{Introduction}
\label{ch:p0}



The main goal of this article
is to study spectral properties
of a random covering map of large degree of a fixed graph, $B$.
The main result is to prove
a relativization of
Alon Second Eigenvalue Conjecture,
formulated in \cite{friedman_relative},
for any {\em base graph}, $B$, that is regular.

Recall that
Alon's Second Eigenvalue Conjecture
says that for
fixed integer $d\ge 3$, and a real $\epsilon>0$,
a random $d$-regular graph on $n$ vertices
has second adjacency eigenvalue at most $2(d-1)^{1/2}+\epsilon$ with
{\em high probability},
i.e., probability than tends to one as $n$ tends to infinity.
The interest in this conjecture is that the conclusion implies that
most graphs have, in a sense, almost optimal spectral properties,
which in turn implies a number of ``expansion'' or ``well connectedness''
properties of the graph.
The conjecture was established with weaker bounds in
\cite{broder,friedman_kahn_szemeredi,friedman_random_graphs}, and finally
settled affirmatively in \cite{friedman_alon}.
All these papers bound not only the second eigenvalue with high
probability, but 
also---with the same bound---the absolute value
of all subdominant
adjacency eigenvalues, i.e., all eigenvalues---including the negative
ones---excepting the dominant (largest)
eigenvalue, namely $d$.

One generalization of the above spectral bounds is
the notion of a relative expander, discussed some two decades ago
in \cite{friedman_relative_boolean};
roughly speaking, for any covering map, we consider its
{\em new adjacency eigenvalues}, i.e., the eigenvalues of the
adjacency matrix of the source not arising from eigenfunctions
pulled back from the target.
A random covering map of degree $n$ to a $d$-regular graph with one vertex
is just a random $d$-regular graph with $n$ vertices;
the subdominant eigenvalues are precisely the new eigenvalues.


The relativized Alon conjecture, regarding random covering maps,
was formulated in \cite{friedman_relative}, and probably received more
attention\footnote{Indeed, the article \cite{friedman_relative_boolean} 
was circulated in a limited fashion, but was rejected for
publication.
By the time of the submission of \cite{friedman_relative}, 
Amit and Linial's 
work \cite{amit1}, as well as related works \cite{amit2,amit3,amit0},
had given relative properties of random covering maps more prominence.}
due to the independent work of Amit and Linial \cite{amit1},
and related work \cite{amit2,amit3,amit0}, which gave prominence to
studies of random covering maps and, implicitly, their relative properties.
Aside from \cite{friedman_relative}, the relativized Alon conjecture
stated there was proven
with weaker new spectral bounds in 
\cite{linial_puder,lubetzky,a-b,puder}.
Friedman introduced relative expansion \cite{friedman_relative_boolean}
without specific applications, rather in anticipated applications, as this
relative notion seemed an intrinsically interesting generalization of
expansion; concretely, one can make the simple observation 
that a covering map of a good expander
with good relative expansion implies that
covering graph is, itself, another good expander.  
However,
some two decades later,
relative expansion
has recently seen a dramatic application to graph theory, 
in the construction of
by Marcus, Spielman, and Srivastava
\cite{mss}, inspired by \cite{bilu} (building on \cite{frieze_malloy}), 
which proves the existence
of $d$-regular bipartite {\em Ramanujan} graphs for all fixed positive
integers $d$, for an infinite set of 
numbers of vertices; here the term ``Ramanujan'' means that all
eigenvalues, excepting $d$ and, for bipartite graphs, $-d$, 
are of absolute value at most
$2(d-1)^{1/2}$.
The existence proof in \cite{mss} uses degree two covers of bipartite
graphs requiring the base graph to be Ramanujan. 

We remark that 
Bilu and Linial \cite{bilu} have pointed out that for certain
regular base graphs
(which ``almost'' decompose into many
disconnected, small graphs), most degree two covers will be very poor
relative expanders.  Our main theorem, by contrast, shows that for any
fixed regular base graph, covers of large degree are, with high probability,
nearly relatively Ramanujan.

Our approach to the relativized Alon conjecture
follows the Broder-Shamir trace method of \cite{broder},
with its refinements of \cite{friedman_random_graphs,friedman_alon},
which we adapt to the more general situation
of random, degree $n$ covering maps of a 
fixed graph, $B$.
However, our proofs significantly simplify the arguments
of \cite{friedman_alon}; in particular,
we replace the cumbersome {\em selective trace} of \cite{friedman_alon}
by the much simpler 
{\em certified trace} of this article; there are
several other new ideas that simplify and strengthen
the methods of \cite{friedman_alon}.

We mention that it is not surprising that one can generalize the
results of \cite{friedman_alon} to obtain the
relativized Alon conjecture in some very special cases.  For example,
if the base graph has two vertices and is $d$-regular, 
then the eigenvalues of its
Hashimoto matrix lie in the set $\{d-1,1-d,1,-1\}$, and it seems
likely that the methods of
\cite{friedman_alon} will generalize in a fairly straightforward way
(however, this article gives the first proof of this result).
But it is not at all clear that the methods of
\cite{friedman_alon} carry over to general $d$-regular base graphs, 
especially when $B$ is not, say, Ramanujan.  Indeed, some of the
previous work, such as \cite{lubetzky}, give estimates that degrade
when $B$ has eigenvalues close to---but less than---$d$.
And, indeed, this article develops some new techniques that are needed
{\em precisely when $B$ is not Ramanujan}.
Furthermore, in this article we obtain much sharper probability
estimates when (and---at present---only when) $B$ is Ramanujan.

This article is organized as follows.
This chapter will state our main results and the historical context
of this conjecture; we will draw upon terminology only to be made
precise in Chapter~\ref{ch:p0.5}.
In Chapter~\ref{ch:p0.5} we shall state some precise terminology,
and give an overview of the trace methods and our proof of the 
relativized Alon conjecture for regular graphs.
Chapter~\ref{ch:p1} of this article will prove the relativized Alon conjecture
in the case where $B$ is regular and without {\em half-loops} (half-loops
in the sense of,
say, \cite{friedman_geometric_aspects}).
Chapter~\ref{ch:p2} will prove the relativized Alon conjecture
for all regular base graphs, in a number of different models;
it will make some further remarks, including refinements of the
probability estimates when the
base graph is Ramanujan, and aspects of the relativized Alon
conjecture for non-regular graphs.

At present we have only weaker versions of the relativized
Alon Conjecture when the base graph is not regular, and our results 
are only interesting for certain non-regular base graphs.

Let us review the relativized Alon conjecture and its historical context.
For brevity,
at time we shall assume some terminology of algebraic graph theory to
be given in Chapter~\ref{ch:p0.5}.


\section{Our Main Results}
\label{se:p0-main}

In this section we will state the main result of this paper, using
some fairly common---although not entirely standard---terminology 
to be made precise later.
For example, in this paper a {\em \gls{graph}} may have multiple edges and two
types of \glspl{self-loop}: \glspl{whole-loop} and \glspl{half-loop}, 
in the sense of
\cite{friedman_geometric_aspects}; however, in much of the graph theory
literature, it is more standard to insist that
a graph have no multiple edges and/or that all self-loops be whole-loops.

The eigenvalues of the adjacency matrix,
\newnot{symbol:AG} $A_G$, 
of a finite undirected graph, $G$, are
real; we order these eigenvalues and denote them as
\newnot{symbol:lambda}
\begin{equation}\label{eq:lambda_i_notation}
\lambda_1(G)\ge\lambda_2(G)\ge\cdots\ge\lambda_n(G),
\end{equation}
where $n$ is the number of vertices in $G$.  If $G$ is $d$-regular, i.e.,
each vertex is of degree $d$,
then $\lambda_1=d$.

For any \gls{covering map}, sometimes called 
a {\em lift of graphs}\footnote{After \cite{amit1}, Amit, Linial, and their
collaborators began using the term {\em lift} to avoid confusion,
since the word {\em covering} has strong connotations in combinatorial
optimization\cite{amit0,amit2,amit3,rozenman}.  
},
$\pi\from G\to B$,
the multiset of eigenvalues of $A_G$ can be partitioned into
two sets: {\em old eigenvalues}\footnote{The old/new eigenvalues is terminology
borrowed
by Friedman \cite{friedman_relative_boolean,friedman_relative}
from the theory of automorphic forms.}, arising via pullback
from $B$ adjacency eigenvectors, and the remaining
{\em new eigenvalues},
of eigenfunctions whose values sum to zero on any vertex fibre of
$\pi$.
We let \newnot{symbol:SpecnewAG}
$\Specnew_B(A_G)$ be the multiset of new eigenvalues of $G$,
and we let \newnot{symbol:rhonewAG} 
$\rhonew_B(A_G)$ be the largest absolute value
occurring in the finite set of real numbers $\specnew_B(A_G)$.
We remark that for any positive integer, $k$, we have
\begin{equation}\label{eq:independent_of_covering_map}
\sum_{\lambda\in\specnew_B(A_G)} \lambda^k 
=
\Tr(A_G^k) - \Tr(A_B^k),
\end{equation}
where
\newnot{symbol:Tr}
$\Tr(\,\cdot\,)$ denotes the trace.
It follows that the $\specnew_B(A_G)$ depends only on $B$ and $G$, not
on the particular covering map $\pi\from G\to B$.

For every graph, $B$, and positive integer, $n$, we will define a
probability space, the {\em \gls{Broder-Shamir model}},
\newnot{symbol:CnB}
$\cC_n(B)$, whose atoms are
random covering maps, $\pi\from G\to B$, of degree
$n$ of $B$; this model closely resembles the model of 
random regular graphs used by
Broder and Shamir in \cite{broder}.  
If $B$ has no half-loops, and $E_B$ denotes the set of
edges of $B$, then $\cC_n(B)$ is formed via
$|E_B|$ independently and uniformly chosen permutations of
$\{1,2,\ldots,n\}$ in the natural way.
In fact, our main theorem will apply to a number of variants of
$\cC_n(B)$.

Our first main result concerns arbitrary regular base graphs.

\begin{theorem}\label{th:main_Alon}
Let $B$ be a $d$-regular graph with $d\ge 3$.
Then for any real 
$\epsilon>0$, there is a constant, $C=C_{B,\epsilon}>0$
for which
\begin{equation}\label{eq:main_Alon_bound}
\prob{G\in \cC_n(B)}{\rhonew_B(A_G) \ge 2\sqrt{d-1}
+\epsilon}
\le C_{B,\epsilon} n^{-1} .
\end{equation}
\end{theorem}

Our second main result concerns regular {\em \glspl{Ramanujan graph}}.

\begin{definition} We say that a $d$-regular graph, $B$, is a
{\em \gls{Ramanujan graph}} if any eigenvalue of $A_B$ other than $d$ and
the possible eigenvalue $-d$ has absolute value at most $2\sqrt{d-1}$.
\end{definition}

\begin{theorem}\label{th:main_Alon_Ramanujan}
Let $B$ be a $d$-regular Ramanujan graph with $d\ge 3$.
There is a finite algorithm to determine an integer 
\newnot{symbol:etafund}
\gls{etafund},
and a real number $\epsilon_0=\epsilon_0(B)>0$ for which
\begin{enumerate}
\item there is 
a constant, $C_1=C_1(B)$, for which
\begin{equation}\label{eq:eta_fund_lower}
\prob{G\in \cC_n(B)}{\rhonew_B(A_G) \ge 2\sqrt{d-1}
+\epsilon_0}
\ge C_1(B)n^{-\etafund(B)}, 
\end{equation}
and
\item for any $\epsilon$ with $0<\epsilon<\epsilon_0(B)$ there is a
constant $C_2=C_2(B,\epsilon)$ for which
\begin{equation}\label{eq:eta_fund_upper}
\prob{G\in \cC_n(B)}{\rhonew_B(A_G) \ge 2\sqrt{d-1}
+\epsilon}
\le C_2(B,\epsilon)n^{-\etafund(B)}.
\end{equation}
\end{enumerate}
\end{theorem}
As we shall explain below, that this theorem gives
an improvement to the results of \cite{friedman_alon} for certain
values of $d$.

The above two theorems, like the results of \cite{lubetzky}, 
gives sharper theorems 
when the base graph is Ramanujan.
We shall state and prove our sharper results, for Ramanujan base graphs,
in Section~\ref{se:p2-fund-exp}.
In this same section we will see, essentially from
Lemma~6.7 of \cite{friedman_alon}, that for any $d$-regular graph, $B$,
we have
$$
\etafund(B) > \sqrt{d-1}.
$$
It follows that Theorem~\ref{th:main_Alon_Ramanujan} gets tighter
bounds on the probability in the Relativized Alon Conjecture
as $d$ increases.
\section{Historical Context and Motivation}
\label{se:p0-hist}

In this section we elaborate on
the historical context and motivation
of our main theorems given earlier, at the beginning of this chapter.

\subsection{The Alon Second Eigenvalues Conjecture}

The theoretical computer science and network theory literature
in the 1970's and 1980's gave rise to expanders and many related
graphs, such as {\em concentrators} and 
{\em superconcentrators}; 
see \cite{hoory_linial_wigderson,krebs_shaheen} and the references given 
there for this (rather long) story.

It is known that $G$ has numerous desirable properties, often called
{\it expansion properties}---such as large ``isoperimetric'' constants,
desirable in communication networks---provided that its subdominant
adjacency eigenvalues are small, meaning 
that $\lambda_2(G)$, and sometimes $\lambda_n(G)$, are
sufficiently close to zero;
see, for example, \cite{hoory_linial_wigderson,krebs_shaheen}.
Historically, this spectral approach to desirable graph properties
via adjacency eigenvalues appears explicitly in
the work of Alon and Milman \cite{alon_milman_FOCS,alon_milman}
and Tanner \cite{tanner}; however, these ideas appear earlier, 
in Gabber and Galil \cite{gabber_galil}, and perhaps
other articles, implicitly in the
proofs of its main theorems.
This lead Noga Alon 
\cite{alon_eigenvalues} to study spectral properties of regular graphs;
there he formulated what we call the {\em Alon Second Eigenvalue Conjecture},
the conjecture that for any $d\ge 3$
and
$\epsilon>0$, as $n\to\infty$ we have that a
random $d$-regular graph on $n$ vertices, $G$, has
\begin{equation}\label{eq:Alon_original}
\lambda_2(G)\le 2\sqrt{d-1}+\epsilon
\end{equation}
with probability tending to one as $n$ tends to infinity.
Alon and Boppana showed that the constant
$2\sqrt{d-1}$ cannot be improved upon (see \cite{alon_eigenvalues,
nilli}, with improvements of Friedman and Kahale 
\cite{friedman_geometric_aspects}).
A number of papers demonstrated a variant of the above conjecture,
with $2\sqrt{d-1}$ replaced 
with a larger function of $d$,
\cite{broder,friedman_kahn_szemeredi,friedman_random_graphs};
the conjecture was finally settled in
\cite{friedman_alon}.  

Before reviewing the main result of \cite{friedman_alon}, we remark
that Alon's conjecture, at least in principle, depends on what model
of random $d$-regular graph on $n$ vertices one takes.
The paper of \cite{broder} insists that $d$ is even, and forms a
random $d$-regular graph on $n$ vertices by independently
choosing $d/2$ permutations
on $\{1,\ldots,n\}$ uniformly; each permutation gives rise to a 
$2$-regular graph 
on $n$ vertices in the natural way, and the $d/2$ permutations
hence give rise to a $d$-regular graph
(therefore possibly with multiple edges and/or
self-loops).
This model is called $\cG_{n,d}$ in \cite{friedman_alon}, and is
used in \cite{friedman_kahn_szemeredi,friedman_random_graphs}
as well as \cite{broder,friedman_alon}.

Broder and Shamir \cite{broder} gave estimates on the expected value
of the traces of adjacency matrix powers; 
their estimates easily imply\footnote{This was not noticed by
Broder and Shamir in \cite{broder}, although this was explained
in \cite{friedman_kahn_szemeredi,friedman_random_graphs}.}
a weaker version of the Alon conjecture,
where the $2(d-1)^{1/2}$ in \eqref{eq:Alon_original} is replaced 
by\footnote{The constant in \cite{broder} is $2d^{3/4}$, although
Friedman explained in
\cite{friedman_kahn_szemeredi,friedman_random_graphs},
a simple modification of their methods to
obtain the slightly stronger constant of \eqref{eq:broder_shamir_constant}}
\begin{equation}\label{eq:broder_shamir_constant}
d^{1/2} 2^{1/2}(d-1)^{1/4} .
\end{equation}
Kahn and Szemeredi 
\cite{friedman_kahn_szemeredi}
gave a modified counting argument to improve this
expression to
$$
C d^{1/2},
$$
with a constant $C$ that was not estimated, while independently
Friedman 
\cite{friedman_kahn_szemeredi,friedman_random_graphs} improved the
trace methods of Broder-Shamir to obtain
$$
d^{1/(r+1)} \Bigl( 2(d-1)^{1/2} \Bigr)^{r/(r+1)}
$$
for $r$ of size roughly $d^{1/2}/2$ (see \eqref{eq:r_bound} for the
precise value), which
for large $d$ is roughly
$$
2(d-1)^{1/2} + 2\log_e(d) + O(1).
$$
The Alon conjecture was finally settled in \cite{friedman_alon}.

In \cite{friedman_alon}, other models of random graphs are described,
including models of $d$-regular graphs on $n$ vertices where $d$ and
$n$ can be of arbitrary parity; however, all these models have a 
certain ``algebraic'' property (see 
Subsection~\ref{sb:intro_to_broder_shamir}
and Section~\ref{se:p2-algebraic} of
this article or \cite{friedman_alon}).  There are many models of random
regular graphs to which the results in \cite{friedman_alon} do not
directly apply; however, by the time of \cite{friedman_alon}---but
not at the time of \cite{broder}---there were enough
{\em contiguity} theorems, which imply that for all the usual models
of regular graphs, Alon's conjecture in any model is equivalent
to Alon's conjecture in the other models; for a discussion
of contiguity theorems see Section~\ref{se:p2-algebraic} and
\cite{friedman_alon}
and the references there.
Unfortunately, contiguity results are not currently available for
random covering maps of a fixed base graph, and hence our main
results are only known to be valid in the ``algebraic'' type of models
that we describe in this article.

The main theorem of \cite{friedman_alon} is the theorem below, although
it is valid only for ``algebraic'' models, such as
$\cG_{n,d}$ described above, or such as the models
$\cH_{n,d},\cI_{n,d},\cJ_{n,d}$ of \cite{friedman_alon};
furthermore, the $\eta=\eta(d)$ in the theorem below depends on the model
(it is roughly twice as large for $\cH_{n,d}$, 
which is the model where we insist that each
of the
$d/2$ permutations forming the random graph are, in their cyclic
representations, each a single cycle of length $n$).

\begin{theorem}[Friedman, \cite{friedman_alon}]
\label{th:Friedman_Alon}
Let $d\ge 3$ be a fixed integer.  Then there
exists a positive integer, $\eta=\eta(d)$, such that
for any sufficiently small positive real number $\epsilon$ there are
positive $C_1,C_2$ for which the following holds:
the probability that a random $d$-regular graph, $G$,
on $n$ vertices has $|\lambda_i(G)|\ge 2(d-1)^{1/2}+\epsilon$ for
at least one $i\le 2$ is between $C_1 n^{-\eta-1}$ and
$C_2 n^{-\eta}$.
\end{theorem}

We note that for ``most'' values of $d$, namely $d$ for which
$d-1$ is not a odd perfect square, the lower bound of $C_1 n^{-\eta-1}$
was improved to $C_1 n^{-\eta}$ in \cite{friedman_alon}.
For the ``exceptional'' values of $d$, where $d-1$ is not an odd perfect
square,
our Theorem~\ref{th:main_Alon_Ramanujan} represents an improvement
of the lower bound of \cite{friedman_alon} by a factor proportional to
$n$, giving upper and lower bounds that match to within a constant
factor (although this constant factor may depend, at least in
our theorems, on $\epsilon>0$).
We also note that the contiguity theorems mentioned earlier speak
only of probabilities that tend to zero as $n$ tends to infinity,
and do not address, at least in their literal definition,
the exponent, $\eta=\eta(d)$, of $n$ in the probability
estimates of Theorems~\ref{th:Friedman_Alon} and
Theorem~\ref{th:main_Alon_Ramanujan} above.
For example, for fixed even $d$ and varying even $n$,
it is known that the spaces $\cG_{n,d}$
and $\cI_{n,d}$ of \cite{friedman_alon} are contiguous, but the
$\eta(d)$ is larger for $\cI_{n,d}$, roughly for the reason that
$\cI_{n,d}$ does not allow self-loops.

One reason that expanders continue to receive attention is that they
are simple to define and related to a number of other fields, such
as random matrices 
(see, for example, 
\cite{tao_vu_survey} for a survey) 
and the type of deviations from the principal term
that one studies in the theory of automorphic forms and number theory
in general.
For example, 
the theory
of Ihara Zeta functions
gives connections
between $p$-adic groups and graphs (see, for example, \cite{st1,st2,st3}).
Furthermore, the field of expanders has seen steady progress
over the years, with many interesting questions still unsolved.

Two different methods have been used to study Alon's above conjecture:
\begin{enumerate}
\item trace methods---akin to those pioneered by Wigner \cite{wigner} to study
random matrices (see \cite{tao_vu_survey} for a survey of 
this large field)---but 
that require a significant adaptation
to give interesting results for random $d$-regular graphs, as first
done by Broder and Shamir \cite{broder}, and 
\item counting methods, as used by 
Kahn-Szemeredi \cite{friedman_kahn_szemeredi}, similarly requiring
a significant adaptation from the standard counting methods to
give interesting results for $d$-regular graphs.
\end{enumerate}

We mention that
there are other generalizations of the Alon second eigenvalue
conjecture to situations
with random graphs other than those we study here.  
Other interesting classes of random graphs include
those of \cite{fjrst98}
(see also the non-random construction of Kassabov \cite{kassabov}).
Another related question on spectral properties
of random graphs arises when the adjacency matrix is replaced with an
arbitrary element of the group algebra of the fundamental group of the
graph, mentioned to us by Lewis Bowen.

\subsection{The Relativized Alon Conjecture}

In the early 1990's, Friedman began to investigate a number of
ideas and general principles of Grothendieck as applied to graph theory,
specifically the topics of expander graphs and Boolean functions which
arise in theoretical computer science.
(see \cite{friedman_geometric_aspects,friedman_relative_boolean},
and later \cite{friedman_relative,friedman_cohomology,
friedman_cohomology2,friedman_linear,friedman_sheaves,friedman_sheaves_hnc,
friedman_memoirs_hnc}).
One compelling idea emerging from Grothendieck's work is the
importance of 
{\em relativization}, which gives rise to the 
relativized Alon conjecture of \cite{friedman_relative}
that is the focus of this paper.

Roughly speaking, to {\em relativize} a theorem means to take a
theorem about objects in a category
and prove an analogous result about morphisms;
usually, one also requires the theorem about morphisms to
implies
the theorem about objects.  Consider,
for example, the following two toy theorems. 

\begin{theorem} A graph, $G$, that appears in the Broder-Shamir
model of a random $d$-regular graph on $n$ vertices
has Euler characteristic
\newnot{symbol:EulerChi}
$\chi(G)=n(2-d)/2$.
\end{theorem}

\begin{theorem} If $\pi\from G\to B$ is a $n$-to-$1$ covering map, 
then $\chi(G)=n\,\chi(B)$.
\end{theorem}

If $B$ is $d$-regular and has one vertex, then its Euler
characteristic is $(2-d)/2$.
Hence, the first theorem regarding $G$, a theorem about objects,
is implied by the second
theorem about morphisms.
Furthermore, consider the category, ${\rm Cov}(B)$,
whose objects are graphs with a covering map to
$B$, whose morphisms are graph morphisms preserving the $B$ maps.
Then $B$ is a terminal object of ${\rm Cov}(B)$.
So, more precisely, the second theorem reduces to the first when
$B$ is the terminal object of ${\rm Cov}(B)$.

Friedman's motivation to study relativization
(discussed in detail in
\cite{friedman_relative_boolean}, some aspects of which, such as
torsors, appear in \cite{friedman_geometric_aspects})
was based on the anticipation that applying
Grothendieck's ideas to graph theory may
yield new interesting research and applications; furthermore,
clearly a covering map to a base graph that is a good expander 
yields a new good expander in the covering graph
iff the covering map has good relative expansion.
Also, Noga Alon suggested 
(see \cite{friedman_relative_boolean}, end
of Subsection~1)
that relativization could be used to construct new expanders by
forming a quotient of a relative expander; we remark that
Alon's idea represents 
a natural
idea to improve
the recent work of \cite{mss}, from their bipartite Ramanujan graphs
to possibly obtaining new, non-bipartite Ramanujan graphs. 
Despite some limited
circulation of \cite{friedman_relative_boolean}, 
the article was rejected for publication\footnote{It was
rejected from the {\sl Journal of Algebraic Combinatorics},
submitted May 14, 1993, and no further formal publication
was pursued.  The paper was re-LaTeXed and posted on the author's
website (along with numerous other papers) around 1995, but does not appear
to have been modified since the May 1993 submission.}, in part because
the article contained no new interesting families of new
graphs at the time, twenty years before the remarkable recent paper
of
Marcus, Speilman, and Srivastava \cite{mss}.
This line of research become more active
some ten years later, with the independent effort begun by
Amit and Linial to study random graph covers
\cite{amit1}, and later \cite{amit2,amit3,amit0,rozenman} (where the
term ``lift'' replaces ``covering map''),
and the article \cite{friedman_relative},
where the relativized Alon conjecture is stated, and the Broder-Shamir
technique is adapted to
prove results analogous to those of \cite{broder} for random graph covers.
The long standing problem
of showing the existence of families of $d$-regular Ramanujan graphs
(i.e., sequences of $d$-regular Ramanujan graphs on an arbitrarily
large number of vertices)
for all integral values of $d$ was recently solved by
Marcus, Speilman, and Srivastava \cite{mss}, via towers of relatively
Ramanujan degree two covers of any bipartite Ramanujan graph,
inspired by the work of Bilu and Linial \cite{bilu}, in turn inspired
by \cite{frieze_malloy}.

The following conjecture is stated in \cite{friedman_relative},
which we give after the following definition.

\begin{definition} If $A$ is an operator on a Hilbert space, we let
\newnot{symbol:Spec}
$\Spec(A)$ be the spectrum of $A$, and
\newnot{symbol:rho}
$\rho(A)$ be the spectral radius of $A$.
If $A$ is a matrix with real entries, indexed on a set, $S$, where $S$ is
either finite or infinite, we view $A$ as acting on the
real Hilbert space $L^2(S)$ with its standard inner product,
$$
(f,g) = \sum_s f(s)g(s).
$$
Hence $\rho(A)$ denotes the spectral radius of $A$ acting on $L^2(S)$ as
above.
In the above, we may allow $A$ to have complex values by working over
the complex Hilbert space $L^2(S)$, where the complex inner product
sums over $\overline{f(s)}g(s)$ instead of $f(s)g(s)$.
\end{definition}

\begin{conjecture}[The relativized Alon conjecture]
\label{co:gen_Alon}
For any fixed graph, $B$, let $\rho(A_{\widehat B})$ be the spectral 
radius of the
adjacency operator on the 
universal cover,
\newnot{symbol:widehatB}
$\widehat B$, of $B$.  
Then for any $\epsilon>0$,
$$
\prob{G\in \cC_n(B)}{\rhonew_B(A_G) \ge \rho(A_{\widehat B})+\epsilon}
$$
tends to zero as $n\to\infty$.
\end{conjecture}
This conjecture reduces to the original Alon Second Eigenvalue Conjecture
in the case where $d$ is even and 
\newnot{symbol:Wd2}
$B=W_{d/2}$ is the 
\gls{bouquet whole},
or where $d$ is arbitrary and 
\newnot{symbol:Hd}
$B=H_d$ is the 
\gls{bouquet half}.
Since the identity map of $B$ represents a terminal
element in the category of graphs with a covering map to $B$,
the above conjecture is truly a
{\em relativization} of Alon's conjecture for the case studied in
by Broder and Shamir, \cite{broder}, which amounts to graphs with 
a covering map to $W_{d/2}$, or, equivalently, graphs formed by
$d/2$ permutations (where we order the permutations or, equivalently,
label each of the $d/2$ permutations with different label chosen from
$\{1,2,\ldots,d/2\}$).

Versions of Theorem~\ref{th:main_Alon} were proven with 
$\rho(A_{\widehat B})$ replaced with larger constants.
Specifically, \cite{friedman_relative} gave a short adaptation of
the methods of Broder and Shamir, proving
Theorem~\ref{th:main_Alon},
for arbitrary $B$, with $\rho(A_{\widehat B})$ replaced with
\begin{equation}\label{eq:geometric_mean}
\lambda_1(B)^{1/2}\rho(A_{\widehat B})^{1/2} ,
\end{equation}
where,
as in \eqref{eq:lambda_i_notation},
$\lambda_1(B)$ denotes the largest eigenvalue of the adjacency matrix
of $B$.
Linial and Puder \cite{linial_puder}
improved this, again for arbitrary $B$, to
$$
3 \lambda_1(B)^{1/3}\rho(A_{\widehat B})^{2/3} ,
$$
again using the Broder-Shamir technique but calculating one extra
term of the associated power series (as in \cite{friedman_random_graphs,
friedman_alon}).
By adapting the Kahn-Szemeredi counting technique to the relative case,
for $B$ $d$-regular,
Lubetzky, Sudakov, and Vu, \cite{lubetzky}, obtained the bound
$$
C (\log d) \bigl( \max(d^{1/2},\lambda_2(B),|\lambda_n(B)|)\bigr),
$$
with a constant $C$ that was not estimated.
Addario-Berry and Griffiths \cite{a-b} improved this to
$$
C d^{1/2},
$$
again, for $d$-regular $B$,
and estimated their $C$ to be at most $430656$.
Recently Puder \cite{puder}, building on
\cite{puder_prev,puder_p},
used trace methods to get the impressive bound
$$
2\sqrt{d-1} + 0.835
$$
in the $d$-regular case (see the paragraph after equation~(6.1) in
\cite{puder}), and the bound
$$
\sqrt{3}\; \rho(A_{\widehat B})
$$
for arbitrary $B$. 
Note that results above are successive improvements, at least for $d$-regular
$B$, with $d$
sufficiently large.

It is also interesting to note that the trace method bounds of
\cite{broder,friedman_random_graphs,friedman_relative,linial_puder},
and probably \cite{puder} as well, can be slightly improved by using
the expected {\em \gls{Hashimoto matrix}} traces.  We shall explain this
in Subsection~\ref{sb:hashimoto_improvement}.


As mentioned before, in Theorem~\ref{th:main_Alon} we establish
the relativized Alon conjecture for any 
$d$-regular base graph, $B$, in which
case it is well-known that 
$$
\rho(A_{\widehat B}) = 2(d-1)^{1/2} .
$$

\subsection{The Hashimoto Matrix and Non-Regular Base Graphs}

At present we are unable to prove Conjecture~\ref{co:gen_Alon}
for arbitrary $B$; we can, for certain non-regular $B$, prove
a weakened form
of Conjecture~\ref{co:gen_Alon}, 
with $\rho(A_{\widehat B})$ replaced by a
larger value.  We shall explain what we can prove; in brief, the problem
is that our
techniques more directly address the expected trace of powers of 
the Hashimoto matrix rather than the adjacency matrix.

We shall prove Theorem~\ref{th:main_Alon} by using 
the fundamental ideas of Broder-Shamir (\cite{broder}), later refined by
Friedman (\cite{friedman_random_graphs,friedman_alon}).
These methods enable one to estimate the expected number of closed,
non-backtracking
walks in a random matrix, $G$, of a given length,
possibly subject to certain additional constraints (such as being
strictly non-backtracking).
This leads us to theorems regarding the eigenvalues of
the {\em Hashimoto
matrix}, 
\newnot{symbol:HG}
$H_G$, of a random graph, $G$.  
For $d$-regular graphs,
there is a direct translation between Hashimoto eigenvalues and 
adjacency eigenvalues, and this can be used
to prove Theorem~\ref{th:main_Alon};
this direct translation is essentially used by
\cite{broder,friedman_random_graphs,friedman_alon}, where results
on the expected number of non-backtracking walks of certain types
are used to infer bounds on the traces of powers of $A_G$ with 
$G$ a random $d$-regular graph on $n$ vertices.

The {\em \gls{Hashimoto matrix}}, $H_G$, of a graph, $G$
is the adjacency matrix of the {\em directed
line graph} or {\em oriented line graph} of $G$.  We will give precise
definitions in Section~\ref{se:p0-precise}; roughly speaking, the oriented line
graph of $G$ is the graph whose vertices are the directed edges of $G$
(i.e., the number of half-loops plus
twice the number of other undirected edges of $G$), and whose edges
consist of non-backtracking walks of length two in $G$.
Throughout this paper, we shall use 
\newnot{symbol:mu1}
$\mu_1(G)$ to denote the
Perron-Frobenius eigenvalue of $H_G$ (i.e., its largest positive 
eigenvalue), and use 
\newnot{symbol:mui}
$\mu_i(G)$ for $i=2,\ldots,m$ to denote
the other eigenvalues of $H_G$, in no particular order
(the order will not matter); here $m$ is the number of directed edges
of $G$, which is twice the number of undirected edges plus the
number of half-loops of $G$.
The ``Ihara determinantal formula'' states that for a connected graph,
$G$, without half-loops, we have
\begin{equation}\label{eq:preZeta}
\det(\mu I -H_G) = \det\bigl(\mu^2 I - \mu A_G + (D_G-I) \bigr)
(\mu^2-1)^{-\chi(G)}
\end{equation}
where $A_G$ is the
adjacency matrix of $G$ and $D_G$ is the diagonal ``vertex degree counting''
matrix of $G$ (i.e., $D_G$ is the degree of $v$ at the $(v,v)$ entry,
and zero otherwise), and where $I$ denotes the appropriate identity
matrix (the $m\times m$ identity on the left-hand-side, with $m$ as above,
and the
$n\times n$ on the right-hand-side, where $n$ is the number of vertices
of $G$);
this was established by Ihara \cite{ihara} for regular graphs, $G$,
and in general by Bass \cite{bass_elegant}, Hashimoto \cite{hashimoto_zeta,
hashimoto1,hashimoto2},
and others (see \cite{terras_zeta}).  The above left- or right-hand-side
is the reciprocal 
{\em Ihara Zeta function} of the graph (see, for 
example, \cite{st1,st2,st3,terras_zeta}).
Consequently, if $G$ is $d$-regular, the Hashimoto eigenvalues,
$\mu_i(G)$, of $G$ are given as the two roots, $\mu$, of
$$
\mu^2 - \mu \lambda_j + (d-1) = 0,
$$
for each adjacency matrix eigenvalue, $\lambda_j$, plus an additional
$-\chi(G)$ multiplicity of the $1$ and $-1$ eigenvalues (the values $\pm 1$
can also occur in the above quadratic equation, namely for
$\lambda_j=\pm d$).  If $G$ has
half-loops, then a similar formula holds with minor modification of 
the $\pm 1$ eigenvalues (see \cite{friedman_alon}, for example).
It follows that for $d$-regular graphs, $G$, knowledge of all the
adjacency eigenvalues,
$\lambda_j$, is, in a sense, equivalent to knowledge of all the
Hashimoto eigenvalues, $\mu_i$.

We shall explain that 
the method of Broder-Shamir \cite{broder} makes essential use of the
fact that
we consider only non-backtracking walks.  It therefore
turns out that all 
methods 
determine information
on the Hashimoto eigenvalues, $\mu_i(G)$, of random covering graphs, $G$,
rather that on the adjacency eigenvalues, $\lambda_j(G)$.
In particular, if $B$ is non-regular, then our theorems may give
better information on Hashimoto eigenvalues of a random $G\in\cC_n(B)$
than adjacency eigenvalues.  
When $A_G$ and $D_G$ in \eqref{eq:preZeta} commute, i.e., when 
$G$ is $d$-regular, then $A_G$ and $D_G$ have a common eigenbasis,
and this gives a rather direct translation between
spectral information of $A_G$ and $D_G$ to spectral information
of $H_G$.  But in the general case, such a translation is not, at 
present, available (see especially,
\cite{angel}).

As an example, our trace methods prove the following theorem.

\begin{theorem} 
\label{th:main}
Let $B$ be an arbitrary connected graph of negative Euler characteristic.
Let $\tau_0$ be any 
positive real number such that
(1) $\tau_0^2\ge\rho(H_B)$, and
(2) 
for every covering
map $\pi\from G\to B$ in $\cC_n(B)$, we have that any non-real eigenvalues
of the Hashimoto matrix, $H_G$, of $G$, are of absolute value at most
$\tau_0$.
Then for any $\epsilon>0$ there is a constant
$C=C_\epsilon$ for which 
$$
\prob{G\in \cC_n(B)}{\rhonew_B(H_G) \ge  \tau_0+\epsilon} \le C_\epsilon n^{-1}
$$
for all $n$.
\end{theorem}
We mention that we may be able to somewhat relax the condition on the non-real
eigenvalues of $G$ by proving a stronger ``side-stepping lemma,''
Lemma~\ref{le:side-step} (a weaker version of which appears
in \cite{friedman_alon}).

In Section~\ref{se:p2-irregular} we use a result of
Kotani and Sunada, in \cite{kotani}, to obtain the following
consequence of Theorem~\ref{th:main}.

\begin{theorem}
\label{th:d_max}
For any graph, $B$, Conjecture~\ref{co:gen_Alon}
holds with $\rho(A_{\widehat B})$ replaced with
$2(d_{\max{}}-1)^{1/2}$, where 
\newnot{symbol:dmax}
$d_{\max}=d_{\max}(B)$ is the maximum degree of a vertex in
$B$
\end{theorem}

For certain $B$ the above theorem gives a very good result.
For example,
if each vertex of $B$ has degree either $d_{\max{}}$ or 
$d_{\max{}}-1$, and $d_{\max{}}$ is very large, then 
$\rho(A_{\widehat B})$ is at
least $2(d_{\min{}}-1)^{1/2}$, where $d_{\min}=d_{\min}(B)$ is the minimum
degree of $B$.
Hence the above theorem gives an improvement of Puder's result \cite{puder}
in this case, which
is the best result to date.

On the other hand, for certain $B$ the above theorem does not give
any non-trivial result.
For example, if $B$, consists entirely
of two long cycles that
meet in a single vertex, then $d_{\max{}}=4$, while
$\rho(A_B)$ can be arbitrarily close to one, as the two cycles' lengths
goes to infinity.  Since $\rho(A_B)$ is an upper bound
on $\rho(A_{\widehat B})$ and on $\rho(A_G)$ for any $G$ admitting a covering
map to $B$, the above theorem, for certain $B$, gives
an interesting result, weaker than the trivial bound on $\rho(A_G)$.

%


\chapter{Precise Terminology and Overview of the Proof}
\label{ch:p0.5}



In this chapter we will make all our terminology precise, and
give an overview of Chapter~\ref{ch:p1}, which gives a proof 
of the Relativized Alon Conjecture in the case of $d$-regular base graphs
without half-loops.
We shall at times refer to arbitrary base graphs, but usually we
shall do so just to illustrate certain ideas, such as what we 
mean by an ``algebraic model'' (see the end of Section~\ref{se:p0-precise}).

\section{Precise Terminology}
\label{se:p0-precise}

In this subsection we give specify our precise
definitions for a number of 
concepts in algebraic graph theory.  We note that such definitions
vary a bit in the literature.
For example,
in this paper graphs may have multiple edges and two types of
self-loops---half-loops and whole-loops---in the terminology of
\cite{friedman_geometric_aspects}, and similar to many other
mathematical formulations of graph theory, such as in 
the work of Stark and Terras 
on Zeta functions of graphs
\cite{st1,st2,st3}.


\subsection{Graphs and Morphisms} 

\begin{defn}\newnot{symbol:graph} A 
\emph{\gls[format=hyperit]{directed graph}} 
(or
\emph{digraph}) is a tuple \( G = (V, \Edir, t, h) \) where \( V \) and
\(\Edir\) are sets---the vertex and 
\gls[format=hyperit]{directed edge} sets---and \( t
\from \Edir \rightarrow V \) is the \emph{tail map} and \( h \from \Edir
\rightarrow V \) is the \emph{head map}. A directed edge \(e\) is called 
\emph{\gls[format=hyperit]{self-loop}} 
if \( t(e) = h(e) \), that is, if its tail is its head.
Note that our definition also allows for \emph{multiple edges}, that is
directed edges with identical tails and heads.  Unless
specifically mentioned, we will only consider directed graphs which have
finitely many vertices and directed edges.
\end{defn}

A graph, roughly speaking, is a
directed graph with an involution that pairs the edges.

\begin{defn} An 
\emph{\gls[format=hyperit]{undirected graph}} 
(or simply a \emph{\gls[format=hyperit]{graph}}) is a
tuple \( G = (V, \Edir, t, h, \iota) \) where \( (V, \Edir, t, h) \) is a
directed graph and where \( \iota \from \Edir \rightarrow \Edir \), called
the \emph{\gls[format=hyperit]{opposite map}} 
or
{\em \gls[format=hyperit]{involution}} of the graph,
is an involution on the set of directed edges
(that is, \( \iota^2=\id_{\Edir} \) is the identity) satisfying \( t \iota
= h \). The directed graph \( G = (V, \Edir, t, h) \) is called the
\emph{underlying directed graph} of the graph \(G\). If \(e\) is an edge,
we denote by \(e^{-1} = \iota(e)\) and call it the \emph{opposite edge}. A
self-loop \(e\) is called a \emph{\gls[format=hyperit]{half-loop}} 
if \( \iota(e) = e\), and otherwise is called a
\emph{\gls[format=hyperit]{whole-loop}}.

The opposite map induces an equivalence relation on the directed edges of
the graph, with $e\in \Edir$ equivalent to $\iota e$;
we call the quotient set, \( E \), the 
\emph{\glspl{undirected edge}} of the 
graph \(G\) (or simply its \emph{\glspl[format=hyperit]{edge}}). 
Given an edge of a
graph, an \emph{\gls[format=hyperit]{orientation}} 
of 
that edge is the choice of a representative
directed edge in the equivalence relation (given by the opposite map).
\end{defn}

\begin{notation} For a graph, $G$, we use the notation
\newnot{symbol:VG}\newnot{symbol:EdirG}\newnot{symbol:EG}
\newnot{symbol:tG}\newnot{symbol:hG}\newnot{symbol:iotaG}
$V_G,E_G,\Edir_G,t_G,h_G,\iota_G$ to denote
the vertex set, edge set, directed edge set, tail map, head map, 
and opposite map of $G$; similarly for directed graphs, $G$.
\end{notation}

\begin{defn} A \emph{\gls[format=hyperit]{morphism of directed graphs}}, 
\( \varphi\from
G\rightarrow H \) is a pair \( \varphi=(\varphi_V, \varphi_E) \) for which
\( \varphi_V \from V_G \rightarrow V_H\) is a map of vertices and \(
\varphi_E \from \Edir_G \rightarrow \Edir_H \) is a map of directed edges
satisfying \( h_H(\varphi_E(e)) = \varphi_V(h_G(e)) \) and \(
t_H(\varphi_E(e)) = \varphi_V(t_G(e)) \) for all \( e \in \Edir_G \). 
We refer to the values of $\varphi_V^{-1}$ as {\em vertex fibres} of $\varphi$,
and similarly for edge fibres.
We often more simply write
\(\varphi\) instead of \(
\varphi_V\) or \( \varphi_E \).

A {\em \gls[format=hyperit]{morphism of graphs}} is defined as a morphism of the underlying
directed graphs, with the additional requirement that \(
\iota_H(\varphi(e)) = \varphi(\iota_G(e)) \) for all \( e \in \Edir_G \).
\end{defn}
The above definitions make graphs and directed graphs into a category;
in both cases, there is a terminal element, namely a graph or directed
graph with one vertex and one edge.

\begin{defn} An \emph{\gls[format=hyperit]{oriented graph}} 
is an undirected graph, \(G\),
with an \gls[format=hyperit]{orientation}
of each of its edges. In this context, when
referring to an edge \( e \in E_G \) we always assume it represents its
underlying directed edge and hence extend the language of directed edges to
this edge (so it has a tail and a head map) and we denote by \( e^{-1} \)
its opposite directed edge.  
\end{defn}

\subsection{Walks and Traces}

The traces of powers of many interesting matrices can be understood
as counting certain types of {\em walks}.

\begin{defn}\label{defn:walks} Let \( G = (V, \Edir, t, h) \) be a directed
graph and \(k \geq 0\) an integer. A \emph{\gls[format=hyperit]{walk} of 
length \(k\) in \(G\)}
is an alternating sequence of vertices and directed edges, \[ v_0, e_1,
v_1, e_2, v_2, \ldots, e_k, v_k \] for which \( t_G(e_i) = v_{i-1} \) and
\(h_G(e_i) = v_i \).
The vertices $v_1,\ldots,v_{k-1}$ are called the 
{\em \glspl[format=hyperit]{interior vertex}} of the walk.
We say that a walk is
\emph{\gls[format=hyperit]{closed}} if \( v_0 = v_k \).

A in a graph is a walk 
in the underlying directed graph.
In this case
$$
v_k, e_k^{-1}, v_{k-1},\ldots, v_2, e_1^{-1}, v_1
$$
is also a walk, which we call the \gls[format=hyperit]{reverse walk}
(or {\em inverse walk}) of $w$, which we denote $w^{-1}$.
We say that a walk, as above, in a graph, $G$, is
\begin{enumerate} 
\item
\emph{\gls[format=hyperit]{non-backtracking}} 
(or \emph{irreducible}) if \( \iota(e_{i+1}) \neq
e_i \) for all \( i = 1, \ldots, k-1 \), 
\item \emph{\gls[format=hyperit]{strictly non-backtracking closed}} 
(or \emph{strongly irreducible}) if it is 
non-backtracking, closed, and if $\iota(e_k)\ne e_1$ (we cannot have 
$\iota(e_k)=e_1$ if the walk is not closed).
\item {\em \gls[format=hyperit]{beaded path}} 
if it is non-backtracking and all interior vertices are traversed
once, and all interior vertices have degree
two in the graph.
\end{enumerate}
\end{defn}

A walk of length at least one can be identified
with its sequence of edges;
a walk of length zero is simply a vertex.

Our main interest lies in the algebraic properties of graphs. We review
some definitions of algebraic graph theory.

\begin{defn}
\newnot{symbol:TrM}\newnot{symbol:rhoAG}
Let $G$ be a directed graph.
The \emph{adjacency matrix}, \(A_G\), of $G$
is the 
square matrix indexed on the vertices, $V_G$, whose \( (v_1, v_2) \)
entry is the number of directed edges whose tail is the vertex \(v_1\) and
whose 
head is the vertex \(v_2\). 
The {\em indegree} (respectively {\em outdegree}) of a vertex, $v$,
of $G$ is the number
of edges whose head (respectively tail) is $v$.

The adjacency matrix of an undirected graph, $G$, is
simply the adjacency matrix of its underlying directed graph. 
For an undirected graph, the indegree of any vertex equals its outdegree,
and is just called its {\em degree}.
The {\em degree matrix} of $G$ is the diagonal matrix, 
\newnot{symbol:DG}
$D_G$, indexed on $V_G$
whose $(v,v)$ entry is the degree of $v$.
We say that $G$ is {\em $d$-regular} if $D_G$ is $d$ times the identity
matrix, i.e., if each vertex of $G$ has degree $d$.
\end{defn}

For any non-negative integer $k$, the number of closed walks of length
$k$ is a graph, $G$, is just the 
trace, $\Tr(A_G^k)$, of the $k$-th power
of $A_G$.

\begin{notation}\label{no:lambda}
Given a graph, $G$, the matrix $A_G$ is symmetric, and
hence the eigenvalues of $A_G$ are real and can be ordered
$$
\lambda_1(G) \ge \cdots \ge \lambda_n(G),
$$
where $n=|V_G|$.  We reserve the notation $\lambda_i(G)$ to denote the
eigenvalues of $A_G$ ordered as above.
When $G$ is a directed graph, we let $\lambda_1(G)$ be the
Perron-Frobenius eigenvalue of $A_G$, and, for $i=2,\ldots,n$,
let $\lambda_i(G)$ be the remaining eigenvalues of $A_G$ in
no particular order
(all concepts we discuss about the $\lambda_i$ for $i\ge 2$ will not depend
on their order).
\end{notation}

\begin{defn}
Let $G$ be a graph.  We define the 
\emph{\gls[format=hyperit]{directed line graph}} or 
{\em \gls[format=hyperit]{oriented line graph}} of $G$,
denoted 
\newnot{symbol:LineG}
\( \Line(G) \), to be the
directed graph \( L=\Line(G) = (V_L, \Edir_L, t_L, h_L) \) given as follows:
its vertex set, \( V_L \), is the set \( \Edir_G \) of directed edges of
\(G\); its set of directed edges is defined by \[
\Edir_L = \left\{ (e_1,e_2) \in \Edir_G \times \Edir_G \mid h_G(e_1) =
t_G(e_2) \text{ and } \iota_G(e_1) \neq e_2 \right\} \] that is,
$\Edir_L$ corresponds to the
non-backtracking walks of length two in $G$. 
The tail and head maps are simply
defined to be the projections in each component, that is by \(t_L(e_1,e_2)
= e_1 \) and \( h_L(e_1,e_2) =e_2 \).

The {\em \gls[format=hyperit]{Hashimoto matrix}} 
of $G$ is the adjacency matrix of its 
directed line graph, denoted
$H_G$, which is, therefore, a square matrix indexed
on $\Edir_G$.
We use the symbol $\mu_1(G)$ to denote the Perron-Frobenius eigenvalue
of $H_G$, and use $\mu_2(G),\ldots,\mu_m(G)$, where $m=|\Edir_G|$,
to denote the remaining eigenvalues, in no particular order (all
concepts we discuss about the $\mu_i$ for $i\ge 2$ will not depend
on their order).
\end{defn}

It is easy to see that
for any positive integer $k$, the number of strictly non-backtracking
closed walks of length
$k$ in a graph, $G$, equals the trace, $\trace(H_G^k)$, of the $k$ power
of $H_G$; of course, the strictly non-backtracking walks begin and end
in a vertex, whereas $\trace(H_G^k)$ most naturally counts walks beginning
and ending in an edge; the correspondence between the two notions can
be seen by taking a walk of $\Line(G)$,
beginning and ending an in a directed edge, $e\in\Edir_G$, and mapping
it to the strictly non-backtracking closed walk in $G$ beginning at,
say, the tail of $e$.

For graphs, $G$, that have half-loops, the Ihara determinantal formula
takes the form (see \cite{friedman_alon,st1,st2,st3}):
\begin{equation}\label{eq:Zetahalf}
\det(\mu I -H_G) = \det\bigl(\mu^2 I - \mu A_G + (D_G-I) )
(\mu-1)^{|{\rm half}_G|}
(\mu^2-1)^{|V_G|-|{\rm pair}_G|},
\end{equation}
where ${\rm half}_G$ is the set of half-loops of $G$, and
${\rm pair}_G$ is the set of undirected edges of $G$ that are not
half-loops, i.e., the collection of sets of the form, $\{e_1,e_2\}$
with $\iota e_1=e_2$ but $e_1\ne e_2$.

\subsection{Covering and Etale (Open Immersion) Maps}

Here we discuss spectral aspects of covering maps of graphs.  We also define
{\em \'etale maps}\footnote{Some articles, such as \cite{stallings83},
prefer the term ``open immersion'' to ``\'etale,'' which are identical
concepts in this context.}, a closely related concept which shall
be of interest in 
Subsection~\ref{sb:limitations}, to understand which graphs can occur
as subgraphs of a graph in $\cC_n(B)$ (with positive probability).

\begin{defn}
\label{de:etale}
A morphism of directed graphs \( \nu \from H \rightarrow G \)
is a \emph{\gls[format=hyperit]{covering map}} (respectively, {\em 
\gls[format=hyperit]{etale map}})
local isomorphism (respectively, injection), that is
for any vertex \(w \in V_H\), the edge morphism \( \nu_E \) induces a
bijection (respectively, injection)
between \( t_H^{-1}(w) \) and \( t_G^{-1}(\nu(w)) \) and a
bijection (respectively, injection)
between \( h_H^{-1}(w) \) and \( h_G^{-1}(\nu(w)) \).
We call \( G \) the
\emph{base graph} and \( H \) a \emph{covering graph of $G$} 
(respectively, {\em graph \'etale over \(G\)}).

If \( \nu\from H \rightarrow G \) is a covering map and \(G\) is connected,
then the \emph{degree} of \(\nu\), denoted \( [H\colon G] \), is the number
of preimages of a vertex or edge in \(G\) under \(\nu\) (which does not
depend on the vertex or edge). If \(G\) is not connected, we insist that
the number of preimages of \(\nu\) of a vertex or edge is the same,
i.e., the degree is independent of the connected component, and we will
write this number as \( [H\colon G] \).
In addition, we often refer to $H$, without $\nu$ mentioned explicitly,
as a {\em covering graph} of $G$.

A morphism of graphs is a \emph{covering map} (respectively,
{\em \'etale map}) if the morphism of the
underlying
directed graphs is a covering map (respectively, \'etale map).  
\end{defn}

Clearly a composition of two covering maps is a covering map, and
similarly for \'etale maps.  Any covering map is \'etale; also,
any inclusion
of a subgraph of a graph (to the graph) is \'etale.
In particular, any morphism that is the composition of an inclusion
with a covering map is \'etale; it is not hard to see that the 
converse is true (see, for example,
\cite{stallings83,friedman_memoirs_hnc}, or 
Proposition~\ref{pr:etale_factorization} below).
The necessary ideas to prove this
are also necessary for us to define what
we call the {\em Broder-Shamir} model, $\cC_n(B)$, for arbitrary
integer $n$ and graph $B$ (possibly with half-loops); 
hence we develop these ideas now.  

\begin{definition} 
\label{de:permutation_assignment}
Let $B$ be a graph.
By a {\em \gls[format=hyperit]{permutation assignment} of 
degree $n$ over $B$} we mean
a map $\sigma\from \Edir_B\to \cS_n$, 
where $\cS_n$ is the set of permutations
on $\{1,\ldots,n\}$, such that
$\sigma(\iota_B e)=\sigma(e)^{-1}$ for all $e\in\Edir_B$.
By a
{\em standard covering of degree $n$ over $B$} we mean the data
$(\pi,\mu)$ where
$\pi\from G\to B$ is a covering map of degree $n$, and
$$
\mu\from V_G\to V_B\times \{1,\ldots,n\}
$$
is a bijection.
To each such standard covering we associate a permutation assignment
by the ``tails-to-heads'' map, meaning
for each edge, $e\in \Edir_B$, we get a permutation, 
$\sigma(e)$, such that for each
$i\in\{1,\ldots,n\}$ we have that $(t(e),i)$ is the tail of an edge whose
head is
$(h(e),\sigma(e)i)$.
\end{definition}

The following proposition is easy, but useful.
\begin{proposition}
\label{pr:associated_covering}
If $B$ is a graph, and $V_G$ is a set and $\mu$ a set theoretic bijection,
$$
\mu\from V_G\to V_B\times \{1,\ldots,n\},
$$
then
any permutation assignment $\sigma\from\Edir_B\to\cS_n$ determines
a unique graph, up to isomorphism, 
$G=(V_G,E_G,t_G,h_G,\iota_G)$ with a covering map,
$\pi\from G\to B$, such that $\sigma$ is the permutation assignment
associated to $(\pi,\mu)$.
\end{proposition}
\begin{proof}
It suffices to consider the case where $V_G$ equals
$V_B\times\{1,\ldots,n\}$ and $\mu$ is the identity map.
In this case we set
$$
E_G = E_B \times \{1,\ldots,n\},
$$
and define
$$
t_G(e,i) = (t_B(e),i) \quad\mbox{and}\quad
h_G(e,i) = (h_B(e),\sigma(e)i) ,
$$
and $\iota_G(e,i) = (\iota_B(e),\sigma(e)i)$.  Then the map
$E_G\to E_B$ via projection gives a covering map $G\to B$.

If $\pi'\from G'\to B$ and $\mu'\from V_{G'}\to \{1,\ldots,n\}$
is any other pair of a covering map of, $\pi'$, of degree $n$, 
and an isomorphism, $\mu'$ yielding
the same permutation assignment, it is straightforward to verify that
$G'$ is isomorphic to $G$: namely, for such $\pi'$ and $\mu'$
we get a natural set theoretic
isomorphism $\alpha\from E_{G'}\to E_G$ such that for $e'\in E_{G'}$,
$\alpha(e')$ is the unique edge whose tail is $\mu'(t_{G'}e')$ and
whose head is $\mu'(h_{G'}e')$;
then we verify that $\alpha$ and $\mu'$ intertwine
the tails
and heads maps and the graph involution.
\end{proof}

The following proposition is noteworthy but immediate, so we state
it without proof.
\begin{proposition}
If $B$ is a graph, then a set theoretic map 
$\sigma\from \Edir_B\to\{1,\ldots,n\}$ is a permutation assignment
iff for each $e\in\Edir_B$ we have
(1) if $e$ is a half-loop, we have
$\sigma(i)$ is an involution (i.e., a permutation equal to its 
inverse), and (2) if $e$ is not a half-loop, then $\sigma(e)$ is an
arbitrary permutation and $\sigma(\iota_B e)$ is the inverse
permutation.
\end{proposition}

Now we give analogous notions of standard coverings for \'etale maps.
\begin{definition}
By a 
{\em standard \'etale map of degree $n$ over $B$} we mean the data
$(\pi,\mu)$ where
$\pi\from G\to B$ is an \'etale map, and
$$
\mu\from V_G\to V_B\times \{1,\ldots,n\}
$$
is an injection.
To each such standard covering we associate a {\em partially defined
permutation assignment}, meaning that for each $e\in \Edir_B$, we
have a {\em partially defined permutation}, i.e.,
a map, $\sigma(e)$, defined on those integers, $i$, for which
$(t(e),i)$ is in the image of $\mu$ (and, in this case, $\sigma(e)i$
is the unique integer such that the edge over $e$ with tail
$(t(e),i)$ has head $(h(e),\sigma(e)i)$); furthermore,
these partially defined
permutations satisfy the property that 
for each such $e$ and $i$, we have $\sigma(\iota_B e)$ is defined
on $\sigma(e)i$ and equals $i$.
\end{definition}

\begin{proposition} 
\label{pr:etale_factorization}
Let $\pi\from G'\to B$ be an \'etale morphism
of graphs, and let $n_0$ be the 
maximum vertex fibre of $\pi$.
Then $\pi$ factors as an inclusion following by a covering
map, and the degree of the covering map can be any integer, $n$, for
which $n\ge n_0$.
\end{proposition}
\begin{proof}
Let $n\ge n_0$.
Our goal is to describe a graph, $G''$, for which $G'$ can be identified
as a subgraph of $G''$, and for which $G''$ has a covering map
to $B$ of degree $n$; first we describe $V_{G''}$, and then $E_{G''}$.

Set
$$
V_{G''} = V_B \times \{1,\ldots,n\},
$$
for each vertex, $v\in V_B$, take an arbitrary injection,
$\pi^{-1}(v) \to \{1,\ldots,n\}$; these injections
gives rise to an injection
$$
\mu\from V_{G'} \to  V_B \times \{1,\ldots,n\} = V_{G''}.
$$
This gives us a partially defined involution, $\sigma(e)$,
for each half-loop, $e\in \Edir_B$,
and partially defined
permutations, $\sigma(e)$ on edges, $e\in \Edir_B$ that are
not half-loops, with $\sigma(e)$ and $\sigma(\iota_B e)$ being
inverses of each other.
It is clear
that any partially defined involution or permutation 
extends to a (fully defined) involution or permutation on
$\{1,\ldots,n\}$; doing so for all the $\sigma(e)$ here (in
any way)
gives $G''$ the structure of a covering graph to $B$
(for each $e\in\Edir_B$ with $e\in \iota_B(e)$, we first extend
either $\sigma(e)$ or $\sigma(\iota_B e)$ to a full permutation,
and infer the other permutation given that the two permutations are
inverses of each other). 
Furthermore, the injection $\mu$ on vertices extends to an
injection of graphs, $G'\to G''$, in view of how we partially
defined $\sigma(e)$ above.
\end{proof}
We remark that the partially defined permutations in the above
proof are crucial to
Proposition~\ref{prop:EsymmProd}, which will be our fundamental
starting point to all our ``\glspl{asymptotic expansion}.''

Covering maps have distinguished spectral properties, which we now
discuss.

\begin{definition} If $\pi\from G\to B$ is a covering map of directed
graphs,
then an {\em \gls[format=hyperit]{old function} (on $V_G$)} is a 
function on $V_G$ 
arising via pullback
from $B$, i.e., a function $f\pi$, where $f$ is a function (usually
real or complex valued), i.e., a function on $V_G$ (usually real or
complex valued) whose value depends only on the $\pi$ vertex fibres.
A {\em \gls[format=hyperit]{new function} (on $V_G$)} is a function 
whose sum on each vertex
fibre is zero.
The space of all functions (real or complex)
on $V_G$ is a direct sum of the old and new functions, an orthogonal
direct sum on the natural inner product on $V_G$, i.e.,
$$
(f_1,f_2) = \sum_{v\in V_G} \overline{ f_1(v)} f_2(v) .
$$
The adjacency matrix,
$A_G$, viewed as an operator, takes old functions to old functions and
new functions to new functions.  The
{\em \gls[format=hyperit]{new spectrum}} of $A_G$, which we often denote
$\specnew_B(A_G)$, is the spectrum of $A_G$
restricted to the new functions; we similarly define the 
{\em \gls[format=hyperit]{old spectrum}}.
As mentioned before, \eqref{eq:independent_of_covering_map} shows that
when $G$ is finite, the new spectrum, meaning the eigenvalues with
their multiplicities, is independent of the covering map.

This discussion holds, of course, equally well if
$\pi\from G\to B$ is a covering
morphism of graphs, by doing everything over the underlying directed
graphs.
\end{definition}

We can make similar definitions for the spectrum of the Hashimoto 
eigenvalues.  
First, we observe that covering maps induce covering maps on
directed line graphs; let us state this formally (the proof is easy).

\begin{prop} Let \( \pi \from G \to B \) be a covering map. Then \( \pi \)
induces a covering map \( \pi^{\Line} \from \Line(G) \to \Line(B) \).  
\end{prop}
Since $\Line(G)$ and $\Line(B)$ are directed graphs, the above
discussion of new and old functions, etc., holds for
$\pi^{\Line} \from \Line(G) \to \Line(B)$; e.g., new and old functions
are functions on the vertices of $\Line(G)$, or, equivalently,
on $\Edir_G$.

\begin{definition}\label{de:new_Hashimoto}
Let \( \pi \from G \to B \) be a covering map.  We define the
{\em new Hashimoto spectrum of $G$ with respect to $B$}, denoted
\newnot{symbol:SpecnewHG}
$\Specnew_B(H_G)$ to be
the spectrum of the Hashimoto matrix restricted to the new
functions on $\Line(G)$, and
\newnot{symbol:rhonewHG}
$\rhonew_B(H_G)$ to be the supremum of the norms of
$\Specnew_B(H_G)$.
\end{definition}
Again, similar to \eqref{eq:independent_of_covering_map}, we have
$$
\sum_{\mu\in \Specnew_B(H_G)} \mu^k = \Tr(H_G^k)-\Tr(H_B^k),
$$
and hence the new Hashimoto spectrum is independent of the
covering map from $G$ to $B$.

\subsection{Variable-Length Graphs}

Variable-length graphs will be used to define our {\em certified trace}
and to prove theorems regarding their expected values.
We refer to \cite{friedman_alon}, Subsection~3.2 for basic facts on
variable-length graphs.  We shall briefly state the facts we need.

\begin{defn} 
Let $G$ be a directed graph, and $\vec k$ a vector indexed on $\Edir$
with non-negative integer components.
We refer to the tuple $(G,\vec k)$ as a 
{\em \gls[format=hyperit]{variable-length graph}}, which
we view as the data of a directed graph where each edge is given a
non-negative real {\em length}; for $e\in \Edir$, $\vec k(e)$ is called
the {\em length of $e$}.
Furthermore, we define
the {\em \gls[format=hyperit]{realization of a variable-length graph}},  
$(G,\vec k)$, which we denote
\newnot{symbol:vlg}
$\VLG(G,\vec k)$, to be the directed graph obtained by replacing
each $e\in \Edir$ by a directed path of length $\vec k(e)$;
in other words, we replace each edge, $e\in \Edir$, by $k=\vec k(e)$
new directed edges, $e_1,\ldots,e_k$,
and $k-1$ new vertices, $v_1,\ldots,v_{k-1}$, so that each new
vertex has indegree and outdegree one, and such that
$$
t(e),e_1,v_1,e_2,v_2,\ldots,e_{k-1},v_{k-1},e_k,h(e)
$$
is a walk in the graph, i.e., 
for all $i=1,\ldots,k-1$ we have $t(e_i)=v_{i-1}$ and $h(e_i)=v_i$,
with the understanding that $v_0=t(e)$ and $v_k=h(e)$.

If $G$ is a graph without half-loops, 
and $\vec k$ a vector indexed in $E_G$ with
non-negative integral components, we make a similar definition;
namely, we define the variable-length graph as the 
a function $\vec k\from E_G\to \reals_{\ge 0}$, and
$\VLG(G,\vec k)$ as the graph obtained by replacing each edge
$e\in E_G$ by a path of length $k=\vec k(e)$; in other words,
we replace $e$ and $\iota_G e$ with with $k$ new edges,
$e_1,\ldots,e_{k-1}$ and $k-1$ new vertices of degree two,
$v_1,\ldots,v_{k-1}$ such that for $i=1,\ldots,k-1$, $e_i$ is
incident upon $v_{i-1}$ and $v_i$, with the understanding that
$v_0$ and $v_k$ are the two endpoints of the discarded edge $e$.
\end{defn}

We remark that when $G$ has half-loops, which only occurs when
$B$ has half-loops, namely in Subsection~\ref{sb:general_base_graphs},
a type remembers all the half-loops; hence, all the half-loops are
unaltered, i.e.,
restricted to having length one.  Hence, in this article we understand
that if $G$ has half-loops, the VLG's we form from $G$ leave all
half-loops alone, and we define lengths only on the edges of $G$ which
are not half-loops.

There is a well-known {\em Shannon's algorithm} for computing 
$\lambda_1(\VLG(G,\vec k))$ (see \cite{friedman_alon}).  
We shall need only the following facts.

\begin{proposition}[Monotonicity]
\label{pr:vlg_greater_than}
If $G$ is a directed graph, and
$\vec k,\vec k'$ are both functions from $\Edir_G$ with
$\vec k\le \vec k'$, i.e., $\vec k(e)\le \vec k'(e)$ for all $e\in \Edir_G$,
then $\lambda_1(\VLG(G,\vec k))\ge \lambda_1(\VLG(G,\vec k'))$.
Similarly for graphs.
\end{proposition}

\begin{proposition}[Continuity]
\label{pr:vlg_limit}
Let $G$ be a directed graph, let $e\in \Edir$, and let 
$\vec k_1,\vec k_2,\ldots$ be functions from $\Edir\to\integers_{\ge 0}$
such that $\vec k_n(e)=n$ and $\vec k_i(e')=\vec k_j(e')$ for any
$i,j$ and any $e'\in \Edir$ with $e'\ne e$.  Then
$$
\lim_{n\to\infty} \lambda_1\bigl( \VLG(G,\vec k_n) \bigr) 
=
\lambda_1 \bigl( \VLG(G',\vec k') \bigr) ,
$$
where $G'$ is $G$ with $e$ discarded, and $\vec k'$ is the restriction of
$\vec k_n$ to $\Edir_{G'}$ (which is independent of $n$).
Similarly for graphs, $G$.
\end{proposition}

\subsection{The Broder-Shamir Model and Related Models}
\label{sb:intro_to_broder_shamir}

For a graph, \( B
\), and a positive integer \( n \), 
we give a model of random cover of \( B \) of
degree \( n \) which slightly generalizes the model used in
\cite{friedman_relative}.
As mentioned before, throughout Chapter~\ref{ch:p1} we assume
that $B$ has no half-loops, and in the case the Broder-Shamir
model is formed from $|E_B|$ independently and uniformly chosen
elements of $S_n$, the group of permutations on $n$ elements, in
the natural fashion.

Hence, the main point of this section is to give the reader some
idea of some of the various random covering models to which our
main theorems, Theorem~\ref{th:main_Alon} and \ref{th:main_Alon_Ramanujan},
apply, especially when $B$ may contain half-loops.
All
this will be elaborated upon (with proofs) in Section~\ref{se:p2-algebraic}.

Recall the definition of a permutation assignment and standard
covering, in
Definition~\ref{de:permutation_assignment}
\begin{defn} 
\label{de:broder-shamir}
Let $B$ be a graph. 
To each permutation assignment $\sigma\from \Edir_B\to\cS_n$, we
associate a graph covering, $\pi[\sigma]\from B[\sigma]\to B$, of degree $n$
as follows: $B[\sigma]$ is the graph given by
$$
V_{B[\sigma]} = V_B \times \{1,\ldots,n\} 
\quad\mbox{and}\quad
E_{B[\sigma]} = E_B \times \{1,\ldots,n\} ;
$$
the respective tail and head maps of $B[\sigma]$ take $(e,i)$ to
$(t_B(e),i)$ and $(h_B(e),\sigma(e)i)$ respectively; 
the involution map takes $(e,i)$ to $(\iota_B e, \sigma(\iota_B e)i)$;
and, finally,
the covering map $\pi[\sigma]\from B[\sigma]\to B$ is given by the natural map,
i.e., projection onto the first component.
In other words, $B[\sigma]\to B$ is just the graph covering determined by
Proposition~\ref{pr:associated_covering} where $\mu$ is the identity
map.

By the {\em \gls[format=hyperit]{Broder-Shamir model} 
of degree $n$ over $B$}, denoted
$\cC_n(B)$ we
mean the probability space of permutation assignments, $\sigma$, such that
\begin{enumerate}
\item for each $e\in \Edir$, $\sigma_e$ is independent of all
$\sigma_{e'}$ for $e'$ not equal to $e$ or $\iota_B e$;
\item if $e$ is not a half-loop, then $\sigma_e$ is uniform over all
permutations; 
\item for each $e\in \Edir$ that is a half-loop, if $n$ is even we
set $\sigma_e$ to be chosen uniformly among all involutions that have
no fixed points; and
\item for each $e\in \Edir$ that is a half-loop, if $n$ is odd we
set $\sigma_e$ to be chosen uniformly among all involutions that have
exactly one fixed point.
\end{enumerate}
When confusion is unlikely to arise,
we also use $\cC_n(B)$ to mean the induced probability space of
covering maps, $\pi[\sigma]$, with notation as above, and of
covering graphs, $B[\sigma]$ as above.
\end{defn}


The above Broder-Shamir model is very similar to some of models discussed
in \cite{friedman_relative,friedman_alon}.
Broder and Shamir defined this model in \cite{broder} in the case where
$B=W_{d/2}$, i.e., where $B$ the graph with one vertex and $d/2$
whole-loops, for an even integer $d\ge 4$; the above definition seems
like the simplest extension of Broder and Shamir's definition to the
case where $B$ is an arbitrary graph; however, our definition when $n$ is
odd and $B$ contains half-loops is a bit ad hoc, and our choice
(like that in \cite{friedman_alon} for $B$ consisting of one vertex
and $d$ half-loops) is chosen mostly to suit our methods.
In Chapter~\ref{ch:p1} we
will assume, for simplicity, that $B$ has no half-loops.

There are many variants of the above model for which all of our main
theorems hold.  The main general requirement we need of a model is a
certain ``algebraic'' property; see Section~\ref{se:p2-algebraic}.
We will not formalize this, but the basic idea can be described as follows:
the probability
that a uniformly chosen
$\varphi\in \cS_n$ assumes any $k$ particular values is
\begin{eqnarray*}
\frac{1}{n(n-1)\ldots (n-k+1)}
& = &
n^{-k}\bigr(1-n^{-1}\bigl)^{-1}
\ldots \bigr(1-(k-1)n^{-1}\bigl)^{-1} 
\\
&=&
n^{-k} p_0(k) + n^{-k-1} p_1(k) + \cdots
\end{eqnarray*}
where $p_0,p_1,\ldots$ are polynomials ($p_i$ is of degree $2i$);
for example, $p_0(k)=1$ and $p_1(k) = \binom{k}{2}$; 
see \cite{friedman_random_graphs}.
A similar example arises when we permit the base graph to have half-loops
and $n$ is even; one way to generate a random $\varphi\in \cS_n$ is
to insist that $\varphi$ is chosen among all involutions without fixed
points.  In this case $\varphi(i)=j$ implies that $\varphi(j)=i$, so that
to $\varphi$ values are specified in pairs.
Any specified $k$ pairs of $\varphi$ values (i.e., any $2k$ values of
$\varphi$) occurs with probability
$$
\frac{1}{(n-1)(n-3)\ldots(n-2k+1)}
=
n^{-k} \tilde p_0(k) + n^{-k-1} \tilde p_1(k) + \cdots
$$
for polynomials $\tilde p_0,\tilde p_1,\ldots$ in $k$.
Roughly speaking, the ``algebraic'' property requires that the probability
that fixing certain values of the permutations of $\cS_n$ under 
consideration
gives rise to power series in $1/n$ with coefficients that are 
polynomials in the number of values fixed. 
However, this is not
strictly necessary: indeed, our Broder-Shamir models for a covering 
of degree $n$ with $n$ odd yield two types of values for a permutation
for a half-loop: (1) the single value that is a fixed point, and
(2) the remaining $n-1$ values, which essentially pairs all the
remaining values into $(n-1)/2$ pairs.  In this case, the probabilities
depend on whether or not the fixed values include the unique fixed point
or not; of course, either case yields probabilities that have algebraic
power series of the type discussed above.

Similarly, the models we work with generally assume that
all permutations given by $\sigma$
(chosen over a set of representatives of $E_B$ in $\Edir_B$)
are chosen independently.
Again, this is not strictly necessary; see Section~\ref{se:p2-algebraic}.


\section{Remarks on the Trace Method}
\label{se:p0-trace}

In this section we review some aspects of the trace method 
of Broder-Shamir \cite{broder} and its various strengthenings
\cite{friedman_random_graphs,friedman_relative,linial_puder,
friedman_alon,puder}.
In order to do so, 
we shall also give some precise definitions and terminology
used throughout this paper.

We shall make one remark that appears to be new: one gets improved
spectral bounds by first working with traces of powers of 
the Hashimoto matrix,
and then translating the spectral bounds to adjacency matrix bounds.

Having given some precise definitions in Section~\ref{se:p0-precise},
we can now give an overview of trace methods in this article
and previous article.

\subsection{Broder and Shamir's Method: A Single Moment Estimate}

In this subsection we describe how expected traces generally give 
eigenvalue results, as in \cite{broder,friedman_random_graphs,
friedman_relative,linial_puder,puder}.

The first works 
\cite{broder,friedman_random_graphs}
considered $d$-regular random graphs with $d$ even, i.e.,
base graph $B=W_{d/2}$, the bouquet of $d/2$ permutations.
The methods of Broder and Shamir \cite{broder} show that
for fixed, even $d$, we have
\begin{equation}\label{eq:broder_shamir_est}
\EE_{G \in \cC_n(B)}[ \Tr(A_G^k) ] = 
P_{-1}(k) n + P_0(k) 
 + {\rm err}(n,k),
\end{equation}
where for fixed $d$ (i.e., $B=W_{d/2}$ fixed) we have
\begin{enumerate}
\item $P_{-1}(k)$ is the number of closed
walks of length $k$ originating at any vertex in the (infinite)
$d$-regular tree;
\item we have 
\begin{equation}\label{eq:broder_shamir_coefficient}
P_0(k)=d^k+e_0(k),
\end{equation}
with 
$$
|e_0(k)| \le C k^C \Bigl( 2\sqrt{d-1} \Bigr)^k;
$$
and
\item ${\rm err}(n,k)$ satisfies the bound
$$
| {\rm err}(n,k) | \le C k^2 d^k / n .
$$
\end{enumerate}
With these bounds one can show that\footnote{The bound in \cite{broder} is
slightly weaker since they use a slightly suboptimal bound on
the number of closed walks of length $k$ originating at any vertex
in the $d$-regular tree; this affects their $P_{-1}(k)$ estimate.}
for any $\epsilon>0$ we have
\begin{equation}\label{eq:broder_shamir_theorem}
\Prob_{G\in \cC_n(B)}\Bigl[\sup_{i>1} |\lambda_i(G)| <
\Bigl( d\; 2\sqrt{d-1}\Bigr)^{1/2} + \epsilon \Bigr] ,
\end{equation}
with $B=W_{d/2}$ fixed, tends to one.

Once we have
\eqref{eq:broder_shamir_est}, we simply choose an even integer $k$
so that the $P_{-1}(k)n$ term and the ${\rm err}(n,k)$ term
are roughly equal.  In other words, we find a sort of trace estimate,
which is interesting for many values of $k$ for a given $n$;
however, for each value of $n$ we apply
\eqref{eq:broder_shamir_est} for a single value of $k$
(of size proportional to $\log n$).

\subsection{Friedman's Asymptotic Expansions}

Friedman \cite{friedman_random_graphs} builds on the methods in
\cite{broder} to obtain the same result with the 
$$
\Bigl( d\; 2\sqrt{d-1}\Bigr)^{1/2} 
$$
in \eqref{eq:broder_shamir_theorem} replaced with
\begin{equation}\label{eq:big_r_improvement}
d^{1/(r+1)} \Bigl( 2\sqrt{d-1}\Bigr)^{r/(r+1)} 
\end{equation}
for any integer $r$ satisfying
\begin{equation}\label{eq:r_bound}
r < 1 + \sqrt{d-1}/2
\end{equation}
(see Theorem~2.18 of \cite{friedman_random_graphs}, noting that
the $r$ and $d$ here are $r+1$ and $d/2$ in \cite{friedman_random_graphs});
this represents an improvement in \cite{broder} for $d\ge 6$.

Now we wish to explain some important similarities and differences
between \cite{broder} and \cite{friedman_random_graphs}.


Both \cite{broder} and \cite{friedman_random_graphs} are similar as
follows:
\begin{enumerate}
\item both articles first estimate the expected number of
non-backtracking walks of a given length in the graph;
non-backtracking walks are fundamentally easier to analyze in the method of
\cite{broder}, used in \cite{friedman_random_graphs},
and related
papers \cite{friedman_random_graphs,friedman_relative,friedman_alon}; 
the papers \cite{broder,friedman_random_graphs}, and, for that matter,
\cite{friedman_relative},
then translate such estimates into estimates for
the trace of powers of 
$$
\EE_{G \in \cC_n(B)}[ \Tr(A_G^k) ] ;
$$
\item both obtain estimates for 
$$
\EE_{G \in \cC_n(B)}[ \Tr(A_G^k) ] ,
$$
for many values of $k$ for a given $n$,
but then for any given $n$ the estimates are applied for a single value
of $k$, whose size is proportional to $\log n$;
\end{enumerate}


The main difference between
\cite{broder} and \cite{friedman_random_graphs} is that \cite{broder}
gives its estimates on $P_{-1}(k)$, $P_0(k)$, and
${\rm err}(k,n)$ by explicit calculations;
on the other hand, the estimates in \cite{friedman_random_graphs} 
follow a two-step process: 
it is shown that 
the expected number of non-backtracking closed walks of length $k$ in a 
$G\in\cC_n(B)$ is given by
an asymptotic series 
$$
Q_0(k)+
Q_1(k)n^{-1} + Q_2(k)n^{-2} + \cdots + Q_{r-1}(k) n^{-r+1} +
{\rm err}(k,n),
$$
for any $r$ satisfying \eqref{eq:r_bound}, where for some constant,
$C$ (depending only on $r$ and $B=W_{d/2}$)
$$
|{\rm err}(k,n)| \le C k^C (d-1)^k n^{-r} 
$$
and each $Q_i(k)$ is given as
\begin{equation}\label{eq:coefficients}
Q_i(k) = (d-1)^k p_i(k) + e_i(k)
\end{equation}
where each $p_i(k)$ is a polynomial in $k$ and the $e_i(k)$ are ``error terms''
bounded as
$$
|e_i(k)| \le Ck^C (d-1)^{k/2} ;
$$
the second phase is to prove, by other considerations, that all the
polynomials, $p_i(k)$ must vanish (provided $r$ satisfies
\eqref{eq:r_bound}); see Theorem~2.18 of \cite{friedman_random_graphs}.
We stress that
one arrives at \eqref{eq:coefficients} by some calculations that 
give, in principle, a method to
compute the each of the polynomials, $p_i(k)$; however such computations,
especially for large $i$, are quite cumbersome, and one infers the
precise polynomials $p_i(k)$
by a different method, namely known expansion properties of random graphs,
and what this implies about the $p_i(k)$ (i.e., that they must vanish).
Such ``non-explicit methods'' occur also in \cite{friedman_alon} and
in this article.

We point out the similarity between
\eqref{eq:broder_shamir_coefficient} and
\eqref{eq:coefficients}; in both cases the coefficients $P_0(k)$ 
and $Q_i(k)$ have a
``principle'' (or ``term''), and an ``error term;'' for the expected
number of walks (i.e., expected $\Tr(A_G^k)$),
the dominant term is roughly of order $d^k$, and the error
term of order $(2(d-1)^{1/2})^k$;
for non-backtracking walks, the analogous principle term is roughly
of size $(d-1)^k$, and analogous error term is roughly of size
$(d-1)^{k/2}$.

\subsection{Hashimoto Traces Give an Improvement}
\label{sb:hashimoto_improvement}

One point seems to have been unnoticed at present, or at least certainly
in \cite{broder,friedman_random_graphs}: for $d$-regular graphs, one
gets a better high probability bound for the subdominant
adjacency eigenvalues by
first getting a such a bound for subdominant Hashimoto eigenvalues,
and then translating the results to the adjacency matrix.

Indeed, consider the Broder-Shamir results in \cite{broder}: there
they divide the non-backtracking closed walks into three types:
(1) those that trace out a cycle (which we review here in 
Section~\ref{se:p0-loop}), (2) those that trace out a cycle plus a
segment, and (3) those that trace out a graph whose Euler characteristic
is at most minus one.
The walks of type (2) represent a closed walk that is
not strictly non-backtracking.  Hence, in estimating the expected
number of strictly non-backtracking walks of length $k$, i.e.,
the expected value of $\Tr(H_G^k)$, we obtain the bound 
\begin{equation}\label{eq:broder_shamir_strict}
\EE_{G \in \cC_n(W_{d/2})}[ \Tr(H_G^k) ] \le (d-1)^k +
C k (d-1)^{k/2}+ Ck^4 (d-1)^k n^{-1} ,
\end{equation}
the first two summands coming from walks of type (1) (essentially
from Theorem~11 of \cite{broder}, compare to our discussion in
Section~\ref{se:p0-loop}), and the last summand from walks of
type (3) (see Lemma~3 of \cite{broder}).
But since $G$ is $d$-regular, we know that the eigenvalues
of $H_G$ either have absolute value $(d-1)^{1/2}$ (generally complex),
of which there are at most $2n-2$, 
or real and of absolute value at most $d-1$.  
It follows that for any $d$-regular graph we have for any even $k$, either
\begin{equation}\label{eq:Ramanujan_and_nonbi}
\max_{i>1} |\mu_i(G)|^k = (d-1)^{k/2}
\end{equation}
(iff $G$ is Ramanujan and non-bipartite) or
\begin{equation}\label{eq:either_way}
\max_{i>1} |\mu_i(G)|^k \le \Tr(H_G^k) + (2n-2)(d-1)^{k/2}.
\end{equation}
However, we claim that \eqref{eq:Ramanujan_and_nonbi} implies 
\eqref{eq:either_way} if $n\ge 2$: indeed, $\Tr(H_G^k)$ counts certain
walks (i.e., strictly non-backtracking closed walks) and is therefore
always non-negative; and if $n\ge 2$
then $2n-2\ge 1$, and our claim follows.
Hence in all cases we have that
\eqref{eq:either_way} holds.  Taking expectations in \eqref{eq:either_way}
yields
$$
\EE_{G\in \cC_n(B)}\Bigl[\max_{i>1} |\mu_i(G)|^k \Bigr] \le 
\EE_{G\in \cC_n(B)}\Bigl[ \Tr(H_G^k) \Bigr] + (2n-2)(d-1)^k.
$$
But using \eqref{eq:broder_shamir_strict} yields, and the fact
that $\mu_1(G)=d-1$, yields
$$
\EE_{G\in \cC_n(B)}\Bigl[\max_{i>1} |\mu_i(G)|^k \Bigr] \le 
C k (d-1)^{k/2}+ Ck^4 (d-1)^k n^{-1} 
+ (2n-2)(d-1)^{k/2}.
$$
Now choosing $k$ an even integer to have the terms $(2n-2)(d-1)^{k/2}$ and
$(d-1)^k/n$ roughly equal shows that for any $\epsilon>0$ we have
$$
\probb{\sup_{i>1} |\mu_i(G)| > (d-1)^{3/4}+\epsilon }
$$
tends to zero as $n\to\infty$.
Using the relation
$$
\mu^2 - \lambda \mu +(d-1) = 0,
\quad\mbox{or}\quad
\lambda= \mu + (d-1)/\mu,
$$
we see that this gives that for any $\epsilon>0$ we have
\begin{equation}\label{eq:broder_shamir_hashimoto}
\probb{\sup_{i>1} |\lambda_i(G)| 
> (d-1)^{3/4}+(d-1)^{1/4}+\epsilon }
\end{equation}
tends to zero.
The above is an improvement over \eqref{eq:broder_shamir_theorem}, which
for large $d$ is a multiplicative 
factor of roughly $2^{1/2}$.

Similarly, for a value of $r$ satisfying \eqref{eq:r_bound}, we can
improve the result of Friedman \cite{friedman_random_graphs}
to obtain
\begin{equation}\label{eq:friedman_hashimoto}
\probb{\sup_{i>1} |\lambda_i(G)| 
> (d-1)^{(r+2)/(2r+2)}+(d-1)^{r/(2r+2)}+\epsilon }
\end{equation}
For large $d$, this represents an improvement over
the bound in \cite{friedman_random_graphs} of a multiplicative factor
of roughly
$2^{r/(r+1)}$.

\subsection{Tangles and the Limitations of the Trace Method}
\label{sb:limitations}

Here we define one of the main concepts in this article and
in \cite{friedman_alon}.
Friedman, in \cite{friedman_alon}, introduced various notions of
``tangles'' in random graphs, to remedy various shortcomings of the
trace method.  Let us introduce the basic notions.
Throughout this subsection,
we work with a fixed connected graph, $B$, without half-loops, and
assume that $\chi(B)\le -1$; in this case the Broder-Shamir model,
$\cC_n(B)$, is formed via $|E_B|$ independent and uniformly chosen
permutations on $\{1,\ldots,n\}$.
We begin with a somewhat technical point.

\begin{definition} 
Let $B$ be a graph without half-loops.
We let the 
{\em \gls[format=hyperit]{subgraphs occurring in a B covering}}, 
\newnot{symbol:OccursB}
denoted 
\gls[format=hyperit]{occurs-B}
to be
those graphs, $\psi$, such that for some $n\ge 1$, $\psi$ occurs as
a subgraph of some $G\in\cC_n(B)$ (i.e., a graph that occurs with 
positive probability in $B$).
\end{definition}

The following observation is not strictly needed for our main theorems,
but sheds light on the set $\Occurs_B$.

\begin{proposition}\label{pr:etale_is_feasible}
Let $B$ be a graph without half-loops.  Then 
$\Occurs_B$
is precisely
the set of graphs, $\psi$, that admit an \'etale map to $B$.
\end{proposition}
\begin{proof}
If $\psi$ occurs as a subgraph of some $G\in\cC_n(B)$, then
the natural projection $G\to B$ gives an \'etale graph morphism 
$\psi\to B$.
The converse was proven in Proposition~\ref{pr:etale_factorization}.
\end{proof}

\begin{definition} 
\label{de:order}
By the 
{\em \gls[format=hyperit]{order}} 
of a graph, $\psi$, we mean
\newnot{symbol:order_psi}
$$
\ord(\psi) = -\chi(\psi) = |E_\psi| - |V_\psi|.
$$
\end{definition}

\begin{definition}
\label{de:tangle_of_B}
Let $B$ be a connected graph of negative Euler
characteristic.  By a 
{\em \gls[format=hyperit]{tangle of B}} or simply
{\em \gls[format=hyperit]{B-tangle}}
we mean a connected graph,
$\psi\in\Occurs_B$, for which
\newnot{symbol:rho_root}
$$
\rho(H_\psi) \ge \rhoroot B,
$$
where $\rhoroot B$ denotes the square root of $\rho(H_B)$;
furthermore, we say that a $\psi$ as above is a
\gls[format=hyperit]{strict tangle of B}
if $\rho(H_\psi)>\rhoroot B$.
We use 
\newnot{symbol:Tangle-B}
$\Tangle_B$ to denote the set of $B$-tangles, and
\newnot{symbol:Tangle-r-B}
$\Tangle_{<r,B}$ to denote those tangles of order less than $r$.
Similarly, for any $\epsilon>0$ we define a 
{\em \gls[format=hyperit]{B-eps-tangle}}
to be 
those graphs, $\psi$,
for which
$$
\rho(H_\psi) \ge \rhoroot B + \epsilon,
$$
and use the notation
\newnot{symbol:Tangle-B-eps}
$\Tangle_{B,\epsilon}$ and 
\newnot{symbol:Tangle-r-B-eps}
$\Tangle_{<r,B,\epsilon}$
analogously.
\end{definition}

We wish to make a few remarks regarding these definitions.
If $\psi$ is a connected, non-empty graph without half-loops,
then we will see that:
\begin{enumerate}
\item if $\ord(\psi)<0$ then $\ord(\psi)=-1$, in which case
$\psi$ is a tree (we consider an ``isolated vertex,'' i.e., the graph
with one vertex and no edges, to be a tree), and in this case
$\rho(\psi)=0$; 
\item if $\ord(\psi)=0$ then $\psi$ is a {\em homotopy cycle} (i.e., 
homotopic to a connected graph, all
of whose vertices have degree two), and in this case $\rho(H_\psi)=1$; and
\item if $\ord(\psi)>0$ then $\rho(H_\psi)>1$.
\end{enumerate}
Hence if $B$ is connected and
$\ord(B)\ge 1$, then any $B$-tangle must have order at least one.
Notice that if $G'\subset G$, i.e., $G'$ is a subgraph of $G$, then
for any non-negative integer, $k$, we have
$$
\Tr(H_{G'}^k)\le \Tr(H_G^k),
$$
since these traces count strictly non-backtracking closed walks in,
respectively, $G'$ and $G$; hence
\begin{equation}\label{eq:rho_comparison}
\rho(H_{G'})\le \rho(H_G).
\end{equation}
In particular, for any graph, $\psi$ with
$$
\rho(H_\psi) > \rhoroot B,
$$
we have that the set of all $G$ that contain $\psi$ as a subgraph have
$\rho(H_G)$ bounded below and away from $\rhoroot B$.

From the above remarks we see that if $G$ contains a 
$(B,\epsilon)$-tangle, $\psi$, then $\rho(H_G)$ is bounded away
from $\rhoroot B$.  It turns out that, for this reason, 
to prove the relativized Alon conjecture for $B$, we will
want to compute expected traces of powers of $H_G$ where we first
discard any graph, $G$, that contains a $(B,\epsilon)$-tangle,
for arbitrarily small $\epsilon>0$.  It turns out that, for technical
reasons, it is easier to discard all $B$-tangles, and in
Chapter~\ref{ch:p1} we will do so.
In Section~\ref{se:p2-fund-exp}, where we refine
Theorem~\ref{th:main_Alon} to obtain the more precise
Theorem~\ref{th:main_Alon_Ramanujan}, we will need to work with
$(B,\epsilon)$-Tanlges for small $\epsilon>0$.

The article \cite{friedman_alon} demonstrates two main points about
$B$-tangles, where $B=W_{d/2}$, 
\gls{bouquet whole}
(with $d$ even): first,
the trace method that applies a single value of $k$ for each $n$
in
$$
\EE_{G \in \cC_n(W_{d/2})}[ \Tr(A_G^k) ] 
\quad\mbox\quad
\EE_{G \in \cC_n(W_{d/2})}[ \Tr(H_G^k) ] 
$$
cannot yield the Alon conjecture, due to certain $B$-tangles;
the second point---which
which comprises most of the work in \cite{friedman_alon}---is
that a trace method which removes graphs with tangles from the expected
values above can be adapted to yield the Alon conjecture.
In this paper we show that the same is true when $W_{d/2}$ is replaced
with any $d$-regular, connected graph with $d\ge 3$.
Let us now give an overview of the our methods.

\jfignore{

**********

expected number of non-backtracking closed walks.  So instead of
\eqref{eq:broder_shamir_est}, one can more directly derive with their
analysis that
$$
\EE_{G \in \cC_n(W_{d/2})}[ \Tr(H_G^k) ] \le (d-1)^k +
C k^C (d-1)^{k/2}+ Ck^C (d-1)^k n^{-1}
$$
for $k/n$ sufficiently small.  However

********

bound of , and therefore gives
a bound for subdominant Hashimoto eigenvalues.
estimate by estimating Hashimoto matrix

the $P_i=P_i(k)$ are functions of $k$ (and $d$, viewed as fixed),
and where the term ${\rm err}(n,k)$ satisfies the bound
$$
|{\rm err}(n,k)|\le f_r(k) n^{-r}
$$
for a function $f_r(k)$ depending only on $k$ (and $d$).  This enables
one to get ``high probability'' bounds on all $\lambda_i(d)$ with $i<1$,
provided that one gets sufficiently good estimates on the $P_i(k)$ and
$f_r(k)$.  Broder and Shamir in \cite{broder} dealt with the
case $r=1$; there they showed $P_{-1}(k)$ is the number of closed
walks of length $k$ originating at any vertex on the
$d$-regular tree; that $f_r(k)$ can be taken to be
$$
f_r(k) = C k^2 d^k 
$$
for a constant $C$ (depending on $d$); and that $P_0(k)$ satisfies the bound
$$
P_0(k) \le d^k + C k^C \Bigl( 2\sqrt{d-1} \Bigr)^k.
$$
For each value of $n$, Broder and Shamir then choose a single
value of $k$ so that the estimates for
$P_{-1}(k)n$ and ${\rm err}(n,k)$ roughly balance.
Friedman \cite{friedman_random_graphs} builds on \cite{broder} to
obtain estimates on $P_1(k),\ldots,P_{r-1}(k)$ and $f_r(k)$ to show
that with probability tending to one, a 
$G\in\cC_n(W_{d/2})$ will have
$$
\sup_{i>1} |\lambda_i(G)| \le d^{1/(r+1)} 
\Bigl( 2\sqrt{d-1}\Bigr)^{r/(r+1)} + \epsilon;
$$
however, Friedman's methods only obtains such estimates for
$r$ up to
Friedman was able to show that $P_1(k),\ldots,P_{r-1}(k)$ were bounded
by $(2\sqrt{d-1})^r$
We wish to make a few remarks on the above strategy.

\begin{equation}\label{eq:broder_shamir_etc}
\EE_{G \in \cC_n(B)}[ \Tr(A_G^k) ] = 
P_{-1}(k) n + P_0(k) + P_1(k) n^{-1} + \cdots + P_{r-1}(k)
n^{1-r} + {\rm err}(n,k),
\end{equation}

\subsubsection{Improvement via the Hashimoto Matrix}

\subsubsection{Explicit Versus Inexplicit Calculations}

\subsubsection{Limitation of One $k$ Value Method}

First, \cite{broder} gives an explicit estimate of $P_{-1}(k)$, 
$P_{0}(k)$, and of $f_r(k)$.  However, \cite{friedman_random_graphs}
obtains results on $P_1(k),\ldots,P_{r-1}(k)$
$P_{-1}(k)$
$B=W_{d/2}$ is fixed), and where the $O(n^{-r})$

The idea is to obtain a $B$-asymptotic expansion 
(as in Definition~\ref{de:asymptotic}) for the expected trace
of a power of the adjacency matrix,
\begin{equation}\label{eq:broder_shamir_etc}
\EE_{G \in \cC_n(B)}[ \Tr(A_G^k) ] = 
P_{-1}(k) n + P_0(k) + P_1(k) n^{-1} + \cdots + P_{r-1}(k)
n^{1-r} + {\rm err}_r(n,k) 
\end{equation}
with the error term in the above satisfying a bound as in
\eqref{eq:basic_error_bound}.

Alon's original conjecture regards the case where $B$ is $d$-regular and
$V_B$ consists of a single vertex.
In this case, for the Broder-Shamir model and several related models,
$P_{-1}(k)$ is just the number of closed walks of length
$k$ from a fixed on the $d$-regular infinite tree, and it is known that
$$
\lim_{k\to\infty} \frac{\log P_{-1}(k)}{\log k} = 2\sqrt{d-1}.
$$
(see, for
example, \cite{kesten_symmetric_random,cartier,
mckay_spanning}).
Furthermore,
Broder and Shamir \cite{broder}, in the case of $d$ even 
and $B=W_{d/2}$, the bouquet of $d/2$ whole-loops,
that an asymptotic expansion of order $1$ exists,
with $P_0(k)$ being $d^k$ plus an error term roughly of size at
most $(2(d-1)^{1/2})$, which implies that
for any $\epsilon>0$ and $k$ even of size roughly
$$
2 \log_\theta n, \quad\theta=\theta_d=d/\Bigl( 2\sqrt{d-1}\Bigr)
$$
we have that for any $\epsilon>0$ we have
$$
\EE_{G \in \cC_n(W_{d/2})}[ \bigl( \rhonew_B(A_G) \bigr)^k ] \le
\Bigl( d\; 2\sqrt{d-1}+ \epsilon \Bigr)^{k/2}
$$
for $n,k$ sufficiently large, and hence we conclude the Alon conjecture
for $d$-regular graphs, weakened by replacing $\rho(A_{\widehat{W_{d/2}}})
=2\sqrt{d-1}$ with
$$
\Bigl( d\; 2\sqrt{d-1}\Bigr)^{1/2} 
$$
(\cite{broder} state a slightly weaker bound since they use a cruder
estimate for $P_{-1}(k)$).
In \cite{friedman_random_graphs} (see discussion around
Theorem~2.8 there), it is noted that if an
asymptotic expansion exists to order $r$, and if 
$P_2,\ldots,P_{r-1}$ as in \eqref{eq:broder_shamir_etc} are
$W_{d/2}$-Ramanujan with vanishing principle part, then we get a similar
bound with $\rho(A_{\widehat{W_{d/2}}})
=2\sqrt{d-1}$ replaced by
$$
d^{1/(r+1)} \Bigl( 2\sqrt{d-1} \Bigr)^{r/(r+1)}.
$$
If this could be done for all integers, $r$, then this would yield
the Alon conjecture; however Friedman \cite{friedman_random_graphs}
was only able to establish
this for $r$ of size proportional to $\sqrt{d}$, and later in
\cite{friedman_alon} showed that this cannot hold when $r$ is large
(roughly proportional to $\log d\sqrt{d}$ for $d$ large).
This is due to certain low probability events, ``tangles'' 
(see \cite{friedman_alon}) which force a graph to have an eigenvalue
bounded above and away from $2\sqrt{d-1}$.
Furthermore these ``tangles'' show that one cannot achieve the
Alon conjecture based solely on a ``one time application'' of
the trace method; i.e., for every $d$ there is an $\epsilon_0=\epsilon_0(d)>0$
for which for any even $k$ we have
$$
\EE_{G \in \cC_n(W_{d/2})}[ \bigl( \rhonew_B(A_G) \bigr)^k ] \ge
\Bigl( 2\sqrt{d-1} + \epsilon_0(d) \Bigr)^r.
$$

In this model of $\cC_n(B)$, similar expected trace calculations
were performed for arbitrary $B$ by
Friedman \cite{friedman_relative}, giving an asymptotic expansion to
order $1$, then Linial-Puder \cite{linial_puder}, essentially giving
an asymptotic expansion to order $2$, and finally the recent work of Puder
\cite{puder} giving higher order expansions.

The trace method of \cite{friedman_alon} and this article differs from
the above methods in that (1) graphs with tangles are removed from
the trace estimates, and (2) for any value of $n$, one applies the
analogous asymptotic expansions for many values of $k$ simultaneously;
intuitively, one applies a polynomial of the operator ``shift in $k$''
to the asymptotic expansion, in order to annihilate the principle
part of the coefficients (functions in $k$, as in
Definition~\ref{de:asymptotic}) of the asymptotic
expansion.

\subsection{Hashimoto Versus Adjacency Trace Methods}

Before describing our trace methods, which most resemble those
of \cite{friedman_alon}, we wish to make a few remarks on trace methods
and the Hashimoto matrix.

First, Broder-Shamir noticed that it was, in a sense, much easier to
analyze the expected number of closed non-backtracking walks, as 
opposed to closed general walks
(see Section~\ref{se:blah} for an explanation).  
So they analyzed closed non-backtracking
walks, and translated those results into results on the number of 
closed walks of a given length in a random graph, thereby estimating
the expected value of 
$\Tr(A_G^k)$ for a random graph, $G$.
Their methods are even simpler (e.g., there are fewer ``types'') if one
restricts one's attention to closed, strictly non-backtracking 
walks.  As such, their methods give, more directly, an estimate of
the expected values of traces of powers of $H_G^k$.
This is also true of the works
\cite{friedman_random_graphs,friedman_relative,friedman_alon}.

Second, it seems that working with $H_G$ traces gives better results
on the Alon conjecture, in the following sense.
Consider the case of $B$ being $d$-regular, and
consider the results of 
\cite{broder,friedman_random_graphs,friedman_relative}, where
an asymptotic expansion of order $r$ is proven, for certain values of
$r$, for the expected value of traces of powers of $A_G$.  Generally
such an expansion is obtained by estimating the expected number of
non-backtracking walks, which imply an order $r$ asymptotic expansions
of the form
$$
\EE_{G \in \cC_n(B)}[ \Tr(H_G^k) ] = 
P_0(k) + P_1(k) n^{-1} + \cdots + P_{r-1}(k)
n^{1-r} + {\rm err}_r(n,k) 
$$
with an error term bounded by $Ck^C(d-1)^k$; notice that in the Hashimoto
type trace, the $nP_{-1}(k)$ is not present
(this is true of the works
\cite{broder,friedman_random_graphs,friedman_relative}).
The fact that all non-real eigenvalues of $H_G$ have absolute value
at most (actually, equal to) $\sqrt{d-1}$, and that there are at most
$2|V_G|=2n|V_B|$ such eigenvalues, implies that
$$
x^k\;\biggl(\prob{G\in \cC_n(B)}{\rhonew_B(H_G) \ge x}\biggr)
$$
$$
\le 
\EE_{G \in \cC_n(B)}[ \Tr(H_G^k) -\Tr(H_B^k) ] 
+2n|V_B|(d-1)^{k/2},
$$
for any even $k$ and real $x> (d-1)^{1/2}$.
It follows that an asymptotic expansion of order $r$
in the Hashimoto matrix expected traces implies that
for any $\epsilon>0$ we have
$$
\prob{G\in \cC_n(B)}{\rhonew_B(H_G) \ge \beta(d,r) + \epsilon}
$$
tends to zero as $n\to\infty$, where
$$
\beta = \beta(d,r) = (d-1)^{(r+2)/(2r+2)}.
$$
However this implies that
$$
\prob{G\in \cC_n(B)}{\rhonew_B(A_G) \ge \alpha(d,r) + \epsilon}
$$
tends to zero as $n\to\infty$, where
$$
\alpha(d,r) = \beta(d,r) + (d-1) \bigl( \beta(d,r) \bigr)^{-1}
= (d-1)^{(r+2)/(2r+2)} + (d-1)^{r/(2r+2)} 
$$
This method of estimation gives a slight improvement over the
estimates of the previous subsection.
For example, this improves the original Broder-Shamir estimate of
\cite{broder} (assuming one uses the best bound for $P_{-1}(k)$ in this
article) from
$$
\Bigl( d\; 2\sqrt{d-1}\Bigr)^{1/2} 
$$
to
$$
(d-1)^{3/4} + (d-1)^{1/4};
$$
for $d$ large, this is an improvement by a multiplicative factor
of roughly $\sqrt{2}$.

\subsection{Our Main Trace Estimates}

At this point we describe our main asymptotic expansion theorems,
generalizing the results of \cite{friedman_alon}, and make some
remarks on how to apply them.

We shall give our primary definition of a {\em tangle}.

\begin{definition}
By the {\em Euler characteristic} of a graph, $\psi$, we mean
$$
\chi(\psi) = |V_\psi| - |\Edir_\psi|/2,
$$
i.e., $|V_\psi| - |E_\psi|$
with the convention that each half-loop contributes $1/2$ to the
edge count $|E_\psi|$.
By the {\em order} of a graph, $\psi$, we mean $-\chi(\psi)$.
(We often use ``order'' only in the context of connected graphs.)
\end{definition}

\begin{definition}
Let $B$ be an arbitrary graph.
A {\em $B$-tangle} is a connected graph, $\psi$, for which
$$
\rho(H_\psi) \ge \rho(H_{\widehat B}).
$$
For any positive integer, $r$, let 
$\HasTangle(B,r)$ be the set of graphs which contain a $B$-tangle, $\psi$,
of order strictly less than $r$, let
$\TF(B,r)$ be set of all {\em $(B,r)$-tangle free}
graphs, i.e., graphs not in $\HasTangle(B,r)$.
\end{definition}

Our main theorem regarding asymptotic expansions is the following.

\begin{theorem}
\label{th:main_asymptotic}
Let $B$ any connected graph.  For the Broder-Shamir
model $\cC_n(B)$, consider
$$
f(n,k)=\expect{G\in\cC_n(B)}{\II_{\TF(B,r)}(G)\;\Tr(H_G^k) },
$$
where $\II_{\TF(B,r)}(G)$ is the indicator function that
$G$ lies in $\TF(B,r)$.
The above function of $n$ and $k$
has an asymptotic expansion in the sense
of Definition~\ref{de:asymptotic}
to order $r$, with error term bounded by
$Ck^C(d-1)^kn^{-r}$ and coefficients that are $B$-Ramanujan.
\end{theorem}

Unfortunately, at present we cannot prove that
the principal part of the coefficients in the above expansion,
beyond the term of degree $0$ 
need vanish.
(The methods of \cite{friedman_relative} show that 
the term of degree $0$, modulo error terms of order
$Ck^C(d-1)^{k/2}$, equals $\Tr(H_B^k)$.)
However, we can use this result for a fixed value of $n$, applied
to numerous values of $k$, to conclude the Alon conjecture.
This type of argument is called side-stepping in \cite{friedman_alon},
and is proven there by applying an appropriate polynomial
of the ``shift in $k$'' operator, $S$, defined as
$$
(Sf)(n,k) = f(n,k+1),
$$
to the asymptotic expansion; the basic idea is
that for any $\mu\in \complex$, 
a sufficiently high power of $S-\mu$ will annihilate any
term of the form $\mu^k$ times a polynomial in $k$ times a power of $n$.
Products of powers of $S-\mu$, where $\mu$ ranges over the 
set of bases of the coefficients, will therefore annihilate all the
principal parts of the polyexponentials involved.
We
shall prove a variant of the ``Side-stepping Lemma''
of \cite{friedman_alon}, which allows for general polyexponential
coefficients; the coefficients in the asymptotic expansions 
in \cite{friedman_alon} had principle parts over the one element
base set $\{d-1\}$; that fact that our set of bases may contain more
than one element makes the proof more complicated.

Furthermore, when $B$ is Ramanujan, then the set of bases will
contain only $d-1$, and it is easy to use standard theorems on
expansion to prove that the principle parts, involving $(d-1)^k$
times a polynomial in $k$, roughly speaking can only arise in tangle
free graphs when the covering graph is disconnected.
This leads to a more refined conclusion, i.e.,
Theorem~\ref{th:main_Alon_Ramanujan}, when the base graph, $B$,
is $d$-regular and Ramanujan, as was done in \cite{friedman_alon}
in the case where $B$ has only one vertex.

\subsection{Remarks on Modified Traces}

One idea in \cite{friedman_alon} can be described as introducing
a ``modified trace,'' which in \cite{friedman_alon} was the
rather cumbersome
{\em selective trace}, and which we will simplify with our
{\em certified trace}.  Let us give the main idea.

In \cite{friedman_alon}, it is shown that asymptotic expansions of
$$
\expect{G\in\cC_n(B)}{\Tr(A_G^k) }
$$
fails to have $B$-Ramanujan coefficients beyond an order, $r$,
of size of order $d^{1/2}\log d$.  For similar reasons, we expect
the similar problems with
$$
\expect{G\in\cC_n(B)}{\Tr(H_G^k) }.
$$
Indeed, even at order $r$ roughly proportional to $\sqrt{d}$,
the technical problem arises that, in rough terms, certain formulas for the
coefficients, $P_i(k)$, of the asymptotic expansion, are given via
infinite sums which fail to converge for $i$ of order $\sqrt{d}$.
However, it turns out that we do get convergence, for coefficients in
an expansion to order $r$ for any $r$,
provided we remove from the $\Tr(H_G^k)$ expectation the contribution
of graphs, $G$, which have tangles of order less than $r$.
One can show this by modifying $\Tr(H_G^k)$ to a count of
closed, strictly non-backtracking walks in $G$ of length $k$,
which do not trace out a tangle.

Let us formalize the notion of a modified trace, to give a general
view of the idea, and to compare the {\em selective trace}
\cite{friedman_alon} with the {\em certified trace} of this article.
Let us give a rather bare definition.

\begin{definition}
\label{de:modified_Hashimoto_trace}
By a {\em modified Hashimoto trace} we mean a function,
$\MT=\MT(G,k,r)$, taking a graph, $G$, and a positive integers, $k,r$, 
returning
a non-negative integer, for which
\begin{enumerate}
\item For each graph, $G$, and positive integers, $k,r$, we have
that 
$$
\MT(G,k,r) \le \SNB(G,k,r),
$$
where $\SNB(G,k,r)$ denotes the number of strictly non-backtracking walks
in $G$ of length $k$ and of order less than $r$;
\item for each $G\in\TF(B,r)$, and any $k,r$, we have that
$$
\MT(G,k,r) = \SNB(G,k,r),
$$
\item for a fixed graph, $B$, and fixed positive integer, $r$,
the Broder-Shamir model, we have
$$
F(n,k)=\expect{G\in\cC_n(B)}{\MT(G,k,r)},
$$
as a function of $n$ and $k$ (with $B$ and $r$ fixed)
has an asymptotic expansion in the sense
of Definition~\ref{de:asymptotic}
to order $r$, with error term bounded by
$Ck^C(d-1)^kn^{-r}$ and coefficients that are $B$-Ramanujan; and
\item
the same is true for $F(n,k)$ replaced by
$$
\widetilde F(n,k) =
\expect{G\in\cC_n(B)}{\II_{\HasTangle(B,r)}(G) \MT(G,k,r)} ,
$$
where $\II_{\HasTangle(B,r)}(G)$ is the indicator function of the
$\cC_n(G)$ event $\HasTangle(B,r)$.
\end{enumerate}
\end{definition}
Of course, the last two items of the above definition imply that
$$
\expect{G\in\cC_n(B)}{\II_{\TF(B,r)}(G) \MT(G,k,r)} = F(n,k)-
\widetilde F(n,k)
$$
also has an asymptotic expansion of the same type as do $F$ and 
$\widetilde F$; but, by item~(2), we also have
$$
\expect{G\in\cC_n(B)}{\II_{\TF(B,r)}(G) \MT(G,k,r)} 
=
\expect{G\in\cC_n(B)}{\II_{\TF(B,r)}(G) \SNB(G,k,r)} ,
$$
and the latter expectation can be shown to differ from
$$
\expect{G\in\cC_n(B)}{\II_{\TF(B,r)}(G) \Tr(H_G^k) }
$$
by an error term of order $Ck^C \rho(H_B)^k n^{-r}$
(see Theorem~\ref{th:blah}).
Hence, it is relatively easy to establish the following theorem.

\begin{theorem}\label{th:modified_Hashimoto_trace}
Assume, possibly for that a modified Hashimoto trace exists, in the sense of
Definition~\ref{de:modified_Hashimoto_trace}.  
Then Theorem~\ref{th:main_asymptotic} holds.
In other words, for each $r$, we have that
\begin{equation}\label{eq:tangle_free_trace}
\expect{G\in\cC_n(B)}{\II_{\TF(B,r)}(G) \Tr(H_G^k) }
\end{equation}
has an asymptotic expansion in the sense
of Definition~\ref{de:asymptotic}
to order $r$, with error term bounded by
$Ck^C\rho(H_B)^kn^{-r}$ and coefficients that are $B$-Ramanujan.
\end{theorem}

We remark that there are numerous variants of the above.  For example,
we can, more generally, speak
of a modified Hashimoto trace that depends on $B$, although here we
will not need this.  We can similarly speak of a modified adjacency trace,
and we can consider modified traces for certain values of $B$ and $r$.
We can (and later will, in Section~\ref{se:p2-fund-exp}) use variants
of $\TF(B,r)$ and $\HasTangle(B,r)$ in the above.

Notice that the conclusion of
Theorem~\ref{th:modified_Hashimoto_trace} does not involve the modified
trace; we require only the modified trace only to establish an asymptotic
expansion for the quantity in \eqref{eq:tangle_free_trace}.

The majority of the work in this article, as in
\cite{friedman_random_graphs,friedman_alon} is to
choose a particular candidate for a modified trace, and then
establish items~(3) and (4) of Definition~\ref{th:modified_Hashimoto_trace}.
Generally speaking, a modified trace works by counting non-backtracking
random walks and discarding walks that exhibit ``some evidence'' of
a tangle.  The selective trace relies on somewhat indirect evidence,
whereas the certified trace simply discards those walks which trace
out a graph whose Hashimoto spectral radius is too large.
In either case, items~(1) and (2) are pretty immediate.
Furthermore, item~(4), in both cases, is pretty much a consequence
of the methods of establishing item~(3).  So the heart of the method
is to establish item~(3).

\begin{definition} 
For any integers $k,r\ge 1$ and graph, $G$, we define
the {\em certified trace of $G$ of length $k$, truncated to order $r$} to be
$\CertTr_{<r}(G,k)$ to be the number of walks, $w$, in $G$ such that 
$\Graph(w)$ is of Euler characteristic greater than $-r$ and the
Hashimoto spectral radius of $\Graph(w)$ is less than
$\rho(H_{\widehat B})$.
\end{definition}

Our main theorem is to show that the above certified trace is a modified
trace.

\mynote{Make some comments comparing selective versus certified trace?}

\subsection{Types, the Certified Trace,
Noetherian Partial Orders, and Finiteness}

In this subsection we shall explain why the certified trace is a useful
idea.

Throughout this article, the main point is that any finite linear
combination of asymptotic expansions, as in
Definition~\ref{de:asymptotic}, is again an asymptotic expansion.
Hence to prove that the certified trace has an asymptotic expansion,
it is enough to write this modified trace as a linear combination of
a finite number of related sums.

This finiteness technique already appear in \cite{broder}.  There they
argue that any closed, non-backtracking walk, $w$, that is of Euler
characteristic $0$, one of
two general ``shapes:'' a cycle, or a cycle with a single ``tail.''
If one restricts this to strictly non-backtracking types, then only
a cycle is possible.  This is generalized in
\cite{friedman_random_graphs,friedman_relative,friedman_alon} by
the notion of a {\em type} of a walk.
Roughly speaking,
the type of the walk remembers only the initial vertex of the walk,
all vertices of degree at least three, and some information about
how the neighbourhood of each vertex is mapped to $B$, and the order
in which vertices
and edges are first encountered in the walk.
In other words, any vertex of degree two that is not the initial (and
final) vertex is ``contracted,'' leaving a graph, $T=T(w)$,
of which all but at
most one vertex has degree three.

The Perron-Frobenius eigenvalue of $H_{\Graph(w)}$ can be inferred from
knowing the underlying graph, $T$, of the type, plus a vector of
positive integers indexed by $E_T$, $\vec k=\{k_e\}_{e\in E_T}$.  If
$\VLG(T,\vec k)$ denotes the graph where each $T$ edge is replaced
with a ``\glspl{beaded path}''---whose interior vertices have degree
two in graph---of length $k(e)$.
Hence the certified trace, for each graph, $T$, of the underlying
type of the walk, (assuming that $T$ is of order less than $r$),
involves precisely those walks whose corresponding vector of lengths
of beaded paths, $\vec k$, lies in 
$$
K(T) = \{ \vec k\from E_T\to\integers_{\ge 1} \;|\;
\rho(H_{\VLG(G,\vec k)})<\rho(\widehat B) \}.
$$
The point is that $K(T)$ is an upper set, and any 
upper set is finitely generated.
Let us make this precise.

Recall that a poset, or partially ordered set, is a tuple $(P,\le)$
consisting of a set, $P$, and a relation $\le$ which is reflexive
(i.e., $p\le p$ for all $p\in P$), transitive (i.e., $p_1\le p_2$
and $p_2\le p_3$ implies that $p_1\le p_3$), and
antisymmetric (i.e., sober as a topological space,
i.e., $p_1\le p_2$ and $p_2\le p_1$ implies that
$p_1=p_2$).  We write $p<q$ to mean $p\le q$ and $p\ne q$.

\begin{definition} 
Let $(P,\le)$ be a poset.  For $p\in P$ we define
the {\em cone at $p$} to be
$$
\Cone(p) = \{ t\in P \ |\ t\ge p \}.
$$
We say that $Q\subset P$ is an {\em upper set} if
$q\in Q$ implies $\Cone(q)\subset P$; i.e. $q\in Q$ and $q\le q'$
implies that $q'\in Q$.
We say that subset, $S$, of $P$ is {\em finitely generated} if
$$
S = \Cone(s_1) \cup \ldots \cup \Cone(s_m)
$$
for some $s_1,\ldots,s_m\in P$ (such a subset is clearly an
upper set).
\end{definition}

For every integer, $m$, let $\integers_{\ge 1}^m$ denote the set
$\integers_{\ge 1}^m$ under the partial order $k^1\le k^2$
iff $k^1_i\le k^2_i$ for all $i=1,\ldots,m$.
The ring $\complex[x_1,\ldots,x_m]$, of polynomials in indeterminates
$x_1,\ldots,x_m$ with complex coefficients, is Noetherian and, hence
all its ideals are finitely generated.
This easily yields the following theorem, which is crucial to the
certified trace.
\begin{theorem} 
Any upper set in $\integers_{\ge 1}^m$ is finitely generated.
\end{theorem}
\begin{proof}
For an upper set $U\subset \integers_{\ge 1}^t$, let $I$ be the
ideal generated by the elements $x^u$ over all $u\in U$
(where $x^u$ denotes $x_1^{u_1}\ldots x_m^{u_m}$).  
This ideal
is finitely generated by elements $f_1,\ldots,f_n$; let $U_0$ consist
of the finite number of $u$ for which $x^u$ appears with a nonzero
coefficient in some $f_i$.
If $u\in U$, then $x^u$ is a linear combination of the $f_i$, and
hence $u$ is greater, under the partial order, to at least one of the
$U_0$.  Hence $U_0\subset U$ generates $U$.
\end{proof}
Of course, the proof that $\complex[x_1,\ldots,x_m]$ is 
Noetherian usually proved
by induction, and one can adapt this argument to give a more direct
proof that $\integers_{\ge 1}^m$ has finitely generated upper sets.
We discuss this and other aspects of finite generation in
Section~\ref{se:p2-poset}.

It follows, my inclusion/exclusion, that to prove
Theorem~\ref{th:main_asymptotic}, it suffices to consider the expected
number of strictly non-backtracking of a given type and in a given
cone.  The point of restricting to a given type, a given cone, and
a, later, ``multiplicity pattern'' (see ????), is that we simplify
the computation, and yet the restrictions yield only a finite number
of cases to consider (of which the modified trace is a finite linear
combination).  Hence it suffices to establish an asymptotic expansion
for each restricted sum.

\subsection{Remarks on the ``Multiplicity Pattern''}

Let us describe another key simplification in this article
over \cite{friedman_random_graphs,friedman_alon}.

Again, we shall divide strictly non-backtracking closed walks in a random
graph, $G\in\cC_n(B)$, by their {\em type}, which collects
walks into a finite set of cases,  essentially remembering only
a graph, $T$, consisting of
the initial vertex of the walk and all vertices of degree three,
plus remembering some addition information.  In this case the walk
traces out a graph, $\Graph(w)$, in $G$, isomorphic to
$\VLG(T,\vec k)$, where $\vec k$ is a vector indexed on the edges,
$E_T$, of $T$, and whose components, $k(e)$, are the positive integers
giving the length of the \gls{beaded path} in $\Graph(w)$ corresponding
to an edge, $e\in E_T$.

One of the remarkable points of this view is that it is almost
immediate that the number of times
an edge of $\Graph(w)$ is traversed depends only the corresponding type
edge; indeed, a non-backtracking walk that begins on the first edge
of a beaded path must traverse each edge in that path once until reaching
the beaded path's end.  We therefore can view this number of times an
edge is traversed as a ``multiplicity,'' denoted $\vec m$, as a function of
the type edges, $\vec m$, again, as $\vec k$, 
a vector of positive integers indexed on $E_T$.
We may view $\vec m=\vec m(w)$ and $\vec k=\vec k(w)$ as 
functions of the walk under
consideration.
If the length of the walk is $k$, then
$$
k = \vec k \cdot \vec m,\quad\mbox{where}\quad
\vec k\cdot\vec m = \sum_{e\in E_T} k(e)m(e)
$$
i.e., the dot product of $\vec k$ with $\vec m$.

Note that if we count walks that are allowed to backtrack then the
edge multiplicity, i.e., number of times the edge is traversed,
can differ on two edges of the same beaded path.  This seems to make
the study of walks that can backtrack, i.e. a direct study 
modified traces of adjacency matrix, significantly more difficult,
and, in particular, beyond the scope of this paper.

We shall use $W(\vec m,\vec k)$ to denote the expected number of strictly
non-backtracking certified walks of a fixed type for $G\in\cC_n(B)$.
We note that if all components of $\vec m$ are at least two, then
the number of such walks should be roughly at most
$\rho(H_B)^{k/2}$, instead of $\rho(H_B)^k$, since every edge is
traversed at least twice.  
It is for this reason that in \cite{friedman_random_graphs,friedman_alon},
we divide the sum over all
$\vec m\in\integers_{\ge 1}^{E_T}$ by the values of $\vec m$; namely
we set $E_{T,1}$ to be edges $e\in E_T$ for which $m(e)=1$, and
let $E_{T,2}$ be the remaining edges, i.e., for which $m(e)\ge 2$.
Then the case $E_{T,2}=E_T$, i.e., the sum over all $\vec m$ with
$m(e)\ge 2$ for all $e\in E_T$, should be easy to bound as an error
term of the $B$-Ramanujan coefficients 
\mynote{(say this better)}.
Note that results in a finite number of sums, corresponding to the
number of partitions of $E_T$ as $E_{T,1}\amalg E_{T,2}$ as above.

In this article we make the following point.
If we fix any $M$, we then we can similarly divide the
sum over all $\vec m$ by partitioning $E_T$ into $M+1$ subsets,
$E_{T,1}\amalg \cdots \amalg E_{T,M+1}$, where
for $i\le M$, we fix $m(e)=i$ for $e\in E_{T,i}$, and 
for $e\in E_{T,M+1}$ we sum over all $m(e)\ge M+1$.
Therefore, in \cite{friedman_random_graphs,friedman_alon}, this
was done with $M=2$.  
It turns out for technical reasons, that the analysis is much simpler
if we allow $M$ to be any fixed integer; for any fixed, $M$, we get
a finite number of terms to sum over.
However, the infinite number of $\vec m$ values in such a sum,
in the cases where
$E_{T,M+1}$ is nonempty,
represent an infinite sum where convergence and desired bounds are
much easier to 
establish for fixed $M$ much larger than two.

To establish the first term, $P_0(k)$, in the asymptotic expansion of
a modified Hashimoto trace, which only involves the type whose graph
is a cycle, it is easy to obtain bounds just taking $M=2$.  However,
for terms $P_i(k)$---where $i$ is large enough so that tangles occur
with probability of order $n^{-i}$ or greater---it is seems likely
to us that taking $M$ large greatly simplifies the argument.
We also note that \cite{friedman_alon} contains an error in the
estimate of 
\mynote{(somewhere around Theorem~8.3, but where exactly???)}.

}

\section{Asymptotic Expansions and The Loop}
\label{se:p0-loop}

The Broder-Shamir result \cite{broder} of \eqref{eq:broder_shamir_constant}, 
for random graphs $\cC_n(B)$ with $B=W_{d/2}$
fixed (and $d$ even), has an analogue valid for all
$\cC_n(B)$, given in \cite{friedman_relative}.  
We shall need some of the tools used in \cite{friedman_relative},
specifically the tools used to prove Lemma~2.3 there.
In this section
we shall give review these tools and results, 
developing some in a more general
context that we need here.
Our discussion is a generalization of the discussion of a {\em loop}
in Section~5.2 of \cite{friedman_alon}; we remark that the term
{\em loop} in \cite{broder} was used differently, namely as
the number of {\em coincidences} in
\cite{friedman_random_graphs,friedman_relative,friedman_alon}
and here (which is one minus the Euler characteristic of the
graph of a walk in $G\in\cC_n(B)$).

Once we develop these tools, in the first part of this 
section, we will be in a better position
to explain
a number of concepts needed in this paper, such
as {\em $B$-Ramanujan functions} and {\em $1/n$-asymptotic expansions}.
Such an explanation is given in the latter part
of this section.

%

%
%
%

\subsection{The Expected Number of Loops}

In \cite{broder}, Broder and Shamir considered closed, non-backtracking
walks and classified them by the ``shape'' or ``type'' of the graph
that the walk traces out.  Let us give some formal definitions.

\begin{definition}\label{de:graph_of_walk}
Let
$$
w=
(v_0, e_1,
v_1, e_2, v_2, \ldots, e_k, v_k )
$$
be a walk in a graph, $G$ (so $v_i\in V_G$ and $e_i\in \Edir_G$).
We define the 
{\em \gls[format=hyperit]{graph of a walk}}, $w$,
denoted 
\newnot{symbol:graph_of_walk}
${\rm Graph}(w)$,
to be the subgraph of $G$ consisting
of the vertices and edges occurring in $w$.
\end{definition}

We shall review the fundamental calculation of Broder and Shamir,
adapted by Friedman in \cite{friedman_relative} for the model
$\cC_n(B)$.

\begin{definition}
Let $w$ be a walk in a graph, $G$.  We say that $w$ is a
{\em \gls[format=hyperit]{loop}} 
if it is a strictly non-backtracking closed walk such that
${\rm Graph}(w)$ is a cycle, i.e., a connected graph such that
all vertices have degree two.
\end{definition}

We remark that if $w$ is a strictly non-backtracking closed walk in $G$,
then each vertex in ${\rm Graph}(w)$ has degree at least two;
hence, either $w$ is a loop, or ${\rm Graph}(w)$ is a graph of
negative Euler characteristic.
In this section we will prove the following result.

\begin{theorem}\label{th:loop_count}
Let $B$ be a connected graph of negative Euler characteristic.
Let $k$ be a positive integer, and let $m$ be the smallest divisor
of $k$ that is greater than one.
Then expected number of loops of length $k$ in a graph of
$\cC_n(B)$ is
$$
\Tr(H_B^k) + O(k^2/n) \Tr(H_B^k) + O(k) \Tr\Bigl(H_B^{k/m}\Bigr)
$$
for $k^2/n$ sufficiently small,
where this smallness and constants in the $O(\ )$ notation depend only on $B$.
\end{theorem}

\begin{theorem}\label{th:broder_shamir_friedman}
Let $B$ be a connected graph of negative Euler characteristic.
Then for any $k,n$ with $k^2/n$ sufficiently small we have
$$
\EE_{G\in\cc_n(B)}\Bigl[ \Tr(H_G^k) \Bigr] =
\Tr(H_B^k) + O(k^4/n) \Tr(H_B^k) + O(k) \Tr\Bigl(H_B^{k/m}\Bigr),
$$
where $m$ is the smallest divisor of $k$ greater than one.
\end{theorem}

Again, the proofs of these theorems 
follows the methods of \cite{friedman_relative}; but here we 
generalize some of these methods in a form that we will need in
this article.

\subsection{Proof of Theorem~\ref{th:loop_count}}
\label{sb:loop_count_proof}

In this section we prove Theorem~\ref{th:loop_count}.

For each loop, $w$, of length $k$
in a $G\in\cC_n(B)$, ${\rm Graph}(w)$ is a cycle in $G$ whose
length, $k'$, divides $k$.  Furthermore, the first $k'$ steps of $w$,
which we denote $w'$,
is a strictly non-backtracking closed walk that determines $w$ and
${\rm Graph}(w)$.  (Notice that it is crucial that $w$ is non-backtracking
here.)
Also, $w'$ traces out $k'$ distinct edges and vertices in $G$ to form
${\rm Graph(w)}$, and the $j$-th vertex of $w'$ is a tuple $(i_j,v_j)$, where
$i_j\in \{1,\ldots,n\}$ and $v_j\in V_B$; and the projection of $w'$ to
$B$ is a strictly non-backtracking closed walk in $B$.

Let us now work backwards: consider a strictly non-backtracking closed walk
$w''$ in $B$ of length $k'$ where $k'$ divides $k$, and let us consider
what is the expected number of walks, $w$, in $G\in\cC_n(B)$ whose
projection to $B$ is $w''$.  The expected number of walks is the 
expected number of $i_j\in\{1,\ldots,n\}$, with $j=0,\ldots,k'-1$
such that
\begin{enumerate}
\item the vertices $(i_j,v_j)$, $j=0,\ldots,k'-1$ are distinct, and
\item for each $j=0,\ldots,k'-1$ the vertex $(i_j,v_j)$ is connected
to the vertex $(i_{j+1},v_{j+1})$ be the $(j+1)$-th edge of $w''$.
\end{enumerate}
The exact formula is given in Proposition~\ref{prop:EsymmProd};
here it suffices to give crude upper and lower
bounds the desired expected value as such:
the vertices $i_j$ clearly take on at most $n^{k'}$ values, and clearly
at least 
$$
n (n-1) \ldots (n-k'+1)
$$
values (simply by choosing $i_0,\ldots,i_{k'-1}$ to be distinct integers).
The probability that they are connected by the desired edges is
greater than $n^{-k'}$ (which would hold exactly if $w''$ consisted of
distinct edges of $B$), and at most
$$
\Bigl( n (n-1) \ldots (n-k'+1) \Bigr)^{-1}
$$
(which would hold if all edges of $w''$ were the same edge of $B$, which
would necessarily be a whole-loop traversed in the same direction).
Hence this expected value is between
$$
1 \Bigl( 1 - (1/n) \Bigr) \ldots \Bigl( 1 - (k'-1)/n \Bigr)
$$
and the reciprocal of the above expression.  But inclusion/exclusion 
on $k'$ events with probabilities $i/n$ with $i=0,\ldots,k'-1$ shows
that 
$$
1 \Bigl( 1 - (1/n) \Bigr) \ldots \Bigl( 1 - (k'-1)/n \Bigr)
$$
$$
\ge 1 - \Bigl(1+2+\cdots + (k'-1)\Bigr)/n  \ge 1 - (k')^2/n;
$$
hence its inverse is at most
$$
1 + (k')^2/n
$$
for $k^2/n$ sufficiently small.

Hence the total number of expected loops in $G\in\cC_n(B)$ is
$$
1 + O(k')^2/n
$$
summed over the number of strictly non-backtracking closed walks in $B$ of 
length $k$.  Since the divisors, $k'$, of $k$ consist of $k$ and at
most $k$ other numbers, each no larger that $k/m$, the theorem
follows.

\subsection{Proof of Theorem~\ref{th:broder_shamir_friedman}}

It turns out that Theorem~\ref{th:broder_shamir_friedman} follows
almost immediately from Theorem~\ref{th:loop_count} by a straightforward
generalization of an idea in \cite{broder}.

\begin{lemma}\label{le:coincidences}
Let $B$ be a fixed, connected graph, and let $r\ge 1$ be an integer.
For any strictly non-backtracking closed walk, $w$, in $B$,
we have that the expected number of closed walks 
over $w$ in $G\in\cC_n(B)$ of 
order at least $r$
bounded above by
\begin{equation}\label{eq:individual_word_bound}
C\binom{k}{r+1} (2k)^{r+1} n^{-r}.
\end{equation}
In particular,
the expected number of strictly non-backtracking closed
walks, $w$, in a graph, $G\in \cC_n(B)$, such
that the Euler characteristic of ${\rm Graph}(w)$ is no more than
$-r$ is bounded by
\begin{equation}\label{eq:first_coincidence_bound}
C \binom{k}{r+1} (2k)^{r+1} n^{-r} \Tr(H_B^k) 
\end{equation}
provided that $k/n$ is sufficiently small (i.e., less than a positive
function of $r$ and $B$),
where $C$ depends only on $r$ and $B$; the above expression is
bounded by
\begin{equation}\label{eq:second_coincidence_bound}
C' k^{2r+2} n^{-r} \Tr(H_B^k)
\end{equation}
for some different constant, $C'$, depending only on $r$ and $B$.
\end{lemma}
\begin{proof}
(Compare Lemma~3 of \cite{broder} and the proof of Theorem~2.18
in \cite{friedman_random_graphs}.)
For any $i_0\in\{1,\ldots,n\}$, consider the unique word, $w'$, 
over $w$, whose initial vertex is $v_0'=(v_0,i_0)$,
$$
w'=(v_0',e_1',v_1',\ldots,v_k') .
$$
Then $w'$ must contain at least $r+1$ coincidences, where a
{\em \gls[format=hyperit]{coincidence}} is
a value, $i$, with
$1\le i\le k$, where the
the head of $e_i'$ was already visited (as a $v_j'$ with $j\le i-1$), 
but the value of $e_i'$ was 
not determined by previous edges (as an $e_j'$ or its inverse, with
$j\le i-1$).
Consider the position of the first $r+1$ coincidences, which can occur
in $\binom{k}{r+1}$ ways.  
Let us fix these $r+1$ coincidence values, $j_1,\ldots,j_{r+1}$,
with 
$$
1\le j_1<\cdots j_{r+1} \le k .
$$

We may view the vertices
$$
v'_j=(v_j,i_j)
$$
in $w'$
as arising from random variables $i_1,\ldots,i_k\in\{1,\ldots,n\}$,
where we successively determine the value of $i_1$, then $i_2$, etc.
Notice that at a coincidence, $j$ (equal to one of the fixed values
$j_1,\ldots,j_{r+1}$ above)
we have that the value of $i_j$ is a random variable that
must take on one of the values $i_0,\ldots,i_{j-1}$, and yet the
value of the edge over $e_j$ with tail $(v_{j-1},i_{j-1})$
has not been determined.
So the probability that $j$ is a coincidence,
given $i_0,\ldots,i_{j-1}$, is at most
$1/(n-j+1)$ (since at most $j-1$ values over $e_j$ have
been fixed via the $e_1',\ldots,e_{j-1}'$).
Hence the probability that
$i_j$ is a coincidence for $j=j_1,\ldots,j_{r+1}$ is at most
$$
j/(n-j+1) \le k/(n-k) \le 2k/n
$$
if $k\le n/2$.
Hence, for $k/n\le 1/2$, the $r+1$
coincidences occur with probability at most 
$$
(2k/n)^{r+1} .
$$
Since there are $n$ possible choices for $i_0$, and the first
$r+1$ coincidences occur in $\binom{k}{r+1}$ locations, 
we conclude the bound in 
\eqref{eq:individual_word_bound}.
Hence the total number of expected closed walks of $r+1$ or more 
coincidences, i.e., of Euler characteristic $-r$ or smaller, is
bounded by the expression in \eqref{eq:individual_word_bound} times
$\Tr(H_B^k)$, the number of
strictly closed non-backtracking walks of length $k$ in $B$.
This yields
the bound involving
\eqref{eq:first_coincidence_bound}.
The statement with the bound
\eqref{eq:second_coincidence_bound}
is an immediate consequence.
\end{proof}

The above lemma is another fundamental part of the method of
\cite{broder,friedman_random_graphs,friedman_relative,friedman_alon}.
Applying this lemma for $r=1$ shows that the expected number of strictly
non-backtracking walks that are not loops is at most
$$
O(k^4/n)\Tr(H_B^k),
$$
and hence
Theorem~\ref{th:broder_shamir_friedman} follows from the above lemma
and Theorem~\ref{th:loop_count}.

\subsection{$1/n$-Asymptotic Expansions and $B$-Ramanujan Functions}

To prove the Alon conjecture, we anticipate that 
Theorem~\ref{th:broder_shamir_friedman} may be refined to give
a ``$1/n$-asymptotic expansion'' as was done for regular graphs
in \cite{friedman_random_graphs,friedman_alon}.  From
Proposition~\ref{prop:EsymmProd} and 
Lemma~\ref{le:coincidences}, it is not hard to see that
$$
\EE_{G\in\cc_n(B)}\Bigl[ \Tr(H_G^k) \Bigr] 
$$
has an asymptotic expansion of the form
$$
P_0(k) + P_1(k) n^{-1} + \cdots + P_{r-1}(k) n^{1-r} + 
O(k^{r+1}) n^{-r} \Tr(H_B^k),
$$
where the $P_i(k)$ are some functions of $k$.
The methods of
\cite{friedman_random_graphs}, show that for $B=W_{d/2}$ and small $r$,
i.e., $r$ satisfying \eqref{eq:r_bound}, one has that 
each $P_i(k)$ has a ``principle part,'' namely $(d-1)^k p_i(k)$
where $p_i(k)$ is a polynomial, and an ``error term''
of size bounded by $Ck^C (d-1)^{k/2}$.
Furthermore, the principle part of $P_0(k)$ is $(d-1)^k$, and
all other principle parts of the $P_i(k)$ vanish.
Theorem~\ref{th:broder_shamir_friedman} shows that for $k$ even we have
$$
P_0(k) = \Tr(H_B^k) + O(k) \Tr(H_B^{k/2}) ,
$$
which for $d$-regular $B$ means that 
$$
P_0(k) = \Tr(H_B^k) + O(k) (d-1)^{k/2}.
$$
Hence, we may expect the principle part to involve powers of $k$ in
all the eigenvalues of $H_B$.
Furthermore, as in \cite{friedman_alon}, we know that such a 
principle part plus error term will not hold when $r$ is large,
and in order to get an asymptotic expansion for all $r$ (which seems
needed to prove a relativized Alon conjecture via trace
methods), we will need to
modify the trace in a way that we omit graphs, $G$, with certain
exceptional behaviour, i.e., avoiding tangles.

In anticipation of such expansions, we shall make some formal definitions.

\begin{definition}
\label{de:B_Ramanujan_one_var}
Let $P=P(k)$ be a function from, $\integers_{\ge 0}$,
the non-negative integers, to itself.  Let $B$ be a graph.  We say
that $P$ is 
a {\em \gls[format=hyperit]{B-Ramanujan function}}
if it can be written as
\begin{equation}\label{eq:Ramanujan_decomp}
P(k)=s(k)+e(k) ,
\end{equation}
where
\begin{enumerate}
\item there are polynomials, $p_\mu(k)$, with $\mu$ ranging over all
the eigenvalues of $H_B$, for which $s(k)$ is given by
$$
s(k) = \sum_{\mu\in\Spec(H_B)} \mu^k p_\mu(k);
$$
$s(k)$ is called the 
{\em \gls[format=hyperit]{principle part}} of the decomposition of
$P(k)$ as in \eqref{eq:Ramanujan_decomp}; and
\item we have that $e(k)$, called the 
{\em \gls[format=hyperit]{error term}} in
\eqref{eq:Ramanujan_decomp}
is such that
for every $\epsilon>0$ there is a $C>0$ for which
\begin{equation}\label{eq:error_term_B_Ramanujan}
|e(k)|
\le
C \Bigl( \Spec(H_B) + \epsilon \Bigr)^{k/2}.
\end{equation}
\end{enumerate}
\end{definition}
It is easy to see that
the $p_\mu(k)$ in the principle part are uniquely determined
(i.e., independent of the decomposition in \eqref{eq:Ramanujan_decomp})
for $\mu$ such that
$$
|\mu| > \Bigl(\Spec(H_B)\Bigr)^{1/2},
$$
and otherwise $p_\mu(k)$ are arbitrary (each choice of which affects
the error term $e(k)$).
We remark that we could replace $C$ by $Ck^C$ and get the same
definition, since for large $k$ we can dominate the $k^C$ contribution
by replacing $\epsilon$ with any $\epsilon'>\epsilon$.

The two basic examples of $B$-Ramanujan functions are as follows:
\begin{enumerate}
\item for arbitrary base graph $B$ which is connected and of 
negative Euler characteristic, Theorem~\ref{th:broder_shamir_friedman}
shows that
\begin{equation}\label{eq:first_example_Ram}
\EE_{G\in\cc_n(B)}\Bigl[ \Tr(H_G^k) \Bigr] =
P_0(k) + 
O(k^4/n) \Tr(H_B^k) ,
\end{equation}
where $P_0(k)$ is $B$-Ramanujan, and
\item for $B=W_{d/2}$ and $d$ even, the methods of \cite{broder} and
\cite{friedman_random_graphs} show that
\begin{equation}\label{eq:second_example_Ram}
\EE_{G\in\cc_n(W_{d/2})}\Bigl[ \Tr(H_G^k) \Bigr] =
P_0(k) + P_1(k)n^{-1} + \cdots + P_{r-1}(k)n^{1-r} +
O(k^C) n^{-r}  (d-1)^k
\end{equation}
for certain values of $r$ (Broder and Shamir show this for $r=1$, and
Friedman obtains this for any $r$ satisfying \eqref{eq:r_bound});
we remark that the eigenvalues of $H_B$ for $B=W_{d/2}$ are
$d-1$, $1$, and $-1$.
\end{enumerate}

\begin{definition}\label{de:asymptotic_expansion}
Let $f(k,n)$ be a function taking two positive integers,
$k$ and $n$, with values in the non-negative integers.  
Let $r$ be a positive integer.  We say that
$f$ has a {\em \gls[format=hyperit]{asymptotic expansion} of order $r$}
if there is an $\alpha=\alpha(r)>0$ and a $C=C(r)$ such that
for all $k,n$ we have
\begin{equation}\label{eq:asymptotic_expansion}
f(k,n)=
P_0(k) + P_1(k)n^{-1} + \cdots + P_{r-1}(k)n^{1-r} + {\rm err}(k,n),
\end{equation}
for some functions $P_i=P_i(k)$, where
for all $k,n$ such that $1\le k\le \alpha n^\alpha$ we have
\begin{equation}\label{eq:weak_error_bound}
|{\rm err}(k,n)| \le  C k^C \rho^k_H(B) 
n^{-r},
\end{equation}
where
\newnot{symbol:rho_k_H_B}
$\rho^k(H_B)$ is shorthand for $(\rho(H_B))^k$.
Moreover, we say that the asymptotic expansion satisfies the 
{\em usual error bound} if ${\rm err}(k,n)$ satisfies the bound
\begin{equation}\label{eq:usual_error_bound}
|{\rm err}(k,n)| \le  C k^{2r+2} \Bigl( \rho_H(B) \Bigr)^k
n^{-r},
\end{equation}
i.e., the error bound \eqref{eq:weak_error_bound} with $k^C$ replaced
with $k^{2r+2}$.
We call the $P_i(k)$ the 
{\em degree $i$ \gls[format=hyperit]{coefficient}} of the
asymptotic expansion.
We say that the expansion 
{\em is \gls[format=hyperit]{B-Ramanujan-coef}} or 
{\em has $B$-Ramanujan coefficients} if its coefficients---i.e., 
the $P_i(k)$---are $B$-Ramanujan functions.
\end{definition}

We remark that the methods of \cite{friedman_random_graphs}
show that a function such as
$$
f(k,n)=
\EE_{G\in\cc_n(B)}\Bigl[ \Tr(H_G^k) \Bigr] 
$$
has an $1/n$-asymptotic expansion to all orders; however, this
fact does not seem useful, unless we can assert something about the
coefficients, $P_i(k)$, of the expansion.

Note that \eqref{eq:first_example_Ram} and 
\eqref{eq:second_example_Ram} are examples of $1/n$-asymptotic expansions
with
$$
f(k,n) = 
\EE_{G\in\cc_n(B)}\Bigl[ \Tr(H_G^k) \Bigr] 
$$
for various $B$.  
The method of Theorem~2.12 of \cite{friedman_alon} shows that for 
any even integer $d\ge 4$, the above expected trace with $B=W_{d/2}$
fails to have $B$-Ramanujan coefficients for $r$
of size proportional to $d^{1/2}\log d$; actually this is done in
\cite{friedman_alon} for the expected number of closed, non-backtracking
walks of length $k$ for a $G\in\cC_n(W_{d/2})$, but it is easy to
modify the argument there to apply to strictly non-backtracking
closed walks, which is just the above expected trace of $H_G^k$.

In \cite{friedman_alon}, 
where $B=W_{d/2}$,
one obtained $1/n$-asymptotic expansions with $B$-Ramanujan coefficients
for arbitrarily large $r$
by considering a different $f(k,n)$, namely the {\em selective trace}
used there; roughly speaking, the selective trace of length $k$
in a graph, $G$, equals the number of strictly non-backtracking closed
walks of length $k$ such that no ``long subwalks'' of the walk trace
out a subgraph that contains a tangle (as in
Definition~\ref{de:tangle_of_B}).
Although the formal definition of a selective trace is rather
cumbersome, the main point is that if $G$ has no tangles---which is
usually the case---then the selective trace of length $k$ equals
$\Tr(H_G^k)$;
when $G$ has one or more tangles,
the selective trace is generally smaller than $\Tr(H_G^k)$.
Selective traces enables
one to get $1/n$-asymptotic expansions to arbitrary order,
as in \cite{friedman_alon}, albeit for the expected value of a
variant of $\Tr(H_G^k)$.
In this paper we make a significant simplification over \cite{friedman_alon}
by replacing the selective traces by a more direct notion of a
{\em certified trace}.
We formally define the {\em certified trace} in the next section.
For the rest of this section we explain a bit more on the trace method,
which will help explain why the certified trace is a simpler variant
of the selective trace, and yet ultimately gives bounds for
the expected values of $\Tr(H_G^k)$ for graphs without tangles.

\subsection{Types and Finite Linear Combinations of
$1/n$-Asymptotic Expansions}
\label{sb:rough_types}

In this paper, like in \cite{friedman_random_graphs,friedman_alon},
we proceed in two steps.  First, we show that certain modifications
of the function
$$
f(k,n) = 
\EE_{G\in\cc_n(B)}\Bigl[ \Tr(H_G^k) \Bigr] 
$$
have $1/n$-asymptotic expansions to arbitrary large order.  In this
first part we know little about the principle parts of the coefficients
of the expansion,
i.e., the $P_i(k)$ of \eqref{eq:asymptotic_expansion} in
Definition~\ref{de:asymptotic_expansion}.
The second step seeks to use the existence of an expansion with
coefficients
$P_i(k)$ being 
$B$-Ramanujan to draw conclusions about ``high probability''
bounds on the eigenvalues
of $H_G$ or $A_G$ for a random $G\in\cC_n(B)$.

A theme throughout \cite{friedman_random_graphs,friedman_alon} and this
paper is that any finite linear combination
of $1/n$-asymptotic expansions is, again, a $1/n$-asymptotic expansion.
This is not generally true of infinite sums or infinite linear combinations.
Hence if $f(k,n)$ is any variant of 
$$
\EE_{G\in\cc_n(B)} \Bigl[ \Tr(H_G^k) \Bigr],
$$
it suffices to write the above, or a variant thereof, as a finite sum
of $1/n$-asymptotic expansions.

The most basic observation about this process is that, by the proof of
Lemma~\ref{le:coincidences}, to obtain a $1/n$-asymptotic expansion
to order $r$
of the number of strictly non-backtracking closed walks of length
$k$, or a subset of such walks,
it suffices to count only those walks, $w$, such that ${\rm Graph}(w)$
has Euler characteristic at least $1-r$.
Hence, for example, it suffices that for each $i=0,1,\ldots,r-1$, the
number of such walks with ${\rm Graph}(w)$ of Euler characteristic exactly
$-i$ has a $1/n$-asymptotic expansion.
Furthermore, as already evident in \cite{broder}, one can often analyze
this number, for a given $i$, in terms to certain essential features of
$w$ and ${\rm Graph}(w)$, such as the starting vertex of $w$ and all
vertices in ${\rm Graph}(w)$ of degree at least three.

Indeed, for ${\rm Graph}(w)$ of Euler characteristic zero, Broder and
Shamir reduce such walks into two cases: (1) those where ${\rm Graph}(w)$
is a simple cycle, and (2) those where ${\rm Graph}(w)$ is a cycle plus
a path, where the starting vertex is of degree $1$ and one other vertex
is of degree three. Case (2) can only yield closed, non-backtracking walks
that are not strictly non-backtracking, since each vertex of ${\rm Graph}(w)$
must be of degree at least two if $w$ is a strictly non-backtracking
closed walk.

Similarly, it is well known (see \cite{linial_puder}) that all
graphs, $G'$, of Euler characteristic $-1$ 
occurring as the graph of a strictly non-backtracking closed walk can be viewed
as three cases: (1) a ``figure 8'' graph (where one vertex is of degree four),
(2) a ``barbell'' graph (where two vertices are of degree three, with only
one path connecting the two vertices), and
(3) a ``theta'' graph (i.e., looks like a $\theta$), 
with two vertices of degree three jointed by three
edge disjoint paths.  The fact that
$$
2\chi(G') = \sum_{v\in V_{G'}} \Bigl(2 - {\rm degree}_{G'}(v) \Bigr)
$$
shows that the above are the only three shapes of a connected graph, $G'$,
each of whose vertices have degree at least two.

Similarly, the ``shape'' of any graph of fixed Euler characteristic
arising as $G'={\rm Graph}(w)$ for a strictly non-backtracking closed walk
can be divided into a finite number of ``shapes,'' according to the starting
vertex of the walk, all vertices of $G'$ of degree greater
than three, and how these vertices are connected by edge disjoint paths
in $G'$.  The {\em \gls{type}} of a walk, $w$, remembers this information as
some other finite amount of information, such as in which order the
vertices and paths are visited; see Subsection~\ref{sb:types_and_forms}
of this article, or similar definitions in
\cite{friedman_random_graphs,friedman_alon}.
The key point is each $1/n$-asymptotic expansions to order $r$
that we study involves summing over a finite number of possible
Euler characteristics, each sum subdivided into a finite number 
of ``types.''

The sum of walks of a given type will be further subdivided into
a linear combination of simpler sums.  Our ``certified trace''
makes this subdivision very simple.

\section{Certified Traces}
\label{se:p0-cert}

In this section we will define the
the certified trace and explain in rough terms its
significant features;
its full significance may not be evident
until Chapter~\ref{ch:p1}.  
The main feature, like the notion of the {\em type} of a walk,
is to divide a complicated sum into a {\em finite} linear combination
of simpler sums.

\subsection{The Certified Trace}
\label{sb:certified_trace}

The following is a self-contained definition of the certified traces
that we will use in Chapter~\ref{ch:p1} (see
Definition~\ref{defn:CertTr}).
\begin{definition} 
\label{de:certified_simple}
For any graph, $G$, we define its
\newnot{symbol:certified_r_k}
{\em $r$-th truncated 
\gls[format=hyperit]{certified trace} of length $k$}, denoted
$$
\CertTr_{<r}(G,k)
$$
to be the number of strictly non-backtracking 
closed walks, $w$, in $G$, of length $k$, such that $\Graph(w)$
is of \gls{order} less than $r$ and 
is not a \gls{B-tangle}.
\end{definition}

As explained in Subsection~\ref{sb:rough_types}, we will classify
all walks, $w$, for which ${\rm Graph}(w)$ has a given
Euler characteristic into its {\em \gls{type}}, which involves various
data about $w$; this will include the 
{\em \gls{type graph}} 
\newnot{symbol:typegraph}
$T={\rm TypeGraph}(w)$, where $V_T$
consists of
the first vertex of $w$ and all vertices of degree three or more, and
where $E_T$ has an edge for each path in ${\rm Graph}(w)$ joining
two vertices of $V_T$.

Let us work backwards: given the graph $T$ as above, a walk, $w$,
for which ${\rm TypeGraph}(w)=T$ is counted in $\CertTr_{<r}(G,k)$ iff
\begin{enumerate}
\item $\ord(T)<r$, since we easily verify that $T$ and
${\rm Graph}(w)$ have the same Euler characteristic; and
\item if each $e\in E_T$ corresponds to a path of length
$k(e)$ in $w$, then the collection
$\vec k = \{k(e)\}_{e\in E_T}$, satisfies
$$
\rho_H\Bigl( \VLG(T,\vec k) \Bigr) < \sqrt{\rho_H(G)}.
$$
\end{enumerate}
We are therefore lead to consider
\begin{equation}\label{eq:SofT}
S(T)=\Bigl\{ \vec k\from E_T\to \integers_{>0} \ |\  
\rho_H\Bigl( \VLG(T,\vec k) \Bigr) < \sqrt{\rho_H(G)}
\Bigr\},
\end{equation}
which is a subset of $\integers_{>0}$.  For each $\vec k\in S(T)$, we
will consider various functions, $f(\vec k)$, and we will want to
conclude that sums of the form
\begin{equation}\label{eq:abstract_sum}
\sum_{\vec k\in S(T)} f(\vec k)
\end{equation}
which give the coefficients of 
\glspl{asymptotic expansion},
are 
\glspl{B-Ramanujan function}.

We shall explain that although $S(T)$ my have a complicated structure,
an abstract sum as in \eqref{eq:abstract_sum} can be written as a
finite linear combination of simpler sums, provided that $S(T)$
has a finite number of {\em minimal elements}.
The rest of this section is a discussion of this point.

\subsection{Minimal Elements in Posets}

At this point we will summarize the discussion in
Subsection~\ref{sb:certifiable}, to explain why the certified trace
is useful.  This discussion applies to a variety
of posets, i.e., partially ordered sets, $P$, but we shall only be
concerned with the case $P=\integers_{>0}^m$, for various values of 
$m$, which becomes a poset under the partial order on two elements
$$
\vec x=(x_1,\ldots,x_m),\quad
\vec y=(y_1,\ldots,y_m)
$$
of $\integers_{>0}^m$ given as 
$$
\vec x\le \vec y \quad\mbox{iff}\quad
\mbox{$x_i\le y_i$ for all $i=1,\ldots,m$}.
$$

Let $P$ be a poset on a countable number of elements, and assume
(1) that the supremum (i.e., least upper bound, i.e., join) of any
two elements exists in $P$, and (2) any {\em upper set} $S\subset P$,
i.e., $s\in S$ and $s\le s'$
implies $s'\in S$, has a finite number of minimal elements.
Then 
inclusions/exclusion shows that any absolutely convergent sum
\begin{equation}\label{eq:Ssum}
\sum_{s\in S} f(s)
\end{equation}
for $f\from S\to\reals$, may be written as a linear combination of 
a finite number of sums
\begin{equation}\label{eq:simplerSum}
{\rm Sum}(s_0,f) = \sum_{s_0 \le s} f(s) .
\end{equation}
It turns out that it \eqref{eq:simplerSum} will be much easier to
analyze that \eqref{eq:Ssum}, and it will be crucial to know that
the upper sets, $S$, of interest to us have a
finite number of minimal elements.

In a bit more detail,
in our situation $f(s)=f(s,k)$ will depend on an element, $s\in S$,
and a positive integer $k$.  It follows that if each
sum in \eqref{eq:simplerSum}, with $f(s)$ replaced with $f(s,k)$,
is $B$-Ramanujan as a function of $k$, then so is the
sum in \eqref{eq:Ssum} with $f=f(s,k)$.
Again, the sums in \eqref{eq:simplerSum} will be much easier to
analyze than those in \eqref{eq:Ssum}, with $f=f(s,k)$.

The essential fact that we will show in
Subsection~\ref{sb:certifiable} is that {\em any} upper set 
in $\integers_{>0}^m$ has a finite number of minimal elements.
This is easy to deduce from the well-known fact that
if $x_1,\ldots,x_m$ are independent transcendentals over
the complex numbers, $\complex$, then any ideal in
$\complex[x_1,\ldots,x_m]$ is finitely generated.

We remark that
\cite{friedman_alon} works with sets like
\begin{equation}\label{eq:nonStrictInequality}
\Bigl\{ \vec k\from E_T\to \integers_{>0} \ |\
g(\vec k) \le \alpha
\Bigr\},
\end{equation}
for various functions $g(\vec k)$ and real numbers $\alpha$; there
a sort of ``compactness'' argument shows that such sets
have a finite number of minimal
elements (see, for example, Lemma~9.2 of \cite{friedman_alon}).
However the strict inequality in \eqref{eq:SofT} means such compactness
don't generally work.

We remark that in $\integers_{>0}^2$, the upper set of pairs
$(k_1,k_2)$ 
for which $k_1+k_2>1000$ has $1000$ minimal elements; replacing
$1000$ with any positive integer we see that number of minimal elements
an upper set of $\integers_{>0}^2$ (and similarly with $2$ replaced
by with any $m\ge 2$) can have an arbitrarily large number of 
minimal elements.

\section{Other New Ideas in This Article}
\label{se:p0-other}

In this section, we briefly explain two other new techniques we use
in this 
article, beyond the methods of \cite{friedman_alon}.
These are
(1) a generalized side-stepping lemma, needed when $B$ is $d$-regular
but not Ramanujan, and
(2) estimates involving ``larger edge multiplicities.''

\subsection{A More General ``side-stepping lemma''}
\label{sb:more_general_side}

Let us review the ``side-stepping lemma'' and its use in
\cite{friedman_alon}, and indicate our more general lemma.

Using the certified trace we will show that for any $B$ and positive
integer, $r$, 
$$
\EE_{G\in\cC_n(B)} \Bigl[ \CertTr_{<r}(G,k) \Bigr]
$$
has a $1/n$-asymptotic expansion to order $r$ with respect to $B$.
It is conceivable that this alone may allow us to deduce the generalized
Alon conjecture with base $B$; but this is not so clear.

Following \cite{friedman_alon} we use the following idea, which requires
some notion.

\begin{notation}
\label{no:has_tangles}
Let
\newnot{symbol:first_has_tangle}
$\HasTangle_{r,B}$
denote the set of graphs that contain a 
\gls{B-tangle} of \gls{order} less than $r$, and
let
\newnot{symbol:first_has_tangle_indicator}
$\II_{\HasTangle_{r,B}}$ denote the indicator function of
$\HasTangle_{r,B}$.  Let
\newnot{symbol:first_tangle_free}
$\TF_{r,B}$ is the complement of $\HasTangle_{r,B}$, i.e., those
graphs free of $B$-tangles of order less than $r$, and
denote the indicator function of $\TF_{r,B}$
\newnot{symbol:first_tangle_free_indicator}
$$
\II_{\TF(r,B)}(G) = 1 - \II_{\HasTangle_{r,B}}(G).
$$
\end{notation}

It is not hard to adapt the techniques used to obtain the above
$1/n$-asymptotic expansion to show that
$$
\EE_{G\in\cC_n(B)} \Bigl[ \II_{\HasTangle_{r,B}}(G) \CertTr_{<r}(G,k) \Bigr]
$$
has a $1/n$-asymptotic expansion to order $r$ with respect to $B$.
It then follows that 
$$
\EE_{G\in\cC_n(B)} \Bigl[ \II_{\TF(r,B)}(G) \CertTr_{<r}(G,k) \Bigr]
$$
also has a $1/n$-asymptotic expansion to order $r$ with respect to $B$.
However, when $G$ is free of tangles of order less than $r$, then
$$
\CertTr_{<r}(G,k)
$$
is the sum of all strictly non-backtracking closed walks of length $k$
and order less than $r$, and it follows---essentially
from
Lemma~\ref{le:coincidences}---that therefore
$$
\EE_{G\in\cC_n(B)} \Bigl[ \II_{\TF(r,B)}(G) \Tr(H_G^k) \Bigr]
$$
has the same
$1/n$-asymptotic expansion to order $r$, modulo the error term, as 
$$
\EE_{G\in\cC_n(B)} \Bigl[ \II_{\TF(r,B)}(G) \CertTr_{<r}(G,k) \Bigr] .
$$

At this point we wish to use the fact that 
\begin{equation}\label{eq:used_for_sidestep}
\EE_{G\in\cC_n(B)} \Bigl[ \II_{\TF(r,B)}(G) \Tr(H_G^k) \Bigr]
\end{equation}
has a $1/n$-asymptotic expansion to order $r$ to draw conclusions about
a high probability bound regarding the eigenvalues of $H_G$ for
$G\in\cC_n(B)$.
For reasons similar to those mentioned in Subsection~\ref{sb:limitations},
we do not expect good results by applying the asymptotic expansion
to one value of $k$ for a given $n$.
Rather, for each $n$ we use the existence of the expansion for many
values of $k$ to draw some 
conclusions about the typical locations of eigenvalues of $H_G$;
this is because we have no a priori information about the principle
parts of the coefficients in the expansion.  This is called
``side-stepping,'' i.e., side-stepping the fact that we have no information
on the principle parts.

The basic idea of the ``side-stepping'' lemma of 
\cite{friedman_alon} is
that each coefficient has a principle part that is a polynomial in
$k$ times $(d-1)^k$ (in this case $B=W_{d/2}$ has eigenvalues $d-1$, $-1$,
and $1$); hence by applying a sufficiently high power
of $S-(d-1)$, where $S$ is the ``shift operator in $k$''
(i.e., $(Sf)(k,n)=f(k+1,n)$), we annihilate the principle part of 
each coefficient.  
The ``side-stepping'' lemma here is a little more involved (and gives
a less information), because if $B$ is not Ramanujan, then the
principle part of the coefficients can contain terms which are polynomials
in $k$ times $\mu^k$ for possibly a number of 
eigenvalues, $\mu$, of $H_B$, that 
are greater than $(\rho_H(B))^{1/2}$;
this contrasts with the case of $\mu=\pm(d-1)$,
where, in Section~\ref{se:p2-fund-exp}, we use {\em spreading}---a
type of expansion
(see Section~\ref{se:p2-spread})---in random covers, to control 
the contribution of the $\mu=\pm(d-1)$ principle parts.
Such a side-stepping is technically more complicated, but very much in
the spirit of the original side-stepping lemma.

The side-stepping lemmas, both here and in \cite{friedman_alon}, show
that any non-zero principle part of a coefficient in the order $r$
$1/n$-asymptotic expansion of the expected value 
in \eqref{eq:used_for_sidestep} actually arises from eigenvalues of
$H_G$ that are concentrated near an $H_B$ eigenvalue (larger than
$\rhoroot B$) with probability proportional to $n^{-i}$, for some $i<r$.
We can use this fact in two ways, for the two main theorems of this paper:
(1) for Theorem~\ref{th:main_Alon}, which follows from
Theorem~\ref{th:main}, the smallest possible value of $i$ is one,
due to the fact that the first term of the $1/n$-asymptotic expansion,
i.e., the $n^0$ term, 
exactly matches the old eigenvalues of $H_G$, i.e., those of $H_B$,
as Theorem~\ref{th:broder_shamir_friedman} shows; and
(2) for Theorem~\ref{th:main_Alon_Ramanujan}, in this case $\pm(d-1)$ are
the only possible
eigenvalues of $H_B$ greater than $\rhoroot B$, and in this case
any concentration of eigenvalues near $d-1$ can be attributed to $G$
being nearly disconnected, and this can be
controlled by the aforementioned {\em spreading}
probabilities estimated in Section~\ref{se:p2-spread}.

We remark that we will apply our side-stepping lemma to
\eqref{eq:used_for_sidestep}, and this equation
makes no reference to the certified trace;
this is similar to \cite{friedman_alon}, where the selective trace is
used to obtain $1/n$-asymptotic expansions, but does not appear once
the side-stepping lemma there is applied.
The point is we need some {\em modified trace}, 
like the certified or selective
trace, to control the coefficients of the 
$1/n$-asymptotic expansions we use.  
Both here and in \cite{friedman_alon}, we use modified traces which
equal $\Tr(H_G^k)$ for most $G\in\cC_n(B)$.
Then, once we control the expansions
for the expected value of a modified
trace, and for the same multiplied by $\II_{\HasTangle_{<r,B}}$, then we
have essentially controlled the expansions
for the expected value of the
tangle free indicator,
$\II_{\TF(r,B)}$, times the Hashimoto trace $\Tr(H_G^k)$.

We remark that it is conceivable that one could construct a more direct
argument regarding the relativized Alon conjecture using only the
expected certified trace $1/n$-asymptotic expansions, without
invoking estimates with the above indicator functions.  At present we
do not know how to do this.

\subsection{The Second Idea: Weaker $B$-Ramanujan Functions,
Simpler Estimates}

The second idea is a bit technical, and will only become clear in the
proof of Lemma~\ref{le:CertTr_B_ramanujan_expansion}.  However we
can give the rough idea.  

In \cite{friedman_random_graphs,friedman_alon}, Friedman defined
$B$-Ramanujan
functions, with $B$ having one vertex and being $d$ regular, but we
require the error term to be bounded by
$$
C k^C (d-1)^{k/2}.
$$
Notice that for $B$ begin $d$-regular, our definition amounts to a
bound, for any positive $\epsilon$, of
$$
C (d-1+\epsilon)^{k/2},
$$
with $C=C(\epsilon)$.  Hence our notion is a bit weaker, but we observe
in this paper that this weaker error term bound
suffices to prove the Alon conjecture or its relativization, such as
Theorems~\ref{th:main_Alon} 
and~\ref{th:main_Alon_Ramanujan}.
The first error term estimate, used in \cite{friedman_random_graphs,
friedman_alon}, is much more difficult to obtain on the coefficients
of $1/n$-asymptotic expansions.

To elaborate, \cite{friedman_alon} 
proves the $1/n$-asymptotic coefficients of the selective traces
are $B$-Ramanujan by using
estimates on two functions there, namely $W$ in Theorem~6.6, and
$f_{\vec m}$, in Theorem~8.5, which are multiplied together.
We point out a minor error there, namely
that Theorem~6.6 is incorrect unless 
we replace $M_2$ there with the quantity $j_2$ as in the
proof of Theorem~8.5; however, to balance this, we 
note that Theorem~8.5 can easily be improved
to have $M_2$ replaced with $j_2$, as well; hence
the fact that the coefficients
are $B$-Ramanujan still holds.
But the point is that any increase in an upper bound for $W$ must be 
compensated by a sharper bound for $f_{\vec m}$.

In this article, 
the roles played by $W$, and $f_{\vec m}$ in
\cite{friedman_alon} are played by our $Q_1(K_1;\vec m^1)$ and
$Q_2(K_2)$ (the $Q_2(K_2)$ incorporates the $W$ into it).
The precise definitions and estimates are given
in Subsection~\ref{sb:asymptotic_expansions}. 
However, roughly speaking, the reason our
estimates are much
simpler is due to the fact that it suffices to show that
$$
|Q_2(K_2)| \le \bigl(\rho(H_B)+\varepsilon\bigr)^{K_2/2}
$$
for any $\varepsilon>0$ and $K_2$ sufficiently large, rather
than to give a more precise $CK_2^C(\rho(H_B))^{K_2/2}$ estimate
(possibly with some factors that trade off between
$Q_2(K_2)$ and $Q_1(K_1;\vec m^1)$, as needed in \cite{friedman_alon}
between $W$ and $f_{\vec m}$).
And the reason this weaker estimate on $Q_2(K_2)$
suffices is due to our weaker notion of $B$-Ramanujan.

\chapter{The $d$-Regular Case Without Half-Loops}
\label{ch:p1}

The point of this chapter is to prove 
the relativized Alon conjecture in the case where the base graph, $B$,
is $d$-regular, for the Broder-Shamir model of a random cover of degree
$n$ of $B$.
In addition, it will ease notation to assume that $B$ has no half-loops
(although this does not change the main techniques in any essential
way).
So in this section we prove Theorem~\ref{th:main_Alon}, allowing $B$ to
be any $d$-regular graph, with multiple edges and whole-loops allowed,
but without half-loops.  The case of allowing $B$ to have half-loops
or more general ``algebraic models'' will be discussed 
in Section~\ref{se:p2-algebraic}.

\section{Introduction and Overview of This Chapter}\label{section:introduction}
\label{se:p1-intro}

We begin by giving an overview of this chapter; throughout we will
assume that the base graph, $B$, is a $d$-regular, connected graph,
for some integer $d\ge 3$, and that $B$ does not have half-loops.
For as long as possible, we will discuss theorems for arbitrary
connected graphs, $B$, of negative Euler characteristic without half-loops.
In fact, we will prove Theorem~\ref{th:main} for all such $B$.
Theorem~\ref{th:main_Alon} follows almost immediately.
We note that the Broder-Shamir model $\cC_n(B)$, and all our results,
require extra care when $B$ has half-loops;
the discussion of Theorem~\ref{th:main} for $B$ with half-loops will
be addressed in Chapter~\ref{ch:p2}, specifically 
Section~\ref{se:p2-algebraic}.

Most of the work in this chapter is devoted to establishing a
\gls{asymptotic expansion} for 
$$
\expect{G\in\cC_n(B)}{\tanglefreeindicator(G)\Tr(H_G^k)} 
$$
i.e., the
expected value of \( \Tr(H_G^k) \), where we replace the trace by zero when
\( G \) contains an element of 
\( \Tangle_{r,B} \)---i.e., a \gls{B-tangle} of \gls{order}
less than $r$---as a subgraph (recall Notation~\ref{no:has_tangles}). 
Let us state the result formally.

\begin{thm} 
\label{th:main_expansion_B}
\label{TH:MAIN_EXPANSION_B}  
Let $B$ be any connected graph of positive \gls{order} without
half-loops.  Then for any positive integer, $r$,
\begin{equation}\label{eq:rr'series} \EE_{G \in \cC_n(B)}[\II_{{\rm
TF}(r,B)}(G) \Tr(H_G^k) ] 
\end{equation}
has a \gls{asymptotic expansion} to order $r$,
satisfying the usual error bound,
\eqref{eq:usual_error_bound}, in the sense of
Definition~\ref{de:asymptotic_expansion}, 
whose \glspl{coefficient} are \glspl{B-Ramanujan function}.
\end{thm}

We will show that the event 
$\HasTangle_{r,B}$ has probability proportional to $n^{-j}$ in 
$\cC_n(B)$ for some
$j\ge 1$ and all $j> r$.  It follows that 
that $P_0(k)$ in the above theorem is the same as the
$P_0(k)$ computed in
Section~\ref{se:p0-loop}.
The same observation about $\HasTangle_{r,B}$ in $\cC_n(B)$
implies that
the $P_i(k)$ cannot have vanishing principle part for all values
of $i$;
this means that we cannot use Theorem~\ref{th:main_expansion_B}
alone to derive conclusions about the relativized Alon Conjecture.

It is conceivable that the the expected value in $\cC_n(B)$
of $\trace(H_G^k)$
conditioned on $G\in\tanglefree$ has an $1/n$-asymptotic expansion
whose coefficients have vanishing principle part. 
This would allow
us to conclude the relativized Alon Conjecture more simply.  However,
our methods do not give explicit values of the coefficients in the
$1/n$-asymptotic expansion; hence in this paper, as well as 
\cite{friedman_alon}, we need a sort of
{\em side-stepping lemma} alluded to in
Subsection~\ref{sb:more_general_side} in conjuction with the above
theorem.

After proving Theorem~\ref{th:main_expansion_B},
and developing appropriate {\em side-stepping} machinery,
the relativized Alon Conjecture, for regular base graphs without
half-loops, will follow quite easily.

Let us give an overview of this chapter, including 
the proof of Theorem~\ref{th:main_expansion_B}.
In Sections~\ref{se:p1-prelim}---\ref{se:p1-tangles} we establish the
most technically difficult theorem below.

\begin{theorem}\label{th:certified_trace_expansion}
\label{TH:CERTIFIED_TRACE_EXPANSION}  
Let $B$ be any connected graph of positive \gls{order} without
half-loops.  Then for any positive integer, $r$,
\begin{equation}\label{eq:certified_expansion}
\expect{G \in \cC_n(B)}{\CertTr_{<r}(G,k)}
\end{equation}
has a \gls{asymptotic expansion} to order $r$,
satisfying the usual error bound,
\eqref{eq:usual_error_bound}, in the sense of
Definition~\ref{de:asymptotic_expansion}, 
whose \glspl{coefficient} are \glspl{B-Ramanujan function}.
\end{theorem}

In Section~\ref{se:p1-prelim} we will give a number of preliminary
and general facts regarding graphs and convolution of functions.
In Section~\ref{se:p1-walk-sums} we will develop a theory of
walk sums, in the spirit of \cite{friedman_alon}, which will be
the starting point for all of our $1/n$-asymptotic expansions.
Section~\ref{se:p1-tangles} develops some technical estimates
regarding tangles and certified traces that allow us to
prove Theorem~\ref{th:certified_trace_expansion}.

Section~\ref{se:p1-with-tangles} will generalize the discussion
in Sections~\ref{se:p1-walk-sums} and Section~\ref{se:p1-tangles}---in
the spirit of Chapter~9 of \cite{friedman_alon}---to
give theorems that are variants of
Theorem~\ref{th:certified_trace_expansion}, including the following
theorem.

\begin{theorem}\label{th:certified_trace_expansion_with_tangles}
\label{TH:CERTIFIED_TRACE_EXPANSION_WITH_TANGLES}  
Let $B$ be any connected graph of positive \gls{order} without
half-loops.  Then for any positive integer, $r$,
\begin{equation}\label{eq:certified_expansion_with_tangles}
\expect{G \in \cC_n(B)}{\II_{\HasTangle_{r,B}}(G)\CertTr_{<r}(G,k)}
\end{equation}
has a $1/n$-asymptotic expansion to order $r$, with $B$-Ramanujan coefficients,
satisfying the usual error bound.
Hence (subtracting \eqref{eq:certified_expansion_with_tangles}
from \eqref{eq:certified_expansion})
the same is true for
\begin{equation}\label{eq:certified_expansion_without_tangles}
\expect{G \in \cC_n(B)}{\II_{{\rm TF}(r,B)}(G)\CertTr_{<r}(G,k)}
\end{equation}
\end{theorem}
We finish Section~\ref{se:p1-with-tangles} by remarking that
Theorems~\ref{th:certified_trace_expansion},
\ref{th:certified_trace_expansion_with_tangles}, and
\ref{thm:walk_sum_expansion}
easily yield Theorem~\ref{th:main_expansion_B}.

From here Section~\ref{se:p1-side-step} gives the {\em side-stepping}
machinery needed here.
At this point all the discussion is valid for an arbitrary graph, $B$,
of positive order and without half-loops.
It is only in Section~\ref{se:p1-main-proof} that we specialize to the
case of $d$-regular $B$, in which we prove the relativized
Alon conjecture, Theorem~\ref{th:main_Alon} (assuming $B$ has no
half-loops).  The case where $B$ is regular but may have 
half-loops is addressed in Section~\ref{se:p2-algebraic}.

For the rest of this section we wish to give an overview of 
the ideas of Sections~\ref{se:p1-walk-sums} and~\ref{se:p1-tangles},
because they involve some rather technical discussion and estimates,
although their underlying ideas are fairly simple.  This will be
done in the next subsection.

\subsection{More on Types and Asymptotic Expansions}
\label{sb:more_on_types}

The basic idea of Broder and Shamir \cite{broder}
and Friedman \cite{friedman_random_graphs,friedman_relative,
friedman_alon} is to organize the walks, \(w\), by 
their 
\gls{type graph}, i.e.,
the
graph \( T \) associated to \( w \) by setting \(V_T\) to be the initial
vertex of \( w \) and all vertices of \( \Graph(w) \) of degree at least 3;
assuming \( w \) is strictly
non-backtracking, all vertices of \( \Graph(w) \) have
degree at least 2, and the vertices not in \( V_T \) can be contracted to
get a graph \( T \) whose vertices are \( V_T \) and whose edges represent
\glspl{beaded path} 
(whose interior vertices are all of degree two) in \( \Graph(w)
\).
There is some additional data that the 
{\em \gls[format=hyperit]{type}} 
remembers, which we
now describe.  A more formal description---in the context of
{\em \glspl{walk sum}}---will be given in Section~\ref{se:p1-walk-sums}.

Before doing so, we emphasize that the notion of type would be much
harder to use if we allowed $w$ to backtrack: if $w$ is 
\gls{non-backtracking},
then at any point that $w$
enters the first edge in a 
\gls{beaded path}, it must continue
along such a path until it reaches the end.  This greatly simplifies
matters: it means that each edge in a beaded path is traversed the
same number of times, and the number of closed, non-backtracking walks
in $\Graph(w)$ of length $k$ can be deduced from the graph, $T$, and
knowing what is the length of each beaded path in $\Graph(w)$ corresponding
to each edge of $T$.

For a walk, \(w \), in \(G \), we remember some additional data beyond the
graph, \(T \), obtained by collapsing the degree two paths in \({\rm
Graph}(w) \); namely we remember \begin{enumerate}[label= (\arabic*)] \item
the vertices of \(B \) lying below (i.e., via the projection \(\pi\from
G\to B \) restricted to \({\rm Graph}(w)\subset G \)) the vertices of \(
{\rm Graph}(w) \) that we remember in \( T \) (i.e., the vertices of \( T
\), which are the initial vertex and all vertices of degree at least
three), \item a lettering, meaning, for each edge incident to a vertex of
\( T \) in \( \Graph(w) \), we remember to which edge in \( B \) it is
mapped, and \item the order in which all type edges and vertices are first
encountered in \( w \) (so the edges and vertices of \( T \) are ordered
sets), and the direction in which each edge is first traversed.
\end{enumerate} We call this data the {\em type} of \( w \), denoting it \(
{\cal T}=\cT(w) \), which consists of an underlying graph, \( T \), along with the
\( B \) structure of the \( 1 \)-neighbourhood of all \( T \) vertices in
\( w \), and the ordering of the vertices and edges of \( T \), and
orientation of edges.

Our main approach works with a certain type of function and various
convolutions of such functions which we now define.

\begin{definition} 
\label{de:polyexponential}
Let $C$ be a finite set of complex numbers.
We say that a function,
$P=P(\vec k)$ from $(\integers_{\ge 1})^m$ to the integers is
{\em \gls[format=hyperit]{polyexponential} with bases $C$} 
if $P$ is given as
$$
P(\vec k) = \sum_{\vec c=(c_1,\ldots,c_m)\in C^m}
p_{\vec c}(\vec k) \vec c^{\,\vec k},
$$
where $p_{\vec c}(\vec k)$ is a polynomial in $\vec k$, and we use the
shorthand
$$
\vec c^{\,\vec k}= c_1^{k_1}\ldots c_m^{k_m}.
$$
If $B$ is a graph, then by a {\em $B$-polyexponential} we mean a
polyexponential with bases being the spectrum of $H_B$.
\end{definition}

It turns out that, as a fairly straightforward extension of the ideas of
\cite{friedman_random_graphs}, we have that the \( P_i(k) \) 
of \eqref{eq:rr'series}
are linear
combinations of sums of the form: \begin{equation}\label{eq:WtimesP}
\sum_{\vec m\cdot \vec k = k} W_{\cal T}(\vec m,\vec k) P(\vec k),
\end{equation} where \begin{enumerate}[label= (\arabic*)] \item \[ \vec m,
\vec k \in \ZZ_{\geq 1}^{E_T}, \] i.e., \( \vec m,\vec k \) are vectors
indexed over \( E_T \) of positive integers; 
here $\vec m$ tells us how many times each edge in $T$ is traversed, and
$\vec k$ tells us the length of each beaded path corresponding to each
edge of $T$;
\item \( P(\vec k) \) is given by
polyexponential in bases \( \Spec(H_B) \), where \( \vec k \) is a vector
indexed over \( E_T \); and \item \( W_{\cal T}(\vec m,\vec k) \) counts
the number of walks in \( T \) that correspond to walks contributing to
strictly non-backtracking closed walks of length \( k \), where edge \(
f\in E_T \) corresponds to a path of length \( k_f \) in \( {\rm Graph}(w)
\) and \( m_f \) is the number of times \( f\in E_T \) is traversed (in
either direction).  
Of course, here
\( W_{\cal T}(\vec m,\vec k) \) is independent
of \( \vec k \).  
\end{enumerate}

The problem with \eqref{eq:WtimesP} is that the function \( W_{\cal T}(\vec
m,\vec k) \) is often so large that the sum in \eqref{eq:WtimesP}
is not a
\gls{B-Ramanujan function}
(or if it is, it is difficult to prove this).
In fact, as previously mentioned,
for some $i$ of size at most roughly $d^{1/2}\log d$,
the expected value of $\Tr(H_G^k)$ has
an asymptotic expansion whose $i$-th coefficient fails to
be a
\gls{B-Ramanujan function}
for $B=W_{d/2}$ (see \cite{friedman_alon}).

To remedy this, we note that \eqref{eq:rr'series} can be generalized to
\begin{equation}\label{eq:cWseries} \EE_{G \in \cC_n(B)}[ \cW_n(G,k) ] =
P^\cW_0(k) + P^\cW_1(k) n^{-1} + \cdots + P^\cW_{r-1}(k) n^{1-r} + {\rm
err}_r(n,k) \end{equation} where \( \cW(G,k) \) counts strictly
non-backtracking closed walks subject to certain restrictions---where
the restrictions (left vague for now)
are formalized as 
\emph{\glspl{walk sum}} 
in \cite{friedman_alon} and in Section~\ref{se:p1-walk-sums} of
this article---and where \( P^\cW_i(k) \) are functions given by linear
combinations of sums similar to \eqref{eq:WtimesP}, namely
\begin{equation}\label{eq:WWtimesP} 
\sum_{\vec m\cdot \vec k = k}
W^\cW_\cT(\vec m,\vec k) P(\vec k), 
\end{equation} with \( P(\vec k) \) are
before, but now \( W^\cW_\cT(\vec m,\vec k) \) counts walks in \( T \), the
underlying graph of \( \cT \), corresponding to strictly non-backtracking
closed walks with the restrictions of \( \cW \).

The idea is that if \( \cW \) counts only the strictly closed,
non-backtracking walks, \( w \), such that don't trace out a tangle, then
\( W^\cW_\cT(\vec m,\vec k) \) should be small enough so that
\eqref{eq:WWtimesP} converges to a \( B \)-Ramanujan function. In
\cite{friedman_alon} this was achieved with \( \cW \) being a
\emph{selective trace}, which is a fairly technical concept; in this paper
we instead use a 
\gls{certified trace}, in the sense of Definition~\ref{de:certified_simple}.
We remark that \( \Graph(w) \) and its associated type, \( \cT \), with
underlying graph \(T\), have the same Euler characteristic, i.e.,
$$
\chi(\Graph(w)) = \chi(\cT) .
$$
It follows that the $\Graph(w)$ which are counted, and hence then number
of $w$ counted, can be expressed purely in terms of $T$ and the path
lengths, $\vec k\in\integers_{\ge 1}^{E_T}$, i.e.,
\begin{equation}\label{eq:limited_sum}
-\chi(T) < r, \quad\mbox{and}\quad
\VLG(T,\vec k\,) < \rhoroot B
\end{equation}
Hence we restrict the sum
$$
\sum_{\vec m\cdot \vec k = k}
W^\cW_\cT(\vec m,\vec k) P(\vec k)
$$
to those $\cT$ whose graph $T$ and whose 
$\vec k\in\integers_{\ge 1}^{E_T}$ satisfy
\eqref{eq:limited_sum}, and
\( W^\cW_\cT(\vec m,\vec k) \) is simply the
original \( W_\cT(\vec m) \).
This gives us the certified trace of 
Definition~\ref{de:certified_simple}.
One crucial aspect of
the certified trace, like other simplifications such as types, is that it
leads to writing the 
\glspl{coefficient}, $P_i^\cW(k)$, of
the \gls{asymptotic expansion} of the certified trace
in terms of a finite number of simpler terms,
as described in Section~\ref{se:p0-cert}.

This certified trace
can be seen to be a good approximation to \( \Tr(H_G^k) \), for
graphs, \( G \), that do not contain tangles from 
$\HasTangle_{r,B}$
as subgraphs. Unfortunately, this sum seems
unpredictable if \( G \) does contain such a tangle. As a remedy, one can
see that if we form the same sum, but this time insist that \( G \)
contains a tangle, we also get terms analogous to \eqref{eq:WWtimesP} that
converge to \( B \)-Ramanujan functions.
Subtracting the two sums, as explained in
Subsection~\ref{sb:more_general_side},
we get a $1/n$-asymptotic expansion to
order $r$ with $B$-Ramanujan coefficients for 
$$
\EE_{G \in \cC_n(B)}[ \CertTr_{<r}(G,k) ]  -
\EE_{G \in \cC_n(B)}[ \II_{\HasTangle_{B,r}[n]} \CertTr_{<r}(G,k) ] 
$$
$$
= \EE_{G \in \cC_n(B)}[ \II_{{\rm TF}_{B,r}[n]} \CertTr_{<r}(G,k) ] ,
$$
which is, within an error term described by
Lemma~\ref{le:coincidences}, of
$$
= \EE_{G \in \cC_n(B)}[ \II_{{\rm TF}_{B,r}[n]} \Tr(H_G^k) ] .
$$
This will prove Theorem~\ref{th:main_expansion_B}, and, with
the appropriate side-stepping lemmas, yields
Theorem~\ref{th:main} and, as an easy consequence,
the relativized Alon conjecture for $d$-regular $B$ without half-loops.

\section{Preliminaries}
\label{se:p1-prelim}

In this section we give various preliminary definitions and 
technical lemmas that will be used to prove the main theorems of
this paper.

%
%

\subsection{The Order of a Graph and Pruned Graphs}
\label{sb:order_and_pruned}

Recall, from Definition~\ref{de:order}, that the 
\gls{order}
of a graph, $G$,
is just
$$
\ord(G) = -\chi(G) = |E_G| - |V_G|.
$$
Theorem~\ref{th:prob_occurs} will show that 
the expected number of 
occurrences (see Definition~\ref{de:num_occurrences}) of $\psi$ in
a graph in $\cC_n(B)$ is roughly
$n^{-\ord(\psi)}$.  To prove this theorem we shall need
some basic facts about the order.

\begin{definition}
We say that a graph is a {\em \gls[format=hyperit]{tree}} 
if it is connected and of
Euler characteristic $-1$.
We say that a graph is {\em \gls[format=hyperit]{treeless}} 
if none of its connected 
components is a tree.
We say that a graph is {\em \gls[format=hyperit]{pruned}} 
if all of its vertices have
degree at least two.
\end{definition}

The term {\em pruned} will be explained in more detail in
Subsection~\ref{sb:occur}.  
Our definition of a pruned graph is like the usual
notion of pruning, except that we do not allow isolated vertices in
the graph.
Our interested in the above notion of ``pruned'' is that we 
will need some theorems, given in this subsection, that are not
true for general graphs; we will apply such theorems to graphs that are
the union of graphs of non-backtracking walks, and such unions are 
pruned.
We remark that historically, pruning, or ``shaving off
trees'' (see \cite{stallings83}), is a procedure of interest in
various parts of graph theory, even including 
trace methods and the Alon conjecture
\cite{broder,friedman_random_graphs}, used to
``reduce'' a general closed walk, $w$,
in a graph
to a non-backtracking walk in the graph.

\begin{proposition}
\label{pr:order_under_inclusion}
For any graph without half-loops we have
\begin{equation}\label{eq:order_as_degree_sum}
\ord(\psi) = \sum_{v\in V_\psi} \bigl( \deg_\psi(v) - 2 \bigr)/2.
\end{equation}
In particular, if $\psi$ is pruned, then it is treeless.
\end{proposition}
\begin{proof} 
the first statement follows from the fact that
$\ord(\psi)=|E_\psi|-|V_\psi|$ and that each edge contributes twice
to the sum of all vertex degrees.
The second statement by applying 
\eqref{eq:order_as_degree_sum} to any connected component of a pruned
graph, i.e., a graph all of whose vertices are of degree at least
two.
\end{proof}

\begin{theorem}
\label{th:pruned_inclusion}
Let $\psi_1$ be a subgraph of 
a pruned graph, $\psi_2$.
Then the order of $\psi_1$ is
at most that of $\psi_2$.
\end{theorem}
\begin{proof}
We have
$$
\ord(\psi_2) =
\sum_{v\in V_{\psi_2}} \bigl(\deg_{\psi_2}(v) -2 \bigr)/2
\ge
\sum_{v\in V_{\psi_1}} \bigl(\deg_{\psi_2}(v) -2 \bigr)/2
$$
since each $v\in V_{\psi_2}$ has degree at least two.
For any $v\in V_{\psi_1}$ we have that its degree in $\psi_1$ is
at most its degree in $\psi_2$, and hence
\begin{equation}\label{eq:psi2_psi1_ineq1}
\ord(\psi_2) \ge 
\sum_{v\in V_{\psi_1}} \bigl(\deg_{\psi_2}(v) -2 \bigr)/2
\end{equation}
\begin{equation}\label{eq:psi2_psi1_ineq2}
\ge 
\sum_{v\in V_{\psi_1}} \bigl(\deg_{\psi_1}(v) -2 \bigr)/2
=\ord{\psi_1}.
\end{equation}
\end{proof}

\begin{theorem}
\label{th:strict_pruned_inclusion}
Let $\psi_1$ be a proper subgraph of 
a pruned graph, $\psi_2$.  Assume that each connected component of
$\psi_2$ has positive order.
Then the order of $\psi_1$ is
strictly less than that of $\psi_2$.
\end{theorem}
\begin{proof}
It suffices to show that when each connected component of $\psi_2$ has
positive Euler characteristic, then in the proof of the previous
theorem at least one of
\eqref{eq:psi2_psi1_ineq1} and 
\eqref{eq:psi2_psi1_ineq2} is a strict inequality.

Assume that some connected component, $\psi_2'$, of $\psi_2$ 
contains no vertex of $\psi_1$.  
Then
\eqref{eq:psi2_psi1_ineq1} is strict,
since the sum over $\psi_1$ vertices completely misses all the $\psi_2'$
vertices, and
the sum over $\psi_2'$ vertices is at least one.

If the assumption of the last paragraph is false, then
each connected component of $\psi_2$ contains at least one vertex of
$\psi_1$.  Since $\psi_1$ is a proper subgraph of $\psi_2$, there is
at least one edge, $e$, of $\psi_2$ that does not lie in $\psi_1$,
and there is at least one vertex, $v$, in that connected component of
$\psi_2$ that also lies in $\psi_1$.  Take a path connecting $e$
to $v$ in $\psi_2$; along this path there must be an edge, $e'$,
which doesn't lie in $\psi_1$, such that $e'$ is incident upon a
vertex, $v'$, that lies in $\psi_1$.  For this vertex, $v'\in V_{\psi_1}$,
we have
$$
\deg_{\psi_1}(v') < \deg_{\psi_2}(v'),
$$
and hence \eqref{eq:psi2_psi1_ineq2} is a strict inequality.

Hence at least one of the inequalities in \eqref{eq:psi2_psi1_ineq1} and
\eqref{eq:psi2_psi1_ineq2} is strict.
\end{proof}

%

\subsection{Convolutions} 

In computing the coefficients of various asymptotic expansions, we
will need some facts about the convolutions of functions on the
positive integers.
It is slightly more convenient to work with convolutions on the
non-negative integers; however, as we show below (see
the remark just after Definition~\ref{de:weighted_and_truncated}), this makes
little difference.

\begin{defn} Let \( g_1 \) and \( g_2 \) be two functions defined on the
non-negative integers. We define their 
\emph{\gls[format=hyperit]{convolution}}, 
\newnot{symbol:convolution}
\( g_1 \ast
g_2\), to be the function on non-negative integers given by \[ (g_1 \ast
g_2) (k) = \sum_{j=0}^k g_1(j)g_2(k-j) \] for any \( k \geq 0\).
\end{defn}

First, we make the following observation.

\begin{thm}\label{thm:conv_polyexponential} The convolution of two
\gls{polyexponential} functions in one variable of given bases is, again, 
a polyexponential function in the same bases.
Furthermore, the degree of the polynomials in the convolution is
at most one plus the sum of the degrees of the polynomials in the
polyexponentials being convolved.
\end{thm} 
\begin{proof} It
suffices to do this for the two functions
\begin{equation}\label{eq:convolution_cases}
g_1(k) =
\alpha^kk^s \quad \text{and} \quad g_2(k) = \beta^kk^t  
\end{equation}
for \( \alpha, \beta \in \CC\) and non-negative
integers \(s,t\). 
In the case \( \alpha = \beta \), we have \begin{align*}
(g_1 \ast g_2)(k) &= \sum_{j=0}^k g_1(j)g_2(k-j) \\ &= \alpha^k
\sum_{j=0}^k (k-j)^s j^t \\ \end{align*} We now examine this sum
\begin{align*} \sum_{j=0}^k (k-j)^s j^t &= \sum_{j=0}^k \sum_{i=0}^s
\binom{s}{i} k^i (-1)^{s-i} j^{s+t-i} \\ &= \sum_{i=0}^s (-1)^{s-i}
\binom{s}{i} k^i \sum_{j=0}^k j^{s+t-i} \\ \end{align*} Sums of the type \(
\sum_{j=0}^k j^p \) are known to be polynomials in \(k\) of degree \(p+1\)
and so we can conclude that \( (g_1 \ast g_2)(k) = \alpha^k p(k) \) for
some polynomial \(p(k)\) of degree at most \(s+t+1\).\\

For the case where \( \alpha \neq \beta \), we remark that 
$$
\sum_{j=0}^k x^{k-j} y^j = \frac{x^{k+1}-y^{k+1}}{x-y},
$$
which proves the theorem for $g_1 \ast g_2$ as in
\eqref{eq:convolution_cases}
for $s=t=0$.
For larger values of $s,t$, we remark that
$$
\sum_{j=0}^k x^{k-j} y^j = \frac{x^{k+1}}{x-y} + \frac{y^{k+1}}{y-x}
=
h(x,y)+h(y,x)
$$
where
$$
h(x,y) = \frac{x^{k+1}}{x-y}.
$$
Differentiating the equation
$$
\sum_{j=0}^k x^{k-j} y^j
=
h(x,y)+h(y,x)
$$
%
\( s \) times in \( x \) and
\( t \) times in \( y \) yields 
\begin{equation}\label{eq:diff_s_t}
\sum_{j=0}^k \binom{k-j}{s} \binom{j}{t}
 x^{k-j-s}y^{j-t} = 
 h_{s,t,k}(x,y) + 
 h_{t,s,k}(y,x), 
\end{equation}
where
$$
h_{s,t,k}(x,y) = \left(\frac{\partial}{\partial y}\right)^t
\left(\frac{\partial}{\partial x}\right)^s
\bigl( x^{k+1}(x-y)^{-1} \bigr)
$$
$$
=\left(\frac{\partial}{\partial x}\right)^s
\left[ x^{k+1} \left(\frac{\partial}{\partial y}\right)^t (x-y)^{-1} \right]
$$
$$
=(t-1)!(-1)^t\left(\frac{\partial}{\partial x}\right)^s 
\left[ x^{k+1} (x-y)^{-1-t} \right] .
$$
By induction on $s$, with base case $s=0$, we see that
$$
h_{s,t,k}(x,y) = 
\sum_{i=0}^s
x^{k+1-i} (x-y)^{-1-t-s+i} g_{s,t,i}(k)
$$
where, for fixed $s,t,i$, $g_{s,t,i}(k)$ 
is a polynomial in $k$ of degree $i$.
It follows that for fixed $\alpha\ne\beta$ we have that for fixed $s,t$,
$$
h_{s,t,k}(\alpha,\beta)=\alpha^k p_{s,t}(k),
$$
where $p_{s,t}$ is a polynomial of degree $s$ in $t$ (depending on
$s,t,\alpha,\beta$).  In view of \eqref{eq:diff_s_t},
for fixed $\alpha\ne\beta$ we have
\begin{equation}\label{eq:binom-alpha-beta}
\sum_{j=0}^k \binom{k-j}{s} \binom{j}{t}
 \alpha^{k-j-s}\beta^{j-t} = 
\alpha^k p_{s,t}(k)
+
\beta^k q_{t,s}(k) 
\end{equation}
where $p_{s,t}(k)$ is as above, and where $q_{t,s}(k)$ is a polynomial
of degree $t$.
But
\begin{equation}\label{eq:widetilde-g-conv}
\sum_{j=0}^k \binom{k-j}{s} \binom{j}{t}
 \alpha^{k-j-s}\beta^{j-t} = 
\alpha^{-s}\beta^{-t}(\widetilde g_1 \ast \widetilde g_2)(k),
\end{equation}
where
$$
\widetilde g_1(k) = \binom{k}{s}\alpha^k, \quad\mbox{and}\quad
\widetilde g_2(k) = \binom{k}{t}\beta^k.
$$
To establish the theorem for \eqref{eq:convolution_cases}, it suffices
to write $k^s$
as a linear combination of
$$
\binom{k}{0} , \binom{k}{1} ,\ldots,\binom{k}{s}
$$
(via Stirling coefficients) and similarly for $k^t$, and to apply
\eqref{eq:widetilde-g-conv} and \eqref{eq:binom-alpha-beta}.
\end{proof}
%
%
%
%
\subsection{Functions of Bounded Growth}

We are interested in results about the convolutions of two (or more)
functions, when one or both functions are arbitrary functions satisfying
some growth restrictions.

\begin{defn} We say that a function \(g = g(k) \) defined on non-negative
integers has \emph{\gls[format=hyperit]{growth} bounded by \( \rho \)} for some real \( \rho > 0
\), provided that there exists a positive constant \(c\) for which \[ |
g(k) | \leq Ck^C \rho^k \] \end{defn}

\begin{thm} The convolution of any finite number of functions of growth \(
\rho \) is again of growth \( \rho \).  \end{thm} \begin{proof} It is
sufficient to consider the convolution of two functions, \( g_1, g_2 \);
the general claim follows by repeated applications of the result for two
functions. So consider \[ g(k) =\!\! \sum_{\substack{k_1+k_2=k\\k_1,k_2
\geq 0}} g_1(k_1)g_2(k_2) \] with \(g_1,g_2\) having growth bounded by \(
\rho \); for some constants \(C_1, C_2\) we have \[ | g_1(k) | \leq C_1
k^{C_1} \rho^k \quad\text{and}\quad | g_2(k) | \leq C_2 k^{C_2} \rho^k \]
There are \( k+1 \) ways of writing \(k\) as the sum of two non-negative
integers, \(k_1,k_1\). For each such pair, \( (k_1,k_2) \), we have \[
|g_1(k_1)g_2(k_2)| \leq C_1 C_2 k_1^{C_1}k_2^{C_2} \rho^{k_1+k_2} \leq C_1
C_2 k^{C_1+C_2} \rho^k \] Since there are at most \(k+1\) pairs, \(
(k_1,k_2) \), we have \[ |g(k)| \leq C_1C_2  (k+1)k^{C_1+C_2} \rho^k \] and
so it follows that \(g\) has growth bounded by \( \rho \).  \end{proof}

For reasons that go back to \cite{friedman_random_graphs}, we will need to
convolve polynomials and polyexponentials with functions for which we only
have a growth rate bound.

\begin{defn} For any polynomial, \( P=P(x) \), with real or complex
coefficients, we define its 
\emph{\gls[format=hyperit]{coefficient norm}}, \( \|P\| \) as the
largest absolute value among its coefficients in its expansion by powers of
\( x \); i.e., \[ \| s_0 + x s_1 + \cdots + x^t s_t  \| = \max_i |s_i| \]
\end{defn}

\begin{thm}\label{th:conv_polyexp_growth} For every non-negative integer \(
D \) there is a constant, \( C_2=C_2(D) \), for which the following holds.
Let \( Q=Q(x) \) be any polynomial of degree at most \( D \), and \( h(r)
\) a function defined on non-negative integers, \( r \), for which \[
|h(r)| \leq  C_1 r^D \rho^r \] for some positive constants, \(C_1 \) and \(
\rho \) with \( \rho<1 \). Then the infinite sum, \[ q(x) =
\sum_{s=0}^\infty Q(x-s)h(s) \] converges in coefficient norm, and has
degree at most that of \( Q \), and satisfies \[ \| q \| \leq C_1 C_2
(1-\rho)^{-2D} \| Q \| \] The same is true of \[ q_k(x) =
\sum_{s=k+1}^\infty Q(x-s)h(s) \] except the norm coefficient bound of \(
q_k \) is given by \[ \| q_k \| \leq C_1 C_2 k^{2D} (1-\rho)^{-2D} \rho^k
\| Q\| \] Hence \[ | (Q * h)(k) - q(k) |  = | q_k(k) | \leq C_1C_2 \|Q\|
k^{3D} \rho^k (1-\rho)^{-2D} \] \end{thm} \begin{proof} The statement on
the convergence and norms of \( q \) and \( q_h \) follow from Lemma~8.8 of
\cite{friedman_alon}; the fact that the sum begins with \( r=0 \) just
affects the constants slightly. For ease of reading, we review the proof,
which goes back to \cite{friedman_random_graphs} (as the last step of the
proof of Sublemma~2.16).  The idea is that we can restrict ourselves to the
cases where \( Q(x)=x^i \) for some \( i \) between \( 0 \) and \( D \).
Then we write \begin{equation}\label{eq:convolve_h} \sum_{r=0}^\infty
Q(x-r) h(r) = \sum_{r=0}^\infty  (x-r)^i h(r) \end{equation} and we expand
the \( (x-r)^i \) via the binomial theorem, and use an identity such as \[
\sum_{r=0}^\infty \binom{r}{\tau} \rho^r =
\frac{\rho^\tau}{(1-\rho)^{\tau+1}} \] to show the convergence to \( q \)
for \( q_k \), and obtain the coefficient norm estimate.

For the last statement, we see that \[ \sum_{r=0}^\infty Q(k-r) h(r) \] is
absolutely convergent, and hence \begin{align*} (Q * h )(k) &= \sum_{r=0}^k
Q(x-r) h(r) \\ &= \sum_{r=0}^\infty Q(k-r) h(r) - \sum_{r=k+1}^\infty
Q(k-r) h(r) \\ &= q(k) - q_k(k) \end{align*} But for \( k\geq 1 \), we have
\( q_k(x) \) is bounded in absolute value by \( D+1 \) times its largest
monomial, and hence \[ |q_k(k) | \leq \|q_k\| (D+1) k^D \leq C_1 C_2 k^{2D}
(1-\rho)^{-2D} \rho^k \| Q\| (D+1)k^D \] \[ \leq C_1 \|Q\| \rho^k
(1-\rho)^{-2D} C_2 k^{3D} \] \end{proof}

\begin{cor}\label{cor:polyexponential_conv_error_Ramanujan} Fix a 
graph, $B$, with $\rho(H_B)>1$.
Fix a
\( D>0 \) and an \( \varepsilon>0 \). 
Then for each \( \mu \in \CC
\) with \( |\mu| > \rhoroot B + \varepsilon\) there exists a constant, \(
C_2 \), for which the following is true.  If \( P(k)=k^i \mu^k \) with
integer \( i \) with \( 0\leq i\leq D \), and \( h(r) \) satisfies \(
|h(r)|\leq C_1 r^D (\rhoroot B+\varepsilon)^r \), for some constant \( C_1
> 0 \), then we have that \[ p(x) = \sum_{i=0}^\infty P(x-i) h(i) \]
converges to a polyexponential function \( p(k) \), and \[ |(P * h)(k) -
p(k) | \leq C_2 (\rhoroot B+\varepsilon)^k \] \end{cor} 
\begin{proof} We
have \[ (P * h)(k) = \mu^k (\widetilde P * \widetilde h), \] where \[
\widetilde P(k) = k^i, \quad \widetilde h(k) = \mu^{-k} h(k). \] Hence \[
|\widetilde h(k) | \leq C_1 r^D \rho^k, \] where \[ \rho \leq
\frac{\rhoroot B+\varepsilon}{|\mu|}\ . \] Hence \( 1-\rho \), for a given
\( d \) and \( \varepsilon \), is bounded away from zero.  Now we apply
Theorem~\ref{th:conv_polyexp_growth}.  \end{proof}

%
 
%
%
%
%
\subsection{Weighted Convolutions of $B$-Ramanujan Functions}
\label{sb:Ramanujan}


\begin{defn} 
\label{de:weighted_and_truncated}
Let \( g_1, \ldots, g_s \) be functions on non-negative
integers. Let \( \vec m = (m_1, \ldots, m_s) \) be a \(s\)-tuple of
positive integers. The 
\emph{\gls[format=hyperit]{weighted convolution} of the \(g_i\) with weights
\( \vec m \)}, denoted \( (g_1 \ast \ldots \ast g_s)_{\vec m} \),
is defined to be \[ (g_1 \ast \ldots \ast g_s)_{\vec m} (k) =
\sum_{\substack{\vec m \cdot \vec k = k\\ \vec k \geq \vec 0}} g(k_1)\ldots
g(k_s); \] 
if \( \vec{k}_0 = (k^0_1, \ldots, k^0_s) \) is a \(s\)-tuple
of non-negative integers, we define the \emph{truncated weighted
convolution}, denoted \( (g_1 \ast \ldots \ast g_s)_{\vec m}^{\vec{k}_0}
\), to be \[ (g_1 \ast \ldots \ast g_s)_{\vec m}^{\vec{k}_0} (k) =
\sum_{\substack{\vec m \cdot \vec k = k\\ \vec k \geq \vec{k}_0}}
g(k_1)\ldots g(k_s). \] This changes very little for the analysis, but is
what we will be using for our certified trace.  \end{defn}
Notice that if we want to work with functions, $g_i$, defined on
positive integers, as opposed to non-negative integers,
we can work with the above truncated weighted convolution with
$\vec k_0$ replaced with $\max(\vec 1,\vec k_0)$, where $\vec 1$
is the vector whose components are all equal to $1$.

\begin{ex} 
\label{ex:not_polyexponential}
Consider the case where \( g_1(k) = g_2(k) = (d-1)^k \) and \(
\vec m = (1,2) \), then the identity on partial sums of geometric series \[
1 + x + \ldots + x^t = \frac{x^{t+1}-1}{x-1} \] with \( x = d-1 \) easily
gives \[ (g_1 \ast g_2)_{\vec m}(k) = \begin{cases} \frac{1}{d-2} \left(
(d-1)^{k+1} - (d-1)^{k/2} \right) & \text{if \( k \) is even,} \\
\frac{1}{d-2} \left( (d-1)^{k} - (d-1)^{(k+1)/2} \right) & \text{if \( k \)
is odd.} \end{cases} \] This is not an exact polyexponential function, but
it is a \( B \)-Ramanujan function.  \end{ex}

In Section~\ref{se:p2-modl} we remark that the above convolution can
be described as a {\em Mod-$S$} function, and the exact value of such
convolutions of polyexponential functions may be of interest.

\begin{thm} 
\label{th:weighted-convolution-is-Ramanujan}
Let \( \vec m \geq \vec 1\) and \( \vec{k}_0 \geq \vec 1 \).
Let \( g_1, \ldots, g_s \) be \(B\)-Ramanujan functions, then the
(truncated) weighted convolution \[ (g_1 \ast \ldots \ast g_s)_{\vec
m}^{\vec{k}_0} \] is \( B \)-Ramanujan as well.  \end{thm} 
\begin{proof}
First, it suffices to prove this in the case $s=2$; the case of general
$s$ follows by repeated application of the $s=2$ result.
Hence $\vec k=(k_1,k_2)$ and $\vec m=(m_1,m_2)$.

Second, observe that a change of variable \[ \kappa_1 = k_1 - k_1^0 + 1
\qquad \kappa_2 = k_2 - k_1^0 + 1 \] allows us, without loss of generality,
to remove the truncated condition. So, we assume that \( \vec k_0 = \vec 1
\).

At this point we need to check a few special cases.  Namely,
if the function \(g_i\) is \(B\)-Ramanujan, then we can write \[ g_i(k) =
\sum_{\ell \in L} p_{\ell,i}(k) \ell^k + \error_i(k) \] where \( p_{\ell,i}
\) is a polynomial,
where $L$ is the set of eigenvalues of $H_B$ (it suffices to take $L$ to
be
the eigenvalues greater than $\rhoroot B$ in absolute),
and where the error bound is such that for any \(
\varepsilon > 0\), we have \[ | \error_i(k) | 
\leq C(\rhoroot B +
\varepsilon)^k \] for some constant \( C = C(i, \varepsilon) \). 
The set \(
L \) is always finite; hence we it suffices to consider
we can assume that for $i=1,2$ we have either $g_i$ is a simple
polyexponential function, of the form
\[ g_i(k) = p_{\ell, i}(k) \ell^k 
\]
or $g_i(k)$ is an ``error term'' function,
i.e., for each $\varepsilon>0$ there is
a $C$ for which
\begin{equation}\label{eq:gi_is_error_term}
|g_i(k)| \le C\bigl( \rhoroot B + \varepsilon \bigr)^k .
\end{equation}
and treat each of the four corresponding
cases for \( i = 1,2 \); of course, there are essentially three
cases, namely the $g_i$ are (1) both ``error terms,'' (2) both
simple polyexponentials as above, and (3) one of each.
Note that in either case, we have that \[ | g_i(k)
| \leq \widetilde C k^{\widetilde C} \rho^k(H_B) \] for some constant \(
\widetilde C \), since \(\rho(H_B)\) is the largest possible Hashimoto
eigenvalue.

Let us consider now, the case where both \( g_i \) are error terms
functions. 
Given any $\widetilde\varepsilon>0$, choose an $\varepsilon$ with
$0<\varepsilon<\widetilde\varepsilon$.
Since the $g_i$ are both error terms,
there is a $C$ such that for $i=1,2$
we have that \eqref{eq:gi_is_error_term} holds, and hence
\begin{align*} | (g_i \ast g_2)(k) | &\leq
\sum_{m_1k_1+m_2k_2=k} C^2 (\rhoroot B + \varepsilon)^{k_1+k_2} \\ &\leq
C^2 \cdot \# \{ (k_1,k_2) \in \ZZ_{\geq 1}^2 \mid m_1k_1+m_2k_2=k \} \cdot
(\rhoroot B + \varepsilon)^k 
\\
& \le C^2 \cdot (k-1)
(\rhoroot B + \varepsilon)^k ,
\end{align*} 
since for each of the $k-1$ possible values of $k_1=1,\ldots,k-1$ there
can be at most one value of $k_2$.
Since we have $\varepsilon<\widetilde\varepsilon$, for some $C'$ we have
$$
C^2 \cdot (k-1)
(\rhoroot B + \varepsilon)^k 
\le C' 
(\rhoroot B + \widetilde\varepsilon)^k 
$$
for all integers $k\ge 0$.
Hence $(g_1 \ast g_2)_{\vec m}$ is $B$-Ramanujan. 
%
%

For each of the remaining
cases regarding the $g_i$, we will consider
separately three subcases: (1) when \(m_1=m_2=1\), (2) when
\(m_1,m_2 \geq 2\) and (3) when \(m_1=1, m_2 \geq 2 \).

Consider (1), when \(m_1=m_2=1\). Then if \( g_1\) and \(g_2\) are both
polyexponential functions, Theorem~\ref{thm:conv_polyexponential} shows
that their convolution is polyexponential as well and hence a
\(B\)-Ramanujan function. In the case where \( g_1 \) is polyexponential
while \( g_2 \) is an error term,
Corollary~\ref{cor:polyexponential_conv_error_Ramanujan} shows that their
convolution is a \(B\)-Ramanujan function.

We can treat both cases when \( m_1, m_2 \geq 2 \) simply using the fact
that \( g_i(k) \leq \widetilde C k^{\widetilde C} \rho^k(H_B)\). Indeed, we
have \begin{align*} | (g_i \ast g_2)(k) | &\leq \sum_{m_1k_1+m_2k_2=k}
|g_1(k_1)g_2(k_2)| \\ &\leq \sum_{m_1k_1+m_2k_2=k} C_1 k_1^{C_1} k_2^{C_1}
\rho^{k_1+k_2}(H_B)) \\ &\leq (\rho(H_B))^{k/2} \sum_{m_1k_1+m_2k_2=k} C_1
k_1^{C_1} k_2^{C_1} \end{align*} since if \(m_1, m_2 \geq 2 \) we have that
\( k_1+k_2 \leq k/2 \).
Since, as before, there are at most $k-1$ tuples $(k_1,k_2)$ for which
$m_1k_1+m_2k_2=k$, we have
Hence, for some constant \( C_2 \) we have
$$
| (g_i \ast g_2)(k) | \le (\rho(H_B))^{k/2} C_2k^{C_2},
$$
which is a $B$-Ramanujan function.

Finally, let us consider (3), when \(m_1 = 1\), \(m=m_2 \geq 2\).
We need to consider the case where $g_1,g_2$ are either both
simple polyexponentials, or one is a simple polyexponential and the
other an error term.  

For the case where $g_1$ is an error term means that for any
$\varepsilon>0$ there is a $C$ for which
$$
g_1(k_1)g_2(k_2) \le C (\rho(H_B)+\varepsilon)^{k/2}
$$
for any $k_1$ and $k_2$ for which $m_1 k_1+m_2 k_2=k$, since $m_2\ge 2$,
and this case is handled
like the case where $g_1,g_2$ are both error terms.

The case where $g_1$ is a polyexponential, means that if we set
$k'=m_2 k_2$, we have 
$$
(g_1 \ast g_2)_{\vec m} = (g_1 \ast g_2')_{(1,1)},
$$
where $g_2'(k')$ is defined to be zero if $k'$ is not a multiple of $m_2$,
and
otherwise $g_2'(k')=g_2(k'/m_2)$.  Hence $g_2'$ satisfies the bounds of
an ``error term,''
and this case reduces to the case where $g_1$ is polyexponential
and $g_2'$ is an error term as above.


\end{proof}

\subsection{Walk Statistics}

One important lemma in \cite{friedman_random_graphs,friedman_alon}
described the number of walks of length $k$ in a graph, and a weighted
sum of such walks, where the weights are polynomials in the number of
time certain vertices are visited.  Such results concern adjacency
matrices in our case, but are much more general.

\begin{lemma}\label{le:matrix_entry_sum}
Let $M$ be an $r\times r$ matrix with complex entries.  Fix an integer,
$t$, and fix any $2t$ integers between $1$ and $r$,
$b_1,\ldots,b_{2t}$.
For non-negative integer, $k$, let
$$
f(k) = \sum_{\substack{k_1,k_2,\ldots,k_t \\ k_1+\cdots+k_t=k}}
(M^{k_1})_{b_1,b_2} \ldots
(M^{k_t})_{b_{2t-1},b_{2t}},
$$
where the sum is over all non-negative integers, $k_i$, whose sum is $k$.
Then we have
\begin{equation}\label{eq:M_poly_p}
f(k) = \sum_{\nu\in\Spec(M)} \nu^k p_\nu(k),
\end{equation}
where $\nu$ ranges over all the eigenvalues of $M$, and 
$p_\nu(k)$ is a polynomial in $k$.
Furthermore, if the Jordan canonical form of $M$ has blocks whose 
block size
are at most $s$, then each $p_\nu$ is of degree at most $ts$;
in particular, if $M$ is diagonalizable, then each $p_\nu$ is of
degree at most $t$.
\end{lemma}
\begin{proof}
It follows from the Jordan canonical form of $M$ that for each
$b,b'\in\{1,\ldots,r\}$ we have that for non-negative integer, $k$, we have
$$
f(k) = (M^k)_{b,b'}
$$
is of the form given in \eqref{eq:M_poly_p}, where the degree of each
$p_\nu$
is bounded by the size of the largest Jordan block of $M$.
Now we apply 
Theorem~\ref{thm:conv_polyexponential}.
\end{proof}

This gives us an important fact about what we
call a ``weighted sum of matrix power entries of $M$'' as follows.
\begin{corollary}\label{co:weighted_sum}
Let $x_1,\ldots,x_s$ be indeterminates, and let $M$ be an $r\times r$
matrix, each
of entries is a complex number times one of the $x_i$ ($1\le i\le s$).
Let $M_1$ be $M$ where each $x_i$ is set to $1$.
For each tuple of non-negative integers, $\vec t=(t_1,\ldots,t_s)$, define
${\rm Weight}_{\vec t}$ on monomials in $x_1,\ldots,x_s$ via
$$
{\rm Weight}_{\vec t\;}( x_1^{e_1}\ldots x_s^{e_s} )
=
e_1^{t_1}\ldots e_s^{t_s} 
$$
where $e_1,\ldots,e_s$ are non-negative integers;
extend ${\rm Weight}_{\vec t}$ to be a function on all polynomials
in $x_1,\ldots,x_s$ by linearity, i.e., for
$$
p(x_1,\ldots,x_s) = \sum_{e_1,\ldots,e_s} c_{e_1,\ldots,e_s}
x_1^{e_1}\ldots x_s^{e_s} ,
$$
where $c_{e_1,\ldots,e_s}\in\complex$, and the above sum is over a 
finite set of tuples of non-negative integers $(e_1,\ldots,e_s)$, we define
$$
{\rm Weight}_{\vec t\;}(p(x))  =
\sum_{e_1,\ldots,e_s} c_{e_1,\ldots,e_s}
e_1^{t_1}\ldots e_s^{t_s}.
$$
Then for any fixed $\vec t$, we have
$$
f(k) = {\rm Weight}_{\vec t\;}\bigl( M^k \bigr) 
$$
is a function of the form 
\eqref{eq:M_poly_p},
where the $\nu$ range over the eigenvalues of $M_1$.
If the largest Jordan block of $M_1$ is of size $\ell$, then the
degree of each $p_\nu$ is at most 
$(t_1+\ldots+t_r)\ell$.
\end{corollary}
\begin{proof}
Again, because $1,z,z^2,\ldots,z^t$ are linear combinations of
$$
1,z,\binom{z}{2},\ldots,\binom{z}{t},
$$ 
it
suffices to show this for 
$$
f(k)=
\Biggl(
\left(\frac{\partial}{\partial x_1}\right)^{t_1} 
\ldots
\left(\frac{\partial}{\partial x_s}\right)^{t_s} 
M^k
\Biggr)_{x_1=\cdots=x_s=1}.
$$
By the product rule, this derivative is the same as the 
sum of each way of applying the
$t=t_1+\cdots+t_r$ partial derivatives to each of the $k$ $M$'s that 
form the product $M^k$; furthermore, any two applications of a partial
derivative to the same $M$ results in zero.  Hence, these $t$
partial derivatives amount to a sum
$$
M_1^{k_0} M[1] 
M_1^{k_1} M[2]
\ldots
M_1^{k_t} M[t]
M_1^{k_{t+1}},
$$
over all $k_1+\cdots+k_t + k_{t+1}= k-t$, 
with a sum over all $M[1],\ldots,M[t]$,
where $M[i]$ is a partial derivative
of $M$ with respect to one of $x_1,\ldots,x_s$, depending on the order
that the differentiation operators are applied to the product of
$k$ of the $M$'s.
For each fixed $M[1],\ldots,M[t]$, expanding in the entries of the 
$M[i]$ yields a sum as in
Lemma~\ref{le:matrix_entry_sum} (with $t$ being $t+1$).
\end{proof}

In this paper, we use Corollary~\ref{co:weighted_sum} only in the following
special case.

\begin{corollary}\label{co:count_occurrences}
Let $B$ be a graph, and for each $e\in E_B$, let us fix a non-negative
integer $\ell_e$.  Consider for any $e_1,e_2\in \Edir_B$
$$
f(k;e_1,e_2) = \sum_{w\in \NB_k(e_1,e_2,B)} \ \  \prod_{e\in E_B}
\bigl( a_e(w)\bigr)^{\ell_e} ,
$$
where the sum is over all non-backtracking walks in 
$B$ of length $k$ whose first edge is $e_1$ and whose last edge
is $e_2$, and $a_e(w)$ denotes the number of times $w$ traverses
a directed
edge in the equivalence class of $e$ (i.e., if $e$ is not a half-loop,
then there are two directed edges in the equivalence class of $e$).
Then $f(k)$ is a function of the form \eqref{eq:M_poly_p}, where
$\nu$ ranges over all Hashimoto eigenvalues of $B$.
\end{corollary}
\begin{proof}
Consider the indeterminates $\{x_e\}_{e\in E_B}$,
and the matrix, $M$, obtained by
taking the Hashimoto matrix of $B$ and replacing each non-zero entry
in the $(e,e')$ entry 
by the entry $x_e$.
Now apply Corollary~\ref{co:weighted_sum}.
\end{proof}

\section{Walk Sums}
\label{se:p1-walk-sums}

In this section, we show that the expected value in $\cC_n(B)$ of the
$k$-th power
Hashimoto trace has an 
\gls{asymptotic expansion} which
satisfies, for any integer \( r \), 
\[ 
\expect{G\in\cC_n(B)}{  \Tr(H_G^k) } = f_0(k) +
\frac{f_1(k)}{n} + \ldots + \frac{f_{r-1}(k)}{n^{r-1}} +
\frac{\error(r,k)}{n^r} \] where the \( f_i \) are functions of \(k\) and
with an upper bound on the error term of \[ | \error(r,k) | \leq ck^{2r+2}
(\rho(H_B))^k \] 
for some $c$ depending only on $r$ and $B$.
Furthermore, such an expansion is obtained for a broader notion
of trace that we call a \emph{walk sum}.

To obtain this result, we observe that the trace of the  
$k$-th
power of the Hashimoto matrix is the number of strictly non-backtracking
closed walks of length \(k\). We generalize this notion with the definition
of a walk sum, which counts the number of walks satisfying certain
properties, and in this context, show that the expected number of walks has
the desired asymptotic expansion. This follows and expands on the theory
that was developed by Friedman in \cite{friedman_random_graphs} and used in
\cite{friedman_alon}.

%
%
%
%
\subsection{Potential Walks}\newnot{symbol:potwalk} We consider a
probabilistic theory describing walks on a random cover of the base graph
\(B\). Consider the probability space \( \cC_n(B) \). For the purpose of
this section, the base graph does not need to be \( d \)-regular.

\begin{defn}
\label{de:potential_walk}
A \emph{$(k,n)$-\gls[format=hyperit]{potential walk}} 
is a pair \( (w; \vec t\,)
\) consisting of a walk, called the \emph{base walk} \[ w =
(v_0,e_1,v_1,e_2, \ldots, e_k, v_k) \] in the base graph \(B\) of length
\(k\); and a vector, called the \emph{trajectory} \[ \vec t = (t_0, \ldots,
t_k) \] with each \(t_i \in \{1, \ldots, n\} \). We refer to \(k\) as the
\emph{length} and \(n\) as the \emph{size} of the potential walk. Given \(
G \in \cC_n(B) \), we say that a $(k,n)$-potential walk \( \potwalk \)
is \emph{attained in \(G\)} if the vector \( \vec t \) represents an
actual walk in $G$,
that is, with the above notations, if 
\begin{equation}\label{eq:attained}
\sigma_{e_i}(t_{i-1}) = t_i \quad
\forall i=1,\ldots, k ,
\end{equation}
where \(G = B[\sigma]\) and 
$\sigma\from\Edir_B\to\cS_n$ is the permutation assignment from
which $G$ arises.
We denote by \( \cE\potwalk \) the event in \(
\cC_n(B) \) that the potential walk is attained in \(G\) and denote by \(
P\potwalk \) the probability of this event. 
A $(k,n)$-potential walk is
\emph{feasible} if \( P\potwalk > 0 \).  
If a potential walk, $\potwalk$,
is attained in some $G\in\cC_n(B)$, then $\potwalk$ represents a
walk, $w'$, in $G$, and we denote its graph, in
the sense of Definition~\ref{de:graph_of_walk} (i.e., the graph formed
by the vertices and edges occurring in $w'$) by 
\newnot{symbol:graphpotwalk}
$\Graph\potwalk$,
and call it the
{\em graph of the potential walk, $\potwalk$};
clearly $\Graph\potwalk$ depends only on $\potwalk$, and this
graph comes with a natural morphism to $B$ via the walk, $w$, in $B$.
\end{defn}


Proposition~\ref{pr:etale_is_feasible} shows that a potential walk
is feasible iff the morphism from the graph of the potential walk to
$B$ is \'etale.  Below we give an alternative characterization 
of the feasibility of a potential walk, akin
to that in \cite{friedman_alon}, Chapter 5, Section 1.  We shall
not need it here, but the reader may find it helpful to gain intuition
regarding this concept.  This lemma is a corollary of
Proposition~\ref{pr:etale_is_feasible}, although it not hard to 
prove without this proposition.

\begin{lem} A $(k,n)$-potential walk \( \potwalk \), with
$$
w=(v_0,e_1,v_1,\ldots,e_k,v_k),
$$
is feasible if and
only if the following two feasibility conditions hold:
\begin{enumerate}[label= (\arabic*)] \item for any \( i, j \) such that 
$e_i=e_j$
we have \[ t_{i-1} = t_{j-i} \iff t_i = t_j \ ;
\] 
and
\item and for any \( i,j \) such that 
$e_i = e_j^{-1}$,
we have \[ t_{i-1} = t_j \iff t_i = t_{j-1}\ . \]
\end{enumerate} \end{lem}

\begin{defn}\label{defn_associated_graph}
Consider a $(k,n)$-potential walk \( \potwalk \),
and its graph, $\Graph\potwalk$. 
If \( e \in E_B \) is an edge in the base graph, we define \( a_e
= a_e \potwalk \) to be the number of edges in the associated graph of the
potential walk in the fibre of the edge \(e\). And if \( v \in V_B \) is a
vertex in the base graph, we define \(b_v = b_v\potwalk \) to be the number
of vertices in the associated graph of the potential walk which is in the
fibre of the vertex \(v\). Hence we have that \[ \sum_{e \in E_B}a_e =
|E_{\Graph\potwalk}| \quad\text{ and }\quad \sum_{v \in V_B}b_v =
|V_{\Graph\potwalk}| \] and we denote by \( \chi\potwalk =
|V_{\Graph\potwalk}| - |E_{\Graph\potwalk}| \) the \emph{Euler
characteristic} of the associated graph. Following the work of Friedman, we
define the \emph{order} of a potential walk \( \potwalk \) to be the
non-negative integer \( \ord\potwalk = - \chi\potwalk \).
\end{defn}\newnot{symbol:orderpotwalk}

We will eventually regroup potential walks with similar graphs under our
notion of forms and types. First we describe an obvious equivalence class
among potential walks and their associated graph.

\begin{defn} Fix a base walk \( w = (v_0, e_1, v_1, \ldots, e_k,v_k) \). We
say that two trajectories of length \(k\) and size at most \(n\), \( \vec t
\) and \( \vec s \), \emph{differ by a symmetry} and denote this by \( \vec
t \sim_w \vec s \) if they only differ by a permutation of the vertices in
each fibre; that is if there exists a collection of permutations of \(n\)
elements \( \tau = \{ \tau_v \}_{v \in V_B} \) such that \( \vec s =
\tau(\vec t\,) \) that is if \[ s_i = \tau_{v_i}(t_i) \quad i=0, \ldots k
\] This is a modification of the original definition in
\cite{friedman_alon} that accounts for the multiple vertices of the base
graph. Indeed, consider the following two examples. First, if the base walk
is a loop, then the trajectories \( (1,1) \) and \( (1,2) \) are not
equivalent since the fist one yields a loop and the second an edge; but if
the base walk is not a loop, then these trajectories are equivalent since
they both yield an edge (they refer to vertices in different fibres of the
base graph). With this definition we guarantee that if \( \vec t \sim_w
\vec s \) then \( P\potwalk = P(w;\vec s\,) \). Define the \( (n,w)
\)\emph{-equivalence symmetry class of} \( \vec t \) to
be\newnot{symbol:vectnw} \[ [\vec t\,]_{n,w} = \{ \vec s \mid \vec s \sim_w
\vec t \text{ and } s_i \leq n, \, i=1,\ldots,k \} \] We may omit the \(n\)
subscript of \(n\) is understood. We denote by \( \potwalkclass  \) the
equivalence class of all potential walk with base walk \(w\) whose
trajectories are in the equivalence class \( [\vec t\,]_{n,w} \), that
is\newnot{symbol:potwalkclass} \[ \potwalkclass = \{ (w; \vec s\,) \mid
\vec s \in [\vec t\,]_{n,w} \} \] and we also say that these potential walk
\emph{differ by a symmetry}. Clearly, if two potential \((k,n)\)-walks
differ by a symmetry, then their associated graphs are canonically
isomorphic (the isomorphism being given by the collection of permutation)
and hence the quantities \(a_e\) and \( b_v \) are identical in the
equivalence class for all \(e\) and \(v\).  \end{defn}
%
%
%
%
\subsection{Walk Sums}

We make the following observation. If \(H_G\) is the Hashimoto matrix of a
random covering graph \(G\), then \[ \EE_n \left[ \Tr(H_G^k) \right] =
\sum_{\potwalk \in \cW}P\potwalk \] where \( \cW \) is the set of all
potential \((k,n)\)-walks \(\potwalk\) whose walk \(w\) is a closed
strictly non-backtracking walk of length \(k\) in the base graph \(B\). We
want to generalize this to allow the use of some other ``traces,'' that is
sums of \( P\potwalk \) over appropriate \(w\)'s and \( \vec t\)'s.

\begin{defn}
\label{de:walk_collection}
\newnot{symbol:walkcollection}\newnot{symbol:Wkn} A 
\emph{walk collection}, 
\( \mathcal{W} = \{ \mathcal{W}(k,n) \}_{k,n \geq 1} \), is a
collection where for any two positive integers \( k \) and \( n \) the
elements of the set \( \mathcal{W}(k,n) \) are $(k,n)$-potential walks
satisfying the following conditions:  
\begin{enumerate}[label= (\arabic*)]
\item \emph{Symmetry:} The walk collection contains all trajectories that
differ by a symmetry; that is, if \( \potwalk \in \mathcal{W}(k,n) \) then
\( \potwalkclass \subseteq \mathcal{W}(k,n) \).  
\item \emph{Size
invariance:} The walk collection reflects the fact that the size of a walk
isn't uniquely determined; that is, if \( \potwalk \in \mathcal{W}(k,n) \)
then \( \potwalk \in \mathcal{W}(k,m) \) for all \( m \geq \max\{ t_i \mid
\vec t = \left( t_0, t_1, \ldots, t_k \right) \} \) 
\item \emph{Strictly non-backtracking closed
walks:} All associated graphs of potential walks in the walk collection 
are strictly
non-backtracking closed walks in the sense of 
Definition~\ref{defn:walks}; that is, in the notation of
Definition~\ref{de:potential_walk}
we have $v_0=v_k$,
$t_0=t_k$, no two successive edges
of $w$  are inverses of each other, and the same for $e_k$ and $e_1$.
\end{enumerate} 

We remark that this is a generalization of the notion of a
{\em SSIIC walk collection}, Definition~5.2 of \cite{friedman_alon}, with 
{\em SSIIC} an acronym for ``symmetric, size invariant, irreducible,
closed;''  however in this paper we are only interested in walk
sums that satisfy conditions~(1)--(4) above, and hence we make this
part of the definition.

Given a walk collection, \( \mathcal{W} \), its associated
\( (k,n) \)\emph{-modified $\mathcal{W}$ trace} (we often drop the
$\mathcal{W}$ when it is understood)
is the random variable, \( \mathcal{W}_n
\), defined on \( \cC_n(B) \), 
given by
\[ \mathcal{W}_n(G,k) = \# \{ \potwalk
\in \mathcal{W}(k,n) \mid \potwalk 
\text{ is attained in } G \} \] if \(G=G[\sigma]
\in \cC_n(B) \),
meaning that \eqref{eq:attained} is satisfied for the permutation
assignment $\sigma$ that gives rise to $G$.
(In general, it is usually simpler to write $G$ rather than
$G[\sigma]$.)

Given a walk collection, \( \mathcal{W} \), its associated
\((k,n) \)\emph{-walk sum} is \[ \WalkSum(\mathcal{W},k,n) = \sum_{\potwalk
\in \mathcal{W}(k,n)} \!\!\!\!\!\!P\potwalk \] Note that this sum is always
a finite sum (assuming the base graph \( B \) is finite).
\end{defn}\newnot{symbol:WnGk}\newnot{symbol:WalkSumWkn}

\begin{ex}
\label{ex:SNBC}
\newnot{symbol:SNBC} We denote by \( \SNBC=\{ \SNBC(k,n)\}_{k,n
\geq 1} \) the walk collection of all strictly non-backtracking closed
walks of length \(k\) and size \(n\). Its associated \( (k,n) \)-modified
trace is simply the $k$-th Hashimoto trace: \[ \SNBC_n(G,k) =
\Tr(H_G^k) \] (thought of as a random variable on \( \cC_n(B) \)) and its
associated \( (k,n) \)-walk sum is simply the expected value of the
$k$-th Hashimoto trace: \[ \WalkSum(\SNBC,k,n) = \EE_n[ \Tr(H^k_G)
] \] This walk collection is the largest possible, that is, any other walk
collection consists of subsets of this one.  \end{ex}

Since walk collections are symmetric, we can regroup the terms of the above
sum by grouping trajectories that differ by a symmetry.  
\begin{defn} \label{de:symmetrized_sum}
Let
\[ \EE_{\rm{symm}}\potwalk_n = \sum_{\vec s \in [\vec t \, ]_{n,w}}
P(w;\vec s\,) \] then we have
\begin{equation}\label{eq:WalkSum_by_walk_classes}
\WalkSum(\mathcal{W},k,n) = \!\!\!\!\! \sum_{(w;[\vec t\,]_{n,w}) \in
\mathcal{W}(k,n)} \!\!\!\!\! \EE_\text{symm}\potwalk_n \end{equation}
\end{defn}

Now we comment on \(\EE_{\rm{symm}}\potwalk_n\) and will later use this to
work out an asymptotic expansion of a walk sum.

\begin{prop}\label{prop:EsymmProd} Let \( \potwalk \) be a potential
\((k,n)\)-walk then we have \begin{equation}\label{eq:walk_probability}
\EE_{\rm{symm}}\potwalk_n = \prod_{v \in V_B} \frac{n!}{(n-b_v)!} \prod_{e
\in E_B} \frac{(n-a_e)!}{n!} \end{equation} with \( a_e \) and \( b_v \) as
in Definition~\ref{defn_associated_graph}, provided that \( B \) has no
half-loops; for each \( e \in E_B \) that is a half-loop, we replace \[
\frac{(n-a_e)!}{n!} \quad \text{ by } \quad \frac{(n-2a_e)\oddf}{n\oddf} \]
for $n$ even, where for $m$ even, $m\oddf$ is the {\em odd factorial}
(also sometimes called the {\em double factorial} in the literature)
$$
m\oddf = (m-1)(m-3)\ldots \cdot 3\cdot 1.
$$
\end{prop} 
\begin{proof} Since all potential walks considered differ by a
symmetry, the probability that they occur in a random cover is the same.
Without loss of generality, we can assume that \( \vec t\, \in [\vec
t\,]_{n,w} \) and hence we have \begin{align*} \EE_{\rm{symm}}\potwalk_n &=
\sum_{\vec s \in [\vec t \, ]_{n,w}} P(w;\vec s\,) \\ &= \left| [\vec
t\,]_{n,w} \right| \cdot P\potwalk \end{align*} and we evaluate each term
separately.

First to compute the size of the equivalence class \( [\vec t\,]_{n,w} \),
we only need to consider how we can shuffle the vertices in the fibre of a
vertex of the base graph. If there are \( b_v \) elements of the trajectory
\( \vec t\, \) in the fibre of the vertex \( v \) then there are \(
n(n-1)\ldots (n-b_v+1) \) ways to choose those vertices among \( n \)
possible and this is true for the fibre of each vertex \( v \in V_B \) of
the base graph.

A random covering \( B[\sigma] \) is given by the choice of a collection of
permutations \( \sigma = \{ \sigma_e \}_{e \in E_B}\). Assume \( B \) has
no half-loops; if there are \( a_e \) edges of the associated graph to the
potential walk that are in the fibre of the base edge \(e\), then it means
that there are \( a_e \) values of the permutation \( \sigma_e \) that are
determined by the walk. The other \( n-a_e \) values can thus be chosen at
random. This implies that there are \( (n-a_e)! \) permutations out of the
possible \( n! \) which satisfy the conditions imposed by the potential
walk on the choice of the permutation \( \sigma_e \). Doing this for each
edge of the base graph concludes the proof. If \( B \) has half-loops, then
we replace the factorials by double factorials
since \( \sigma_e \) is a random perfect matching.
\end{proof}
%
%
%
%
\subsection{Asymptotic Expansion}

Again, begin with the assumption that \( B \) has no half-loops. We can
rewrite \eqref{eq:walk_probability} as \[ n^{\chi\potwalk}\prod_{v
\in V_B} n(n-1) \ldots (n-b_v+1) \prod_{e \in E_B} \frac{1}{n(n-1)\ldots
(n-a_e+1)}, \] 
where $\chi\potwalk$ is the Euler characteristic of the graph traced out
(i.e., subgraph in $G\in\cC_n(B)$) of  $\potwalk$.
Consider the power series expansions about \( x = 0 \) of \[
(1-x)(1-2x) \ldots (1-mx) \quad\text{ and }\quad \frac{1}{(1-x)(1-2x)
\ldots (1-mx)} \] The \(x^i\) coefficient is a polynomial of degree at most
\(2i\) in \(m\); see
\cite{friedman_random_graphs}, Lemmas~2.8 and 2.9.
Using this for \( x = 1/n \), we can deduce that there
must exist polynomials \(p_0, p_1, \ldots \) in the variables \(a_e\) and
\(b_v\) for \( e \in E_B \) and \( v \in V_B \) such that
\begin{equation}\label{equ:expansion_polynomials} \EE_{\rm{symm}}\potwalk_n
= n^{\chi\potwalk} \sum_{i \geq 0} p_i(\vec a, \vec b\,) \frac{1}{n^i}
\end{equation} 
\begin{defn} 
\label{de:expansion_polynomials}
We call the polynomials \( p_0, p_1, \ldots \)
the \emph{expansion polynomials} of the potential walk.  \end{defn}

We now comment on this expansion and in particular, its truncated form.

\begin{thm}\label{thm:Esymm_expansion} For any $(k,n)$-potential walk
\( \potwalk \) with \( k \leq n/2\) and any integer \( r \geq 1 \) we have
\begin{equation}\label{eq:Esymm_expansion} 
\EE_{\rm{symm}}\potwalk_n =
n^{\chi\potwalk} \left( p_0 + \frac{p_1}{n} + \ldots +
\frac{p_{r-1}}{n^{r-1}} + \frac{\error(k,n)}{n^r} \right)  
\end{equation}
where the \( p_i \) are the expansion polynomials of the potential walk \(
\potwalk \) in the variables \( \{ a_e\}_{e \in E_V}, \{b_v\}_{v \in V_B}
\) and with 
\begin{equation}
\label{eq:Esymm_expansion_error}
 | \error(k,n) | \leq ck^{2r} 
\end{equation}
where \(c\) is a constant
that depends only on \(r\) (and $B$).  
\end{thm} 
\begin{proof} Consider a function of
the form \[ g(x) = (1-\alpha_1 x) \ldots (1- \alpha_s x) (1- \beta_1
x)^{-1} \ldots (1-\beta_t x)^{-1} \] where \( \alpha_i \) and \( \beta_j \)
are positive constants. The \(i^{\text{th}}\)-derivative of this function
satisfies the bound \[ \left| \frac{g^{(i)}(x)}{i!} \right| \leq \left( 1-x
\, \beta_{\rm{max}} \right)^{-i} \left( \sum \alpha_j + \sum \beta_j
\right)^i \] where \( \beta_{\rm{max}} \) is the maximum of the \( \beta_k
\) (by equation (6) in \cite{friedman_random_graphs} on page 339). Using
Taylor's theorem about \( x = 0 \) we obtain that there must exist some \(
\xi \in [0, 1/n] \) such that \[ | \error(k,n) | \leq \left( 1-\xi \,
\beta_{\rm{max}} \right)^{-r} \left( \sum \alpha_j + \sum \beta_j \right)^r
\] On the interval \( [0,1[ \) we have that \( (1-x)^{-1} \leq e^{x/(1-x)}
\) and since \( \beta_{\rm{max}} \leq k \) we have \( \left( 1-\xi \,
\beta_{\rm{max}} \right)^{-r} \leq e^{rk/(n-k)} \). If \(k \leq n/2\) we
have that \( e^{rk/(n-k)} \leq e^r \), hence a constant that only depends
on \(r\). The \( \sum \alpha_j \) represents the sum for all \(v \in V_B\)
of all sums of \( 0, 1, \ldots, b_v-1\), and the \( \sum \beta_j \)
represents the sum of all \( e \in E_B \) of \( 0, 1, \ldots, a_e-1 \),
both of which are at most \( 0 + 1 + \ldots + k-1 = \binom{k}{2} \). Since
\( 2 \binom{k}{2} \leq k^2 \) we get the second part of the claimed error
term.  \end{proof}

We now study the error term of the truncated expansion.

\begin{defn}\label{de:truncate_walk_collections}
Let \( \cW \) be a walk collection and \(r \geq 1\) a positive
integer. 
We define the {\em $r$-truncated form of $\cW$},
$\cW_{<r}$ to be the walk collection
of those elements of $\cW$ of order less than $r$.
We define its \emph{\(r\)-truncated \( (k,n) \)-modified trace} to
be the random variable, \( \cW_{n, <r} \), defined on \( \cC_n(B) \) with
\[ \cW_{n,<r}(G,k) = \{ \potwalk \in \cW(k,n) \mid \potwalk \text{ is
attained in } G \text{ and } \ord\potwalk < r \} \] and we define its
associated \emph{\(r\)-truncated \((k,n)\)-walk sum} to be \[
\WalkSum_{<r}(\cW,k,n) = \sum_{\substack{\potwalk \in \cW(k,n)\\
\ord\potwalk < r}}\!\!\!\!\! P\potwalk \] \end{defn}


We note that the proof of Lemma~\ref{le:coincidences} generalizes to
walk sums.
\begin{lem}\label{le:walk_sum_coincidences}
Let $B$ be a graph without half-loops.
For any walk collection, \( \cW \), and any positive integer,
\( r \geq 1 \) and any $k\le n/2$, we have that 
\[ 
\WalkSum(\cW,k,n) = \WalkSum_{<r}(\cW,k,n)
+ O({k^{2r+2} n^{-r}}) \Tr(H_B^k), \] 
where the constant in the $O(\ )$ notation depends only on $r$ and $B$.
\end{lem} 
\begin{proof}
Since any walk collection consists of a subcollection of the strictly
non-backtracking closed walks in $G$, this is immediate from
the bound of Lemma~\ref{le:coincidences}.
\end{proof}

We can use this now to prove an expansion theorem for any walk sum.

\begin{thm}\label{thm:walk_sum_expansion} Let \( B \) be a base graph, let
\( \cW \) be a walk collection, and fix an integer \( r \geq 1 \). Then for
all \( k \leq n/2 \) we have 
\begin{equation}\label{eq:walk_sum_expansion}
\WalkSum(\cW,k,n) = f_0(k) +
\frac{f_1(k)}{n} + \ldots + \frac{f_{r-1}(k)}{n^{r-1}} + \error(r,n,k) 
\end{equation}
where 
\begin{equation}\label{eq:truncate_f_i}
f_i(k) = \sum_{j=0}^{r-1-i} \sum_{(w; [\vec t\,]) \in \cW^j(k)}
\!\!\! p_{i-j} \potwalk 
\end{equation}
and where \( \cW^j(k) \) is the set of all
equivalence classes of potential walks of length \(k\) and of order \(j\),
and
with the error term satisfying 
\begin{equation}\label{eq:error_bound_wse}
| \error(r,n,k) | \leq ck^{2r+2} 
\Tr(H_B^k) 
n^{-r} ,
\end{equation} 
where \(c\) is a constant that depends only on $r$ and $B$.
\end{thm} 
The point of this theorem is that we can truncate the asymptotic
expansion of each $f_i$ as in \eqref{eq:truncate_f_i} and still 
retain the usual bound on the error term.
\begin{proof} 
(Compare with the proof 
of Theorem~5.9 in \cite{friedman_alon}.) 
By Lemma~\ref{le:walk_sum_coincidences}, it suffices to establish
\eqref{eq:walk_sum_expansion} with
$\WalkSum(\cW,k,n)$ replaced with
$\WalkSum_{<r}(\cW,k,n)$.

Fix a strictly non-backtracking closed word, $w$, in $B$.
Using the notation in 
Lemma~\ref{le:coincidences} and the notion of a {\em \gls{coincidence}} there,
we fix a value $i_0\in\{1,\ldots,n\}$, and view the successive vertices,
$(v_j,i_j)$, $j=1,\ldots,k$,
in the walk over $w$ as random variables.
Note that each successive value of $i_j$ 
is either (1) determined by the previous
vertex values (a ``forced choice''), 
(2) a coincidence, where for some $m<j$ we return
to a vertex $(v_m,i_m)=(v_j,i_j)$, or (3) the choice of a new vertex
$(v_j,i_j)$, and hence a new value of the permutation over the $j$-th edge,
$e_j$ of $w$ (a ``free, non-coincidence choice'').
We note that if for each $j$ we know whether the $j$-th vertex falls
into cases (1), (2), or (3), and in case of (2) for which $m<j$ we
have $(v_j,i_j)=(v_m,i_m)$, then we 
can construct the graph of the walk $(w;\vec t)$, up to isomorphism.
Notice also that whether we are in case (1) or (3) for a given $i_j$
is determined by the previous $i_0,\ldots,i_{j-1}$.
Hence the number of equivalence classes 
$(w;[\vec t\,])$ for a fixed $w$, of order $s\le r-1$, is determined
by a choice of up to $\binom{k}{s+1}$ coincidences and at most
$k$ choices of the values of $m$ in each case (2).  Hence the total number
of classes $(w,[\vec t\,])$ is at most
$$
\binom{k}{s+1} k^{s+1} \le  k^{2s+2}.
$$
Combining equations
\ref{eq:WalkSum_by_walk_classes} and \ref{eq:Esymm_expansion}, we truncate
each symmetric $(w,[\vec t\,])$ expectation expansion at the 
$n^{-r+1}$ term, which yields the $f_i$ claimed in the theorem,
with an error term bounded
by
$$
c_s k^{2r -2s } 
$$
where $c_s$ is a constant that depends only on $B$ and $s\le r$.
It follows that the total error introduced by
truncating the expansions is at most
$$
\sum_{s=1}^r c_s k^{2r-2s+2} k^{2s}   \le c k^{2r+2}.
$$
for a $c$ depending only on $r$ and $B$.
\end{proof}

Such an expansion would give us the Alon Conjecture if we could show that
the coefficients, \( f_i(k) \) are \(B\)-Ramanujan functions for all values
of \(r\), for the walk sum consisting of all strictly non-backtracking
closed walks.
As we have explained, this is not the case for some sufficiently large
$i$, and so we take a smaller walk sum, namely the certified trace.
%
%
%
%
\subsection{Types and Forms}
\label{sb:types_and_forms}

We finish this section with a subsection that describes how to rearrange the
walk sum by examining more closely the shapes that the walks can take.
Friedman first introduced the notions of types and forms in
\cite{friedman_random_graphs} and then refined them in
\cite{friedman_alon}. We maintain the language and simplify the definitions
here. 

\begin{defn}\label{de:type}
\newnot{symbol:type} A \emph{\gls[format=hyperit]{type}} is a 
tuple \( \cT = (T,
\gamma_V, \gamma_E, \cL) \), where \(T\) is an oriented graph
(sometimes called the {\em \gls{type graph}} of $\cT$), \( \gamma_V
\) an ordering of the vertices of \(T\), \( \gamma_E \) an ordering of the
edges of \(T\) and \( \cL = (\cL^t, \cL^h) \), called the \emph{lettering}
of the type, where \( \cL^t, \cL^h \from E_T \to E_B \) are maps of
oriented graphs called the \emph{lettering at the tail} and the
\emph{lettering at the head} respectively. Additionally, we require the
following two conditions to hold: \begin{enumerate}[label= (\arabic*)]
\item All the vertices, except possibly the first vertex (in the ordering
given by the type), are of degree at least three, and 
\item the lettering is
\emph{consistent}, that is, for any vertex \(v \in V_T\) and any edges \(
f_1, f_2 \in E_T \) such that \( t_T(f_1) = t_T(f_2) \) (respectively \(
h_T(f_1) = h_T(f_2) \)) we have that \( \cL^t(f_1) \neq \cL^t(f_2) \)
(respectively \( \cL^h(f_1) \neq \cL^h(f_2) \)). In other words, the
lettering at the tail for two edges sharing their tail vertex is distinct,
and similarly at the head.  
\item $T$ is pruned (i.e., its first vertex is of degree at least two),
and $T$ is connected.
\end{enumerate} \end{defn}

We remark that the lettering of a type determines which vertices in
$V_T$ are mapped to which vertices in $B$: indeed, if $v\in V_T$,
then $v$ is incident upon some edge, $e\in E_T$; 
so if $he=v$ (recall that $T$ is oriented, so either $he=v$ or $te=v$)
then $v$ must be taken to $\cL^h(e)$, and similarly if $te=v$.

\begin{defn}\newnot{symbol:form} A \emph{\gls[format=hyperit]{form} 
of type \(\cT\)} is a tuple
\( \cF = (F,\ell, \cE) \) where 
\begin{enumerate}
\item $F$ is the underlying oriented graph of $\cT$;
\item $(F,\ell)$ is a variable-length graph, i.e.,
$\ell\from E_F\to\integers_{\ge 1}$; and
\item $\cE$ associates to each directed edge, $f\in \Edir_F$, 
of $F$ a non-backtracking walk in $B$,
$w(f)$, such that 
\begin{enumerate}
\item $w(f)$ has length $\ell(f)$, 
\item $w(f)$ is the 
\gls{reverse walk} of $w(\iota_F e)$ (in the sense of 
Definition~\ref{defn:walks}), and 
\item $\cE$ is consistent
with the lettering, $\cL$, of $\cT$, in that the first edge in $w(f)$ is
$\cL^t(f)$ and the last edge is $\cL^h(f)$.
\end{enumerate}
\end{enumerate}
\end{defn}

Given a potential walk, \( \potwalk \), there is a unique type, denoted \(
\cT\potwalk \)\newnot{symbol:typepotwalk}, and a unique
form, denoted \( \cF\potwalk
\)\newnot{symbol:formpotwalk}, associated to it. 
The type and form are obtained by
considering the graph, $\Graph\potwalk$, of the potential
walk (as in Definition~\ref{de:potential_walk}), the walk $w'$
in $\Graph\potwalk$ given by $\potwalk$, and then:
\begin{enumerate}
\item making $V_T$ consist only of the initial vertex of $w'$ and
all vertices of degree at least three in $\Graph\potwalk$;
\item replacing each maximal beaded path between $V_T$ vertices
in $\Graph\potwalk$ by a single edge of $T$;
\item the vertices and edges of $T$
are ordered as they occur in the walk, $w'$, and the edges of $V_T$
are oriented in the direction they are first traversed in $w'$;
\item each edge 
of the form takes its length and 
non-backtracking walk from the beaded path corresponding to the edge 
in $\Graph\potwalk$; and
\item the lettering of an edge in $V_T$ is taken from
first and last edges of the
corresponding beaded path in $\Graph\potwalk$ in the direction it is
first traversed (from which the edge gets its orientation, as 
mentioned above).
\end{enumerate}

Recall that we defined the order of a potential walk, \( \potwalk \) to be
minus the Euler characteristic of the $\Graph\potwalk$.
Since, this only depends on the
number of edges and vertices, and is invariant when beaded paths are
replaced with single edges, we
can define the order of a form or a type to be the order of any of its
associated potential walk and we denote it by \( \ord\potwalk =
\ord(\cF\potwalk) = \ord(\cT\potwalk) \).

\begin{prop}\newnot{symbol:Typesr} For any positive integer \(r\), there
are only finitely many types \( \cT \) of order less than \(r\). We denote
that finite set, \( \Types[r] \).  \end{prop} \begin{proof} Besides the
initial vertex, which might be of degree two, 
all other vertices in a type are
of degree at least three. 
Applying this to the handshaking lemma, we obtain
\begin{equation}\label{eq:ETVT-inequality} 
2|E_T| = \sum_{v \in V_T} \deg(v) \geq 2 + 3(|V_T|-1) 
\end{equation}
and since we
require that \( |V_T|-|E_T| >-r \) we have
$|V_T|> |E_T|-r$, which we apply to the right-hand-side of 
\eqref{eq:ETVT-inequality} to obtain
$$
2|E_T| > 2 + 3 (|E_T| -r-1),
$$
and hence, 
$$
|E_T| < 3r+1
$$
and hence, again using \eqref{eq:ETVT-inequality}, 
$$
|V_T| \le (2|E_T|+1)/3   < 2r + 1.
$$
This
implies that there are finitely
many graphs on which to base a type of order at most \(r\). 
And for each such graph, clearly there are only finitely many
ways to orient its edges,
finitely many ways to order the vertices and the edges, and
finitely many
letterings (since the base graph, $B$, is fixed).
\end{proof}

We are going to reorganize walk sums by types and forms.

\begin{defn} Let \( \cF = \cF\potwalk \) be a form associated to some
potential walk \( \potwalk \) with underlying oriented graph \( F \). We
say that a walk in \( F \) is \emph{consistent in \( \cF \)} if it is
strictly non-backtracking closed,
if it reaches every vertex and edge of \( F \),
and if it reaches them in the order specified by the orderings on \( \cF
\) and it first traverses each undirected edge in the direction of
its orientation.  \end{defn}

\begin{prop}
\label{pr:one_to_one}
Let \( \potwalk \) be a potential walk and let \( \cF =
\cF\potwalk \) be its associated form. There is a natural one-to-one
correspondence between \begin{enumerate}[label= (\arabic*)] \item the set
of potential walk classes \( \potwalkclass \) whose form is \( \cF \), and
\item the set of consistent walks in \( \cF \).  \end{enumerate} \end{prop}
\begin{proof} This is immediate: a consistent walk uniquely defines a base
walk and an equivalence class of trajectories and reciprocally.
\end{proof}

\begin{lem}\label{lem:lettering_polyexponential}\newnot{symbol:NBeek} Let
\( \cT \) be a type, with underlying graph \( T \) and let \( e_1\) and \(
e_2 \) be two directed edges in \( T \). 
Let \( \mathrm{NB}(e_1,e_2,k)
\) denote all non-backtracking
walks of length \( k \) from \( e_1 \) to \(e_2\) in \( T
\). For each such walk, \( w \) and each edge, \( e \in E_B \), of the
graph \( B \), let \( a_e(w) \) be the number of times \( e \) appears in
the walk \( w \). Then for any polynomial \( Q=Q( \{ a_e\}_{e \in E_T}) \)
in the variables \( a_e \), then we have that \[ \sum_{w \in
\mathrm{NB}(e_1,e_2,k)} Q(a_e(w)) \] is a \(B\)-polyexponential function in
\(k\).  \end{lem} 
\begin{proof}
This follows immediately from Corollary~\ref{co:count_occurrences},
analogous to Lemma~2.7 of
\cite{friedman_random_graphs}.
\end{proof}

\begin{defn}\newnot{symbol:WFWkm}\newnot{symbol:WTmkF} Fix a walk
collection, \( \cW \), a type, \( \cT \) with underlying graph \( T \). For
every form, \( \cF \), of type \( \cT \), and every \( \vec m \in \ZZ_{\geq
1}^{E_T} \), let \( W_\cF(\cW,k,\vec m\,) \) be the number of consistent
potential walk classes, \( \potwalkclass \in \cW(k,n) \), of form \( \cF
\), that traverse each edge, \( f \in E_T \), \(m_f\) times. We say that \(
\cW \) is \emph{multiplicity and length determined} if \( W_\cF(\cW,k, \vec
m\,) \) is just a function of \( \cT, \vec m, \vec k\) where \( \vec k =
\vec k(\cF) \) is the lengths of the \( E_T\)-edges in the form \( \cF \),
i.e.  \begin{equation}\label{equ:multiplicity_determined} W_\cF(\cW,k, \vec
m\,) = W_\cT(\vec m, \vec k(\cF)) \;. \end{equation} 
\end{defn}

\begin{ex} If \( \cW \) is the walk collection \( \SNBC \) of all strictly
non-backtracking closed walks, then it is multiplicity and length
determined. Furthermore, $W_\cT(\vec m,\vec k(\cF))$
is a function of \( \vec m \) alone.
Furthermore, if the subcollection, \( \cW_0 \), of \( \SNBC \) of all
walks, \( w \), for which the number of edges in \( \Graph(w) \) is a prime
number, then it is multiplicity and length determined; in fact, more
generally, for every type, \( \cT \), with underlying graph \(T \), say one
is a given a subset \( \cK(\cT) \subset \ZZ_{\geq 1}^{E_T} \); let \( \cW_0 \)
be the walk subcollection of \( \SNBC \) of walks \( w \) for which the
form of the walk has vector of length, \( \vec k = \vec k(\cF) \), on the
\( E_T \) edges that lie in \( \cK(\cT) \), then the walk subcollection,
\( \cW_0 \), is multiplicity and length determined.  \end{ex}

\begin{defn} We say that a walk collection, \( \cW \), \emph{decouples}, if
for every type, \( \cT \) with underlying graph \(T\),  the walk collection
\( \cW \) satisfies the following property: if two forms of type \( \cT \)
induce the same edge lengths on \( E_T \), then the first form belongs to
\( \cW \) if and only if the second one does too. In other words, to each
type, \( \cT \), there corresponds a subset of \( \ZZ_{\geq 1}^{E_T} \)
such that a form of type \( \cT \) belongs to \( \cW \) if and only if the
induced vector of edge lengths on \( T \) is in that subset.  \end{defn}

This allows us to state the main theorem of this subsection, namely an
expansion of the walk sum in terms of types.
Recall the notion of a $B$-polyexponential for a graph, $B$,
in Definition~\ref{de:polyexponential}

\begin{thm}\label{thm:type_form_expansion} Given a walk collection \( \cW
\) that is multiplicity and length determined, and a positive integer \( r
\), we have an expansion, called the \emph{type-form expansion of \( \cW
\)}: \[ \WalkSum_{<r}(\cW,k,n) = \sum_{\cT \in \Types[r]} n^{-\ord(\cT)}
\sum_{i=0}^{r-\ord(\cT)} Q_{\cT,i}(k) n^{-i} + \error(r, n, k) \] where the
\( Q_{\cT,i} \) are functions satisfying \[ Q_{\cT,i}(k) =
\sum_{\substack{\vec m, \vec k \in \ZZ_{\geq 1}^{E_T} \\ \vec m \cdot \vec
k \,=\, k}} W_\cT(\vec m, \vec k\,) P_i(\vec k\,) \] where the \( P_i( \vec
k\,) \) are \(B\)-polyexponential functions, and \( W_\cT(\vec m, \vec k\,)
\) is the function given in \eqref{equ:multiplicity_determined},
and where
the error term satisfies a bound as in 
Theorem~\ref{thm:walk_sum_expansion}, that is 
\[ | \error(r,n,k) | \leq ck^{2r+2}
\rho^k(H_B) n^{-r} \] 
for some constant $c=c(r,B)$.
\end{thm} 
\begin{proof} Given a form \( \cF \), we denote
by \( W_\cF(\cW,k,n) \) the number of potential walk classes, \(
\potwalkclass \) of length \(k\) in the walk collection whose form is \(
\cF \). 
Proposition~\ref{pr:one_to_one}
shows that this is the same as the number
of consistent walks of length \(k\) in the from.

Given a potential walk, \( \potwalk \), its associated form, \( \cF \),
contains information that is equivalent to the \( a_e \) and \( b_v \) and
hence, \( \EE_{\rm{symm}}\potwalk_n \) only depends on the form. So we
define \( \EE[\cF]_n = \EE_{\rm{symm}}\potwalk_n \). This allows to rewrite
our walk sums as \[ \WalkSum(\cW,k,n) = \sum_{\cT} \sum_{\cF \in \cT}
W_\cF(\cW,k,n) \,\EE[\cF]_n \;. \] 
From this point and onward, we will only sum over types of
at most some order, so we will sum over only finite types.

For a fixed type, \( \cT \), define \( \EE[\cT]_{n,k} \) to mean \[
\EE[\cT]_{n,k} = \sum_{\cF \in \cT} W_\cF(\cW,k,n) \,\EE[\cF]_n \] Since we
usually do not work with more than one walk collection at a time, we drop
its mention from the notation.

For a \( \vec k = \{ k_e\}_{e \in E_T} \), let \( \cT(\vec k\,) \) denote
the set of all forms, \( \cF \), of type \( \cT \), such that the length of
the edges of the forms is given by \( \vec k \). For each edge \( e \in E_T
\) of the type, fix an integer \( m_e \geq 1 \), and denote \( \vec m = \{
m_e \}_{e \in E_T} \). This allows us to write \[ \EE[\cT]_{n,k} =
\sum_{\vec m \cdot \vec k \,=\, k} W_\cT(\vec m, \vec k\,) \sum_{\cF \in
\cT(\vec k\,)} \EE[\cF]_n \]

Since each of the \( \EE[\cF]_n \) have a \( 1/n \) expansion, by adding
expansions, we get an asymptotic expansion \[ \sum_{\cF \in \cT(\vec k\,)}
\EE[\cF]_n = n^{-\ord(\cT)} \sum_{\cF \in \cT(\vec k\,)} \left( P_0(\cF) +
\frac{P_1(\cF)}{n} + \ldots \right) \] where \[ P_i(\cF) = p_i(\vec a(\cF),
\vec b(\cF)) \] where the \( p_i \) are the expansion polynomials in
\eqref{equ:expansion_polynomials}, and $\vec a(\cF)$ is the vector
whose components are $a_e(\cF)$, indexed on the $e\in E_B$, representing
the number of times $e$ edges occur in $\cF$, and similarly for
$\vec b(\cF)$, indexed on the $f\in V_B$. 
Note that the vector \( \vec
b(\cF) \) is an affine linear function of the variables \( \vec a(\cF) \)
for all forms \( \cF \) of a given type \( \cT\ \), so we may
write \[ P_i(\cF)
= P_{i,\cT}(\vec a(\cF) )  .
\] 
Rearranging the sum yields \[ \EE_n[\cT]_{n,k}
= n^{-\ord(\cT)} \left( Q_0(k) + Q_1(k)n^{-1} + \ldots \right) \] where 
\begin{equation}
\label{eq:Q_i_k}
Q_i(k) = \sum_{\substack{\vec m, \vec k \\ \vec m \cdot \vec k \,=\, k}}
W_\cT(\vec m, \vec k\,) \sum_{\cF \in \cT(\vec k)} P_{i,\cT}(\vec a(\cF))
\end{equation} 
For each form, \( \cF \), of type \( \cT \), let \( A_{e,f}(\cF) \) be
the number of occurrences of a label \( e \in E_B \) on an edge \( f \in
E_T \). We have that for any form $\cF\in\cT(\vec k\,)$
\begin{equation}
\label{eq:a_e_are_sums}
a_e(\cF) = \sum_{f \in E_T}
A_{e,f}(\cF), 
\end{equation} 
and 
\begin{equation}
\label{eq:k_f_are_sums}
k_f = \sum_{e \in E_B} A_{e,f}(\cF). 
\end{equation}
To finish the proof of the theorem, \eqref{eq:Q_i_k} shows that it
suffices to show that
\[ \sum_{\cF \in
\cT(\vec k)} P_{i,\cT}(\vec a(\cF)) \] is polyexponential in \(k\). 
Let us prove this; the proof makes essential use 
of the fact that $\cW$ decouples.

By expanding the $P_{i,\cT}$ is suffices to show that with $\vec k$ fixed
we have that for any set, $\{s_e\}_{e\in E_B}$, of non-negative integers,
we have
$$
\sum_{\cF\in\cT(\vec k)}\prod_{e\in E_B} a_e^{s_e}(\cF)
$$
is polyexponential in $k$ (with bases being the eigenvalues of $H_B$).
Using \eqref{eq:a_e_are_sums}, we may further expand the powers of 
$a_e(\cF)$, and reduce the above sum to linear combinations of a
finite number of terms
$$
\sum_{\cF\in\cT(\vec k)}\prod_{e\in E_B}\prod_{f\in E_T} A_{e,f}^{s_{e,f}}(\cF)
$$
for some non-negative integers $s_{e,f}$.

But since $\cW$ decouples, the above equals 
\begin{equation}
\label{eq:product_over_E_T}
\prod_{f\in E_T} \left( \sum_{w_f \in \mathrm{NB}(e_1(f), e_2(f),k_f ) }
\left( \prod_{e\in E_B} \bigl( a_{e,f}(w_f) \bigr)^{s_{e,f}} \right) \right) ,
\end{equation}
where \( e_1 \) and \( e_2 \) are determined by the lettering. 
However, by
Lemma~\ref{lem:lettering_polyexponential}, for all $f\in E_T$ we have
$$
\sum_{w_f \in \mathrm{NB}(e_1(f), e_2(f),k_f ) }
\left( \prod_{e\in E_B} \bigl( a_{e,f}(w_f) \bigr)^{s_{e,f}} \right) 
= P_f(k_f)
$$
for some polyexponential functions $\{P_f(k_f)\}$. 
Hence the quantity in \eqref{eq:product_over_E_T} is just
$$
\prod_{f\in E_T}
P_f(k_f), 
$$
a polyexponential function in the variable $\vec k$, with bases in the
spectrum of $H_B$.  
\end{proof}

\section{Tangles and the Certified Trace}
\label{se:p1-tangles}

Tangles are low-probability events that prevent our asymptotic expansion of
the expected trace to have \(B\)-Ramanujan coefficients. This was first
observed by Friedman in \cite{friedman_alon} and lead him to define a
\emph{selective trace}, which despite being quite cumbersome, showed to be
effective to settle the Alon Conjecture. In our work, we simplify his
selective trace and develop our certified trace in this section. 

%
%
%
%
\subsection{Tangles}\label{tangles:subsection:tangles}

We mainly use the theory developed by Friedman in
\cite{friedman_alon}, however some of our terminology and methods
differ.
Recall the various notions of \glspl{B-tangle} of 
Definition~\ref{de:tangle_of_B} and the discussion below there;
this explains that \glspl{strict tangle of B} are troublesome in the
trace method, and---for technical ease---in this chapter our certified
trace discards all walks giving rise to a \gls{B-tangle};
in Section~\ref{se:p2-fund-exp} we will only discard (certain)
strict tangles, namely \glspl{B-eps-tangle} for various values of
$\epsilon$.
Recall also
Notation~\ref{no:has_tangles}.

First, we show that if we wish to remove all tangles up to 
some fixed order, then there are only finitely many of them to consider.

\begin{definition} For a given graph, $B$, and positive integer, $r$,
let 
\newnot{symbol:TangleminBr}
$\HasTanglemin$ be the set of elements of $\Tangle_{<r,B}$ that
are minimal with respect to inclusion, i.e., the set of all $B$-tangles
of order less than $r$ that do not contain any proper subgraph that is
a $B$-tangle.
\end{definition}

\begin{thm}\label{thm:tangles_finite}
For any graph, $B$, and any positive integer, $r$,
the set \( \HasTanglemin \) is
finite, and any tangle of order less than $r$ contains at least one
such minimal tangle.  \end{thm} 
\begin{proof} 
Our proof uses the same argument as in
Lemma~9.2 of \cite{friedman_alon}.
We repeat here the main points of his proof for completeness. 

First note that any tangle is a finite graph, and has finitely many
subgraphs, and all subgraphs without edges are not tangles.
Hence any tangle contains some element of 
$\HasTanglemin$.

Now assume by
contradiction that the set 
$\HasTanglemin$ is not finite. Then there must exist a type \(
\cT \), with underlying graph \( T \) and an infinite sequence of tangles
of type \( \cT \), which each correspond to a choice \( \vec k \in
\ZZ_{\geq 1}^{E_T} \) of lengths of the edges of \( T \). Let us refer to
that infinite sequence by \( \{ \vec k_i \}_{i \geq 1} \); 
by passing to a subsequence, we may assume that 
for each edge \( f \in E_T \), the
corresponding sequence of lengths for this edge is either constant or tends
to infinity; furthermore, there must exist at least one edge whose length
sequence tends to infinity. Denote by \( \psi_\infty \) this limiting
graph, where we discard the edges with length tending to infinity. Since
tangles have their Hashimoto spectral radius at least 
$\rho^{1/2}(H_B)$, the
same must hold for the limiting graph and hence \( \psi_\infty \) is a
$B$-tangle as well. Furthermore, the order of \( \psi_\infty \) is at most that
of the tangle in the sequence. Hence $\psi_\infty$ is a $B$-tangle.
But $\psi_\infty$ is properly
contained in every
element of $VLG(T,\vec k)$, which contradicts the minimality of these
graphs.
\end{proof}

We mention that in the above proof it is crucial that $\Tangle_{r,B}$
requires its elements, $\psi$, to satisfy the inequality
$$
\rho(H_\psi) \ge \rhoroot B     
$$
which is not a strict inequality.  If the inequality were strict, then
in the above proof we could not conclude that $\psi_\infty$ is also a
$B$-tangle.  This, in turn, leads us to define the certified trace
with strict inequality, which means that one cannot use a simple argument
as above to conclude that upper sets of interest in the certified trace
have a finite number of minimal elements.



\subsection{The Occurrence of Subgraphs}
\label{sb:occur}

The point of this subsection is to show that if $\psi\in\Occurs_B$
is any connected, feasible
graph of order 
$r\ge 1$, and if $\phi\from\psi\to B$ is an \'etale morphism,
then for large $n$ we have
that $\psi$ {\em occurs} in $\cC_n(B)$, in a way consistent
with its ``$B$-structure'' from $\phi$, with probability
\begin{equation}\label{eq:psi_probability}
n^{-r}/c + O(n^{-r-1}) 
\end{equation}
for some integer $c$ which we will specify below.
Our proof is almost identical to the proof of
Theorem~4.7 of \cite{friedman_alon}.  
This theorem has two applications.  The most immediate application
is to give a lower bound that tangles occur in $\cC_n(B)$.
A much subtler application is that to estimate
$$
\EE_{G\in\cC_n(B)} \Bigl[ \II_{\HasTangle_{r,B}}(G) \CertTr_{<r}(G,k) \Bigr]
$$
we shall need to estimate certified traces for those graphs which
contain a 
tangle; in this case we will need the definitions and the methods
used in this subsection to extend the methods used to prove 
a $1/n$-asymptotic expansion with $B$-Ramanujan coefficients for
$$
\EE_{G\in\cC_n(B)} \Bigl[ \CertTr_{<r}(G,k) \Bigr]
$$
to do the same for the expected certified trace multiplied by
the indicator function of 
$\HasTangle_{r,B}$.  This more subtle application will become
clear in Section~\ref{se:p1-with-tangles}.

After we prove the theorems that we will need in this paper, we shall
describe some more general facts that can be derived with similar
methods;  these more general facts will help put our theorems
in perspective, although the more general facts will not be used here.

\begin{definition} Let $B$ be a graph.  By a 
{\em \gls[format=hyperit]{B-graph}} we mean
a morphism $\phi\from\psi\to B$; if $\phi$ is understood, we will abusively
refer to $\psi$ as the $B$-graph.
\end{definition}

The notion of $B$-graphs is very important for the following reasons.
\begin{definition}
Let $\psi\in\Occurs_B$.  Then $\psi$ is
isomorphic to a subgraph of at least one graph, $G_0$, than occurs with
positive probability in $\cC_{n_0}(B)$ for some $n_0$.
Then any injection $\psi\to G_0$ (given by any isomorphism of $\psi$
as a subgraph of $G_0$), composed with
the projection of $G_0$ to $B$, gives rise to a morphism
$\phi\from\psi\to B$, i.e., gives $\psi$ the structure of a $B$-graph.
We say that $\phi$ (or, abusively, $\psi$) {\em is occurrence induced}.
\end{definition}

We now wish to count the number of times a $B$-graph {\em occurs} in
a graph, $G\in\cC_n(B)$.  This only depends on the structure of
$G$ as a $B$-graph.

\begin{definition}\label{de:num_occurrences} 
Let $\phi_1\from\psi_1\to B$
and $\phi_2\from\psi_2\to B$ be two $B$-graphs.  By a 
{\em \gls{morphism of B-graphs}},
$\phi_1\to\phi_2$, we mean a morphism
$\nu\from\psi_1\to\psi_2$ which respects the $B$ structure, i.e.,
for which $\phi_1=\psi_2\circ\nu$.
We call the morphism an {\em injection} if $\nu$ is an injection.
We alternatively call an injection an {\em occurrence} of $\phi_1$
in $\phi_2$, and the {\em number of occurrences} of $\phi_1$ in
$\phi_2$, 
denoted $\NumOccurs(\phi_1,\phi_2)$,
is the number of all such injections, $\nu$ (where two injections
are distinct unless they agree on all their edge values and vertex
values).
\end{definition}

For example, in the above let $\psi_1$ consist of two vertices joined
by some number of edges (so $\psi_1$ has no self-loops).  Let
$\tau\from\psi_1\to\psi_1$ be the (iso)morphism exchanging the vertices
of $\psi_1$ and taking each directed edge to its inverse.
Then if $\psi_1$ is given any $B$-graph structure,
for any inclusion $\nu$ of $\psi_1$ into any $B$-graph,
$\nu\tau$ is a distinct inclusion of $\psi_1$ into the same graph.
Hence the number of occurrences of $\psi_1$ into any graph is even,
owing to the fact that if $\psi_1$ is given any $B$-graph structure,
then $\tau$ gives a nontrivial automorphism of $\psi_1$ as a $B$-graph.

This notion of automorphism gives the integer $c$ in 
\eqref{eq:psi_probability}.  Let us make this precise.

\begin{definition} Let $\phi\from\psi\to B$ be a $B$-graph.  An
{\em automorphism of $\phi$} is morphism $\nu$ from $\psi$ to itself,
as a $B$-graph, which is an isomorphism.
The set of automorphisms becomes a group under composition, and we 
denote this group $\Aut(\phi)$.
\end{definition}

%

\begin{thm}\label{thm:prob_tangle}
\label{th:prob_occurs}
Let $\phi\from\psi\to B$ be an \'etale morphism (hence $\psi\in \Occurs_B$,
and occurs in $G\in\cC_n(B)$, with its $B$-structure, with positive
probability).
Then, for large $n$, we have that 
\begin{equation}\label{eq:psi_expectation}
\expect{G\in\cC_n(B)}{\NumOccurs(\phi,B)} =
n^{-r} + O(n^{-r-1});
\end{equation}
hence 
\begin{equation}\label{eq:psi_prob_no_assumptions}
\prob{G\in\cC_n(B)}{\mbox{$\phi$ occurs in $G$}} \le
n^{-r} + O(n^{-r-1}).
\end{equation}

Furthermore, assume that $\psi$ is pruned and that
each connected component of $\psi$
has order at least one.
Then, for large $n$, the probability that
$\phi$ occurs (at least once) as a $B$-subgraph of $G\in\cC_n(B)$
is, for large $n$, equal to
\begin{equation}\label{eq:actual_probability}
n^{-r}/c + O(n^{-r-1}),
\end{equation}
where $c$ is the order of the $\Aut(\phi)$.
\end{thm} 

\begin{proof}
The proof as essentially the same 
as in Theorem~4.7 in \cite{friedman_alon};  we
repeat here the argument for completeness. 
Let \( V_\psi =
\{ u_1, \ldots, u_s \} \).
The number of occurrences of $\psi$ in $G=G[\sigma]\in\cC_n(B)$, is 
precisely the number of $s$-tuples,
$$
\vec t = (t_1,\ldots,t_s),
$$
of distinct integers
between \(1\) and \(n\), such that there is an injection 
$\psi\to G$ which takes $u_i$ to the vertex $(\phi(u_i),t_i)$ in $G$;
equivalently, this means that for each $e\in V_\psi$ we have
$$
\sigma(\phi(e))t_i = t_j,
$$
where $i$ and $j$ are given by $u_i=\phi(t_\psi e)$ and
$u_j=\phi(h_\psi e)$;
i.e., for each $e\in V_\psi$,
the edge, $\phi(e)$, in $B$, takes $t_i$ to $t_j$,
where $u_i$ is the $B$ vertex corresponding to the tail of $e$,
and $u_j$ is that corresponding to the head of $e$.
For each $\vec t$, let \( \iota_{\vec t} \) denote the event in 
$G=G[\sigma]\in\cC_n(B)$
that this injection of $\psi\to G$ occurs.
If \(a_e\) is the number of edges of \( \psi \)
labelled with the edge \( e \in E_B \), then the event \( \iota_{\vec t} \)
involves setting \( a_e \) values of the corresponding permutations which
will occur with probability \[ \prod_{e \in E_B}
\frac{(n-a_e)!}{n!} = n^{|E_\psi|} + \BigOn{ n^{-|E_\psi|-1} }. \] Since the
sum of the \( a_e \) is \( |E_\psi| \). For any vertex \( v \in V_B \),
denote by \( v_b \) the number of vertices of \( \psi \) in the fibre of \(
v \), then the number of choices for \( \iota_{\vec t} \) is \[ \prod_{v
\in V_B} \frac{n!}{(n-b_v)!} = n^{|V_\psi|} + \BigOn{n^{|V_\psi|-1}} \] and
so we can conclude that the expected number of occurrences of \( \psi \) is
\[ n^{-\ord(\psi)} + \BigOn{n^{-\ord(\psi)-1}}. \] 
This proves \eqref{eq:psi_expectation}, and 
\eqref{eq:psi_prob_no_assumptions} is an immediate consequence.

It remains to prove
\eqref{eq:actual_probability}, under the added assumptions on $\psi$.
So fix a $B$-graph structure $\phi\from\psi\to B$, and let 
$$
{\rm NumSubgraphs}(\psi,G)
$$
denote the number of subgraphs of $G$ that are isomorphic, as $B$-graphs,
to $\psi$.  We claim that 
$$
{\rm NumSubgraphs}(\psi,G) = c \; \NumOccurs(\psi,G).
$$
Indeed, each subgraph, $G'$, of $G$ as a $B$-graph that is isomorphic to
the $B$-graph $\psi$ gives rise to at least $c$
occurrences, using the automorphisms of $\psi$.  Furthermore, if 
$\phi\from\psi \to G$ is a fixed injection, with image $G'$, then any map
$\phi'\from\psi\to G$ which has the same image as $\psi$ gives rise
to an automorphism of $\psi$ as a $B$-graph
given by $\phi$ composed with the
inverse of $\phi'$ (restricted to $G'$).
Hence each subgraph, $G'$, of $G$ isomorphic as a $B$-graph to
$\psi$ gives rise to exactly $c$ occurrences of $\psi$ in $G$.

Now, again, 
let \( V_\psi =
\{ u_1, \ldots, u_s \} \) and for a $s$-tuple \( \vec t \) of distinct integers
between \(1\) and \(n\), let \( \iota_{\vec t} \) denote the event that the
map \( \iota \) mapping \( u_i \) to the vertex at height \( t_i \) above
the matching vertex of the base graph \(B\), is an occurrence of \( \psi \)
in \( G \in \cC_n(B) \).
By the preceeding paragraph, each $s$-tuple gives identical
to $c-1$
other $s$-tuples; hence the $s$-tuples $\vec t$ naturally are divided
into equivalence classes, each of order $c$; for each equivalence class
choose one $s$-tuple $\vec t$.  This gives us a set of equivalence class
representatives, $T$ of order $n^s/c$; each $\vec t\in T$ describes an
event, $\iota_{\vec t}$, each of which describes an occurrence of
$\psi$ as a subgraph of $G\in\cC_n(B)$.

For any integer $i=1,2$, let $E_{\ge i}=E_{\ge i}[n]$ 
(respectively, $E_i=E_i[n]$) be the events that ${\rm NumSubgraphs}(\psi,G)$
is at least $i$ 
(respectively, exactly $i$).  By inclusion-exclusion we have
$$
\prob{G\in\cC_n(B)}{E_{\ge 1}[n]} \le \sum_{\vec t\in T}
\prob{G\in\cC_n(B)}{\iota_{\vec t}} ,
$$
and 
$$
\prob{G\in\cC_n(B)}{E_{\ge 1}[n]} \ge \sum_{\vec t\in T}
\prob{G\in\cC_n(B)}{\iota_{\vec t}} - 
\sum_{\vec t,\vec t'\in T}
\prob{G\in\cC_n(B)}{\iota_{\vec t}\cap \iota_{\vec t'}} ,
$$
where 
the sum over $\vec t,\vec t'\in T$ requires $\vec t\ne\vec t'$, and where
each pair $(\vec t,\vec t')$ is counted once (not twice).
By the first part of the theorem, we have
$$
\sum_{\vec t\in T}
\prob{G\in\cC_n(B)}{\iota_{\vec t}} 
= n^{-r}/c + O(n^{-r-1}).
$$
Hence to complete the proof of the theorem, it suffices to show
that 
\begin{equation}\label{eq:needed_vec_t}
\sum_{\vec t,\vec t'\in T}
\prob{G\in\cC_n(B)}{\iota_{\vec t}\cap \iota_{\vec t'}} 
= O(n^{-r-1}).
\end{equation}

For each $\vec t,\vec t'\in T$ with $\vec t\ne\vec t'$,
the event $\iota_{\vec t}\cap \iota_{\vec t'}$
corresponds to
two distinct copies,
$\psi_1,\psi_2$,
of the $B$-graph $\psi$ in $G$.  The union of $\psi_1,\psi_2$,
in $G$ is a $B$-graph, $\widetilde\psi$, which comes
with injections from $\psi_1$ and $\psi_2$.
Let us show that
\begin{enumerate}
\item each vertex of $\widetilde\psi$ has degree at least two,
i.e., $\widetilde\psi$ is pruned;
\item each connected component of $\widetilde\psi$ has order at
least one;
\item $\widetilde\psi$ strictly contains $\psi_1$ (or, for that
matter, $\psi_2$).
\end{enumerate}
Item~(1) follows from the fact that each vertex of $\widetilde\psi$
lies in $\psi_1$ or $\psi_2$, which has an injection into $\widetilde\psi$,
and therefore has degree at least two.
Item~(2) follows similarly, since any vertex, $v$, of $\widetilde\psi$
is in the target of the injection from $\psi_1$ or from $\psi_2$; hence
the connected component of $v$, in $\widetilde\psi$, contains, via the
injections, either a connected component from $\psi_1$ or $\psi_2$
which is of order at least one.  Hence any connected component of 
$\widetilde\psi$ contains a subgraph of order one, and,
by Proposition~\ref{pr:order_under_inclusion}, has order at least one.
Item~(3) is immediate from the fact that $\psi_1$ and $\psi_2$
are distinct $B$-subgraphs of $G$.
From Items~(1)--(3) it follows that Theorem~\ref{th:strict_pruned_inclusion}
implies that the order of $\psi$ is at least $r+1$.  Hence
$$
\expect{G\in\cC_n}{\NumOccurs(\widetilde\psi,G) } \le
n^{-r-1} + O(n^{-r-2}).
$$
But since $\widetilde\psi$ has at most twice the number of vertices of $\psi_1$
and twice the number of edges, there are only a finite number of
possible $B$-graphs $\widetilde\psi$ that arise as a union of $B$-graphs.
It follows that
$$
\sum_{\widetilde\psi} \expect{G\in\cC_n(B)}{\NumOccurs(\widetilde\psi,G)}
\le O(n^{-r-1}),
$$
where $\widetilde\psi$ ranges over all $B$-graphs that arise as the
union of two distinct $B$-subgraphs in any $G\in\cC_n(B)$.
Furthermore, we notice that any $\widetilde\psi$ has a finite number
of subgraphs isomorphic to $\psi$, and hence can arise from at most
a constant, $C$, of events $\iota_{\vec t}\cap\iota_{\vec t'}$ with
$\vec t,\vec t'\in T$.
Hence, we have
$$
\expect{G\in\cC_n}{\sum_{\vec t\ne \vec t'} \iota_{\vec t}\cap\iota_{\vec t'}}
\le
C
\sum_{\widetilde\psi}
\expect{G\in\cC_n}{\NumOccurs(\widetilde\psi,G) } \le
O(n^{-r-1}) ,
$$
which establishes \eqref{eq:needed_vec_t}, and therefore completes
the proof.
\end{proof}


In the case where for $G=G[\sigma]\in\cC_n(B)$ we have that
$\sigma(e)1=1$ for all $e\in E_B$, then
the resulting graph has Hashimoto eigenvalue $\rho(H_B)$ and is therefore a
tangle.  Hence there always exists a finite tangle.
It follows that among all the tangles
there is a smallest order that occurs.

\begin{definition}
\label{de:fundamental_order_defined}
For a connected graph, $B$, of negative Euler
characteristic, we define 
the {\em \gls[format=hyperit]{fundamental order}}, 
denoted \gls[format=hyperit]{etafund},
to be the smallest order of a \gls{strict tangle of B}, i.e.,
$$
\etafund(B) = \min\{ \ord(L) \ |\
\rho(H_L) > \rhoroot B  \} .
$$
\end{definition}

\begin{prop}\label{prop:prob_tanglemin} 
Let $B$ be a fixed graph, and fix an integer
$r\ge \etafund(B)$.  Then 
there is a $C>0$ such that for large
$n$ we have
$$
\prob{G\in\cC_n(B)}{\mbox{$G$ contains a strict tangle of $B$}} \ge 
C n^{-\etafund(B)}.
$$
\end{prop}

\begin{remark}
In Section~\ref{se:p1-with-tangles} we will extend
our inclusion/exclusion further to get asymptotic expansions.  As
an example, we will see that 
for every $r>0$ and $m$ we have
an expansion
$$
\prob{G\in\cC_n(B)}{\tanglefree} =
1 - q_1n^{-1}-q_2n^{-2} - \cdots - q_m n^{-m} + O(n^{-m-1}),
$$
where the $q_i$
are some constants with $q_i=0$ for
$i<\taufund(B)$.
\end{remark}

%
%
%
%
%
%
\subsection{Finitely Certifiable Partial Orders}
\label{sb:certifiable}

In this subsection, we revisit cones in partially ordered sets. The main
theorem of this subsection is Theorem \ref{thm:main_mobius}, which shows that
under some assumptions (that we call being \emph{finitely certifiable}),
some infinite sums can be dealt with using only finitely many \emph{cone
sums}. This result will be used in the following subsection to define a
\emph{certifiable trace}.

\begin{defn} Let \( (P,\leq) \) be a partially ordered set. For \( p \in P
\) we define the \emph{cone at \(p\)} to be \[ \Cone(p) = \{ t \in P \mid t
\geq p \} \] and we define the \emph{open cone at \(p\)} to be \[
\Cone^\ast(p) = \{ t \in P \mid t > p \} \] \end{defn}

\begin{defn} Let \( (P,\leq) \) be a partially ordered set and let \(p,q
\in P\) be two distinct elements of the set. We say that \(p\) and \(q\)
have a \emph{maximum} (or \emph{least upper bound}) if there exists an
element, denoted \( \max(p,q) \), such that \[ \Cone(p) \cap \Cone(q) =
\Cone(\max(p,q)) \] We say that \( (P,\leq) \) \emph{has maximums} if every
pair of distinct elements \(p,q \in P\) have a maximum.  \end{defn}

\begin{defn} A subset \( Q \subseteq P \) of a partially ordered set \(
(P,\leq) \) is an \emph{upper set} if \( q \in Q \) implies that \(
\Cone(q) \subseteq Q \). We say that \( (P,\leq) \) has \emph{finitely
generated upper sets} if every upper set, \( Q \subseteq P \), has a finite
number of minimal elements, which we denote \( \Minimal(Q) \); that is, if
every upper set is the union of finitely many cones.  \end{defn}

\begin{defn} We say that a partially ordered set \( (P,\leq) \) is
\emph{Noetherian} if every infinite sequence of upper sets in \(P\) \[ Q_1
\subseteq Q_2 \subseteq \ldots \] stabilizes, that is there exists a
positive integer \(N\) such that \( Q_i = Q_j \) for all \(i,j \geq N\).
\end{defn}

We make this simple observation.

\begin{thm} A partially ordered set is Noetherian if and only if it has
finitely generated upper sets.  \end{thm} \begin{proof} Assume that \(
(P,\leq) \) is Noetherian. Let \( Q \subseteq P \) be an upper set and
assume by contradiction that \( \Minimal(Q) \) is infinite. Then, there
exists a countable sequence of distinct minimal elements \(q_1, q_2, \ldots
\). Consider the sets \[ Q_i = \bigcup_{j=0}^i \Cone(q_j) \] Each \( Q_i \)
is an upper set since it is a union of cones and clearly \( Q_i \subseteq
Q_{i+1} \) for all \(i \geq 1\) but this sequence cannot stabilize since it
is constituted of minimal elements, which contradicts \( (P,\leq) \) being
Noetherian.

Assume now that \( (P,\leq) \) has finitely generated upper sets and
consider an infinite sequence of upper sets \( Q_1 \subseteq Q_2 \subseteq
\ldots \)

Consider the set \[ Q = \bigcup_{i \geq 0} Q_i \] Being the union of upper
sets makes it an upper set as well, and hence it has a finite set of
minimal elements, \( \Minimal(Q) = \{q_1, \ldots, q_n\} \). Since there are
finitely many of them, they must all belong to some \( Q_j \) for \( j \)
large enough and thus have the sequence stabilize which shows that \(
(P,\leq) \) is Noetherian.  \end{proof}

\begin{defn} A partially ordered set \( (P,\leq) \) is \emph{finitely
certifiable} if it has maximums and is Noetherian.  \end{defn}

\begin{prop}\label{prop:noether_subset} Let \( (P,\leq) \) be a finitely
certifiable partially ordered set, and let \( Q \subseteq P \) be an upper
set of \(P\). Then \( (Q,\leq) \) is finitely certifiable.  \end{prop}
\begin{proof} Since \( Q \) is an upper set, we are guaranteed that the
maximum of any two elements in \(Q\) is in \(Q\) as well. And similarly, if
\( Q_1 \subseteq Q_2 \subseteq \ldots \) is a sequence of upper sets in
\(Q\), then since it stabilizes in \(P\) it does in \(Q\) as well.
\end{proof}

\begin{defn}\label{defn:poset_product} Let \( (P_1, \leq_1) \) and \( (P_2,
\leq_2) \) be two partially ordered set. Define their \emph{product} to be
the partially ordered set \( (P_1 \times P_2, \leq) \) where \[ (p_1,p_2)
\leq (q_1,q_2) \iff p_1 \leq_1 q_1 
\quad\mbox{and}\quad p_2 \leq_2 q_2 \]
\end{defn}

We want to show now that the finite product of finitely certifiable
partially order sets is finitely certifiable as well. For this, we first
show the following lemma.

\begin{lem}\label{lem:lemma_for_poset_product} Let \( (P,\leq) \) be a
Noetherian partially ordered set, and let \( P' \subseteq P \) be any
infinite subset of \(P\). Then there is an infinite sequence \( p_1, p_2,
\ldots \) of elements of \( P \) for which \[ p_1 < p_2 < \ldots \] that
is, for all \(i\) we have \( p_i \neq p_{i+1} \) and \( p_i \leq p_{i+1}
\). In particular, there is no infinite subset of incomparable elements.
\end{lem} \begin{proof} Let \( Q \) be the union of the cones of the
elements of the set \( P' \). Since \(Q\) is an upper set, it has a finite
set of minimal elements, \( \Minimal(Q) \), which has to be a subset of \(
P' \). So there must exist at least one element of \( \Minimal(Q) \), say
\( p_1 \), whose cone meets infinitely many elements of \( P' \). Let \(
P'' \) be those elements of \( P' \) in the cone of \( p_1 \) that are not
equal to \( p_1\). Applying the same argument to the set \( P'' \) yields
an element, \( p_2 \), in \( P'' \) whose cone meets an infinite number of
elements of \( P'' \) and clearly we have that \( p_1 < p_2 \). Iterating
this argument yields the desired infinite sequence.  \end{proof}

\begin{thm}\label{thm:poset_product} Using the notation of 
Definition~\ref{defn:poset_product}, 
the product of two finitely certifiable partially
ordered sets is finitely certifiable.  \end{thm} \begin{proof} Let \( Q \)
be an upper set of \( (P, \leq) \) and assume by contradiction that the set
\( \Minimal(Q) \) is infinite. This implies that there is an infinite
sequence \[ (p_1^1,p_2^1), (p_1^2,p_2^2), (p_1^3, p_2^3), \ldots \] of
distinct elements of \( \Minimal(Q) \). Then we can find an infinite
sequence of distinct elements in either \( P_1\) or \(P_2\). Without loss
of generality, let us assume that the set \[ P' = \{ p_1^1, p_1^2, \ldots
\} \subseteq P_1 \] is infinite and that its element are all distinct, that
is, for all \( i \neq j \) we have \( p_1^i \neq p_1^j \). Passing to a
subsequence and using Lemma~\ref{lem:lemma_for_poset_product}, we can
assume without loss of generality that we have \[ p_1^1 < p_1^2 < p_1^3 <
\ldots \] But then we can conclude that for each \( i \neq j \) we have \(
p_2^i \neq p_2^j \) or else one of \( (p_1^i,p_2^i) \) and \( (p_1^i,
p_2^j) \) would not be minimal. Hence the set \( \{ p_2^1, p_2^2, \ldots \}
\subseteq P_2 \) is an infinite set. Using 
Lemma~\ref{lem:lemma_for_poset_product} again and passing to a subsequence, we
can assume without loss of generality that we have \[ p_2^{\tau(1)} <
p_2^{\tau(2)} < p_2^{\tau(3)} < \ldots \] for some permutation, \( \tau \),
of the integers. But this is impossible: consider \( j_0 \) such that \(
\tau(j_0) = 1 \) then for any \( j > j_0 \) we have \[ p_2^1 =
p_2^{\tau(j_0)} < p_2^{\tau(j)} \] and hence \[ (p_1^1,p_2^1) <
(p_1^{\tau(j)},p_2^{\tau(j)}) \] which contradicts the minimality of \(
(p_1^{\tau(j)},p_2^{\tau(j)}) \).  \end{proof}

Our main interest in this paper consists of the following example.

\begin{proposition} 
Let \( \ZZ^t_{\geq 0}\) be the set of non-negatives \(t\)-tuples
and equip it with the partial order \[ (a_1, \ldots, a_t) \leq (b_1,
\ldots, b_t) \iff a_i \leq b_i \quad \forall i \] then this partially
ordered set is finitely certifiable.  
\end{proposition} 
\begin{proof} Using Theorem~\ref{thm:poset_product}, 
we simply need to show that the partially ordered
set \( (\ZZ_{\geq 0}, \leq) \) is finitely certifiable. 
The maximum of two elements is just the usual maximum.
If $U\subset\ZZ_{\ge 0}$ is an upper set, then it contains some
nonzero element, $a$, and therefore has a minimum element
$a'$; in this case $a'$ is the unique minimal element.
Hence $\ZZ_{\ge 0}$ is finitely certifiable, and hence also
$\ZZ^t_{\ge 0}$.
\end{proof}

We now study functions on partially ordered sets and aim to conclude that
when they are finitely certifiable, functions can be expressed as finite
linear combinations.

\begin{defn} Let \( (P,\leq) \) be a partially ordered set and let \(
L^1(P) \) denote the set of functions \( f \from P \to \CC \) for which \[
\| f \|_1 = \sum_{p \in P}|f(p)| \] is finite. For each function \(f \in
L^1(P)\) and element \( p \in P \) define the \emph{cone sum of \(f\) at
\(p\)} to be \[ \widehat f(p) = \sum_{q \in \Cone(p)} \!\!\!\!\!\!f(q) \]
\end{defn}

We now state an inclusion-exclusion theorem, giving M\"obius-type
coefficients.

\begin{thm}\label{thm:mobius_proof} Let \( (P,\leq) \) be a finitely
certifiable partially ordered set and let \( f \in L^1(P) \). Then for all
\(p \in P\), there are integers \( \mu(p) \) for which
\begin{enumerate}[label= (\arabic*)] \item \( \mu(p) = 0 \) for all but
finitely many elements, \(p \in P\), and \item for any function \( f \in
L^1(P) \) we have \[ \sum_{p \in P} f(p) = \sum_{p \in P} \mu(p) \widehat
f(p) \] \end{enumerate} Specifically, \( \mu(p) \) is the sum of the number
of times \(p\) appears as the maximum of an odd number of minimal points of
\(P\), minus the number of times it appears as the maximum of an even
(non-zero) number of minimal points.  
\end{thm} 
\begin{proof} For any fixed \(
\varepsilon > 0 \), we will show that \[ \left| \sum_{p \in P} f(p) -
\sum_{p \in P} \mu(p) \widehat f(p) \right| < \varepsilon \] for some
coefficients, \( \mu(p) \), that we will describe below.

We begin by replacing \( f \) by a function which is very close but has the
advantage of having a finite support, which we need in order to rearrange
what would otherwise be infinite sums. Since \[ \sum_{p \in P} |f(p)| \] is
finite, there exists a function \( f_\varepsilon \) of finite support such
that \( \| f - f_\varepsilon \|_1 < \varepsilon \). Now, we observe that \(
| \widehat f(q) - \widehat f_\varepsilon(q) | < \varepsilon \), since \[ |
\widehat f(q) - \widehat f_\varepsilon(q) | \leq \sum_{p \in \Cone(q)}
\!\!\!\!\! |f(p) - f_\varepsilon(p) | \leq \| f - f_\varepsilon \|_1 <
\varepsilon \] Now since \( f_\varepsilon \) has a finite support, we can
use the finite inclusion-exclusion principle on the cones of the minimal
elements of \( P \) to obtain \begin{multline*} \sum_{p \in P}
f\varepsilon(p) = \sum_{1 \leq i_1 \leq n} \widehat f_\varepsilon(p_{i_1})
- \!\!\!\!\! \sum_{1 \leq i_1<i_2 \leq n} \!\!\!\!\! \widehat
f_\varepsilon(\max(p_{i_1},p_{i_2})) + \ldots \\ + (-1)^{n+1} \!\!\!\!\!
\sum_{1 \leq i_1< \ldots <i_n \leq n} \!\!\!\!\! \widehat
f_\varepsilon(\max(p_{i_1}, \ldots, p_{i_n})) \end{multline*} This sum
being finite, we can rearrange it and write as \[ \sum_{p \in P} \mu(p)
\widehat f_\varepsilon(p) \] where \( \mu \) is a function on \(P\) whose
support is included in the set of all possible combinations of maximums of
minimal elements of \( P \), that is \[ \support(\mu) \subseteq \{
\max(p_{i_1}, \ldots, p_{i_k}) \mid k=1,\ldots,n \quad i_j \in \{ 1,
\ldots, n\} \} \] As claimed in the theorem's statement, one can see that
\( \mu(p) \) is the sum of the number of times \(p\) appears as the maximum
of an odd number of minimal points of \(P\), minus the number of times it
appears as the maximum of an even (non-zero) number of minimal points.

We now have \begin{align*} \left| \sum_{p \in P} f(p) - \sum_{p \in P}
\mu(p) \widehat f(p) \right| &\leq \left| \sum_{p \in P} f(p) -
f_\varepsilon(p) \right| + \left|\sum_{p \in P}
f_\varepsilon(p)-\mu(p)\widehat f_\varepsilon(p) \right| \\ &+
\left|\sum_{p \in P} \mu(p)\widehat f_\varepsilon(p) - \mu(p) \widehat f(p)
\right| \\ &\leq \| f-f_\varepsilon \|_1 + 0 + \varepsilon \sum_{p \in P}
|\mu(p)| < C \varepsilon \end{align*} For some fixed constant \( C \) which
only depends on \( P \) and thus concludes the proof.  \end{proof}

This allows to state the main theorem of this subsection. It follows
immediately from the above theorems and its application suggests a very
simple modified trace that we will describe in the next subsection.

\begin{thm}\label{thm:main_mobius} Let \( (P,\leq) \) be a finitely
certified partially ordered set, let \( Q \) be an upper set of \( P \) and
let \( f \in L^1(P) \). Then there exists 
``M\"obius coefficients,'' \( \mu_Q
\from Q \to \ZZ \), such that \[ \sum_{q \in Q} f(q) = \sum_{q \in Q}\mu(q)
\widehat f(q) \] where \( \mu_Q(q) = 0 \) for all but finitely many
elements of \(Q\).  \end{thm} \begin{proof} This is immediate, applying
Theorem \ref{thm:mobius_proof} to \(Q\) itself since 
Proposition~\ref{prop:noether_subset} 
guarantees it is a finitely certifiable partially
ordered set itself.  \end{proof}
%
%
%
%
\subsection{The Certified Trace}

\begin{defn}\newnot{symbol:CertSNBCr} Consider the collection of sets \[
\CertSNBC = \{ \CertSNBC(k,n) \}_{k,n \geq 1} \] defined by \[
\CertSNBC(k,n) = \{ \potwalk \in \SNBC(k,n) \mid \rho_H(\Graph\potwalk) <
\rhoroot B \} \] that is, all the potential walks that are closed, strictly
non-backtracking, of order less than \( r \) and which are not tangles.
\end{defn}

\begin{prop} The collection of sets \( \CertSNBC \) is a walk collection.
\end{prop} \begin{proof} This is immediate since both the Hashimoto
spectral radius is preserved up to symmetry and size increase.  \end{proof}

\begin{defn}
\label{defn:CertTr}
\newnot{symbol:CertTr}\newnot{symbol:CertTr_r} We define the
\emph{\(k\)th certified trace of size \(n\)}, \( \CertTr(G,k) \), of a
graph \( G \in \cC_n(B) \) to be the \( (k,n) \)-modified trace associated
to \( \CertSNBC \) and we write its associated walk sum as \[
\EE_{G\in\cC_n(B)}[\CertTr(G,k)] = \WalkSum(\CertSNBC(k,n) ) \] and for any positive
integer, \( r \geq 1 \), we define the \emph{truncated certified trace} to
be the associated truncated walk sum \[ \EE_{G\in\cC_n(B)}[\CertTr_{<r}(G,k)] =
\WalkSum_{<r}(\CertSNBC(k,n)) \] which only sums certified potential walks
of order less than \(r\).  \end{defn}

\begin{prop}\label{prop:tangle_upset} Let \( T \) be the graph of a type \(
\cT \). Then the set of vectors \[ 
U = \{ \vec k \in \ZZ_{\geq 1}^{E_T} \mid
\VLG(T, \vec k\,) \notin \Tangle_B \} \] is an upper set of \( \ZZ_{\geq
1}^{E_T} \). 
\end{prop}
\begin{proof}
It suffices to show that for any $\vec k\le\vec k'$, then
$\vec k\in U$ implies that $\vec k'\in U$.  So fix
a $\vec k\le \vec k'$ with $\vec k\in U$.
We will show that $\vec k'\in U$.

By definition, $T$ is a connected graph.  
Fix a vertex, $v\in T$.
For any $\vec k\in\ZZ_{\ge 1}^{E_T}$, $v$ can also be viewed as 
a vertex in $\VLG(T,\vec k)$.  Since $T$ is connected, we have that
$$
\mu_1\bigl( \VLG(T,\vec k) \bigr) =
\lambda_1\bigl( \VLG(\Line(T),\Line(\vec k) \bigr),
$$
where $\Line(T)$ is the oriented line graph of $T$, and
$\Line(\vec k)$ is the vector of edge lengths on $\Line(T)$
which to an edge, $(e,e')$, in $\Line(T)$ (so $e,e'\in E_T$)
assigns the length $k(e)$.
Since $\vec k\le\vec k'$, then 
$$
\Line(\vec k) \le \Line(\vec k').
$$
Now we apply Proposition~\ref{pr:vlg_greater_than} to conclude that
$$
\mu_1\bigl( \VLG(T,\vec k') \bigr) 
\le 
\mu_1\bigl( \VLG(T,\vec k) \bigr)  .
$$
Since $\vec k\in U$, we have
$$
\mu_1\bigl( \VLG(T,\vec k) \bigr)  < \rhoroot{B},
$$
and hence
$$
\mu_1\bigl( \VLG(T,\vec k') \bigr)  < \rhoroot{B}.
$$
It follows that $\vec k'\in U$.
\end{proof} 

\begin{cor} For each type, \( \cT \), there is a
finite number of \emph{tangle certificates}, that is, a finite family of
vectors, \( \vec k_{0,1}, \ldots, \vec k_{0,s} \) such that \[ \VLG(T, \vec
k) \notin \Tangle_B \iff \vec k \geq \vec k_{0,i} \text{ for some }
i=1,\ldots,s \] where \( T \) is the underlying graph of the type \( \cT \)
\end{cor}

We now have all the elements required to state the main theorem of this
subsection.

\begin{thm}\label{thm:certified_walk_expression} Consider the type-form
expansion described in Theorem \ref{thm:type_form_expansion}. The walk sum
coefficients
of the certified trace can be written as a finite linear combination of
functions of the form \[ \Term [ \cT, \vec k_0, r ] (k) = \sum_{\vec k \geq
\vec k_0} \sum_{\vec m \cdot \vec k \,=\, k} W_\cT(\vec m, \vec k\,)
P_i(\vec k\,) \] where \( W_\cT(\vec m, \vec k\,) \) counts the number of
strictly non-backtracking walks in the type \( \cT \) with edge lengths
defined by \( \vec k \) and that traverses each edge \( f \in E_T \) of the
type \( m_f \) times; and \( P_i(\vec k\,) \) is a polyexponential
function.  
\end{thm} 
\begin{proof} Since \( \CertSNBC \) is clearly
multiplicity and length determined, we can apply
Theorem~\ref{thm:type_form_expansion} to obtain an expansion \[
\CertTr_{<r}(\cW,k,n) = \sum_{\cT \in \Types[r]} n^{-\ord(\cT)}
\sum_{i=0}^{r-\ord(\cT)} Q_{\cT,i}(k) n^{-i} + \error(r, n, k) \] where the
\( Q_{\cT,i} \) are functions satisfying \[ Q_{\cT,i}(k) =
\sum_{\substack{\vec m, \vec k \in \ZZ_{\geq 1}^{E_T} \\ \vec m \cdot \vec
k \,=\, k}} W_\cT(\vec m, \vec k\,) P_i(\vec k\,) \] where the \( P_i( \vec
k\,) \) are \(B\)-polyexponential functions, and with error term \[ |
\error(r,n,k) | \leq ck^{2r+2} \rho(H_B)^k n^{-r} \] Now, we simply apply
Theorem~\ref{thm:main_mobius} to the upper set \[ \{ \vec k \in \ZZ_{\geq
1}^{|E_T|} \mid \VLG(T, \vec k\,) \notin \Tangle_B \} \] and obtain that
there are finitely many certificates, \( \vec k_0 \), such that we can
rewrite the \( Q_{\cT,i}(k) \) as a finite linear combination of cone sums
of the form \[ \Term [ \cT, \vec k_0, r ] (k) = \sum_{\vec k \geq \vec k_0}
\sum_{\vec m \cdot \vec k \,=\, k} W_\cT(\vec m, \vec k\,) P_i(\vec k\,) \]
\end{proof}
%
%
%
%
\subsection{Asymptotic Expansions and a Proof of
Theorem~\ref{th:certified_trace_expansion}}
\label{sb:asymptotic_expansions}

In this section we give the somewhat technical estimates
to establish that terms in Theorem~\ref{thm:certified_walk_expression}
are $B$-Ramanujan.  After proving this,
we collect the various theorems proven which
immediately establish Theorem~\ref{th:certified_trace_expansion}.

\begin{thm}\label{thm:certified_B_Ramanujan} The terms of the form
\begin{equation}\label{equ:ML_sum_Qi} \Term [ \cT, \vec k_0, r ] (k) =
\sum_{\vec k \geq \vec k_0} 
\ \sum_{\vec m \cdot \vec k \,=\, k} W_\cT(\vec
m, \vec k\,) P_i(\vec k\,) \end{equation} given in
Theorem~\ref{thm:certified_walk_expression} are \(B\)-Ramanujan functions.
Hence the $r$-th truncated certified trace admits a
$1/n$-asymptotic expansion to order $r$ with \(B\)-Ramanujan coefficients.
\end{thm}

These types of sums are similar to those in Theorem~8.2 of
\cite{friedman_alon} and Theorem~2.18 of \cite{friedman_random_graphs}. In
both those theorems, a key idea is to divide the vector \( \vec m \) by
those components equal to \( 1 \) and those components that are at least \(
2 \). Here we notice that by a similar division, into components at most \(
M \) and components at least \( M+1 \), for a fixed value of \( M \), we
greatly simplify the computations. The role of tangles on the asymptotic
expansion becomes clear as well.

\begin{proof}[Proof of Theorem~\ref{thm:certified_B_Ramanujan}]
Fix an integer \( M \) such that \( M \geq 1/\varepsilon \) and fix a
partition of the edges of the type \[ E_T = E_{T,1} \amalg E_{T,2} \] We
call the edges in \( E_{T,1} \) the \emph{short} and the edges in \(
E_{T,2} \) \emph{long}. We write \( \vec m \in ML(k,M) = ML(k,M,E_{T,1},
E_{T,2}) \) if \begin{enumerate}[label= (\arabic*)] \item \( \vec m \in \{
\vec m \mid \vec m \cdot \vec k = k,\, \vec k \geq \vec k_0, \text{ for \(
\vec k_0 \) a certificate} \} \), \item \( m_f \leq M \) if \( f \in
E_{T,1} \), \item \( m_f \geq M+1 \) if \( f \in E_{T,2} \).
\end{enumerate} and so we can rewrite equation \eqref{equ:ML_sum_Qi} as \[
\Term [ \cT, \vec k_0, r ] (k) = \sum_{E_{T,1} \amalg E_{T,2} = E_T}
\sum_{\vecmk \in ML(k,M,E_{T,1}, E_{T,2})} W_\cT(\vec m\,) q_i(\vec k\,) \]
For each \( \vecmk \in ML(k,M) \) define \begin{equation}\label{eq:K_i}
K_1\vecmk = \sum_{f \in E_{T,1}} m_f k_f \quad \text{and} \quad K_2\vecmk =
\sum_{f \in E_{T,2}} m_f k_f \end{equation} and so it follows that \(
K_1\vecmk + K_2\vecmk = k \) for all \( \vecmk \in ML(k,M) \). Since there
are finitely many partitions of the edges of the type, we only need to
consider expressions of the form \begin{equation}\label{eq:partition_term}
\Term(k,E_{T,1},E_{T,2}) = \sum_{\vecmk \in ML(k,M)} W_\cT(\vec m\,)
q_i(\vec k\,) \end{equation} and show that they are \( B \)-Ramanujan
functions, which would prove our Theorem~\ref{thm:certified_B_Ramanujan}

\begin{lemma}\label{le:CertTr_B_ramanujan_expansion} Let \( \cT \in
\Types[r] \) be a type of order less than \( r \) and let \(
(E_{T,1},E_{T,2}) \) be a partition of the edges of the type. Then for any
\( \varepsilon > 0 \) and integer \( M \geq 1/\varepsilon \) we have that
the function \( \Term(k, E_{T,1}, E_{T,2}) \), as described above, is \( B
\)-Ramanujan.  
\end{lemma} \begin{proof} We distinguish two cases for this
proof: \begin{enumerate}[label= (\arabic*)] \item When \( E_{T,1} =
\emptyset \) \item When \( E_{T,1} \neq \emptyset \) \end{enumerate} So,
let us start with the first case: we consider the term in
\eqref{eq:partition_term}, in the case where $E_{T,1}=\emptyset$ and
$E_{T,2}=E_T$.  For every $(\vec m,\vec k)\in ML(k,M)$ we have $$
W_\cT(\vec m) \leq Ck^C \rho(H_B)^{k/2}, $$ since each walk in $W_\cT(\vec m)$
gives a walk of length $k$ in ${\rm Path}(T,\vec k)$, which corresponds to
walk of length at most $k$ in ${\rm Path}(T,\vec k^0)$, and the number of
such walks is bounded by $Ck^C \rho(H_B)^{k/2}$ since the largest Hashimoto
eigenvalue of ${\rm Path}(T,\vec k^0)$ is at most $\rho(H_B)^{-1/2}$.

Now we note the following crude bound.  \begin{lem} The number of pairs
$(\vec m,\vec k)$ for which $\vec m, \vec k \geq \vec 1$ and $\vec
m\cdot\vec k=k$ is at most $$ \binom{k+|E_T|-1}{|E_T|-1}  k^{|E_T|}.  $$
\end{lem} \begin{proof} For each $(\vec m,\vec k)$, set $s_f=m_f k_k$ for
each $f\in E_T$.  Then the sum of the $s_f$ is $k$, and there are
$\binom{k+|E_T|-1}{|E_T|-1}$ ways of representing $k$ as the sum of
integers $\{s_f\}_{f\in E_T}$ with $s_f\geq 0$.  Furthermore, each $s_f$
can be represented as at most $k$ products $m_f k_f$, with $m_f,k_f\geq 1$.
\end{proof}

Finally, we notice that if $m_f\geq M+1$ for all $f\in E_T$, and $\vec
m\cdot\vec k=k$, then we have $$ S=\sum_{f\in E_T} k_f \leq k/(M+1), $$ and
so $$ P(\vec k) \leq \prod_{f\in E_T} \rho(H_B)^{k_f} Ck_f^C \leq C'k^{C'}
\rho(H_B)^{\sum_{f\in E_T} k_f} \leq C' k^{C'} \rho(H_B)^{k/(M+1)}.  $$ It follows
that, in the notation of \eqref{eq:partition_term}, we have $$ | {\rm
Term}(\emptyset,E_T) | \leq |ML(k,M,\emptyset,E_T)| \ \!\!\! \max_{\vec
m\in ML(k,M\emptyset,E_T)} \!\!\! |W_\cT(\vec m)| \   \!\!\! \max_{\vec
k\in ML(k,M,\emptyset,E_T)} \!\!\!\!\!\!\!\!\!\!\!\!P(\vec k) $$ $$ \leq
Ck^C \ \rho(H_B)^{k/2} Ck^C \ C' k^{C'} \rho(H_B)^{(k/2) + (k/(M+1))} .  $$ $$ \leq
C'' k^{C''} \rho(H_B)^{k((1/2)+\varepsilon)}.  $$ It follows that for any
$\widetilde{\varepsilon}>0$ and some (new) constant, $C$, and sufficiently
large $M$ we have $$ | {\rm Term}(\emptyset,E_T) |  \leq Ck^C
(d-1+\widetilde{\varepsilon})^{k/2} .  $$

It suffices to give a similar bound for ${\rm Term}(E_{T,1},E_{T,2})$ for
the other partitions of $E_T$ as $E_{T,1}\amalg E_{T,2}$.  But this case
alone shows how the certification seems essential, in that if $W_\cT(\vec
m)$ could be as large as $(d-1+\delta)^{k/2}$ for some fixed $\delta>0$, we
could not give a good enough bound just on ${\rm Term}(\emptyset,E_T)$.

\subsubsection*{A Specialized Bound For $W_\cT(\vec m)$}

To bound ${\rm Term}(E_{T,1},E_{T,2})$ when $E_{T,1}$ is not empty, we will
use Theorem~\ref{th:conv_polyexp_growth}.  To do so we would like to bound
$W_\cT(\vec m)$ for each $(\vec m,\vec k)\in ML(k,M)$, in terms of
$K_1,K_2$ given by \eqref{eq:K_i}.  Specifically, we which for $m_f$ fixed
for $f\in E_{T,1}$, gives a bound for $W_\cT(\vec m)$ in terms of the
exponent $K_2$, rather than $k$.  Here is what we will need.

\begin{lem}\label{le:crucial_walk_bound}
Fix $\cT,\vec k^0,M$, and a partition $E_{T,1}\amalg E_{T_2}$
of $E_T$.  Then there exists a constant, $C$, for which $$ W_\cT(\vec m)
\leq C K_2^C \rho(H_B)^{K_2/2} $$ provided that $(\vec m,\vec k)\in
ML(k,M,E_{T,1},E_{T,2})$ \end{lem} \begin{proof} Consider a walk, $w$, in
$\cT$.  Such a walk traverses a bounded number of edges in $E_{T,1}$,
namely at most $B=M |E_{T,1}|$ edges.  Let $w^0$ be the walk in $E_{T,2}$
edges taken before the first $E_{T,1}$ edge is taken in $w$, and for $1\leq
i\leq B$, let $w^i$ be the walk taken between the $i$-th and $(i+1)$-th
occurrence of an $E_{T,1}$ edge in $w$.  (Some of the $w^i$ will be empty
when $w$ traverses two or more $E_{T,1}$ edges in a row; and some of the
$w^i$ will be empty if $w$ traverses fewer than $B$ edges of $E_{T,1}$.)
Let $K_2^i$ denote the number of edges traversed in $w^i$ when transferred
to the graph ${\rm Path}(T,\vec k)$.  First, we have $$ K_2^0 + K_2^1 +
\cdots  + K_2^B  = K_2.  $$ Second, for fixed $K_2^i$, the number of walks
in ${\rm Path}(T,\vec k)$ of length at most $K_2^i$ is bounded by
\begin{equation}\label{eq:w^i_bound} C (K_2^i)^C \rho(H_B)^{K_2^i/2}
\end{equation} Finally, between $w^i$ and $w^{i+1}$, we traverse one edge
of $E_{T,1}$, which can happen in at most a constant number of ways (namely
one less than the maximum degree of $T$).  It follows that
\begin{equation}\label{eq:mess_W} W_\cT(\vec m) \leq \sum_{\substack{
K_2^i\geq 0 \\ K_2^0+\cdots+K_2^B = k} } C^B \prod_{i=0}^B [ C(K_2^i)^C
\rho(H_B)^{K_2^i/2} ] , \end{equation} where the $C^B$ term reflects the $B$
choices of $E_{T,1}$ element, and the product over $i$ comes from the
\eqref{eq:w^i_bound}.  We bound the right-hand-side of \eqref{eq:mess_W} as
follows: $$ \prod_{i=0}^B [ C(K_2^i)^C \rho(H_B)^{K_2^i/2} ] \leq \bigl( CK_2^C
\bigr)^{B+1} \rho(H_B)^{(K_2^0+\cdots+K_2^B)/2} \leq C' K_2^{C'} \rho(H_B)^{K_2/2}
.  $$ Finally the number of ways of writing $K_2$ as a sum of the $K_2^i$,
i.e., $B+1$ non-negative integers is exactly $\binom{K_2+B}{B+1}$, a
polynomial in $K_2$.  Hence \eqref{eq:mess_W} yields $$ W_\cT(\vec m) \le
\binom{K_2+B}{B+1} C' K_2^{C'} \rho(H_B)^{K_2/2}.  $$ \end{proof}

It remains to show an asymptotic expansion for $W_\cT(\vec m)P(\vec k)$ for
$B$-Rama\-nujan $P$.  Let us partition $\vec m$ and $\vec k$ by the
$E_{T,1}\amalg E_{T,2}$ partition; namely, let $$ \vec k^i = \{ k_f
\}_{f\in E_{T,i}}, $$ so we may write $\vec k = (\vec k^1,\vec k^2)$, and
similarly $\vec m = (\vec m^1,\vec m^2)$.  It suffices to show that for
$P_1(\vec k^1)$ and $P_2(\vec k^2)$ polyexponentials, we have that $$
g(E_{T,1},E_{T,2},k) = \sum_{(\vec m,\vec k)\in ML(k,M,E_{T,1},E_{T,2}) }
W_\cT(\vec m) P_1(\vec k^1) P_2(\vec k^2) $$ is $B$-Ramanujan.  So for any
integer, $K_1$, and any $\vec m^1\in \{1,\ldots,M\}^{E_{T,1}}$, let $$
Q_1(K_1;\vec m^1) = \sum_{ \substack{ \vec m^1 \cdot \vec k^1 = K_1 \\ \vec
m^1\geq \vec 1,\ \vec k^1\geq \vec k^{1,0} } } P_1(\vec k^1), $$ and $$
Q_2(K_2) = \sum_{ \substack{ \vec m^2 \cdot \vec k^2 = K_2 \\ \vec m^2\geq
\vec 1,\ \vec k^2\geq \vec k^{2,0} } } W_\cT( (\vec m^1,\vec m^2) )
P_2(\vec k^2).  $$ We have $Q_1$ is a $B$-Ramanujan, and $Q_2$ is of growth
at most $(d-1+\varepsilon)^{1/2}$. Hence, for each $\vec m^1\in
\{1,\ldots,M\}^{E_{T,1}}$ we have $Q_1 * Q_2$ is $B$-Ramanujan.  Hence the
sum over all $\vec m^1$ is $B$-Ramanujan.

This completes the proof of Lemma~\ref{le:CertTr_B_ramanujan_expansion}.
\end{proof}
Since each term in Theorem~\ref{thm:certified_B_Ramanujan} is a finite
sum of terms in 
Lemma~\ref{le:CertTr_B_ramanujan_expansion}, and clearly a finite sum
of $B$-Ramanujan functions is again $B$-Ramanujan,
we conclude
Theorem~\ref{thm:certified_B_Ramanujan}.
\end{proof}
%
\section[Certified Traces With Tangles;
Proofs of Theorems~\ref{th:certified_trace_expansion_with_tangles}
and \ref{th:main_expansion_B}]{
Certified Traces In Graphs With Tangles and 
The Proof of Theorems~\ref{th:certified_trace_expansion_with_tangles}
and \ref{th:main_expansion_B}
}
\label{se:p1-with-tangles}

%
%
%
%

In this section we widen the notion of a \emph{type} to prove 
Theorem~\ref{th:certified_trace_expansion_with_tangles},
theorem, based on the ideas of Section~9 of \cite{friedman_alon}.  
We then give the short argument to
prove Theorem~\ref{th:main_expansion_B}.
Actually, we will prove the following more general form
of Theorem~\ref{th:certified_trace_expansion_with_tangles}.

\begin{thm} \label{th:with_without} 
Let $B$ be any connected graph of positive order without half-loops.
Let \( \Psi \) be any 
finite collection of 
connected, pruned, \( B \)-graphs, each of order at least one.
Let \( I_\Psi(G) \) be the function that is one or
zero according to whether or not \( G \) has a subgraph isomorphic to an
element of \( \Psi \).  Then for any positive integer, $r$, we have that
$$
\expect{G\in\cC_n(B)}{I_\Psi(G) {\rm CertTr}_{<r}(G,k) }
$$
has a $1/n$-asymptotic expansion with the usual error term and with
$B$-Ramanujan coefficients.
\end{thm}

We remark that in the above theorem, we can replace the condition that
each $\psi\in\Psi$ be pruned by the condition that each 
connected component of each $\psi$ has order at least one;
however, assuming that each $\psi$ is pruned simplifies the argument,
and is sufficient for our purposes.

This section also contains a related theorem,
Theorem~\ref{th:no_walks} below, which will be needed
to complete the proof of 
Theorem~\ref{th:main};
it is a consequence of the methods of this section.

\begin{thm}\label{th:no_walks} 
Let $B$ be any connected graph of positive order without half-loops.
Let \( \Psi \) be any 
finite collection of 
connected, pruned, \( B \)-graphs, each of order at least one.
Let \( I_\Psi(G) \) be the function that is one or
zero according to whether or not \( G \) has a subgraph isomorphic to an
element of \( \Psi \).  Then for any positive integer, $r$, we have that
$$
\EE_{G\in\cC_n(B)}[I_\Psi(G)] = 
p_1 n^{-1} + p_2 n^{-2} + \cdots + p_{r-1}n^{1-r} +
O(n^{-r}) 
$$
for some constants \( p_i \).  
\end{thm}

At this point we simply adapt the proof of Theorem~9.3 of
\cite{friedman_alon} to this situation.
The first idea is that we want to develop the theory of walk-sums
and asymptotic expansion which allows us to condition upon the
existence of a certain subgraph occurring in $G\in\cC_n(B)$.
The conditioned walk-sums will then be used to express,
via inclusion/exclusion,
the indicator function, $I_\Psi(G)$, in
Theorems~\ref{th:with_without} and~\ref{th:no_walks}.
Let us give some precise terms and prove the basic estimates.

\subsection{Potential Graph Specializations}

\begin{definition}
A \emph{\gls[format=hyperit]{potential graph specialization}} is a pair, \(
(\Omega,\nonsigma) \), where \( \Omega \) is a \( B \)-graph, and \(
\nonsigma\from V_\Omega \to \{1,\ldots,n\} \) is a map.  
We may view $(\Omega,\nonsigma)$ as an event in $\cC_n(B)$, namely the
event that $\Omega$ occurs in $G\in\cC_n(B)$ in a way so that
the map on vertices of the occurrence is given by $\nonsigma$.
We say that a potential graph specialization is 
{\em feasible} if it is an event of positive probability in 
$\cC_n(B)$.
\end{definition}
Alternatively, we may describe the event $(\Omega,\nonsigma)$,
as the set of $G\in \cC_n(B)$
for which the following is true:
for each $e\in E_\Omega$, let $e_B\in E_B$ be the $B$-edge corresponding to
$e$ in its structure as a $B$-graph;
let the tail of $e$ be $v_1$, and its head $v_2$,
and let $v_{1,B}$ and $v_{2,B}$ be the $B$-vertices
corresponding to $v_1$ and $v_2$ respectively;
the event $(\Omega,\nonsigma)$ is the event that for all $e\in E_\Omega$,
and ensuing $e_B,v_1,v_2,v_{1,B},v_{2,B}$ as above,
the (unique) edge over $e_B$ in
$G$ with tail $(v_{1,B},\nonsigma(v_1))$ must have its head being
$(v_{2,B},\nonsigma(v_2))$.

Note that a potential graph specialization is equivalent to an
event $\iota_{\vec t}$
in the proof of Theorem~\ref{thm:prob_tangle}; the above definition
and notation is a bit better suited to the theorems in this section
and their proofs.

\begin{definition}
Given a potential
walk, $\potwalk$, and a potential graph specialization,
$(\Omega,\nonsigma)$, we define
$$
\prob{G\in\cC_n(B)}{\potwalk,(\Omega,\nonsigma)}
$$
to be
the probability that events \( \potwalk \) and \( (\Omega,\nonsigma) \) 
both occur 
in $G\in\cC_n(B)$, i.e., both $(\Omega,\nonsigma)$ and the walk
$\potwalk$ occur in $G$.
\end{definition}
In effect, both $\potwalk$ and $(\Omega,\nonsigma)$, if they occur,
describe subgraphs of $G$.  It may well
happen that the subgraphs of $G$ described
by $\potwalk$ and 
$(\Omega,\nonsigma)$ intersect on some 
of their vertices and edges; it may also happen
that $(w\;\vec t)$ and $(\Omega,\nonsigma)$ describe contradictory information
in $G$ and hence can never occur together.

\begin{definition}
Given a potential
walk, \( \potwalk \), and a potential graph specialization,
$(\Omega,\nonsigma)$, we may view $(\vec t,\nonsigma)$ as two
maps (from the vertices of $w$ and the vertices of $\Omega$) to
$\{1,\ldots,n\}$.  We say that
$(\vec t',\nonsigma')$ is {\em equivalent} to
$(\vec t,\nonsigma)$ if they differ by a permutation
of $\{1,\ldots,n\}$, i.e., there is a permutation $\pi\in S_n$ 
of $\{1,\ldots,n\}$
for which $\vec t'=\pi\circ\vec t$ and $\nonsigma'=\pi\circ\nonsigma$.
(If $\pi\in S_n$ fixes all elements of $\{1,\ldots,n\}$ in the range
of $\vec t$ and $\nonsigma$ then we get the same vertex maps
with $\vec t'=\pi\circ\vec t$ and $\nonsigma'=\pi\circ\nonsigma$, and
we consider $(\vec t',\nonsigma')$ to be the same as
$(\vec t,\nonsigma)$.)
We define
$$
\EE_{\rm symm}[\potwalk,(\Omega,\nonsigma)]_n
$$
to be the sum
$$
\sum_{(\vec t',\nonsigma')\sim(\vec t,\nonsigma)}
\prob{G\in\cC_n(B)}{(w;\vec t'),(\Omega,\nonsigma')},
$$
summing over all $(\vec t',\nonsigma')$ equivalent to
$(\vec t,\nonsigma)$ (analogous to Definition~\ref{de:symmetrized_sum}).
\end{definition}

Clearly 
Proposition~\ref{prop:EsymmProd} generalizes as follows.
\begin{proposition}\label{pr:gen_EsymmProd}
Let $B$ be a graph without half-loops.
Given a potential
walk, \( \potwalk \), and a potential graph specialization,
$(\Omega,\nonsigma)$, we have
$$
\EE_{\rm symm}[\potwalk,(\Omega,\nonsigma)]_n
$$
is zero if $\vec t$ and $\nonsigma$ imply contradictory values for the
permutation over any edge, $e\in E_B$, in $G\in\cC_n(B)$, and otherwise
equals
\begin{equation}\label{eq:gen_walk_probability}
\EE_{\rm{symm}}\potwalk_n = \prod_{v \in V_B} \frac{n!}{(n-b_v)!} \prod_{e
\in E_B} \frac{(n-a_e)!}{n!},
\end{equation}
where for $v\in V_B$, $b_v$ describes the number of vertices over $B$
in $\vec t\cup\nonsigma$, meaning that vertices where $\vec t$ and 
$\nonsigma$ agree are counted once,
and where for $e\in E_B$, $a_e$ counts the number of edges over 
$e$ occurring
in the union of the subgraphs determined by
$\potwalk$ and $(\Omega,\nonsigma)$.
\end{proposition}
Again, if $B$ has half-loops with $n$ even, then each half-loop gives rise
to an involution rather than a general permutation, and we can replace
the factorials for $a_e$ by odd factorials, as described in
Proposition~\ref{prop:EsymmProd}.

For the rest of this section we need to do two things.
First, we will show that in Theorems~\ref{th:with_without} and
\ref{th:no_walks}, using a ``M\"obius function,'' one can
replace the function $I_\Psi(G)$ by a linear combination of 
sums of the form
$$
\EE_{\rm symm}[\potwalk,(\Omega,\nonsigma)]_n
$$
where $\Omega$ ranges over a finite set of graphs, and the
$\potwalk$ has at most a certain order.
Then we will take
the union of the subgraphs determined by $\potwalk$ and
$(\Omega,\nonsigma)$, and divide them into {\em generalized types},
which consist of a finite number of fixed edges and vertices via
$(\Omega,\nonsigma)$, plus the remaining edges traversed by
$\potwalk$, with the union of the subgraphs divided into 
a simple generalization of the notion of a type.
Then we need to check most all of the results 
in Sections~\ref{se:p1-walk-sums} and \ref{se:p1-tangles}
carry over to our generalized notion of types, and to the sum of
$$
\EE_{\rm symm}[\potwalk,(\Omega,\nonsigma)]_n
$$
over all $\potwalk$ of that type.
Let us begin with the M\"obius function argument.

\begin{definition}
Let $B$ be a graph.  
For a collection $B$-graphs, $\Psi$, we
say that 
a \emph{set of derived graphs of \( \Psi \)}, denoted \( \Psi^+ \), 
is the set
\( B \)-graphs in \( {\rm Occurs}_B \) which can be expressed as the finite
union of graphs, each isomorphic to an element of $\Psi$; furthermore
we insist that $\Psi^+$ contain exactly one graph in each isomorphism class
of $B$-graphs.
We let
$\Psi^+_{<r}$ be those elements of \( \Psi^+ \) of order less than \( r
\). 
\end{definition}
If $\Psi$ is a finite collection of finite graphs, then $\Psi^+$ always
exists (without the axiom of choice); we will choose some arbitrary
$\Psi^+$ for each $\Psi$ we consider.
We emphasize that in defining $\Psi^+$,
it is simplest to work with one representative
in each isomorphism class of $B$-graphs,
since in each isomorphism class there are infinitely many
graphs.

\begin{theorem} Let $\Psi$ be a nonempty, finite collection of 
pruned, connected, non-empty graphs
of order at least one.
For any $r$, $\Psi_{<r}^+$ is finite.
\end{theorem}
\begin{proof}
Let $\Psi[i]$ denote the subset of $\Psi^+$
that can be obtained
as a union of $i\ge 1$ copies of the $B$-graphs in $\Psi$ (i.e., the union
of $i$ of its subgraphs, each isomorphic to an element of $\Psi$),
but not fewer than
$i$ copies.
Then $\Psi[1]$ is just a set of (isomorphism) representatives of 
the set of $B$-graphs, $\Psi$.  
We remark that since any element in $\Psi^+$ can be obtained as
a finite union of elements of $\Psi$, each element of $\Psi^+$
can be written as a minimal number of copies of elements of $\Psi$,
and hence each element of $\Psi^+$ lies in $\Psi[i]$ for exactly
one integer $i\ge 1$.
We will now establish a number of properties of $\Psi[i]$ for all $i$,
and these properties will prove the theorem.

First note that every vertex, $v$, in a graph
of $\psi\in\Psi^+$ lies
in a copy of (i.e., subgraph of $\psi$ isomorphic to) a graph
of $\Psi$, and hence $v$ is of degree at least two in $\psi$.
Hence every element of $\Psi^+$ is pruned.
Now if $\psi\in\Psi^+$, then any connected component of $\psi$
contains at
least one copy of an element of $\Psi$; hence,
by Theorem~\ref{th:pruned_inclusion}, any connected component of
$\psi$ has order at least one.

Second, we claim that for any graph $\psi\in\Psi[i]$, the order of
$\psi$ is at least $i$.  We can prove this by induction on $i$.
Indeed, this holds for $i=1$, since $\Psi[1]=\Psi$.  Assuming this holds
for some value of $i$, then any $\psi\in\Psi[i+1]$ is the union of
the $B$-graphs $\psi'\in\Psi[i]$ plus one copy of an element of
$\Psi$, and, by definition, the inclusion of $\psi'$ in $\psi$ is
not surjective.
Since the order of $\psi'$ is, by induction, at least $i$,
Theorem~\ref{th:strict_pruned_inclusion} implies that the
order of $\psi$ is strictly larger than that of $\psi'$, i.e., at least
$i+1$.
Hence the order of any element of $\Psi[i+1]$ is at least $i+1$.
Therefore, by induction,
the order of any element of $\Psi[i]$ is at least $i$.

Third, we claim that the set $\Psi[i]$ is finite for each $i$;
again we use induction on $i$.
Indeed, this is true for $i=1$, since $\Psi[1]=\Psi$.
Let us assume that $\Psi[i]$ is finite.
For any $i\ge 1$, each element of $\Psi[i+1]$ can be written as
a union of an element, $\psi_i$, of $\Psi[i]$ and an element 
of $\psi\in\Psi$.  Since $\psi_i$ and $\psi$ are finite $B$-graphs,
their union gives rise to finitely many (isomorphism classes of)
$B$-graphs, since each is obtained as the
disjoint union of $\psi$ and $\psi_i$, followed with the identification
of certain vertices and edges of $\psi$ with those of $\psi_i$
(and the number of such possible identifications is finite).
Since $\psi_i$ ranges over finitely many graphs of $\Psi_i$, and
$\psi$ the same of $\Psi$,
$\Psi[i+1]$ is a finite set.
Hence, by induction, for each $i$ we have that $\Psi[i]$ is finite.

Since each element of $\Psi^+_{<r}$ belongs in one of the
$\Psi[i]$ for $i<r$, we have that $\Psi^+_{<r}$ is finite.
\end{proof}

\begin{definition}
For any two finite graphs, \( \psi,\psi' \), let \( N(\psi,\psi') \)
denote the number of injections of \( \psi \) in \( \psi' \) as \( B
\)-graphs.  Of course, $N(\psi,\psi')$ depends only on the
isomorphism class of $\psi$ and $\psi'$, so we may also view
$N$ as defined on pairs of isomorphism classes of graphs.
\end{definition}

Clearly $N(\psi,\psi')$ is finite if $\psi,\psi'$ are finite graphs.

\begin{definition}\label{de:B_graph_order}
If $\psi,\psi'$ are finite $B$-graphs, we write
$\psi\le_B\psi'$ if $N(\psi,\psi')>0$.
\end{definition}

It is easy to check that the above relation is 
reflexive and transitive, and that
$\psi\le_B\psi'$ and $\psi'\le \psi$ implies that
$\psi$ and $\psi'$ are isomorphic.
Hence $\le_B$ is a partial order on isomorphism classes of $B$-graphs.

Since each element of $\Psi^+$ is a finite $B$-graph, 
then, under $\le_B$, each such element is greater than only
a finite number of other graphs in $\Psi^+$.
Hence 
M\"obius inversion yields the following lemma (compare with
Proposition~9.5 in
\cite{friedman_alon}).  

\begin{lem} For any collection, $\Psi$, of connected, pruned graphs
of order at least
one, there exist real numbers 
$\mu_\Omega$, indexed on elements, $\Omega$, of $B$-graph classes $\Psi^+$, 
such that for any $B$-graph, $G$, we have
\begin{equation}\label{eq:mobius_omega}
I_\Psi(G) = 
\sum_{\Omega} N(\Omega,G)\mu_\Omega
\end{equation}
and the above sum is finite
(and the sum is over one $\Omega$ in each isomorphism class of $B$-graphs).
\end{lem}
\begin{proof}
For each $\Omega\in\Psi^+$, we define $\mu_\Omega$ inductively with
respect to $\le_B$, via
$$
\mu_\Omega \  N(\Omega,\Omega)  
= 1 - \sum_{\Omega'<_B\Omega} N(\Omega',\Omega) \mu_{\Omega'},
$$
where $<_B$ means $\le_B$ and not isomorphic to; in other words, first we
define $\mu_\Omega$ for those minimal elements of $\Psi^+$; second
we define $\mu_\Omega$ for those
elements, $\psi\in\Psi^+$, for which $\psi'<_B \psi$ and $\psi'\in\Psi^+$
implies that $\psi'$ has already been defined, i.e., $\psi'$ is minimal;
etc.
More formally, each $\psi\in\Psi^+$ has a maximal length chain
of subgraphs
$$
\psi_1<_B \psi_2 <_B \cdots \psi_{m-1} <_B  \psi_m=\psi
$$
with $\psi_i\in\Psi^+$ and $m$ as large as possible 
(the above $m$ is finite since $\psi_i$ is a proper subgraph
of $\psi_{i+1}$, and on the
$m$-th ``step'' of this process we define $\mu_\psi$ for all $\psi$ whose
maximal length chain is $m$.

Let us now show that \eqref{eq:mobius_omega} holds for all $G\in\Psi^+$,
by induction on the number of edges of $G$.  We remark that if $G$ has
two or more edges, then $G$ is the union of all subgraphs of $G$ with
one less edge than $G$.

To begin the induction, we establish \eqref{eq:mobius_omega} if
$G$ has no edges or one edge.
In this case,
then all
of $G$'s connected components have order $-1$ or $0$ ($0$ can occur
if $G$ has a self-loop).  But we claim that if $\psi\in\Psi^+$, then
$N(\psi,G)$ is zero.
Indeed, if $N(\psi,G)>0$, then any connected component of $\psi$ must inject
into a component of $G$.  But since each vertex of $\Psi^+$ is
of degree at least two, then this component of $G$ must consist of
one vertex and one self-loop, the only edge in $G$.  But then this
connected component of $\psi$ must consist entirely of one vertex and
one self-loop, which contradicts the fact that each connected component
of $\Psi^+$ has positive order.
It follows that $N(\psi,G)=0$ for all $\psi\in\Psi^+$, and hence
\eqref{eq:mobius_omega} holds for any $G$ with zero or one edges.

Now we can prove \eqref{eq:mobius_omega} for any $G$, by induction
on the number of edges in $G$.
Assume that \eqref{eq:mobius_omega} holds whenever $G$ has at most
$i\ge 1$ edges.  If $G$ is isomorphic to an element of
$\Psi^+$, then \eqref{eq:mobius_omega} holds
by its definition.  Otherwise, no injection of an element of
$\Psi^+$ can reach all the edges of $G$.  Hence, if
$G_1,\ldots,G_m$ denote all graphs obtained by deleting one edge of
$G$ (and not deleting any vertices), then we have that
$I_\Psi(G)$ is one iff $I_\Psi(G_i)$ is one for some $i$.
Hence, by inclusion/exclusion, we have that
$$
I_\Psi(G) = \sum_{i=1}^m (-1)^{i+1} \sum_{j_1<\cdots<j_i}
I_\Psi(G_{j_1}\cap \cdots \cap G_{j_i})
$$
and also
$$
N(\Omega,G) = 
\sum_{i=1}^m (-1)^{i+1} \sum_{j_1<\cdots<j_i}
N(\Omega,G_{j_1}\cap \cdots \cap G_{j_i}),
$$
and now we use the inductive assumption, multiply the above equation
by $\mu_\Omega$ and sum over $\Omega$ to conclude
\eqref{eq:mobius_omega}.
\end{proof}

It follows that 
\[ \EE_{G\in\cC_n(B)}[I_{\Psi}(G) {\rm CertSNBC}_r(G,k) ] 
= \sum_{\Omega\in\Psi^+} \mu_\Omega 
\EE_{G\in\cC_n(B)}[{\rm CertSNBC}_r(G,k)\, N(\Omega,G) ] 
\]
\begin{equation}\label{eq:sumPsi1} = 
\sum_{(w,\vec t)\in {\rm CertSNBC}_r(G,k)} 
\ \ \   \sum_{(\Omega,\nonsigma),\ \Omega\in\Psi^+}
\EE_{\rm symm}[\potwalk,(\Omega,\nonsigma)]_n \mu_\Omega ,
\end{equation}
where the symmetric expectation is summed over one element
$(\vec t,\nonsigma)$ for each isomorphism class.
Similarly 
\[ \EE_{G\in\cC_n(B)}[I_{\Psi}(G)] \]
\begin{equation}\label{eq:sumPsi2} = \sum_{(\Omega,\nonsigma),\
\Omega\in\Psi^+} \EE_{\rm symm}[(\Omega,\nonsigma)]_n \; \mu_{\Omega} ,
\end{equation}
summed over one representative $\nonsigma$ in each
equivalence class (compare~(37) and~(38) of \cite{friedman_alon}).

\begin{lem}\label{le:indicator_less_than_r}
Let \( \Psi \) be a set of connected 
\( B \)-graphs, each pruned and of order at
least one. 
Then in \eqref{eq:sumPsi2}, we may replace \( \Psi^+ \) with \(
\Psi^+_{<r} \), with a difference of at most $C n^{-r}$
(in absolute value),
where $C$ depends on $\Psi$, $r$, and $B$.
\end{lem} 
\begin{proof}
We begin by establish the following claim, which we
will use repeatedly: if $\psi\in\Psi^+$ is
of order at least $r$, then we claim that
$\psi$ contains, as a $B$-subgraph,
an element of $\Psi^+$ of order at most $r+s-1$, where
$s$ is the maximum number of edges of an element of $\Psi$.
Indeed, for any $B$-graph, $\psi'$, and any $\psi\in\Psi$, we
have that
$$
\ord(\psi'\cup\psi)\le \ord(\psi')+s,
$$
since adding $\psi$ to $\psi'$ increases the number of edges in
$\psi'$ by at most $s$, and does not decrease the number of vertices,
and hence
$$
\ord(\psi'\cup\psi) = |E_{\psi'\cup\psi}| - |V_{\psi'\cup\psi}|
\le |E_{\psi'}|+ s - |V_{\psi'}| = \ord(\psi')+s.
$$
Since any $\psi\in\Psi^+$
can be written as a union
$\psi_1\cup\cdots\cup\psi_m$, then if $j$ is the smallest integer for which
the order of 
$\psi_1\cup\cdots\cup\psi_j$ is greater than $r-1$,
then this order can be no greater than $r-1+s$.

Let $\Psi^+_{\ge r}$ be those elements of $\Psi^+$ of order at least $r$,
and let $E$ be the set of $B$-graphs that contain at least one element
of $\Psi^+_{\ge r}$.  If $E'$ is the complement of $E$, i.e., those
graphs that do not contain any element of $\Psi^+_{\ge r}$, then we have
$$
I_\Psi(G) = \sum_{\Omega\in\Psi^+_{<r}}  N(\Omega,G) \mu_\Omega.
$$
In other words, for all $G$ we have
$$
I_{E'}(G)I_\Psi(G) = 
I_{E'}(G) \sum_{\Omega\in\Psi^+_{<r}}  N(\Omega,G) \mu_\Omega.
$$
and hence
\begin{equation}\label{eq:multiply_by_indicator_E'}
\expect{G\in\cC_n(B)}{I_{E'}(G)I_\Psi(G)} =
\sum_{\Omega\in\Psi^+_{<r}}
\expect{G\in\cC_n(B)}{I_{E'}(G)N(\Omega,G)}\mu_\Omega.
\end{equation}
To prove the lemma, i.e., that
\begin{equation}\label{eq:replace_with_less_than_r}
\expect{G\in\cC_n(B)}{I_\Psi(G)} =
\sum_{\Omega\in\Psi^+_{<r}}
\expect{G\in\cC_n(B)}{N(\Omega,G)}\mu_\Omega,
\end{equation}
it suffices to show that
$$
\expect{G\in\cC_n(B)}{I_{E'}(G)I_\Psi(G)} =
\expect{G\in\cC_n(B)}{I_\Psi(G)} + O(n^{-r})
$$
and that for every $\Omega\in\Psi^+_{<r}$ we have
$$
\expect{G\in\cC_n(B)}{I_{E'}(G)N(\Omega,G)} =
\expect{G\in\cC_n(B)}{N(\Omega,G)} + O(n^{-r}).
$$
Since $I_{E'}(G)+I_{E}(G)=1$, it suffices to show that
\begin{equation}\label{eq:unlikely_large_order}
\expect{G\in\cC_n(B)}{I_E(G) I_\Psi(G) } = O(n^{-r}) ,
\end{equation}
(of course $I_{E}(G)I_\Psi(G)$ is just $I_{E}(G)$)
and that for each $\Omega\in\Psi^+_{<r}$ we have
\begin{equation}\label{eq:unlikely_Omega}
\expect{G\in\cC_n(B)}{I_E(G) N(\Omega,G)} = O(n^{-r}) .
\end{equation}

First let us prove \eqref{eq:unlikely_large_order}.
If $G\in \Psi^+_{\ge r}$, then
$G$ contains an element of $\Psi^+$ of order
between $r$ and $r+s-1$.  But the set of elements of $\Psi^+$ of
order between $r$ and $r+s-1$ is a finite set, and the probability
of $G$ containing any particular element of this finite set is,
by Theorem~\ref{th:prob_occurs},
at most $n^{-r}+O(n^{-r-1})$. 
Hence
$$
\expect{G\in\cC_n(B)}{I_E(G) I_\Psi(G) } 
\le
\sum_{\substack{\psi\in\Psi^+ \\ r\le\ord(\psi)\le r+s-1}}
\bigl( n^{-r}+O(n^{-r-1}) \bigr)
$$
which establishes \eqref{eq:unlikely_large_order}.

Next let us similarly show \eqref{eq:unlikely_Omega}.
If $G\in E$, then $G$ has subgraphs $\psi_1,\ldots,\psi_m$,
each isomorphic to an element of $\Psi$, for which
$$
\psi=\psi_1\cup\cdots\cup \psi_m 
$$
is of order at least $r$.
Fix an $\Omega\in\Psi^+_{<r}$. Consider any inclusion of $\Omega$
in $G$, and let $G'$ be its image.
Then by successively adding to $G'$ the graphs $\psi_i$, for $i=1,\ldots,m$,
at some point we obtain a graph $\Omega'\in\Psi^+$,
$$
\Omega'=\Omega\cup\psi_1\cup\cdots\cup\psi_j .
$$
for some $j$, which is of order between $r$ and $r+s-1$.
It follows that
\begin{equation}\label{eq:N_iterated}
N(\Omega,G) \le \sum_{\Omega'} N(\Omega,\Omega') N(\Omega',G),
\end{equation}
with the sum ranging over the finite set of elements, $\Omega'$,
of $\Psi^+$ 
of order between $r$ and $r+s$.  Taking expectations yields
$$
\expect{G\in\cC_n(B)}{N(\Omega,G)} \le 
\sum_{\Omega'} N(\Omega,\Omega') \expect{G\in\cC_n(B)}{N(\Omega',G)},
$$
and applying Theorem~\ref{th:prob_occurs} shows that the right-hand-side
is at most $O(n^{-r})$.
This establishes \eqref{eq:unlikely_Omega} and completes the proof.
\end{proof}

\begin{lem}\label{le:truncate_Psi_at_r}
Let \( \Psi \) be a set of connected 
\( B \)-graphs, each pruned and of order at
least one. 
Then in \eqref{eq:sumPsi1}, we may replace \( \Psi^+ \) with \(
\Psi^+_{<r} \), with a difference of at most \( C\mu_1(B)^k n^{-r}k^{2r+2} \)
(in absolute value),
where $C$ depends on $\Psi$, $r$, and $B$.
\end{lem} 
\begin{proof}
We use the same ideas and notation as in the proof
of Lemma~\ref{le:indicator_less_than_r}.  With $E$ as in the proof there,
it suffices to show that
\begin{equation}\label{eq:unlikely_large_order_cert}
\expect{G\in\cC_n(B)}{I_E(G) {\rm CertSNBC}_r(G,k) } = O(\rho(H_B)^k) n^{-r} ,
\end{equation}
and that for each $\Omega\in\Psi^+_{<r}$ we have
\begin{equation}\label{eq:unlikely_large_cert}
\expect{G\in\cC_n(B)}{I_E(G) {\rm CertSNBC}_r(G,k)\, N(\Omega,G) }
= O(n^{-r}).
\end{equation}

Equation~\ref{eq:unlikely_large_order} shows that
that the probability that $G\in\cC_n(B)$ lies in $E$ is
$O(n^{-r})$.  However, 
${\rm CertSNBC}_r(G,k)$ is never more than $\Tr(H_B^k)$, which
is at most a constant times $\rho(H_B)^k$.  This establishes
\eqref{eq:unlikely_large_order_cert}.

For each $\Omega\in\Psi^+_{<r}$, in the proof of
Lemma~\ref{le:indicator_less_than_r} we have seen that
\eqref{eq:N_iterated}
implies that the expected value of $I_E(G)N(\Omega,G)$ is $O(n^{-r})$.
Since ${\rm CertSNBC}_r(G,k)$ is never more than $\Tr(H_B^k)$,
we have
$$
\expect{G\in\cC_n(B)}{I_E(G) {\rm CertSNBC}_r(G,k)\, N(\Omega,G) }
\le
\Tr(H_B^k) 
\expect{G\in\cC_n(B)}{I_E(G) \, N(\Omega,G) }
$$
$$
\le
\rho(H_B)^k\, O(n^{-r}).
$$
\end{proof}

Hence, to prove Theorems~\ref{th:with_without} and \ref{th:no_walks}, it
suffices to prove an asymptotic expansion, for any fixed \(
\Omega\in\Psi^+_{<r} \) for \[ \sum_{(w,\vec t)\in {\rm CertSNBC}_r}
\sum_\nonsigma \EE_{G\in\cC_n(B)}[\potwalk,(\Omega,\nonsigma)] \] and the same without the
\( (w,\vec t) \).

\subsection{Proof of Theorem~\ref{th:no_walks}.}

\begin{proof}[Proof of Theorem~\ref{th:no_walks}.]
First we will show that for each $B$-graph
$\Omega\in\Psi^+$, any two (feasible)
potential specializations, $\nonsigma$ and $\nonsigma'$,
which are maps $V_\Omega\to\{1,\ldots,n\}$, yield 
equivalent events $(\Omega,\nonsigma)$ and $(\Omega,\nonsigma')$.
Indeed, these events give a $B$-isomorphism of $\Omega$ with some
subgraph of some element of $\cC_n(B)$, and hence the
$B$-graph that $(\Omega,\nonsigma)$ and
$(\Omega,\nonsigma')$ describe are isomorphic as $B$-graphs.
Hence the sum over all feasible classes $[\nonsigma]$ in
$$
\sum_{[(\Omega,\nonsigma)]} \EE_{\rm symm}[(\Omega,\nonsigma)]_n
$$
consists of precisely one class $[\nonsigma]$.  Furthermore, this
class has a $1/n$-asymptotic expansion given by
\eqref{eq:gen_walk_probability}
(where $a_e$ counts the number of edges over $e$ in $\Omega$, 
and similarly for $b_v$).

By Lemma~\ref{le:indicator_less_than_r} we have
$$
\expect{G\in\cC_n(B)}{I_\Psi(G)} =
O(n^{-r}) + \sum_{\Omega\in\Psi^+} \sum_{[(\Omega,\nonsigma)]} 
\expect{G\in\cC_n(B)}{N(\Omega,G)} \; \mu_\Omega.
$$
But $\Psi^+_{<r}$ is a finite set, and by the preceeding paragraph
each term in the above equation with a fixed $\Omega$ has a
$1/n$-asymptotic expansion.  Hence we conclude that the expected
value of $I_\Psi(G)$ has a $1/n$-asymptotic expansion to order 
$r$ for any integer $r>0$.
\end{proof}

\subsection{The Proof of Theorem~\ref{th:with_without}}

This subsection is devoted to proving Theorem~\ref{th:with_without}.
Like the proof of Theorem~\ref{th:no_walks}, the point is that we
may fix an $\Omega\in\Psi^+_{<r}$, of which there are only finitely
many. 
We then need to consider a joint event of a potential walk, $(w,\vec t)$,
and a potential graph specialization,
$(\Omega,\nonsigma)$.  We essentially have to show that all the
methods of Sections~\ref{se:p1-walk-sums} and~\ref{se:p1-tangles}
carry over to this more general case.
Then we will appeal to Lemma~\ref{le:truncate_Psi_at_r} to complete
the proof.

This was done in \cite{friedman_alon}, where it was
a rather long---albeit straightforward---process.
In this article we shall simplify this process by
defining the generalized types and forms as
automatically
incorporating the data of $\Omega$ into them.
This means that our definitions will be more complicated than the
corresponding Definitions~9.7 and 9.8 in \cite{friedman_alon},
but then the verification Theorem~\ref{th:with_without} is 
(a bit) simpler.

Let us give the rough idea.  If $(w,\vec t)$ is a potential walk,
and $(\Omega,\nonsigma)$ a potential specialization of $\Omega$,
an {\em $\Omega$-type} should remember 
all the vertices and edges of $\Omega$, as well as all
beaded paths left upon deleting any other vertices that are not the
starting vertex and not a vertex of degree three or greater.
Aside from this, we should be able to reconstruct all the type data,
i.e., the vertices and edges traversed by $w$, in their order, and
the labelling, or local data around all the type vertices.

\begin{proof}[Proof of Theorem~\ref{th:with_without}]

\begin{definition}
Let $\Omega$ be a $B$-graph for a graph, $B$.
By an {\em \gls[format=hyperit]{omega-type}}, we mean the data
$\cal T=(T_1,T,\gamma_V,\gamma_E,\cal L,\nu_\Omega)$, such that
\begin{enumerate}
\item $T_1$ is a graph, and $\nu_\Omega$ is an injection of $\Omega$
into $T_1$;
\item $T$ is a subgraph of $T_1$, and $\gamma_V,\gamma_E$ are an ordering
of the vertices and edges of $T$, along with an orientation for each
edge of $T$;
\item $T_1$ is the union of $T$ and the image of $\Omega$;
\item $\cal L$ is a labelling of all the vertices in $T_1$, which respects
the $B$-vertices the $B$-edges in the structure of $\Omega$ under
$\nu_\Omega$; and
\item the vertices of $T$ consist of the union of the following:
(1) all vertices in the image of $\nu_\Omega$;
(2) all vertices of degree at least three; and
(3) the first vertex, in the order given by $\gamma_V$.
\end{enumerate}
We say that an edge of $T$ is {\em variable length} if it is not
shared with $\Omega$, i.e., not in the image of $\nu_\Omega$; 
otherwise we say that the edge is
{\em unit length}.
\end{definition}

\begin{definition}
Let $\cal T=(T_1,T,\gamma_V,\gamma_E,\cal L,\nu_\Omega)$ be an
$\Omega$-type.
By an {\em \gls[format=hyperit]{omega-form} of type $\cT$}
we mean a $B$-graph, $\Gamma$, obtained from $\cT$ by taking
every variable length edge of $T$ and replacing it by a beaded path,
with a map to $B$, which is consistent with the lettering, $\cal L$.
\end{definition}

\begin{definition} 
By a {\em specialization (in $\cC_n(B))$} 
of a $\Omega$-form, $\Gamma$, we mean any
map $\vec t\from V_\Gamma\to\{1,\ldots,n\}$.
We say that two specializations, $\vec t,\vec t'$
are {\em equivalent}, denoted $\vec t\sim\vec t'$, if their values
differ by a permutation of $\{1,\ldots,n\}$.
To any pair of a
potential walk, $(w,\vec t)$ and an $\Omega$-specialization in $\cC_n(B)$,
$(\Omega,\nonsigma)$, 
we associate a (unique) $\Omega$-type, $\cal T=(T_1,T,\gamma_V,\gamma_E,
\cL,nu_\Omega)$, 
with $T_1$ being the union of $\Graph(w)$ and the image of $\Omega$
under $\nonsigma$, and $T$ being the edges and vertices of $T_1$ in
the image of $\Graph(w)$, and $\gamma_V,\gamma_E,\cL$ all arising
from the order and orientation in which $\Graph(w)$ edges are traversed.
\end{definition}

Now we wish to describe how to modify the methods of 
Sections~\ref{se:p1-walk-sums} and \ref{se:p1-tangles}
to prove Theorem~\ref{th:with_without}.
The results in Section~\ref{se:p1-walk-sums}
will hold for any walk collection, 
$\cW=\{ \cW(k,n) \}_{k,n\ge 1}$ and any
collection of graphs, $\Psi$, of pruned, connected graphs of
order at least one.  (It is simpler to assume that the elements of $\Psi$
are connected, but one can also prove this if the elements of $\Psi$
are not connected, provided that each of their connected components is
of order at least one.)
So fix any walk collection, $\cW$, and graph collection, $\Psi$, as 
such.

Fix an $\Omega\in\Psi^+_{<r}$.
We first consider the results in Sections~\ref{se:p1-walk-sums}, and
describe how they apply to
$\Omega$-forms and $\Omega$-types:
\begin{enumerate}
\item Proposition~\ref{prop:EsymmProd} clearly holds for $\Omega$-forms, with
$a_e$ and $b_v$ counting all elements over $B$ edges and vertices in
the $\Omega$-form.
\item Theorem~\ref{thm:Esymm_expansion} holds in this setting, i.e.,
for any potential walk, $\potwalk$ and potential specification
$(\Omega,\nonsigma)$, we have
\begin{equation}
\EE_{\rm{symm}}[\potwalk(\Omega,\nonsigma)]_n =
n^{-\ord(\cF)} \left( p_0 + \frac{p_1}{n} + \ldots +
\frac{p_{r-1}}{n^{r-1}} + \frac{\error(k,n)}{n^r} \right)  ,
\end{equation}
where $\cF$ is the $\Omega$-form corresponding to
$\potwalk$ and $(\Omega,\nonsigma)$,
and the $p_i$ are polynomials of $a_e=a_e(\cF)$
and $b_v=b_v(\cF)$, or just of the $a_e$ alone, with
$\error(k,n)$ satisfying
\eqref{eq:Esymm_expansion_error}
\item We now prove Theorem~\ref{thm:walk_sum_expansion} with the
following modification: we consider, for each
equivalence class, $[(w;\vec t\,),(\Omega,\nonsigma)]$ the union of
the graph of $(w;\vec t\,)$ union the image of $\Omega$, and we
first determine the vertices of the graph over $\Omega$,
i.e., we first determine the random variables $\{i'_v\}_{v\in V_\Omega}$
given by $\nonsigma\from V_\Omega\to \{\pi(v),i'_v\}$, where
$\pi\from\Omega\to V_B$ is given by the $B$-graph structure of $\Omega$.
After determining all of the $i'_v$ (in some fixed order imposed on
$V_\Omega$) we then determine the variables $i_0,i_1,\ldots,i_k$---in
the notation of the proof of Lemma~\ref{le:coincidences}
or Theorem~\ref{thm:walk_sum_expansion}---where $i_0$ may already be
fixed by the $\{i'_v\}_{v\in V_\Omega}$ or, otherwise, may take
any value from $\{1,\ldots,n\}$ that is not assumed previously
as $(\pi(v),i'_v)$ with $v\in V_\Omega$, i.e., where $\pi(v)=v_0$
and $t(v_0)=i'_v$.
When determining the random variables
$\{i'_v\}$ with $v\in V_\Omega$, and then
$i_0,\ldots,i_k$, throughout we regard each of these random variables
($i_0$ is not really a random variable, but a variable whose value
is either fixed or varying among at least $n-|V_\Omega|$ values) 
using the same notion of coincidences (and ``fixed choices'' and
``free, non-coincidence choices,'' as in the proof of 
Theorem~\ref{thm:walk_sum_expansion}); before considering the
$i_0,\ldots,i_k$, if the $i'_v$ are determined in any fixed order,
then their choice falls again into the categories
(1)--(3) of the proof of Theorem~\ref{thm:walk_sum_expansion}
(there is nothing about the classifications (1)--(3) that requires
successive $i'_v$ values to come from a (strictly non-backtracking closed)
walk in $G$).
Determining the $i'_v$ and $i_j$ involves a choice of 
$k+|V_\Omega|$ random variables, and hence we get the same
expansion theorem as in Theorem~\ref{thm:walk_sum_expansion},
except that the $k$ needs to be replaced with $k+|V_\Omega|$, which,
since $\Omega$ is fixed, means that we get the same bound,
except that the constant in the error bound \eqref{eq:error_bound_wse}
depends on $B$, $r$, and $\Omega$.
\item We then deduce
Theorem~\ref{thm:type_form_expansion} with the following minor changes:
(1) all $Q_{\cT,i}(k)$ and $P_i(\vec k)$ and constant $c=c(r,B)$
now become dependent also on $\Omega$;
(2) the $\vec k$ and $\vec m$ now lie in $\integers_{\ge 1}^{{\Var_T}}$,
where $\Var_T$ are the variable length edges of $T$ in $\cT$.
Of course, a sequence of edges in $E_T$ that determine a beaded
path in $T$ must have the same $\vec m$ values, but this is not
important (this information is
incorporated into the function $W_\cT(\vec m)$).
\end{enumerate}

Next we give the modifications necessary to obtain the results in
Section~\ref{se:p1-tangles}.
\begin{enumerate}
\item We view $\vec m,\vec k$ indexed over $E_T$, although for the
unit length edges, $e\in E_T$, $k(e)$ is always $1$.
\item A tangle refers to Hashimoto eigenvalue bounds on
$\VLG(T,\vec k)$.
\item The same finiteness of minimal tangle elements and
minimal certificates for the certified trace hold; both involve the
partially ordered set $\integers_{\ge 1}^{\Var_T}$, where $\Var_T$
is the set of variable length edges of $T$, as determined by $\cT$,
both arising from the variable $\vec k'$ which is the restriction of
$\vec k$ to $\Var_T$.  However, it becomes simpler to refer to
certificates and minimal tangles 
over $\integers_{\ge 1}^{E_T}$, by adding in the unit
length edges of $E_T$
\item The polyexponential functions are functions of the variables
$\vec k'$, with the other $\vec k$ values fixed at the value one.
\item Since $\vec k,\vec m$ take values in $E_T$, we still have
$$
\vec m\cdot\vec k = k .
$$
\item Theorem~\ref{thm:certified_walk_expression} still holds,
with $\vec k_0\in\integers_{\ge 1}^{E_T}$.
\item We still divide all edges of $E_T$ (including the unit length
edges) into the ``short'' and ``long'' edges; the crucial
Lemma~\ref{le:crucial_walk_bound} goes through verbatim,
and hence we conclude Theorem~\ref{thm:certified_B_Ramanujan} by
the same division of $E_T$ into ``short'' and ``long'' edges.
\end{enumerate}

We conclude 
Theorem~\ref{th:with_without}.

\end{proof}

\subsection{Concluding the Proofs of 
Theorems~\ref{th:certified_trace_expansion_with_tangles}
and \ref{th:main_expansion_B}}

\begin{proof}[Proof of 
Theorem~\ref{th:certified_trace_expansion_with_tangles}]
Apply Theorem~\ref{th:with_without} with
$\Psi$ taken to be $\Tanglemin$, which we know is finite
(by Theorem~\ref{thm:tangles_finite}) and consists of connected graphs
of order at least one.
\end{proof}

\begin{proof}[Proof of 
Theorem~\ref{th:main_expansion_B}]
If $G$ has no $B$-tangles of order $r$, then
$$
\CertTr_{<r}(G,k) = \SNBC_{<r}(G,k) ,
$$
with $\SNBC$ the walk collection of Example~\ref{ex:SNBC} (of all
strictly non-backtracking closed walks) and
$\SNBC_{<r}$ is its truncated form
(of Definition~\ref{de:truncate_walk_collections}).
Hence
$$
\expect{G \in \cC_n(B)}{\II_{{\rm TF}(r,B)}(G)\CertTr_{<r}(G,k)} =
\expect{G \in \cC_n(B)}{\II_{{\rm TF}(r,B)}(G)\SNBC_{<r}(G,k)} .
$$
It follows from Theorem~\ref{thm:walk_sum_expansion} it follows
that
$$
\expect{G \in \cC_n(B)}{\II_{{\rm TF}(r,B)}(G)\SNBC_{<r}(G,k)} 
$$
$$
= \expect{G \in \cC_n(B)}{\II_{{\rm TF}(r,B)}(G)\SNBC(G,k)} +
O(k^{2r+2} \Tr(H_B^k) n^{-r}).
$$
But $\SNBC(G,k)$ is just $\Tr(H_G^k)$.
\end{proof}

\section{The Side-Stepping Lemma}
\label{se:p1-side-step}

In this section, we prove a new version of the side-stepping lemma in
\cite{friedman_alon}. We improve this lemma to the more general case where
we consider polyexponential functions with bases being all the Hashimoto
eigenvalues strictly greater than \( \rhoroot B \), which will allow us to
use the $1/n$-asymptotic expansions with
\(B\)-Ramanujan coefficients have to conclude
Theorem~\ref{th:main_Alon}.

The key idea is to use the fact that we can have an expansion for several
consecutive values of \(k\) and apply the shift operator, described in the
next subsection, several times. This side-stepping lemma is stated for an
abstract collection of finite probability spaces.
%
%
%
%
\subsection{The Shift Operator}

Let \(S\) denote the \emph{\gls[format=hyperit]{shift operator in k}}, 
meaning the operator
on functions, \(f(k)\), defined on non-negative integers, \( k \), taking
\(f\) to \[ (Sf)(n,k) = f(n,k+1) \] For any integer \( i \geq 0 \) we let
\( S^i \) denote the \(i\)-fold application of \(S\), so that \[ (S^i
f)(n,k) = f(n,k+i) \] For any polynomial (with real or complex
coefficients), \[ Q(z) = q_0 + q_1 z + \cdots + q_t z^t \] we define
\(Q(S)\) in the natural way, i.e., \[ Q(S) = q_0 S^0 + q_1 S^1 + \cdots +
q_t S^t \]

The utility of polynomials in \(S\) is due to the following simple
observations.

\begin{prop} 
\label{pr:basic-facts-shift}
Let \( Q(z),Q_1(z) \) be a polynomials, with real
or complex coefficients.  Then we have
\begin{enumerate}[label= (\arabic*)] 
\item \( (Q(S))(\mu^k) = Q(\mu) \mu^k
\) for any \( \mu \in \CC \); 
\item $Q(S)Q_1(S)=Q_1(S)Q(S)$ as operators; and
\item \(
(S-\mu)^D(f) \equiv 0 \) for any function \( f(k) = \mu^k p(k) \) where \(
\mu \in \CC \) and \( p \) is a polynomial of degree at most \( D-1 \).
\end{enumerate} \end{prop} \begin{proof} (1) and (2) are straightforward
and (3) is a consequence of the following identity, \[ \sum_{k=0}^D
\binom{D}{k} (-1)^k k^d = 0 \] for any \( N \geq 0 \) and any \( d < D \).
\end{proof}

When \(h(k)\) is a function for which we do not have a simple expression,
but is rather bounded by a polyexponential function (such as in our "error
terms"), we employ the following crude bound.  \begin{lem} \label{le:crude}
Let \(h(k)\) be a function for which \[ |h(k)| \leq C \tau^k k^C, \] for
some positive real numbers \(C\) and \(\tau\) with \(\tau\geq 1\). Let
\(Q(x)\) be any complex polynomial of degree at most \(D-1\).  Then, for \(
k \geq 0 \), we have \[ |Q(S)h(k)| \leq C D \|Q\| \tau^{k+D}(k+D)^C \]
Where \(\|Q\|\) denotes the largest  absolute value of \(Q\)'s
coefficients. In particular, for fixed \(Q\) and \(\tau\) we have for all
\( k \geq 0 \) \begin{equation}\label{eq:Qh_bound} |Q(S)h(k)| \leq C'
\tau^k k^{C'} \end{equation} for some constant \( C' \).  \end{lem}
\begin{proof} Denote \[ Q(x) = q_0 + q_1 x + \cdots + q_{D-1} x^{D-1} \]
and for any non-negative integer \(i\) we clearly have \[ |S^i h(k) | \leq
C\tau^{k+i} (k+i)^C \leq C \tau^{k+D} (k+D)^C \] Hence \begin{align*} |Q(S)
h(k) | &\leq \sum_{i=0}^{D-1} | q_i | \; |S^i h(K) | \\ &\leq D\|Q\| C
\tau^{k+D} (k+D)^C \end{align*} Clearly the last equation of the lemma,
\eqref{eq:Qh_bound}, follows.  \end{proof}
%
%
%
%
\subsection{Statement of the Side-Stepping Lemma}

The main lemma of this subsection, Lemma~\ref{le:side-step}, generalizes the
Side-Step\-ping Lemma of \cite{friedman_alon}, both of which can be viewed
as abstract lemmas in probability theory.  Of course, these abstract lemmas
are motived by our trace methods, and we will need Lemma~\ref{le:side-step}
to prove the Hashimoto version of the generalized Alon conjecture for all
base graphs, \( B \), which are \( d \)-regular.  Also,
Lemma~\ref{le:side-step}, especially \eqref{eq:pellj_limit}, is, arguably,
more direct and simpler than the Side-Stepping Lemma in
\cite{friedman_alon}.

This section involves a number of positive, real constants, depending on
various parameters, that need to be chosen sufficiently large; we will
generally use the letters \( C,C',C'' \), for various of these constants,
when no confusion is likely to arise.

\begin{defn} By a {\em finite probability space} we mean a pair, \(
(\Omega,\nu) \), where \( \Omega \) is a finite set, and a probability
measure, \( \nu \), which we view as a function \( \nu\from\Omega\to\RR \),
such that \( \nu(\omega)\geq 0 \) for each \( \omega\in\Omega \), and \[
\sum_{\omega\in\Omega} \nu(\omega)=1.  \] We will view \( \nu \) as giving
a measure to each subset \( E\subset\Omega \) in the usual way, i.e., \[
\nu(E)=\sum_{\omega\in E} \nu(\omega).  \] By a {\em complex random
variable} on \( (\Omega,\nu) \) we mean a function \( \mu\from \Omega\to
\CC \); similarly we define a {\em real random variable}.  \end{defn} We
often write \[ \Prob_{\Omega}[E] \text{ for } \nu(E) \] and, for a random
variable, \( \mu=\mu(\omega) \), \[ \EE_{\omega \in \Omega}[\mu(\omega)]
\text{ for } \sum_{\omega \in \Omega}\nu(\omega) \mu(\omega) \] As such,
the measure \( \nu \) will often be omitted in notation, being implicit in
the notation \( \Prob_\Omega[E] \) and \(
\expect{\omega\in\Omega}{\mu(\omega)} \).

Now we come to the main definition which describes the hypothesis to
our side-stepping lemmas as an {\em abstract partial trace} setup,
designed to be applied to Theorem~\ref{thm:certified_B_Ramanujan},
at least for regular $B$.

\begin{defn} \label{de:abstract} Let \( \tau_0,\tau_1,\gamma \) be positive
real numbers for which \( \tau_0<\tau_1 \); let \( L \) be a finite
collection of real numbers, each of whose absolute value is strictly larger
than \( \tau_0 \) and at most \( \tau_1 \); let $r$ be a positive
integer.  By an {\em \gls[format=hyperit]{abstract partial trace}
with parameters \( (\tau_0,\tau_1,\gamma, L,r) \)} we mean the
following data for each positive integer, \( n \): \begin{enumerate} \item
a finite probability space, \( (\Omega_n,\nu_n) \); \item complex random
variables \( \mu_i=\mu_i(\omega) \) on \( (\Omega_n,\nu_n) \) such that for
all \( i=1,\ldots,\gamma n \) and \( \omega\in\Omega_n \) we have \[
\mu_i(\omega) \in R_1\cup R_2, \] where \[ R_1 = \{ z\in\CC\ |\ |z|\leq
\tau_0 \} ,\quad R_2 = \{ z\in\RR\ |\ \tau_0 < |z|\leq \tau_1 \} ; \] \item
there are polyexponentials \( P_0(k),P_1(k),\ldots \) with bases in \( L
\), for which we have
\begin{equation}\label{eq:side_series} \sum_{\omega\in E_n} \nu_n(\omega)
\sum_{i=1}^{\gamma n} \mu_i^k(\omega) = P_0(k) + P_1(k)n^{-1} + \cdots +
P_{r-1}(k) n^{1-r} + {\rm err}_r(n,k) \end{equation} for all positive
integers \( k,n \) with \( 1\le k\le \gamma n^\gamma \), 
where \( {\rm err}_r(n,k) \) is a
function such that for any \( \varepsilon>0 \) there is a constant, \( C \)
(depending on \( \varepsilon \) and \( r \)) for which \[ |{\rm
err}_r(n,k)| \leq C k^C \Bigl( \tau_1^k n^{-r} + (\tau_0+\varepsilon)^k
\Bigr).  \] \end{enumerate} \end{defn}

\begin{ex} We will apply the above setting to \( \Omega_n=\cC_n(B) \) with
the following parameters when \( B \) is \( d \)-regular: \(
\tau_0=\rhoroot B \), \( \tau_1=\rho(H_B \), \( \gamma=|\Edir_B| \); the choice
of \( \mu_i(\omega) \) is a bit intricate: we will fix an \( \varepsilon'>0
\) and take the \( \mu_i(\omega) \) be either (1) all zero, if \( G \)
contains a \( (B,\varepsilon') \)-tangle, or (2) the Hashimoto eigenvalues
of \( G \), except that the old Hashimoto eigenvalues will be set to zero.
Note that the \( \varepsilon' \) in this paragraph is different from the \(
\varepsilon \) in Lemma~\ref{le:side-step}.  \end{ex}

We remark that above setting essentially permits the random variables \(
\mu_i(\omega) \), to range over \( i=1,\ldots,f(n) \) for any function, \(
f(n) \), bounded by \( \gamma n \) for some fixed \( \gamma \); this can be
done simply by introducing new ``dummy'' random variables, \(
\mu_i(\omega)=0 \) for \( i=f(n)+1,\ldots,\gamma n \).

We shall prove a lemma which will apply \eqref{eq:side_series} for numerous
values of \( k \) to show that it is ``unlikely'' for \( \omega\in\Omega_n
\) that some \( \mu_i(\omega) \) is larger than \( \tau_0+\varepsilon \) in
absolute value but not near any \( \ell\in L \).  We shall also obtain
important information on the ``dominant'' part of the \( P_i(k) \)
corresponding to the \( \ell^k \) term for each \( \ell\in L \).  To
describe this, we need some further terminology.

\begin{defn} 
\label{de:partial_trace_sets}
Consider an abstract partial trace as defined above with
parameters \( (\tau_0,\tau_1,C,L,r) \) and with \( \Omega_n,\nu_n,\mu_i \) as
above.  For any \( \varepsilon,\delta>0 \) and integer \( n \), set \[ {\rm
Out}(\delta,\varepsilon) =  \{ z\in\RR\ |\ \mbox{\(
\tau_0+\varepsilon<|z|\leq \tau_1 \) and \( |z-\ell| > \delta \) for all \(
\ell\in L \)} \} \] and \[ {\rm Exception}_n(\delta,\varepsilon) = \{
\omega\in \Omega_n\ | \ \mbox{\( \mu_i(\omega)\in {\rm
Out}_n(\delta,\varepsilon) \) for some \( i \)} \}  \] and \[ {\rm
AbsoluteException}_n(\varepsilon) = \{ \omega\in \Omega_n\ | \ \mbox{\(
|\mu_i(\omega)|>\tau_0+\varepsilon \) for some \( i \)} \}   .  \] Also set
\[ {\rm Near}_n(\ell,\delta) = \expect{\omega\in\Omega_n}{\ \Bigl|\bigl\{ i
\ \bigm|\  |\mu_i(G)-\ell|\leq \delta \ \bigr\}\Bigr|\ } , \] i.e., the
expected number of \( \mu_i(\omega) \) which lie within \( \delta \) of \(
\ell \).  \end{defn}

\begin{lem}[Improved Side-Stepping Lemma] \label{le:side-step}
\label{LE:SIDE-STEP}  
Consider an abstract partial trace with
parameters \( (\tau_0,\tau_1,C,L,r) \) as above, with \(
\Omega_n,\nu_n,\mu_i,P_i \) as above.  For any \( \alpha,\varepsilon>0 \)
there exists a \( \theta>0 \) for which
\begin{equation}\label{eq:first_part_of_side}
\Prob_n[ {\rm Exception}_n (n^{-\theta},\varepsilon)] \leq n^{-\alpha}
\end{equation} 
for \( n \) sufficiently large,
provided that
$$
r\ge F(\alpha,\varepsilon,\tau_1,\tau_0),
$$
for some function $F$.
Furthermore, write \[ P_i(k) = \sum_{\ell\in L} p_{\ell,i}(k) \ell^k, \]
i.e., let \( p_{\ell,i}(k) \) be the polynomial coefficient of \( \ell^k \)
in \( P_i(k) \).  Assume that for some \( j \) and \( \ell\in L \) we have
that \[ p_{\ell,0}(k)=\cdots=p_{\ell,j-1}(k)=0.  \] Then \( p_{\ell,j}(k)
\) is a constant (i.e., independent of \( k \)), and this constant is the
limit \begin{equation}\label{eq:pellj_limit} p_{\ell,j}= \lim_{n\to\infty}
n^j \;  {\rm Near}_n\bigl(\ell,n^{-\theta} \bigr)  \end{equation} for any
\( \theta>0 \) sufficiently small (i.e., there is a \( \theta_0>0 \) such
that \eqref{eq:pellj_limit} holds for any \( \theta \) with \(
0<\theta<\theta_0 \)), provided that
\begin{equation}\label{eq:H_defined}
r > H(j,L,\tau_1,\tau_0,\varepsilon)
\end{equation}
for some function, $H$.
Finally, if \( j \) is any positive integer for which \[ P_0(k) = \cdots =
P_{j-1}(k) = 0, \] then for each \( \varepsilon>0 \) there is a \(
C=C_\varepsilon \) for which \begin{equation}\label{eq:final_side} \Prob_n[
{\rm AbsoluteException}_n(\varepsilon) ] \leq C_\varepsilon n^{-j} .
\end{equation} \end{lem}

It is an important aspect of the above lemma that it hold for all $r$
sufficiently large, and that we explicitly describe the parameters
involved in stating how large $r$ has to be.  The point is that
Theorem~\ref{th:with_without} holds for any fixed $r$, but for each value of
$r$ it produces $B$-Ramanujan coefficients whose principle part
involves various polynomial in $k$.
We have to make sure that nothing about these polynomials
(other than, say, $P_0,\ldots,P_j$ for some fixed $j$) affects the
condition on $r$.  This will become clear in Section~\ref{se:p1-main-proof},
when we prove the generalized Alon Conjecture for regular graphs, $B$.

Lemma~\ref{le:side-step} will be proven in the subsections that follow.  The
proofs are based on the simple idea of applying certain polynomials of the
``shift operator in \( k \)'' to \eqref{eq:side_series}.  The basic idea
behind these proofs is very simple: we let \( S \) denote the ``shift
operator in \( k \),'' taking a function, \( f(n,k) \), to the shifted
function \( (Sf)(n,k)=f(n,k+1) \); we begin our proof by establishing
\eqref{eq:first_part_of_side}, by applying \[ Q(S) = \prod_{\ell\in L}
(S-\ell)^D \] to both sides of \eqref{eq:side_series}, where \( D \) is an
even integer that is sufficiently large to annihilate the \( P_i(k) \) for
\( i=0,\ldots,r-1 \) (this is explained in the next subsection, but is
essentially Item~(3) of
Proposition~\ref{pr:basic-facts-shift}); we then prove
the second part, regarding the \( p_{\ell,i}(k) \), by applying, for each
\( \ell\in L \), \[ Q_{\ell,D}(S) = \prod_{\ell'\in L,\ \ell'\ne\ell}
(S-\ell')^D \] to both sides of \eqref{eq:side_series}, to isolate the
effect of the \( \mu_i \) near \( \ell \).  

A very similar set of ideas appears in \cite{friedman_alon}, Section~11,
where the side-stepping lemma there is proven only in the case \(
L=\{\tau_1\} \), which suffices for random \( d \)-regular graphs, i.e.,
base graph \( B=W_{d/2} \) or \( B \) being one vertex and \( d \)
half-loops.  The proof of Lemma~\ref{le:side-step} is primarily complicated
by the fact that \( L \) can have many elements, unlike the case in
\cite{friedman_alon} where $L$ consists entirely of $d-1$.

%
%
\subsection{Proof of the First Exception Bound}

In this subsection we will prove \eqref{eq:first_part_of_side}.  In other
words, fix an \( \varepsilon>0 \); we will show that for any \( \alpha>0
\), there exists a \( \theta=\theta(\alpha)>0 \) for which
\eqref{eq:first_part_of_side} holds. First, let us describe \(
\theta=\theta(\alpha) \) explicitly, in rather unmotivated terms.

Let \( \phi \) be any real number with \( \phi>2\alpha \).  Let \(
\phi',\phi'' \) be any real numbers satisfying
\begin{equation}\label{eq:givephi1} \phi' >
\frac{\phi+1}{\log\frac{\tau_0+\varepsilon}{\tau_0+(\varepsilon/2)}}
\end{equation} and \begin{equation}\label{eq:givephi2} \phi'' >
\frac{\phi+2}{\log\frac{\tau_0+\varepsilon}{\tau_0}}.  \end{equation} Let
\( \rho>0 \) be any real number for which
\begin{equation}\label{eq:giverho} \rho < (1/2) \log(\tau_1/\tau_0).
\end{equation} Let \( r \) be any positive integer for which
\begin{equation}\label{eq:giver} r/2 > \max(\phi,\phi'/\rho,\phi''/\rho) ,
\end{equation} 
since $\phi,\phi',\phi'',\rho$ depend only on 
$\tau_0,\tau_1,\varepsilon,\alpha$, our stipulation on $r$ can be viewed
as a condition of the form
$$
r\ge F(\alpha,\varepsilon,\tau_1,\tau_0)
$$
for some function, $F$.
Let \( P_0,\ldots,P_{r-1} \) be as in
Definition~\ref{de:abstract}, and \( D \) be an even integer bounding the
polynomial degree of the \( p_{\ell,i} \) ranging over all \( \ell\in L \)
and \( i=0,\ldots,r-1 \); we take \( \theta \) to be any positive real
number for which \begin{equation}\label{eq:givetheta} \theta < \alpha /
(|L|\,D); \end{equation} hence \( \theta \), and \(
D,r,\rho,\phi,\phi',\phi'' \), depend only on \(
\alpha,\varepsilon,\tau_0,\tau_1,L \).  So \( \theta \) can be viewed as a
function depending only on \( \alpha,\varepsilon,\tau_0,\tau_1,L \).  For
the rest of this subsection we will prove that this value of \( \theta \)
satisfies \eqref{eq:first_part_of_side}.  

Our general approach is to note that \[ Q(S) = \prod_{\ell\in L}
(S-\ell)^D, \] satisfies \( Q(S)P_i(k)=0 \) for all \( i \); we shall apply
\( Q(S) \) to both sides of \eqref{eq:side_series}.

For the rest of this subsection, the letters \( C,C',C'' \) will refer to
various constants that are independent of \( k \) and \( n \), but may
depend on any of \( \alpha,\varepsilon,\tau_0,\tau_1,L \), and the above \(
\phi,\phi',\phi'',\rho,r,D,Q,\theta \), which are functions of \(
\alpha,\varepsilon,\tau_0,\tau_1,L \).

Let us begin with a few immediate remarks about \( Q(S) \):
\begin{enumerate} \item for \( \mu\in\RR \), \( Q(\mu) \) is real and
non-negative; \item we have \( |Q(\mu)| \), for any \( \mu \) with \(
|\mu|\le\tau_1 \) is bounded by a constant depending on \( Q \) and \(
\tau_1 \); \item for \( \mu\notin B_\delta(L) \), where
$$
B_\delta(L) = \bigcup_{\ell\in L} B_\delta(\ell), \quad\mbox{where}
\quad
B_\delta(\ell)=\bigl\{z \ \bigm|\ |z-\ell|\le \delta\bigr\},
$$
we have \[ |Q(\mu)| \geq
\delta^{|L|\,D}; \] \item for any real \( \tau>0 \), if \( h(k)=\mu^k \)
with \( |\mu|\leq \tau \), then \[ |Q(S)h(k)| \leq C\tau^k \] for some \( C
\) depending only on \( Q \); \item furthermore, if \( h(k) \) is any
function bounded by \( Ck^C \tau^k \) for some \( C \), then we have \[
|Q(S)h(k)| \leq C'k^{C'} \tau^k  \] for a constant, \( C' \), depending
only on \( C \) and \( Q \).  \end{enumerate}

We are now ready to apply \( Q(S) \) to both sides of
\eqref{eq:side_series}.  Beginning the with left-hand-side, we write: \[
Q(S) \expect{\omega\in \Omega_n}{ \sum_{i=1}^{\gamma n} \mu_i^k(\omega) } =
\expect{\omega\in \Omega_n}{ \sum_{i=1}^{\gamma n} Q(S)\mu_i^k(\omega)}; \]
note that for each \( i \) and \( \omega \) for which \(
\mu_i(\omega)\notin\RR \) we have \( |\mu_i(\omega)|\leq \tau_0 \), and
hence \[ |Q(S)(\mu_i^k(\omega)) | \leq C' \tau_0^k, \] where \( C' \) is
independent of \( n,k \).  Since there are \( \gamma n \) values of \( i \)
in the above, we have \[ \sum_{i=1}^{\gamma n} Q(S)\mu_i^k(\omega) \ge
\sum_{i,\ \mu_i\in\RR} Q(S)\mu_i^k(\omega) - \sum_{i,\ \mu_i\notin\RR}
|Q(S)\mu_i^k(\omega) | \] \[ \geq \sum_{i,\ \mu_i\in\RR}
Q(S)\mu_i^k(\omega) - C''n\tau_0^k \] for some constant \( C'' \)
independent of \( n,k \).  Furthermore, since \( Q(\mu)\geq 0 \) for \( \mu
\) real, we have for any \( j=1,\ldots,\gamma n \) and \( \omega\in\Omega_n
\) for which \( \mu_j=\mu_j(\omega) \) is real \[ Q(\mu_j)\mu_j^k \leq
\sum_{i,\ \mu_i\in\RR} Q(S)\mu_i^k(\omega) \leq C''n\tau_0^k + Q(S)
\sum_{i=1}^{\gamma n} Q(S)\mu_i^k(\omega).  \] But for each \(
\omega\in{\rm Exception}_n(\delta,\varepsilon) \) we have \(
\mu_j(\omega)\in{\rm Out}(\delta,\varepsilon) \) for some \( j \), and
hence \[ Q(\mu_j(\omega))\mu_j^k(\omega) \] is non-negative, and bounded
from below by \[ \delta^{|L|\,D} (\tau_0+\varepsilon)^k.  \] It follows
that for \( \omega\in{\rm Exception}_n(\delta,\varepsilon) \) we have \[
\Prob_n [{\rm Exception}_n(\delta,\varepsilon)] \delta^{|L|\,D}
(\tau_0+\varepsilon)^k \leq \expect{\omega\in\Omega_n}{ \sum_{i,\
\mu_i(\omega)\in\RR} Q(S)\mu_i^k(\omega)} \] \[ \leq C''n\tau_0^k + Q(S)
{\rm LHS}, \] where \( {\rm LHS} \) is the left-hand-side of
\eqref{eq:side_series}, i.e., \[ {\rm LHS} =
\expect{\omega\in\Omega_n}{\sum_{i=1}^{\gamma n} \mu_i^k(\omega)}.  \] But
the left-hand-side of \eqref{eq:side_series} equals its right-hand-side,
RHS, and \( Q(S) \) applied to any \( P_i(k) \) vanishes;  hence
\begin{equation}\label{eq:rhs1} \Prob_n [{\rm
Exception}_n(\delta,\varepsilon)] \delta^{|L|\,D}  (\tau_0+\varepsilon)^k
\leq C''n\tau_0^k + Q(S){\rm RHS} \end{equation}
\begin{equation}\label{eq:rhs2} \le C''n\tau_0^k + \Bigl|
Q(S)\bigl(h_1(k)+h_2(n,k)\bigr) \Bigr|, \end{equation} for any \( k,n \)
with \( k+D\leq n/2 \), where \[ |h_1(k)| \leq Ck^C
\bigl(\tau_0+(\varepsilon/2)\bigr)^k , \quad |h_2(n,k)| \leq Ck^C \tau_1^k
n^{-r} \] for some constant, \( C \), independent of \( n,k \).
Lemma~\ref{le:crude} implies that \begin{equation}\label{eq:Qh1}
|Q(S)h_1(k)| \leq Ck^C \bigl(\tau_0+(\varepsilon/2)\bigr)^k \end{equation}
and \begin{equation}\label{eq:Qh2} |Q(S)h_2(n,k)| \leq C k^C \tau_1^k
n^{-r} \end{equation} for a new constant, \( C \), independent of \( n,k
\).  Combining these estimates with \eqref{eq:rhs1} and \eqref{eq:rhs2} we
have \begin{equation}\label{eq:frac_mess} \Prob_n[{\rm
Exception}_n(\delta,\varepsilon)] \le \delta^{-|L|\,D} C \left(\frac{
n\tau_0^k + k^C \bigl(\tau_0+(\varepsilon/2)\bigr)^k + k^C \tau_1^k
n^{-r}}{ (\tau_0+\varepsilon)^k }\right), \end{equation} with \( C \)
independent of \( n,k \).

We claim that setting \( k \) to be \begin{equation}\label{eq:givek} k
=f(n)= \ \mbox{the smallest even integer greater than \(
\max(\phi',\phi'')\log n \)}, \end{equation} for sufficiently large \( n \)
we have that \eqref{eq:frac_mess} implies
\begin{equation}\label{eq:conclude_alpha} \Prob_n[{\rm
Exception}_n(\delta,\varepsilon)] \le 3 \delta^{-|L|\,D} n^{-\phi};
\end{equation} 
if we can establish this claim,
then \eqref{eq:first_part_of_side} follows by taking
\( \delta=n^{-\theta} \), for then, in view of \eqref{eq:givetheta} \[
\Prob_n[{\rm Exception}_n(n^{-\theta},\varepsilon)] \le 3 n^{-|L|\,D\theta}
n^{-\phi} \leq n^{-\phi+\alpha} \] which, since \( \phi>2\alpha \), is at
most \( n^{-\alpha} \) for sufficiently large $n$.

So we will finish this subsection, i.e., the proof of
\eqref{eq:first_part_of_side}, by establishing \eqref{eq:conclude_alpha}.

To show \eqref{eq:conclude_alpha}, from \eqref{eq:frac_mess} is suffices to
show that each of the three expressions \[
\frac{Cn\tau_0^k}{(\tau_0+\varepsilon)^k } ,\quad
\frac{Ck^C\bigl(\tau_0+(\varepsilon/2)\bigr)^k}{(\tau_0+\varepsilon)^k },
\quad \mbox{and}\quad \frac{Ck^C\tau_1^k n^{-r}}{(\tau_0+\varepsilon)^k }
\] is bounded by \( n^{-\phi} \) for \( n \) sufficiently large.  For the
first expression, we note that \[ \frac{Cn\tau_0^k}{(\tau_0+\varepsilon)^k
} \le Cn \bigl(\tau_0/(\tau_0+\varepsilon)\bigr)^k \le Cn
\bigl(\tau_0/(\tau_0+\varepsilon)\bigr)^{\phi''\log n} \leq Cn
n^{-\phi-2}=C n^{-\phi-1}, \] using \eqref{eq:givek} and
\eqref{eq:givephi2}, which is therefore bounded by \( n^{-\phi} \) for
sufficiently large \( n \).  For the second expression, we note that \[
\frac{Ck^C\bigl(\tau_0+(\varepsilon/2)\bigr)^k}{(\tau_0+\varepsilon)^k }
\le Ck^C
\left(\frac{\tau_0+(\varepsilon/2)}{\tau_0+\varepsilon}\right)^{\phi'\log
n} \leq Ck^C n^{-\phi-1}, \] using \eqref{eq:givek} and
\eqref{eq:givephi1}, which is therefore bounded by \( n^{-\phi} \) for
sufficiently large \( n \).  For the third expression, we note that \[ k
\leq \max(\phi',\phi'')\log n + 2 \] by virtue of \eqref{eq:givek}, and
hence \[ k \leq \rho (r/2) \log n \] for sufficiently large \( n \), by
virtue of \eqref{eq:giver}, and hence \[ \frac{Ck^C\tau_1^k
n^{-r}}{(\tau_0+\varepsilon)^k }  \le Ck^C n^{-r} ( \tau_1/\tau_0)^k \le
Ck^C n^{-r} n^{r/2} = Ck^C n^{-r/2}, \] using \eqref{eq:giverho}, and by
virtue of \eqref{eq:giver}, \( Ck^C n^{-r/2} \) is bounded by \( n^{-\phi}
\) for sufficiently large \( n \).

Hence we conclude \eqref{eq:conclude_alpha}, which, as mentioned just below
this equation, implies \eqref{eq:first_part_of_side}.

%
%
%
%
\subsection{The Proof of the Limit Formula}

In this subsection we will prove the limit formula in \eqref{eq:pellj_limit}.
Since we have established \eqref{eq:first_part_of_side}, we will use
\eqref{eq:first_part_of_side} and perform a slight variant of the technique
in the previous subsection.  However, there is an added subtlety which we now
explain.  Given \( r \), let \( D_0 \) be a positive integer for which
bounds the degrees of the polynomials which are the coefficients of the \(
P_i(k) \) over all \( i=0,\ldots,r-1 \).  For each \( \ell\in L \), and
each even, positive integer \( D \) with \( D\geq D_0 \), consider
\begin{equation}\label{eq:giveQellD} Q_{\ell,D}(z) = \prod_{\ell'\in
L\setminus\{\ell\}} (z-\ell')^D.  \end{equation} We shall apply \(
Q_{\ell,D}(S) \) to both sides of \eqref{eq:side_series}, the rough new
ideas being: \begin{enumerate} \item \( Q_{\ell,D}(S)P_i(k) \) annihilates
the part of \( P_i(k) \) involving a polynomial in \( k \) times \(
(\ell')^k \), for any \( \ell'\in L \) with \( \ell'\ne\ell \); \item \(
Q_{\ell,D}(S)P_i(k) \) does not annihilate the  part of \( P_i(k) \)
involving a polynomial in \( k \) times \( \ell^k \); and \item for \( \mu
\) within \( n^{-\theta} \) of any \( \ell'\in L \) with \( \ell'\ne\ell
\), we have that \( Q(\mu) \) is bounded by roughly \( n^{-D\theta} \).
\end{enumerate} The new ideas (1) and (2) indicate that our choice of \(
Q_{\ell,D} \) can isolate the \( \ell^k \) terms in the \( P_i(k) \); the
new idea (3) is subtle, in that we will fix \( j,r,\alpha>0 \) first,
deducing a value \( \theta>0 \) for which \eqref{eq:first_part_of_side}
holds, and then will we choose \( D \) sufficiently large so that \(
D\theta \) is large enough, compared to \( j \).

Let us, as in the previous subsection, now fix a number of variables in a
somewhat unmotivated fashion.  First, we fix an \( \varepsilon>0 \)
sufficiently small so that \( \tau_0+\varepsilon \) is strictly less than
the absolute value of any element of \( L \).  It follows that for
sufficiently small \( \delta>0 \), we have that the closed balls of radius
\( \delta \) about each element of \( L \) are disjoint, and also disjoint
from the set of complex numbers of absolute value at most \(
\tau_0+\varepsilon \).  Fix an \( \ell\in L \) and an integer \( j\geq 1 \)
for which \[ p_{\ell,0}(k) = \cdots = p_{\ell,j-1}(k)=0.  \] Let \[ \kappa
=
\frac{(j+2)\log(\tau_1/|\ell|)}{\log\Bigl(|\ell|/\bigl((\tau_0+\varepsilon)\bigr)
\Bigr)} \] which is non-negative, and equals zero iff \( |\ell|=\tau_1 \).
(The illustrates the fact that the case \( |\ell|=\tau_1 \) is, in a sense,
easier than \( |\ell|<\tau_1 \).) Choose \( \alpha \) so that
\begin{equation}\label{eq:givealpha2} \alpha-1-j-2 > \kappa \end{equation}
and choose \( r \) so that \begin{equation}\label{eq:giver2} r-j-1 > \kappa
; \end{equation} 
notice, since $\kappa$ is a function of $j,|\ell|,\tau_1,\tau_0,\varepsilon$
that the above conditions are of the form
$$
\alpha\ge G(j,L,\tau_1,\tau_0,\varepsilon),
$$
and hence
$$
r > F(\alpha,\varepsilon,\tau_1,\tau_0)
$$
becomes a condition
$$
r > H(j,L,\tau_1,\tau_0,\varepsilon)
$$
for some function, $H$;
let \( \theta>0 \) be chosen so that
\eqref{eq:first_part_of_side} holds for sufficiently large \( n \).  Choose
an even integer \( D \) so that \begin{enumerate} \item \( D \) is greater
than the degree of all polynomials \( p_{\ell',i} \) with \( i=0,\ldots,j-1
\) and \( \ell'\in L\setminus\{\ell\} \), and \item
\begin{equation}\label{eq:giveD} \theta D-j-2 > \kappa .  \end{equation}
\end{enumerate} Given the above, we can find a real number \( \nu>0 \) for
which \begin{equation}\label{eq:givenu}
\log\Bigl(|\ell|/\bigl((\tau_0+\varepsilon)\bigr) \Bigr) > \nu - (j+2)
> (j+2)\log(\tau_1/|\ell|) \min(X).
\end{equation} Where \( X = \{ \alpha-1-j-2,r-j-1,\theta D-j-2 \} \). Let
\( Q_{\ell,D}(z) \) be as in \eqref{eq:giveQellD}.  Our strategy will be so
apply \( Q_{\ell,D}(S) \) to both sides of \eqref{eq:side_series}, and then
choose \begin{equation}\label{eq:givek2} k =f(n)= \ \mbox{the smallest even
integer greater than \( \nu\log n \)}.  \end{equation} We claim that
\eqref{eq:pellj_limit} will follow.  Before going through this calculation,
let us note for future use that the choice of \( \nu \) as above,
\eqref{eq:givenu}, shows that \begin{equation}\label{eq:nufacts}
\max(\delta^D,n^{-\alpha+1},n^{r-1}) \,\tau_1^k,\ (\tau_0+\varepsilon)^k
\quad\mbox{are both}\quad O(n^{-j-2}) \ell^k \end{equation} where \(
\delta=n^{-\theta} \), for a constant in the \( O(n) \) notation that is
independent of \( n,k \).

Again, we use \( C,C' \) to denote various constants that are independent
of \( n,k \), but may depend upon \(
L,\gamma,\ell,\tau_0,\tau_1,\varepsilon,j \), and therefore depending on \(
\alpha,r,\theta,D,\kappa,\nu \) as above.

We shall make some simple observations about \( Q_{\ell,D}(z) \), analogous
to the ones make of \( Q(z) \) in the last subsection.  There is only one set
of new estimates, which we state as a lemma.

\begin{lem} \label{le:QellD_est}
The following bounds hold for, say, all \( \delta\leq 1 \):
\begin{enumerate} \item if for some \( \ell'\in L \) with \( \ell'\ne \ell
\) we have \( |z-\ell'|\leq \delta \), then \[ |Q_{\ell,D}(z)| \leq
(2\tau_1+\delta)^{(|L|-2)D} \delta^D \leq C\delta^D, \] where \( C \) is a
constant independent of \( n,k \).  \item for any \( \tau_0 \), we have
that \( |z|\leq \tau_0 \) implies that \[ |Q_{\ell,D}(z)| \leq
(\tau_0+\tau_1)^{(|L|-1)D} \leq C, \] where \( C \) is a constant
independent of \( n,k \).  \end{enumerate} \end{lem} \begin{proof} For (1),
we see that if \( \ell''\in L \) with \( \ell''\ne \ell',\ell \), then \[
|z-\ell''|\leq |z| + |\ell''| \leq (\tau_1+\delta) + \tau_1=
2\tau_1+\delta; \] hence \[ |Q_{\ell,D}(z)| \leq |z-\ell'|^D \
\prod_{\ell''\in L\setminus\{\ell,\ell'\}} |z-\ell''|^D \leq \delta^D \
(2\tau_1+\delta)^{(|L|-2)D}.  \] For (2) we note that \( \ell'\in L \)
implies that \[ |z-\ell'| \leq |z| + |\ell'| \leq \tau_0+\tau_1, \] and (2)
follows.  \end{proof}

We compile a list of simple observations about \( Q_{\ell,D}(z) \),
analogous to the ones make of \( Q(z) \) in the last subsection:
\begin{enumerate} \item for \( \mu\in\RR \), \( Q_{\ell,D}(\mu) \) is real
and non-negative; \item for any \( \mu \) with \( |\mu|\le\tau_1 \) we have
\begin{equation}\label{eq:allmu} |Q_{\ell,D}(\mu)| \leq
(2\tau_1)^{(|L|-1)D} \leq C \end{equation} where \( C \) is a constant
independent of \( n,k \) (by (2) of Lemma~\ref{le:QellD_est}); \item \(
Q_{\ell,D}(\ell)\ne 0 \); \item for real \( z \) with \( |z-\ell|\leq
\delta \) and \( \delta\leq 1 \), we have \begin{equation}\label{eq:Qderiv}
Q_{\ell,D}(z) = Q_{\ell,D}(\ell) +  O(\delta), \end{equation} where the
constant in the \( O() \) is independent of \( n,k \) (and, in fact,
depends only on \( L \) and \( D \)), which follows using the
mean-value-theorem, with the \( O() \) constant being any bound on the
derivative \( Q'_{\ell,D}(z) \) over all \( z \) within \( 1 \) of \( \ell
\); \item for any real \( \tau>0 \), if \( h(k)=\mu^k \) with \( |\mu|\leq
\tau \), then \[ |Q_{\ell,D}(S)h(k)| \leq C\tau^k \] for some \( C \)
depending only on \( L \) and \( D \); and \item furthermore, if \( h(k) \)
is any function bounded by \( Ck^C \tau^k \) for some \( C \), then we have
\[ |Q_{\ell,D}(S)h(k)| \leq C'k^{C'} \tau^k  \] for a constant, \( C' \),
depending only on \( C \), \( L \), and \( D \).  \end{enumerate} In
addition to these simple estimates, we will use Lemma~\ref{le:QellD_est}.

Returning to Lemma~\ref{le:side-step}, the goal of the rest of this subsection
is to establish \eqref{eq:pellj_limit}.  As in the previous subsection, we let
\( {\rm LHS} \) and \( {\rm RHS} \), respectively, be the left-hand-side
and right-hand-side of \eqref{eq:side_series}.  Again, the basic idea is to
apply \( Q_{\ell,D}(S) \) to \( {\rm LHS}={\rm RHS} \).  As before, we see
that \begin{equation}\label{eq:lhsD} Q_{\ell,D}(S){\rm LHS} =
\expect{\omega\in\Omega_n}{ \sum_i Q_{\ell,D}\bigl(\mu_i(\omega)\bigr)
\mu_i^k(\omega) }, \end{equation} and we wish to estimate this expectation
(for various values of \( n,k \)).

For \( \delta>0 \) sufficiently small (ultimately we will take \(
\delta=n^{-\theta} \)) for each \( i \) and \( \omega\in\Omega_n \) we have
that exactly one of the following holds: \begin{enumerate} \item for some
\( \ell'\in L\setminus\{\ell\} \) we have \( |\mu_i(\omega)-\ell'| \leq
\delta \); \item we have \( |\mu_i(\omega)-\ell|\leq \delta \); \item we
have \( |\mu_i(\omega)| \leq \tau_0+\varepsilon \); \item we have \[
\mu_i(\omega)\in {\rm Out}(\delta,\varepsilon).  \] \end{enumerate} Assume
that $\delta$ is sufficiently small for this to hold.  In view of
\eqref{eq:lhsD}, let us estimate the contribution to the expected value of
$Q_{\ell,D}(\mu)\mu^k$ for $\mu=\mu_i(\omega)$ in each of the above four
cases.

For case~(1), i.e., $|\mu_i(\omega)-\ell'|\le\delta$ with $\ell'\in
L\setminus\{\ell\}$, use the estimate \[ \expect{\omega\in\Omega_n}{
\sum_{i\ s.t.\ |\mu_i(\omega)-\ell'|\leq \delta}
Q_{\ell,D}\bigl(\mu_i^k(\omega)\bigr) \mu_i^k(\omega) } \leq {\rm
Near}_n(\ell',\delta) \delta^D C (\ell')^k, \] using
Lemma~\ref{le:QellD_est}, where $C$ is independent of $n,k$.  Summing over
all $\ell'\in L\setminus\{\ell\}$, and using the crude bound \[
\sum_{\ell'\in L} {\rm Near}_n(\ell',\delta) \leq \gamma n, \] we conclude
that \begin{equation}\label{eq:k_power_error1} \expect{\omega\in\Omega_n}{
\sum_{i\ s.t.\ |\mu_i(\omega)-\ell'|\leq \delta {\rm \ for\ some }\ell'\in
L\setminus\{\ell\}} Q_{\ell,D}\bigl(\mu_i^k(\omega)\bigr) \mu_i^k(\omega) }
\end{equation} \[ \leq \gamma  n \delta^D (2\tau_1+\delta)^{(|L|-2)D}
\max_{\ell'} (\ell')^k \leq C n\delta^D \tau_1^k, \] where $C$ is
independent of $n,k$.

For case~(2), i.e., $|\mu_i(\omega)-\ell|\le\delta$, with $\delta>0$
sufficiently small, we note that \[ |\mu_i^k(\omega)-\ell^k | \leq
|\mu_i(\omega)-\ell| \ |\mu_i^{k-1}(\omega) + \mu_i^{k-2}(\omega)\ell +
\cdots + \ell^{k-1}| \] \[ \leq |\mu_i(\omega)-\ell| \,k (\tau_1+\delta)^k
\leq \delta k (\ell+\delta)^k = \delta k \ell^k
\bigl(1+(\delta/\ell)\bigr)^k.  \] Hence we have \[ \mu_i^k(\omega)= \ell^k
\bigr( 1 + \delta k\, O(1) \bigl), \] with the $O(1)$ being at most $e$,
provided that \[ (1/k) \geq \ell/\delta; \] given $k$ as in
\eqref{eq:givek2}, this holds for all $n$ sufficiently large.  Then
\eqref{eq:Qderiv} implies that \[ \expect{\omega\in\Omega_n}{ \sum_{i\
s.t.\ |\mu_i(\omega)-\ell|\leq \delta}
Q_{\ell,D}\bigl(\mu_i^k(\omega)\bigr) \mu_i^k(\omega) } \] \[ = {\rm
Near}_n(\ell,\delta) Q_{\ell,D}(\ell) \bigl(1+O(\delta)\bigr) \ell^k\bigr(
1 + \delta k\, O(1) \bigl), \] and so
\begin{equation}\label{eq:k_power_error2} \expect{\omega\in\Omega_n}{
\sum_{i\ s.t.\ |\mu_i(\omega)-\ell|\leq \delta} \!\!\!\!\!\!\!\!\!\!
Q_{\ell,D}\bigl(\mu_i^k(\omega)\bigr) \mu_i^k(\omega) } = {\rm
Near}_n(\ell,\delta) Q_{\ell,D}(\ell) \ell^k \bigr( 1 + \delta k O(1)
\bigl).  \end{equation} where the constant in the $O(1)$ is independent of
$n,k$ for $\delta=n^{-\theta}$ and $k$ as in \eqref{eq:givek2}.

For case~(3), \eqref{eq:allmu} implies that
\begin{align}\label{eq:k_power_error3} \expect{\omega\in\Omega_n}{\sum_{i\
s.t.\ |\mu_i(\omega)|\leq \tau_0+\varepsilon}
Q_{\ell,D}\bigl(\mu_i^k(\omega)\bigr) \mu_i^k(\omega) } & \le n\gamma
(2\tau_1)^{(|L|-1)D} (\tau_0+\varepsilon)^k \\ & \le n C
(\tau_0+\varepsilon)^k \end{align} with $C$ independent of $n,k$.

Finally, for case~(4), we use the estimate \[ \expect{\omega\in\Omega_n}{
\sum_{i\ s.t.\ \mu_i(\omega)\in{\rm Out}(\delta,\varepsilon)}
Q_{\ell,D}\bigl(\mu_i^k(\omega)\bigr) \mu_i^k(\omega) } \] \[ \leq n\gamma
C\tau_1^k \Prob_n[{\rm Exception}_n(\delta,\varepsilon)], \] since
$\mu_i(\omega)\in{\rm Out}(\delta,\varepsilon)$ implies that $\omega\in{\rm
Exception}_n(\delta,\varepsilon)$, and in this case there are at most
$\gamma n$ values of $i$ for which $\mu_i(\omega)$ lies in ${\rm
Out}(\delta,\varepsilon)$, and for each such $i$ we have
$Q_{\ell,D}\bigl(\mu_i^k(\omega)\bigr)$ is at most a constant, $C$, by
\eqref{eq:allmu}.  Hence, using \eqref{eq:first_part_of_side}, we have
\begin{equation}\label{eq:k_power_error4} \expect{\omega\in\Omega_n}{
\sum_{i\ s.t.\ \mu_i(\omega)\in{\rm Out}(\delta,\varepsilon)}
Q_{\ell,D}\bigl(\mu_i^k(\omega)\bigr) \mu_i^k(\omega) } \end{equation} \[
\leq Cn\tau_1^k n^{-\alpha}.  \] for $n$ sufficiently large, and $C$
independent of $n,k$.

Combining \eqref{eq:k_power_error1}--\eqref{eq:k_power_error4}, and
\eqref{eq:lhsD}, we get that for $\delta=n^{-\theta}$, for all $k$ with
$k\leq \ell/\delta$ and $n$ sufficiently large we have and \[
Q_{\ell,D}(S){\rm LHS} - {\rm Near}_n(\ell,\delta) Q_{\ell,D}(\ell) \ell^k
\bigr( 1 + \delta k O(1) \bigl) \] \begin{equation}\label{eq:yuck} = O(n)
\bigl( \delta^D \tau_1^k + (\tau_0+\varepsilon)^k + n^{1-\alpha} \tau_1^k
\bigr), \end{equation} where the constant in the $O(n)$ is independent of
$n,k$.  But \eqref{eq:nufacts} implies that \eqref{eq:yuck} is $O(\ell^k
n^{-j-1})$ for $k$ as in \eqref{eq:givek2}, and hence for $k$ as such we
have \begin{equation}\label{eq:simplify} Q_{\ell,D}(S){\rm LHS} - {\rm
Near}_n(\ell,\delta) Q_{\ell,D}(\ell) \ell^k \bigr( 1 + \delta k O(1)
\bigl) = O(\ell^k  n^{-j-1}).  \end{equation}

%

As in the previous subsection, we apply $Q_{\ell,D}(S)$ to ${\rm RHS}$, and
find \begin{equation}\label{eq:Qrhs_error} Q_{\ell,D}(S){\rm RHS} =
\sum_{i=0}^{r-1} Q_{\ell,D}(S) P_i(k) + O(n^{-j-1})\ell^k \end{equation}
provided that for any $C$ independent of $n,k$ we have \[ k^C
\bigl(\tau_0+(\varepsilon/2)\bigr)^k + k^C \tau_1^k n^{-r} =
O(n^{-j-1})\ell^k \] (for the same reasoning as \eqref{eq:Qh1},
\eqref{eq:Qh2}, \eqref{eq:rhs2}); but this is implied by
\eqref{eq:nufacts}.
%
Hence, dividing \eqref{eq:Qrhs_error} by $\ell^k$, we have \begin{align*}
Q_{\ell,D}(\ell)\sum_{i=0}^{r-1} p_{\ell,i}(k) n^{-i} &= \ell^{-k}
Q_{\ell,D}(S){\rm RHS} + O(n^{-j-1})\\ &= \ell^{-k} Q_{\ell,D}(S){\rm LHS}
+ O(n^{-j-1})\\ &= {\rm Near}_n(\ell,\delta) Q_{\ell,D}(\ell) \bigr( 1 +
\delta k O(1) \bigl) + O(n^{-j-1}) \end{align*} for $k$ as in
\eqref{eq:givek2}.  Given that $p_{\ell,0}(k),\ldots,p_{\ell,j-1}$ all
vanish, we have, upon dividing by $n^{-j}Q_{\ell,D}(\ell)$, that \[
p_{\ell,j}(k) + O(n^{-1}) = n^j {\rm Near}_n(\ell,n^{-\theta}) + O(n^{-1}).
\] But $n^j {\rm Near}_n(\ell,n^{-\theta})$ is independent of $k$, and $k$
is proportional to $\log n$.  Hence taking $n\to\infty$ we conclude
\eqref{eq:pellj_limit}.

We note that the only requirement on our choice of $\theta>0$ is that
$\theta$ satisfies \eqref{eq:first_part_of_side} for our chosen value of
$\alpha$ (in \eqref{eq:givealpha2}).  Hence if $\theta_0$ is such a value
of $\theta$, then for any $\theta$ such that $0<\theta<\theta_0$, then \[
\Prob_n[{\rm Exception}_n(n^{-\theta},\varepsilon)] \le \Prob_n[ {\rm
Exception}_n(n^{-\theta_0},\varepsilon) ] \leq n^{-\alpha}, \] for $n$
sufficiently large.  Hence any $\theta$ with $0<\theta<\theta_0$ also
satisfies \eqref{eq:first_part_of_side} for $\alpha$ (in
\eqref{eq:givealpha2}), and hence satisfies \eqref{eq:pellj_limit}.

\subsection{The End of The Proof of Lemma~\ref{le:side-step}}

It remains to prove \eqref{eq:final_side}, having established
\eqref{eq:first_part_of_side} and \eqref{eq:pellj_limit}.  First we note
that \eqref{eq:pellj_limit} implies that for any $\ell\in L$ and
$\varepsilon>0$ (and given $L,\gamma,\tau_0,\tau_1,j$) there is a $C>0$ for
which \begin{equation}\label{eq:thetaNear} {\rm Near}_n(\ell,n^{-\theta})
\leq C (n^{-j}) \end{equation} for any fixed $\theta>0$ for which
\eqref{eq:first_part_of_side} holds with $\alpha$ given in
\eqref{eq:givealpha2}.  Hence this also holds with $\theta$ replaced by any
smaller, positive value of $\theta$.

\begin{proof}[Proof of \eqref{eq:final_side}, therefore
completing the proof of Lemma~\ref{le:side-step}.]
Given an $\varepsilon>0$, we can therefore choose a
$\theta>0$ for which \begin{enumerate} \item for this $\varepsilon$, and
with $\alpha=j$, we have \eqref{eq:first_part_of_side} holds; and \item we
have that \eqref{eq:thetaNear} holds for all $\ell\in L$.  \end{enumerate}

We have \[ \Prob_n[ \{ \omega\ | \ \mbox{$|\mu_i(\omega)-\ell|\leq
n^{-\theta}$ for some $i$} \} ] \leq {\rm Near}_n(\ell,n^{-\theta}) \leq C
\, n^{-j} \] for a constant, $C$, independent of $n,k$.  Hence, summing
over all $\ell$, we have \[ \Prob_n[ E_n \leq C\,n^{-j}] \] where \[ E_n =
\{ \omega\ |\ \mbox{$|\mu_i(\omega)-\ell|\leq n^{-\theta}$ for some $i$ and
some $\ell\in L$} \} .  \]

Since $\theta$ also satisfies \eqref{eq:first_part_of_side} with
$\alpha=j$, we have \[ \Prob_n[{\rm Exception}_n(n^{-\theta},\varepsilon) ]
\leq n^{-j}, \] for $n$ sufficiently large.  But clearly \[ \Prob_n[{\rm
AbsoluteException}_n(\varepsilon)] \] \[ \le \Prob_n[{\rm
Exception}_n(n^{-\theta},\varepsilon) ] + \Prob_n[ E_n ], \] since if
$\omega\in\Omega_n$ has $|\mu_i(\omega)|>\tau_0+\varepsilon$ for some $i$,
then either $\mu_i(\omega)$ lies in ${\rm Out}(n^{-\theta},\varepsilon)$ or
$\mu_i(\omega)$ is within $n^{-\theta}$ of some $\ell\in L$.  Hence \[
\Prob_n[{\rm AbsoluteException}_n(\varepsilon)] \] \[ \le \Prob_n[{\rm
Exception}_n(n^{-\theta},\varepsilon) ] + \Prob_n[ E_n ] \leq C n^{-j}.  \]
\end{proof}
\section{Proof of the Relativized Alon Conjecture}
\label{section:main_proof}
\label{se:p1-main-proof}

%

This section is devoted to completing the proof of
Theorem~\ref{th:main}, for base graphs, $B$, without half-loops;
Theorem~\ref{th:main_Alon} for 
regular base graphs, $B$, without half-loops will easily follow.
Chapter~\ref{ch:p2}
has variants of this result when $B$ contains half-loops, for
other models related to the Broder-Shamir model, and
weaker results for non-regular
$B$.

First, let us explain why Theorem~\ref{th:main_Alon} follows for
any regular $B$ for which Theorem~\ref{th:main} holds.
If $G$ is $d$-regular, then any non-real eigenvalue
of $G$ has absolute value $(d-1)^{1/2}$, by \eqref{eq:preZeta}
and the discussion below this equation.  
It follows that for any regular graph,
$B$, we can take $\tau_0=\rhoroot B$, and all the hypotheses of
Theorem~\ref{th:main} are satisfied.  Hence the conclusion of
Theorem~\ref{th:main} holds, which is just
Theorem~\ref{th:main_Alon}.

\begin{proof}[Proof of Theorem~\ref{th:main}, for $B$ without
half-loops] (And therefore of Theorem~\ref{th:main_Alon}, for $B$
without half-loops.)

According to Theorem~\ref{thm:tangles_finite}, for any $r$, the set of
minimal tangles of order less than $r$ is finite; since the Perron-Frobenius
eigenvalue of the Hashimoto matrix of a cycle equals $1$, any such
minimal tangle is a connected graph of order at least one.
It follows from Theorem~\ref{th:no_walks} that
\begin{equation}\label{eq:tangle_free_only}
\expect{G\in\cC_n(B)}{\tanglefreeindicator(G)} =
p_1 n^{-1} + p_2 n^{-2} + \ldots + p_{1-r} n^{-r+1} + O(n^{-r})
\end{equation}
for any $r$; in fact it follows that $p_1=\ldots=p_j=0$
if the smallest order of tangle is $j+1$.
From Theorem~\ref{th:main_expansion_B}, if follows
that for any 
integer $r>0$ we have
\begin{equation}\label{eq:power_series_proof}
\expect{G\in\cC_n(B)}{\tanglefreeindicator(G)H_G^k} =
P_0(k) + P_1(k)n^{-1} + \cdots 
+ P_{r-1}(k) n^{1-r} + {\rm err}_r(n,k),
\end{equation}
where the \( P_i(k) \) are \( B \)-Ramanujan, and there is a \( C \)
independent of \( n,k \) for which 
\begin{equation}\label{eq:error_bound_exp}
|{\rm err}(n,k)| \leq Ck^{2r+2}
\rho(H_B)^k n^{-r}. 
\end{equation}
Theorem~\ref{th:broder_shamir_friedman} implies that, up to the error
term, the principle part of $P_0$ is given by
$$
P_0(k) = \Tr(H_B^k) 
$$
(say by taking $k,n\to\infty$ with $k \ll \log n$).
It follows that 
\begin{eqnarray}
&& \expect{G\in\cC_n(B)}{\tanglefreeindicator(G)
\sum_{\mu\in\Specnew_B(H_G)} \mu^k} \\[.2cm]
&=& 
\expect{G\in\cC_n(B)}{\tanglefreeindicator(G)
\Bigl(\Tr(H_G^k)-\Tr(H_B^k)\Bigr)} \\[.2cm]
&=&
\widetilde P_0(k) + \widetilde P_1(k)n^{-1} + \cdots 
+ \widetilde P_{r-1}(k) n^{1-r} + {\rm err}_r(n,k),
\end{eqnarray}
where 
\begin{equation}\label{eq:widetilde_P_defined}
\widetilde P_i(k) = P_i(k) + p_i \Tr(H_B^k),
\end{equation}
and where
$$
|{\rm err}(n,k)| \leq Ck^{2r+2}
\rho(H_B)^k n^{-r}. 
$$
In particular, the $\widetilde P_i(k)$ are $B$-Ramanujan functions with
bases $\Spec(H_B)$, and the
principle part of $\widetilde P_0(k)$ vanishes.

We now wish to apply the side-stepping-lemma,
Lemma~\ref{le:side-step}.
So choose an arbitrary $\varepsilon>0$.
Consider the following abstract partial trace $(\tau_0,\tau_1,C',L,r)$,
where: 
\begin{enumerate}
\item $\tau_0$ as in the hypothesis
of Theorem~\ref{th:main};
\item $\tau_1=(\tau_0)^{1/2}$;
\item $L$ is the set of Hashimoto eigenvalues of $B$ (it suffices
to take those of absolute 
value greater than $\tau_1$);
\item we choose any $r$ with
\begin{equation}\label{eq:how_big_is_r}
r>H(1,L,\tau_1,\tau_0,\varepsilon)
\end{equation}
with $H$ as in
\eqref{eq:H_defined};
\item $\gamma$ is the $\alpha=\alpha(r)$ in the conclusion of
Theorem~\ref{th:main_expansion_B}
\item $C'=C'(r)$ is a constant for which
\begin{enumerate}
\item each error term of the 
$\widetilde P_i(k)$ is bounded in absolute value by
by 
$$
C' \bigl(\tau_0+\varepsilon/2 \bigr)^k  
$$
(which is possible since $\tau_0\ge \rhoroot B$);
and
\item the
error term bound in \eqref{eq:error_bound_exp} is at most
$$
C' k^{2r+2} \rho(H_B)^k n^{-r}  
$$
for all $k,n$ with $1\le k\le \gamma n^\gamma$
\end{enumerate}
(the existence of such a $C'=C'(r)$ is guaranteed by
Theorem~\ref{th:main_expansion_B})
and
\item the $\mu_i=\mu_i(G)$ are random variables ranging over 
$G\in\cC_n(B)$ given by the {\em new eigenvalues} of 
$H_G$ over $H_B$.
\end{enumerate}
It is easy to verify then that the hypotheses (1)--(3) of
Lemma~\ref{le:side-step} are satisfied; hence we may apply this lemma.
The last part of the lemma can be applied with $j=1$ in the lemma,
since the principle part of $\widetilde P_0(k)$ vanishes; hence we 
conclude that
$$
\Prob_n[ {\rm AbsoluteException}_n(\varepsilon) ] \leq C_\varepsilon n^{-1} ,
$$
in the language of Definition~\ref{de:partial_trace_sets}.  However,
the event
$$
{\rm AbsoluteException}_n(\varepsilon)
$$
is just the event that $G\in\cC_n(B)$ has a new Hashimoto eigenvalue
of absolute value at least $\tau_0^{1/2}+\varepsilon$.
Hence we conclude Theorem~\ref{th:main}.

\end{proof}

%
%
\chapter{Generalizations and Further Directions}
\label{ch:p2}

In this chapter we give a number of generalization and refinements 
of Theorem~\ref{th:main_Alon} and the proof techniques 
of Chapter~\ref{ch:p1}.  
In Section~\ref{se:p2-irregular} we prove a result regarding
the new adjacency eigenvalues of graphs that are not regular.
In Section~\ref{se:p2-spread} we prove some results about the
{\em spreading} (a type of expansion property) of elements of
$\cC_n(B)$;
in Section~\ref{se:p2-fund-exp} we use these spreading estimates
to prove Theorem~\ref{th:main_Alon_Ramanujan}.
In Section~\ref{se:p2-algebraic} we discuss the case where $B$ may have
half-loops, and some variants of the Broder-Shamir model, $\cC_n(B)$, to
which are theorems apply.
In Section~\ref{se:p2-modl}, we explain that some of the 
lower order coefficients
of our $1/n$-asymptotic expansions have, at least for certain $B$,
well defined polyexponential-type terms
of the form $p(k)(d-1)^{k/2}$ for $k$ even, and $q(k)(d-1)^{(k-1)/2}$
for $k$ odd; these coefficients would be interesting to compute;
in principle such terms can be computed with our techniques, although
we do not know how easy it is to make such computations.
In Section~\ref{se:p2-future} we conclude with some possible future
directions for research.

\section{Irregular Graphs}
\label{se:p2-irregular}

In this section we make prove Theorem~\ref{th:d_max}.
This follows from Theorem~\ref{th:main} and the following
theorem of Kotani and Sunada, in \cite{kotani}.

\begin{theorem}
Let $G$ be a graph whose maximum degree is $d_{\max}$.  Then any
non-real eigenvalue of $H_G$ has absolute value at most
$\sqrt{d_{\max}-1}$.
\end{theorem}

Their proof is short and clever: take the inner product of the equation
$(\mu_2 - \mu A_G  + (D_G-I))v=0$ with $v$, and divide by $|v|^2$;
their result follows from the fact that this quadratic equation in $\mu$
has constant term equal to the Rayleigh quotient of $v$ for the
matrix $D_G-I$, which is clearly bounded by $d_{\max{}}-1$.
See \cite{kotani} for details.

\begin{proof}[Proof of Theorem~\ref{th:d_max}]
Since each vertex of $\Line(B)$ has degree at most $d_{\max}-1$, we have
$\rho(H_B)\le d_{\max}-1$.
We now apply Theorem~\ref{th:main} with $\tau_0=d_{\max}-1$.
\end{proof}

\section{Spreading in Random Covers of Regular Graphs}
\label{se:p2-spread}

If $B$ is $d$-regular, then, 
in Chapter~\ref{ch:p1}, we have identified the probability that
a random cover in $\cC_n(B)$ has new adjacency eigenvalues of absolute
value greater than $2(d-1)^{1/2}+\epsilon$ in terms of the
principle part of the \glspl{coefficient} of certain
\glspl{asymptotic expansion}.  
It turns out that the contribution of the bases $d-1$, and $1-d$ if
$B$ is bipartite, can be understood via the notion
of {\em $\gamma$-spreaders}, in Chapter~12 of \cite{friedman_alon};
the exact same notion is called a
{\em $\gamma$-expander} in \cite{friedman_random_graphs}
in Lemma~3.1, and this notion plays the same role in both papers.


In Section~\ref{se:p2-fund-exp} these results will
be crucial to the proof of Theorem~\ref{th:main_Alon_Ramanujan};
the basic point is that if $B$ is Ramanujan, then the only bases
contributing to the coefficients of the $1/n$-asymptotic expansions
are $\pm(d-1)$, we which understand via spreading.

Our goal in this section is to prove the following two theorems.

\begin{theorem}\label{th:tangle_spreading}
Let $B$ be a connected,
$d$-regular graph with $d\ge 3$.
Then for any integer, $j>0$,
there exists an $\epsilon>0$, $C>0$, and an integer $r$ such that
\begin{enumerate}
\item if $B$ is not bipartite, then
the probability that
$G\in\cC_n(B)$ in contains no connected component of fewer than
$r$ vertices and has
$$
|\lambda_i(G)|\ge d-\epsilon
$$
for some $i>1$ is at most $Cn^{-j}$; and
\item if $B$ is bipartite, then the same is true with the condition
``\;$i>1$'' replaced with
``\;$i\ne 1,|V_G|$'' (when $B$ is bipartite, then $G$ automatically has
an eigenvalue equal to $-d$).
\end{enumerate}
\end{theorem}

The idea behind the proof is to fix a spanning tree, $T$, for $B$.
Then we reduce spreading in $B$ to spreading in the graph, $B[T]$, which
is $B$ where $T$ is {\em contracted} to a single vertex; the graph
$B[T]$ has one vertex, and edges $E_B\setminus E_T$.
Then information regarding spreading in random covers of
$B[T]$, which was essentially established
in \cite{friedman_alon}, easily establishes spreading in $B$.

We will need some results on $d$-regular graphs that 
are slightly stronger results than those of 
\cite{friedman_random_graphs}, but that follow from the methods there.
So we will review all the terminology
and proofs there.  We will use these results to establish spreading
theorems for $\cC_n(B)$ for any $d$-regular $B$.

\subsection{Spreaders}
\label{sb:spreaders}

\begin{definition}
Let $G$ be a graph, and $A\subset V_G$.  We define the
{\em neighbourhood of $A$}, denoted $\Gamma_G(A)$, to be the subset
of $V_G$ consisting of those vertices joined by
an edge of $G$ to a vertex of $A$.
\end{definition}

\begin{definition}
Say that a $d$-regular graph, $G$, on $n$ vertices is a {\em $\gamma$-spreader}
if for every subset,
$A$, of at most $n/2$ vertices we have
$$
|\Gamma_G(A)|\ge (1+\gamma)|A|.
$$
\end{definition}

\begin{theorem}\label{th:separation}
Let $G$ be a $d$-regular $\gamma$-spreader.
Then for all $i>1$ we have 
$$
\lambda_i^2(G) \le d^2 -\frac{\gamma^2}{4+2\gamma^2}.
$$
\end{theorem}
\begin{proof}
See \cite{friedman_alon}; this is a pretty easy consequence of
Alon's work on {\em magnifiers} \cite{alon_eigenvalues}.
\end{proof}

\begin{definition} 
\label{de:k_walk_graph}
Let $G$ be a graph and $k\ge 1$ an integer.
By the {\em graph of $k$-length walks in $G$}, denoted $G[k]$, we mean
the graph whose vertices
are $V_G$, and whose edges are walks of length $k$ in $G$, with the
tail of the edge the first vertex in the walk, and the head of the
edge the last vertex of the walk.
Hence $A_{G[k]}$, the adjacency matrix of
$G[k]$, is just $A_G^k$.  Also, if $G$ is $d$-regular, then
$G[k]$ is $d^k$-regular.
\end{definition}

\begin{corollary}\label{co:weak_separation}
Let $G$ be a $d$-regular graph, and $k$ a positive integer.
Then if $G[k]$ is a $\gamma$-spreader,
then for all $i>1$ we have
$$
\lambda_i^{2k}(G) \le d^{2k} -\frac{\gamma^2}{4+2\gamma^2}.
$$
\end{corollary}
\begin{proof} Apply Theorem~\ref{th:separation} to $G[k]$.
\end{proof}

\subsection{Preliminary Lemmas}

In this subsection we give some general, simple facts in graph theory and
spreading to be used later.  

\begin{lemma}\label{le:paths}
Let $B$ be a connected graph with at least one edge.  Fix a vertex,
$v\in V_B$.  Then
\begin{enumerate}
\item for any even integer, $k\ge 0$, there is a closed walk of length
$k$ originating and terminating in $v$; 
\item $B$ is bipartite iff every closed walk originating and terminating
in $v$ has even length;
\item if $B$ is not bipartite, then there exists a closed walk
originating and terminating in $v$ of odd length, $k$, with
$k\le 2|V_B|$.
\end{enumerate}
\end{lemma}
\begin{proof}
Item~(1): $v$ is incident upon some edge $e$; if $e$ is a self-loop, we
may traverse it $k$ times; otherwise we may traverse $e$ back and
forth $k/2$ times.

Item~(2) is standard:  the ``if'' direction is clear.  The ``only if''
directions follows because $B$ is connected: every $v'\in B$
is connected to $v$ by a walk of some length, and all walks from $v'$ to
$v$ must have the same parity.  This parity
gives a bipartition of the 
vertices.

Item~(3) is also standard: by Item~(2), there exists a closed walk
of odd length.
Let $w$ be a closed walk of minimum
odd length in $V_B$; then $w$ is of length at most $|V_B|$, since
if a vertex of $V_B$ occurs twice in a closed non-backtracking
walk (we count the occurrence of the first and last vertex in the closed
walk as a single occurrence), then the walk breaks into two non-backtracking
walks, one of which must be of odd length.
Let $w'$ be a walk of minimum length from $v$ to a vertex of $w$;
the length of $w'$ is at most $|V_B|-{\rm length}(w)$.
Then the walk $(w')^{-1} w w'$ is a closed walk of odd length,
originating and terminating in $v$, of length at most $2|V_B|$.
(This bound is tight when $B$ consists of a single ``path'' plus
a self-loop at one of its endpoints.)
\end{proof}

\begin{corollary}\label{co:gamma_inclusion}
Let $B$ be a connected graph with at least one edge.  Then
$$
A\subset \Gamma^2_G(A) \subset \Gamma^4_G(A) \subset \cdots
\quad\mbox{and}\quad
\Gamma^1_G(A) \subset \Gamma^3_G(A) \subset \cdots
$$
\end{corollary}
\begin{proof}
This follows from Item~(1) of the lemma above.
\end{proof}

\begin{lemma}\label{le:disconnected}
Let $G$ be a $d$-regular graph, with $d\ge 1$.  
Then for all $A\subset V_G$ we have
$$
|A| \le |\Gamma_G(A)|,
$$
and equality holds iff the subgraph of $G$ induced on the
vertex subset $A\cup \Gamma(A)$ is disconnected from the rest of $G$.
\end{lemma}
\begin{proof}
For each $A$, $\Gamma_G(A)$ can be described as the heads of all directed
edges whose tail lies in $A$, or, the same with ``head'' and ``tail''
interchanged.

Let $E_A$ be the set of
directed edges whose tail lies in $A$.  Then $|E_A|=d|A|$.
Any vertex occurs as a head of edges in $E_A$ at most $d$ times.
Hence 
$$
|\Gamma_G(A)|\ge |E_A|/d \ge |A|.
$$
If the above holds with equality, then each vertex in $\Gamma_G(A)$
is the head of $d$ elements of $E_A$.  It follows that
$\Gamma^2_G(A)$ which is the set of vertices appearing as a tail
of an edge whose head lies in $\Gamma_G(A)$, is precisely $A$; by repeating
this argument (or just by applying $\Gamma_G$ repeatedly) we conclude
that
$$
A = \Gamma^2_G(A) = \Gamma^4_G(A) = \cdots
\quad\mbox{and}\quad
\Gamma^1_G(A) = \Gamma^3_G(A) = \cdots
$$
Hence all edges with tails or heads in the set of vertices
$A\cup\Gamma_G(A)$ have both their tails and heads in this set.
\end{proof}

\begin{corollary}\label{co:disconnected}
Let $G$ be a regular graph, and $A\subset V_G$.  If for 
a real $\gamma \le 1/|A|$ we have 
$$
|\Gamma(A)| < |A| (1+\gamma),
$$
$|\Gamma(A)|=|A|$, and
$A\cup \Gamma(A)$ is disconnected from the rest of $G$.
\end{corollary}

\subsection{Spreading in Random Graphs}

Now we wish to show that random covers of graphs with be spreaders.
The proofs of 
Theorems~12.3 and 12.4 of \cite{friedman_alon} imply some
slightly stronger theorems that we will need.  Hence we will
state these strengthenings and outline their proofs.
Here is the strengthening of Theorem~12.3 that we need here.

\begin{definition} 
We say that a graph, $B$, is a {\em bouquet of self-loops}, if $B$ has one
vertex; hence $E_B$ consists entirely of self-loops.
Specifically, by the {\em bouquet of $i$ whole-loops and $j$ half-loops}
we mean the bouquet of self-loops with $i$ whole-loops and $j$ half-loops;
in this case, $B$ is $(2i+j)$-regular, and $\cC_n(B)$ is a model
of a random, $(2i+j)$-regular graph.
\end{definition}

\begin{theorem} 
\label{th:spreader_power}
Let $B$ be a bouquet of self-loops such that the degree of the vertex in
$B$ is at least three.  Let $s$ be any positive integer.
The there exists an integer $m$ and a real $\gamma>0$ such that the 
following is true:
the probability, $E(n,2m,\gamma,B)$,
that an element of $\cC_n(B)$ has no connected component of
at most $2m$ vertices and is not a $\gamma$-spreader, is at most
$n^{-s}$ for all $n$ sufficiently large.
\end{theorem}

This theorem is a slight improvement of
Theorems~12.2 and 12.3 in \cite{friedman_alon}; it is also more general,
since \cite{friedman_alon} proves this only for a bouquet which is 
either entirely whole-loops or half-loops.  However, the above theorem
follows easily from the methods used in the proofs 
Theorems~12.2 and 12.3 in
\cite{friedman_alon}.
For ease of reading, we will summarize the main points there; we will
also correct
a minor error there.
We shall prove the above theorem in stages: first for a bouquet of
whole-loops, and then mixture of whole-loops and half-loops.

\begin{proof}[Proof of Theorem~\ref{th:spreader_power} in the case of
a bouquet of whole-loops.]

So fix $B$, a bouquet of $d/2$ whole-loops, where $d\ge 4$ is an
even integer.


By Corollary~\ref{co:disconnected},
as long as 
$$
1/\gamma\ge m,
$$
we can estimate the probability $E(n,2m,\gamma)$ in the statement of the
theorem by bounding the probability that there exist subsets
$A,B\subset\{1,\ldots,n\}$ such that for $d/2$ random permutations
$\pi_1,\ldots,\pi_{d/2}$ (used to form an element of 
$\cC_n(W_{d/2})$) we have
$$
m\le |A|\le n/2,\quad |B|\le \lfloor A(1+\gamma) \rfloor ,\quad
\mbox{and}\quad \Gamma(A)\subset B.
$$

So fix $A,B$ with $|A|,|B|$ as above.
Let $C_1=A\cap B$, $C_2=A\setminus B$, $C_3=B\setminus A$, and let
$c_i=|C_i|$.
Let $\pi$ be $\pi_i$ for some $i=1,\ldots,d/2$,
i.e., a random, uniformly
chosen element of $S_n$.
For let $r$ be the number of elements of $C_1$ that $\pi$
maps to elements of $C_1$.  It then follows that $\pi$
maps $c_1-r$ values of $C_1$ to $C_3$, and $c_1-r$ values of
$C_3$ to $C_1$, $c_2$ values of $C_2$ to $C_3$, and $c_2$ values of
$C_3$ to $C_2$; if the $C_i$ and $r$ are fixed, then the probability
of this is exactly
$$
p=p(c_1,c_2,c_3,r)=\left[ \binom{c_1}{r}^2 r! \right]
\left[ \binom{c_3}{c_1-r} (c_1-r)! \right]^2 \times
$$
$$
\left[ \binom{c_3-c_1+r}{c_2} c_2! \right]^2
\left[ n(n-1)\cdots(n-2c_1-2c_2+r+1)\right]^{-1} 
$$
(for details see the proof of Theorem~12.2 in \cite{friedman_alon}).
So let $r_i$ be the number of elements that $\pi_i$ maps to $C_1$.
Let
$$
b=b(c_1,c_2,c_3,r,n) = \binom{n}{c_1,c_2,c_3,n-c_1-c_2-c_3},
$$
representing the number of choices of $C_1,C_2,C_3$ with fixed
sizes $c_1,c_2,c_3$.
Since each of $c_1,c_2,c_3$ take at most $n$ values and similarly
for each of $r_1,\ldots,r_{d/2}$
(assuming $n\ge 2$ so that $1+(n/2)\le n$), we have that 
probability of this happening is at most $n^{3+(d/2)}$ times the
maximum value of $bp^{d/2}$ over all choices of $c_i$ and $r_i$, i.e.,
$$
E(n,r,\gamma)\ \le n^{3+(d/2)} \;
\max_{c_i,r} [b(c_1,c_2,c_3,n)
\; p^{d/2}(c_1,c_2,c_3,r) ]
$$
where the maximum is over all $c_1,c_2,c_3$ yielding $A$ and $B$ of 
appropriate size, i.e.,
$$
c_1+c_2=a,\quad
c_1+c_3=a+\lfloor \gamma a\rfloor,\quad
r \le c_1,
$$
for some $a$ with
$$
1/\gamma \le  a \le n/2
$$
(there are $d/2$ values $r_1,\ldots,r_{d/2}$, but for a given
$c_1,c_2,c_3$, the function $p$ takes its maximum at some value $r$).
(We remark that in \cite{friedman_alon} the fact that there are 
more than one $r_i$ appears to be missed, although this only affects
the $E(n,r,\gamma)$ estimate by a constant power of $n$.)

At this point we write $b$ and $p$ above in terms of factorials,
and use Stirling's approximation.  Namely, we have
$$
b=\frac{n!}{c_1!\;c_2!\;c_3!\;(n-c_1-c_2-c_3)!}
$$
and
$$
p=\frac{(c_1!\;c_3!)^2\;(n-2c_1-2c_2+r)!}{
\bigl( (c_1-r)!\;(c_3-c_1-c_2+r)!\bigr)^2\; r!\;n!},
$$
and approximating $m!$ by $(m/e)^m$ changes each factorial by
a multiple of at most a constant times $\sqrt{m}$; these factors,
like the $n^5$ for $c_1,c_2,c_3,r_1,r_2$, can be absorbed into 
$n^{-s}$ by adding a constant to $s$.

The remarkable aspect this approach is that if 
one sets
$$
\nu_i = c_i/n, \quad
\delta = \lfloor \gamma n \rfloor /n,
$$
then one has
$$
\frac{\log b}{n} \quad\mbox{and}\quad 
\frac{-\log p}{n} 
$$
are both (!)
$$
= h(\nu_1,\nu_2)+
O\left(|\delta\log\delta|+\frac{\log n}{n}\right) ,
$$
where
$$
h(\nu_1,\nu_2) = -\nu_1\log\nu_1-2\nu_2\log\nu_2
-(1-\nu_1-2\nu_2)
\log(1-\nu_1-2\nu_2)
$$
(see equation~(63) in \cite{friedman_alon}).
This coincidence of $\log b$ and $-\log p$ is indicative of the
fact that with one permutation, i.e., a $d$-regular graph with $d=2$,
one never gets a spreader, but with $d\ge 4$ one does.
From there one shows that
$$
h(\nu_1,\nu_2) \ge -(\alpha/2)\log(\alpha/2),
$$
where $\alpha=a/n=(c_1+c_2)/n$ to conclude that
$bp^2$ is can be made smaller than any give power, $n{-s}$, of $n$, 
provided
that $a\ge m$ for $m=m(s)$ sufficiently large.
One concludes the theorem, for $r=2m$ and $\gamma=1/(2m)$.
See \cite{friedman_alon} for details.
\end{proof}

The above theorem is sufficient for the case where $B$ is $d$-regular
without half-loops.  

\begin{proof}[Proof of Theorem~\ref{th:spreader_power} in the general
case.]
This follows from the above methods plus a calculation in
the proof of Theorem~12.3 of \cite{friedman_alon}.
For half-loops and $n$ even, we replace $p$ as above with
\begin{equation}\label{eq:tilde-p-defined}
\tilde p=\tilde p(\{c_i\},r,n)=
\left[ \binom{c_1}{r} r\oddf \right]\;
\left[ \binom{c_3}{c_1-r} (c_1-r)! \right]\;
\end{equation}
$$
\times
\left[ \binom{c_3-c_1+r}{c_2} c_2! \right]\;
\frac{ (n-2c_1-2c_2-r)\oddf }{n\oddf},
$$
where $m\oddf$ is the odd factorial:
\begin{equation}\label{eq:odd_factorial}
m\oddf = (m-1)(m-3)\cdots 3 = \frac{m!}{2^{m/2}(m/2)!},
\end{equation}
and where Stirling's approximation implies one can replace $m\oddf$ with 
$(m/e)^{m/2}$.  The same analysis shows, again remarkably, that
$$
\frac{-\log \tilde p}{n} = (1/2) h(\nu_1,\nu_2)+
O\left(|\delta\log\delta|+\frac{\log n}{n}\right) ,
$$
where the $\nu_i$, $\delta$ are as before,
and $h$ is the exact same function (!) as before.
Hence the same calculation shows that for $2i+j\ge 3$ we have
$bp^i(\tilde p)^j$
is less than any fixed power of $n$.

For $n$ odd, there is one fixed point, which takes on a specific
value from $\{1,\ldots,n\}$ with probability $1/n$, and then
$\tilde p$ is the same, up to an additive difference of one in the $c_i$
and $r$, with $n$ replaced by $n-1$ in \eqref{eq:tilde-p-defined}.
Hence the same conclusions hold, by absorbing the $j$ $1/n$'s into the 
$s$ of $n^{-s}$.
\end{proof}

Theorem~\ref{th:spreader_power} remains true in certain modifications
of the model $\cC_n(B)$.
Let us give one example.

\begin{corollary}
\label{co:spreader_power_cycle}
Theorem~\ref{th:spreader_power} also holds for the model of random graph
where $d\ge 4$ is an even integer and
$d/2$ permutations are chosen from among those permutations whose
cyclic structure is that of a single cycle.
\end{corollary}
\begin{proof}
Every single cycle occurs in $\cS_n$ with probability $1/n$, and so
$d/2$ independent permutations are all cycles with probability $1/n^{d/2}$.
So any event occurring in $\cC_n(B)$ with probability $n^{-s}$ can
occur in the single cycle model with probability at most $n^{-s+(d/2)}$.
\end{proof}

\subsection{Spreading in $\cC_n(B)$ for non-bipartite $B$}

Now we will establish spreading for graphs in $\cC_n(B)$.
In this section we work with non-bipartite $B$.
Our spreading results easily imply
Theorem~\ref{th:tangle_spreading} for $B$ not
bipartite.

We remark that if $B$ is bipartite, then any
$G$ which admits a covering map to $B$ (or even a graph morphism to $B$)
is bipartite, and hence has $-d$ as an eigenvalue.
It follows from Theorem~\ref{th:separation}
that any bipartite graph cannot be a $\gamma$-spreader for any value of
$B$.

Our basic strategy is as follows: with notation as in
Corollary~\ref{co:weak_separation}, for a graph, $G$, let
$\Gamma^k(A)$ be the vertices connected to a vertex in $A$ by
a path of length $k$; then
$$
\Gamma^k_G(A) = \Gamma_{G[k]}(A).
$$
Now let 
$\pi\from G\to B$ is be a
covering map of degree $n$, and $A\subset \pi^{-1}(v)$ for a vertex
$v\in V_B$; if $B$ is connected, and $k$ is even, then
$$
A\subset \Gamma^k_G(A).
$$
It follows that $G[k]$ will be a spreader if we can show that for
any $A\subset V_G$, there is some $v\in V_B$ such that
\begin{equation}\label{eq:fibre_spread}
\Bigl| \Bigl( \Gamma^k_G(A_v)\cap\pi^{-1}(v)\Bigr)\setminus A_v \Bigr|
\ge \theta |A|,
\end{equation}
where $A_v=A\cap\pi^{-1}(v)$ and $\theta>0$ is a real number depending
only on $B$.

Second, if $T$ is a spanning tree in $B$, then setting $B/T$ to be the
graph where we {\em contract} $B$ along $T$, then $B/T$ has one vertex,
and a random $G\in \cC_n(B)$ gives rise to a random, regular graph
which is a cover of degree at least $3$ over $B/T$.  
Then the fact that such a random,
regular graph is a spreader will be seen to imply that
$G[k]$ is a spreader for any 
$$
k\ge 2|V_B|.
$$
Let us begin with
the first part of our strategy: we state the first part above, and
then give a slightly more useful form of this statement.


\begin{lemma}\label{le:nonbipartite_spreading}
Let $\pi\from G\to B$ be covering map of a $d$-regular graphs, with
$B$ not a bipartite graph.
Let $q\ge 2|V_B|$ be an odd
integer, and $m>0$ is an integer.  
Assume that for some $\theta>0$ the following is true:
for any $v\in V_B$, and any
$A_v\subset \pi^{-1}(v)$, we have that if $m\le |A_v|\le n/2$ then
\begin{equation}\label{eq:spreading_in_fibre}
|\Gamma^q_G(A_v)\cap \pi^{-1}(v) | \ge |A_v|(1+ \theta).
\end{equation}
Then either
\begin{enumerate}
\item $G$ has a connected component of size at most $2m|V_B|$; or
\item $G[q+1]$ is a $\gamma$-spreader for any $\gamma$ such that
\begin{equation}\label{eq:gamma_inequalities}
\gamma < \min\Bigl( 1/(m|V_B|), 1/\bigl(4|V_B|^2\bigr)  ,
\theta/(4|V_B|),\theta/(8|V_B|^2)  \Bigr) .
\end{equation}
(Of course, the fourth expression in the above $\min$ is always smaller
than the third expression, but it will be convenient to leave both
in for ease of reading.)
\end{enumerate}
\end{lemma}
\begin{proof}
By \eqref{eq:gamma_inequalities}
$$
\gamma < 1/(m|V_B|);
$$
but by Corollary~\ref{co:disconnected} this implies that
if $A\subset V_G$ has $|A|\le m|V_B|$, then either
$$
|\Gamma_G(A)|\ge |A|(1+\gamma)
$$
or else $G$ has a connected component of size at most $2m|V_B|$.

Hence to prove the theorem, it suffices to show that if
$A\subset V_G$ and $m|V_B|\le |A|\le |V_G|/2$, then
$$
|\Gamma_G(A)|\ge |A|(1+\gamma) .
$$
So consider 
$$
f(v)=|\pi^{-1}(v)|
$$
as $v$ varies over all vertices in $B$;
let $v_{\max}$ be a vertex where $f$ attains its maximum, and
$v_{\min}$ one where $f$ attains its minimum.  Since $B$ is 
connected, there exists a walk from $v_{\max}$ to $v_{\min}$,
$$
v_0=v_{\max},e_1,v_1,\ldots, v_{r-1},e_r,v_r=v_{\min},
$$
where $r+1\le |V_B|$ (if a vertex, $v$, appears twice in the walk, we can
discard the segment of the walk between the first and last
appearance of $v$).
If $r$ is odd, then by walking along a closed walk of odd length
about $v_{\min}$ we may assume
that in the above walk we have $r \le 3|V_B|$ and $r$ is even.

We claim that if there is a path of length two in $B$ from $v\in V_B$ to
$v'\in V_B$, then
$$
|\Gamma^2_G(A)\setminus A\bigr | \ge f(v)-f(v');
$$
indeed, if this path has edges $e_1e_2$, then in a permutation 
assignment $\sigma\from \Edir\to\cS_n$ we have $\sigma(e_2)\sigma(e_1)$
takes the vertex fibre of $A$ over $v$ to a set of vertices over
$v'$ of size $f(v)$.  Applying the same argument from $v'$ to $v$ shows
that
$$
|\Gamma^2_G(A)\setminus A\bigr | \ge |f(v)-f(v')| \ .
$$
Hence for $i=0,1,\ldots,r-2$ we have
$$
|\Gamma^2_G(A)\setminus A\bigr | \ge |f(v_i)-f(v_{i+2})|.
$$
For at least one value of $i=0,2,\ldots,r-2$, we must have
$$
|f(v_i)-f(v_{i+2})| \ge |f(v_{\max})-f(v_{\min})|/(r/2)
\ge |f(v_{\max})-f(v_{\min})|(2/3)(1/|V_B|).
$$
Hence
\begin{equation}\label{eq:first_spreading_bound}
|\Gamma^2_G(A)| \ge |A| (1+ \delta_A),
\end{equation}
where $\delta_A$ is defined by
\begin{equation}\label{eq:delta_A_definition}
|A| \delta_A \le  |f(v_{\max})-f(v_{\min})|/(2|V_B|).
\end{equation}
Since $G$ is $d$-regular, we have that
$$
|\Gamma^{i}_G(A)|
\le |\Gamma^{i+1}_G(A)|
$$
for any $i$.  It follows that 
$$
|\Gamma^{q+1}_G(A)| \ge |A|(1+\delta_A).
$$
This gives us a lower bound for $|\Gamma^k_G(A)|$ if $\delta$ as above
is bounded away from zero.  Let us give a bound for $|\Gamma^k_G(A)|$
when $\delta$ above is small.
Assuming that $m|V_B|\le |A| \le |V_G|/2$, we have
$$
m \le |A|/|V_B| \le n/2.
$$
Since the average value of $f(v)$ over all vertices is $|A|/|V_B|$,
we have
$$
f(v_{\min}) \le |A|/|V_B| \le f(v_{\max}).
$$
Since
$$
|f(v_{\max})-f(v_{\min})| = |A| \delta_A 2 |V_B|,
$$
for those $A$ with 
\begin{equation}\label{eq:delta_A_small}
\delta_A \le 1/(4|V_B|^2) 
\end{equation}
we have
$$
|f(v_{\max})-f(v_{\min})| \le (1/2) |A|/|V_B|,
$$
and hence
\begin{equation}\label{eq:squeeze_f}
(1/2) (|A|/|V_B|) \le f(v_{\min}) \le |A|/|V_B| \le f(v_{\max})
\le (3/2) (|A|/|V_B|) .
\end{equation}
Hence we have
$$
f(v_{\max})\le 3 f(v_{\min}).
$$
It follows that either
$$
f(v_{\min}) \ge m
$$
or 
$$
f(v_{\max}) \le 3m.
$$
Hence if $m\le n/6$, then either $v=v_{\min}$ or $v=v_{\max}$ satisfies
$$
m\le f(v) \le n/2,
$$
and
$$
f(v) \ge (1/2)(|A|/|V_B|).
$$
By \eqref{eq:spreading_in_fibre} applied to $A_v=A\cap\pi^{-1}(v)$, we have
\begin{equation}\label{eq:start_q_even}
|\Gamma^q_G(A)\cap \pi^{-1}(v) | \ge |A_v| (1+\theta) = f(v)(1+\theta).
\end{equation}
Now take any vertex, $v'$, joined to $v$ by an edge (we can take $v'=v$
if $v$ is incident upon a self-loop).  We get
$$
|\Gamma^{q+1}_G(A)\cap \pi^{-1}(v') | \ge f(v)(1+\theta).
$$
Since $q$ is odd, $A\subset\Gamma^{q+1}_G(A)$.  It follows that
$$
|(\Gamma^{q+1}_G(A)\setminus A)\cap \pi^{-1}(v') | \ge 
f(v)(1+\theta)-f(v').
$$
Hence
$$
|\Gamma^{q+1}_G(A)| =
|A| + 
|\Gamma^{q+1}_G(A)\setminus A| 
\ge |A| + f(v)(1+\theta) - f(v').
$$
Now we have
\begin{equation}\label{eq:lower_bound_outside_A}
f(v)(1+\theta)-f(v')
\le f(v)(1+\theta) - f(v_{\max})
\le f(v)\theta - \bigl(f(v_{\max})-f(v)\bigr),
\end{equation}
which by \eqref{eq:delta_A_definition} is at least 
$$
f(v)\theta - |A|\;2\;|V_B|\; \delta_A.
$$
Hence
$$
|\Gamma^{q+1}_G(A)| \ge
|A| (1-2|V_B|\;\delta_A) + 
f(v)\theta,
$$ 
which by \eqref{eq:squeeze_f} is at least
\begin{equation}\label{eq:to_be_improved_upon}
|A| (1-2|V_B|\;\delta_A) + 
(1/2)(|A|/|V_B|) \theta
=
|A| \bigl(  1 + (1/2)(\theta/|V_B|) - 2|V_B|\;\delta_A \bigr).
\end{equation}
This last expression is at least
$$
|A| (1 + \gamma),
$$
provided that
$$
\gamma \le \theta/(4|V_B|) \quad\mbox{and}\quad
2|V_B|\;\delta_A \le \theta/(4|V_B|).
$$
By \eqref{eq:gamma_inequalities} the first inequality holds; if the
second inequality doesn't hold, then
$$
\delta_A \ge \theta/(8|V_B|^2),
$$
and we can apply
\eqref{eq:first_spreading_bound}
to conclude that 
$$
|\Gamma^{q+1}_G(A)| \ge
|\Gamma^2_G(A)| \ge |A| (1+ \delta_A)
\ge |A| \bigl(1+ \theta/(8|V_B|^2)\bigr),
$$
which is at least $|A|(1+\gamma)$ by
\eqref{eq:gamma_inequalities}.
\end{proof}

The second part of our strategy uses a spanning tree in $B$ to reduce
spreading in covers of $B$ to spreading in regular graphs.
Recall that we define the Euler characteristic in graphs, $B$,
which may have
half-loops as
$$
\chi(B) = |V_B| - |\Edir_B|/2.
$$

\begin{lemma}\label{le:usually_spreads_in_fibre}
Let $B$ be a connected graph (which may have half-loops)
with negative Euler characteristic.  For any
integer $t>0$, there is a $\theta>0$ and an integer $r>0$ such that
for sufficiently large $n$ the following is true: 
with probability at least $1-n^{-t}$ we have that
a $G\in\cC_n(B)$ satisfies at least one of the following properties:
\begin{enumerate}
\item $G$ has a connected component with fewer than $r$ vertices; or
\item we have that for every $v\in V_B$, any set $A_v\in\pi^{-1}(v)$
with $m\le |A_V|\le n/2$ satisfies \eqref{eq:spreading_in_fibre}
holds any odd $q> 6|V_B|$.
\end{enumerate}
\end{lemma}

\begin{proof}
Since $B$ is connected, $B$ has a spanning tree, i.e., a subgraph 
$T\subset B$ that is a tree (so $T$ has no self-loops) containing each
vertex of $B$.

Let $B[T]$ be the contraction of $B$ along $T$, i.e., the graph with
one vertex, and whose directed edge set is $\Edir_B\setminus \Edir_T$,
with all heads and tail maps taken to the single vertex of $B[T]$, and
with the edge involution being the restriction of the involution of $B$
(so all the half-loops in $B$ occur as half-loops in $B[T]$, and all the
whole-loops in $B$ and the edges not in $\Edir_T$ occur in $B[T]$ as
whole-loops).
We have
$$
|\Edir_{B[T]}| =
|\Edir_B| - |\Edir(T)|
=
|\Edir_B|-2|V_B|+2
=
2-2\chi(B) \ge 3.
$$
Hence $B[T]$ is a bouquet of self-loops of degree $d\ge 3$.

A graph $G\in\cC_n(B)$ arises from a permutation assignment
$\sigma\from\Edir_B\to\cS_n$.
For each $e\in \Edir_T$ fix an arbitrary value
for $\sigma(e)$, and for $e\in E_{B[T]}$
we view the $\sigma(e)$ as a random variable (a permutation or involution).
We now wish to relate spreading of the random $G\in\cC_n(B)$ arising
$\sigma$ (on all of $E$) to
that of a random regular graph.

Fix a ``base point,'' $v_0\in V_B$.
For any $v\in V_B$, there is a unique non-backtracking walk in $T$ from $v_0$
to $v$,
$$
w(v) = (v_0,e_1,\ldots,e_k,v_k=v).
$$
To each $e\in E_{B[T]}$, we may view $e$ an edge in $E_B$, with head
$h_B(e)$ and tail $t_B(e)$, and we associate to $e$ the associated
walk
$$
\tilde a=\tilde a(e) = w(t_B(e)) e w(h_B(e))^{-1} 
$$
which is a closed walk from $v_0$ to itself.
Since $B$ is not bipartite, at least one of the $\tilde a(e)$ must be of
odd length; so fix an $e_0$ such that $\tilde a(e_0)$ is of odd length,
and let
$$
a(e) = \left\{ \begin{array}{ll}
\tilde a(e) & \mbox{if the length of $a(e)$ is odd, and} \\
\tilde a(e_0)\tilde a(e) & \mbox{if the length of $a(e)$ is even.} 
\end{array} \right.
$$

Let us extend $\sigma$ from a function on $\Edir_B$ to a function
on any walk in $B$ the natural way: namely, if $w$ is a walk
whose successive edges are
$$
e_1,\ldots,e_k
$$
we set
$$
\sigma(e_1,\ldots,e_k) = \sigma(e_k)\cdots\sigma(e_1),
$$
where the permutations act on the left (which is why their order is
reversed).

Our first claim is that if we fix permutations $\sigma(e)$ over all
$e\in E_T$, then the $\sigma(a(e))$, viewed as random variables over
$\sigma(e)$ for $e\in E_{B[T]}$, are independent, and are involutions or
permutations, according to whether or not $e$ is a half-loop or not.
Indeed, view the variables $\sigma(a(e))$ as being determined by first
fixing $\sigma(e)$ for $e\in E_{B[T]}$ with $\tilde a(e)$ for odd length,
and then fixing the rest of the values of $\sigma$.
If $e$ is self-loop, then $t_B(e)=h_B(e)$, so
$$
\tilde a(e) = w e w^{-1} 
$$
for some walk $w$ for which $\sigma(w)$ is a permutation (any walk in $T$
consists of edges that are not half-loops) has been determined;
then $\tilde a(e)$ is of odd length, and so $\tilde a(e)=a(e)$;
it follows that if $e$ is a half-loop, then $\sigma(a(e))$ is
a uniformly chosen involution.
Otherwise if $\tilde a(e)$ is odd, then
$$
a(e) = \tilde a(e) = w_1 e w_2^{-1},
$$
where the permutations $\sigma(w_1)$ and $\sigma(w_2)$ have been fixed,
and so $a(e)$ is uniformly chosen among all permutations.
Furthermore, the $\sigma(a(e))$ for $\tilde a(e)$ odd is
determined by $\sigma(e)$, and all these $\sigma(e)$ are independent.
If we fix all such $\sigma(e)$, then the remaining $a(e)$ are of the
form
$$
a(e) =  \tilde a(e_0) w_1 e w_2^{-1},
$$
where $\sigma$ has been determined on $\tilde a(e_0)$ and on
$w_1$ and $w_2$; hence $\sigma(a(e))$ is a random permutation
(all $e$ that are self-loops have $\tilde a(e)$ of odd length)
depending only on $\sigma(e)$,
and hence the remaining $\sigma(a(e))$ are independent.
It follows that for fixed $\sigma(e)$ with $e\in E_T$, each
$\sigma(a(e))$ takes on each permutation or involution with
the same probability, and hence the $\sigma(a(e))$ are independent.

It follows that any $v_0\in V_B$ and any permutation assignment
$\sigma\from E_B\to \cS_n$ we associate a $d$-regular graph,
$\widetilde G(v_0)$, with $d\ge 3$, according to the distribution
$\cC_n(B[T])$ 
It follows 
from Theorem~\ref{th:spreader_power} that with $t$ fixed, there is
an $m$ and a $\theta>0$ such that for sufficiently large $n$ 
we have with probability at least
$1-n^{-t}$ that $\widetilde G$ has a connected component of no more than
$2m$ vertices, or is a $\theta$-spreader.
From the union bound it follows that this condition holds for all 
$\widetilde G(v_0)$, ranging over all $v_0\in V_B$, with probability
at least $1-|V_B|n^{-t}$.

If for any $v_0$ we have that $\widetilde G(v_0)$ has a connected component
on a set of vertices, $V'$, then it follows that the vertex subset
of $G\in \cC_n(B)$ whose $v$ fibre is $\sigma(w(v))V'$
is a connected component of $G$, of size $|V_B|\;|V'|$.
Hence, with probability $1-|V_B|n^{-t}$ (for sufficiently large $n$)
either
(1) $G\in \cC_n(B)$ has a connected component with at most
$|V_B|\;2m$ vertices,
(2) each graph $\widetilde G(v_0)$ is a $\theta$-spreader.
So assume condition (2) holds.

If $A_v\subset \pi^{-1}(v)$ with $r\le |A_v| \le n/2$, then
we have
$$
U=\bigcup_{e\in E_{B[T]}} \sigma(a(e)) A_v
$$
is of size at least $|A_v|(1+\theta)$, and
lies in $\pi^{-1}(v)$ and in $\Gamma^q_G(A)$ for
any odd $q$ larger than the longest length among the walks 
$a(e)$; since this length is at 
most that of $\tilde a(e_0)$ plus that of $\tilde a(e)$,
this length is at most $6|V_B|$.  Hence if $q$ is any odd number
greater than $6|V_B|$, we have
\eqref{eq:spreading_in_fibre} is satisfied.
\end{proof}

\begin{proof}[Proof of Theorem~\ref{th:tangle_spreading} for $B$ not
bipartite.]
Immediate from Lemmas~\ref{le:nonbipartite_spreading}
and~\ref{le:usually_spreads_in_fibre} and
Corollary~\ref{co:weak_separation}.
\end{proof}

\subsection{Spreading in $\cC_n(B)$ for bipartite $B$}

In this subsection we briefly describe the modifications needed
to the proof in the last subsection for $B$ bipartite.

Of course, if $B$ is bipartite then any $G\in\cC_n(B)$ is bipartite,
and therefore $-d$ occurs as an eigenvalue.  Hence, by
Corollary~\ref{co:weak_separation}, there is no $k$ for which
$G[k]$ is a $\gamma$-separator for any $\gamma>0$.

Let $B$ be a connected, bipartite graph, which therefore has an
essentially
unique partition $V_B=V_1\amalg V_2$ (unique up to exchanging
$V_1$ with $V_2$) of its vertices so that $E_B$ has no self-loops,
and each directed edge of $B$ has a head or tail in each of $V_1$
and $V_2$.
In this case $B[2]$ has exactly two connected components,
$B_1$ and $B_2$, which are the subgraphs induced on the edge
sets $V_1$ and $V_2$ respectively.
For any $G\in\cC_n(G)$ with $\pi\from G\to B$ the covering map,
we have that $G[2]$ has naturally divides into two $d$-regular graphs,
$G_i=\pi^{-1}(B_i)$ for $i=1,2$.

The same arguments as in the previous subsection can be applied to
a subset $A\subset V_{G_1}$ with 
$m\le |A|\le |V_{G_1}|/2$ with the following modifications:
\begin{enumerate}
\item
In Lemma~\ref{le:usually_spreads_in_fibre} we replace $q$ odd with $q$ even;
all the $\tilde a(e)$ are of even length, and we take
$a(e)=\tilde a(e)$ for all $e\in E_{B[T]}$.
\item
We claim that
Lemma~\ref{le:nonbipartite_spreading} holds with $q$ replaced by
an even integer (and one can replace the second claim about
$G[q+1]$ with the (stronger) claim regarding $G[q]$);
indeed, we follow the exact same proof until
\eqref{eq:start_q_even}; then we note that $A\subset\Gamma^q_G(A)$
since $q$ is even, and hence
$$
|\Gamma^q_G(A)\cap \pi^{-1}(v) | \ge |A_v| (1+\theta) = f(v)(1+\theta)
$$
implies that
$$
|(\Gamma^q_G(A)\cap \pi^{-1}(v))\setminus A | 
\ge |A_v| \theta = f(v)\theta
$$
since $A\cap\pi^{-1}(v)=A_v$.  From there we have
$$
|\Gamma^q_G(A)| =
|A| + 
|\Gamma^q_G(A)\setminus A| 
\ge |A| + f(v)\theta ,
$$
which is the same estimate as in
\eqref{eq:lower_bound_outside_A}
except that the $f(v_{\max})-f(v)$ doesn't appear, meaning that we have
$$
|\Gamma^q_G(A)| \ge 
|A| \bigl(  1 + (1/2)(\theta/|V_B|) \bigr)
$$
which is an improvement over \eqref{eq:to_be_improved_upon}.
Hence the estimates which suffice to prove that this quantity is
at least $|A|(1+\gamma)$ in Lemma~\ref{le:nonbipartite_spreading}
for $q$ odd, must also hold here, for $|\Gamma^q(A)|$ and $q$ even.
\end{enumerate}

\begin{proof}[Proof of Theorem~\ref{th:tangle_spreading} for $B$
bipartite.]
Let $B$ be a connected, bipartite graph; let $V_1$ and $V_2$ be a
(the essentially unique)
bipartition of $B$'s vertices; for each $G\in\cC_n(B)$ and $i=1,2$,
let $G_i$ be the subgraph of $G[2]$ induced from those vertices
of $G$ lying over $V_i$.
From the above modified versions of 
Lemmas~\ref{le:nonbipartite_spreading}
and~\ref{le:usually_spreads_in_fibre}, 
we have that for any $t$ there are integers $r,q$ and $\gamma>0$ with $q$
even such
that the following holds for $n$ sufficiently large: for 
$G\in\cC_n(B)$,
we have that with probability at least
$1-|V_B|n^{-t}$ that $G_1[q/2]$ and $G_2[q/2]$ are both 
$\gamma$-spreaders or else $G$ has a connected component with fewer
than $r$ vertices.  
Now we apply 
Corollary~\ref{co:weak_separation}.
\end{proof}

\section{The Fundamental Order and Ramanujan Bases}
\label{se:p2-fund-exp}

If $B$ is $d$-regular and Ramanujan we can give upper and lower bounds on
$$
\prob{G\in \cC_n(B)}{\rhonew_B(A_G) \ge \rho(A_{\widehat B})+\epsilon},
$$
that are optimal to within a multiplicative constant.
Note that in \cite{friedman_alon}, the upper and lower bounds differed
by a factor of $n$ is certain ``exceptional'' cases, namely
for the Broder-Shamir over the base $W_{d/2}$, for even values of $d$
for which $\sqrt{d-1}$ is an odd integer (e.g., $d=10$).
Hence this result gives an improvement on \cite{friedman_alon}, even
just in the case of $d$-regular graphs, for some values of $d$.

Recall the definition of $\etafund(B)$, the 
\gls{fundamental order},
of Definition~\ref{de:fundamental_order_defined}, namely
$$
\etafund(B) = \min\{ \ord(L) \ |\ 
\rho(H_L) > \rhoroot B  \} ,
$$
i.e., the smallest order of a 
\gls{strict tangle of B}.

By the discussion of \eqref{eq:rho_comparison}, we see that
if $\psi$ is (isomorphic to) a subgraph of $G$, then
$$
\rho(H_G) \ge \rho(H_\psi) = \rhoroot B + \epsilon_0
$$
for some $\epsilon_0>0$.
Hence, by Theorem~\ref{thm:prob_tangle}
we conclude the following simple observation.

\begin{proposition} Let $B$ be a connected graph with no half-loops.
Then there is an $\epsilon_0=\epsilon_0(B)>0$ and a $C=C(B)>0$ for which
$$
\Prob_{G\in \cC_n(B)}\{ \rho(H_G) \ge \rhoroot B +\epsilon_0  \} 
\ge  C n^{-\etafund(B)} .
$$
\end{proposition}

To get a matching upper bound, to within a constant, we shall prove
the following theorem.

\begin{theorem} Let $B$ be a connected graph with no half-loops,
and assume $B$ is $d$-regular for some $d\ge 3$.
Then 
for every $\epsilon>0$ with 
$$
\epsilon<(d-1)-(d-1)^{1/2}
$$
there is a $C=C(\epsilon)$ for which
$$
\Prob_{G\in \cC_n(B)}\{ \rho(H_G) \ge \rhoroot B +\epsilon  \}
\le  C(\epsilon) n^{-\etafund(B)} .
$$
\end{theorem}

\begin{proof}
We follow the proof of Theorem~\ref{th:main} in
Section~\ref{se:p1-main-proof}, except that we will replace the set of
$B$-tangles, $\Tangle_B$, with the set of $(B,\epsilon)$-tangles.

So fix an $\epsilon>0$ and an integer $r>0$;
$\tanglefreeeps$ be the set of graphs which contain no subgraphs
in $\Tangle_{<r,B,\epsilon}$, and let
$\tanglefreeindicatoreps$ be its indicator function.
Let $p_i$, $P_i$, and $\widetilde P_i$, respectively, be as in 
\eqref{eq:tangle_free_only}, 
\eqref{eq:power_series_proof}, and
\eqref{eq:widetilde_P_defined}, respectively, with 
$\tanglefreeindicator$ replaced 
with $\tanglefreeindicatoreps$.
Now we consider the abstract partial trace $(\tau_0,\tau_1,C',L,r)$
as described after these
equations, except that $\tau_0$ is specifically $\rho(H_B)+\epsilon$
and that we replace the choice of $r$
in \eqref{eq:how_big_is_r} we choose
$$
r >  \max\bigl( H(1,L,\tau_1,\tau_0,\varepsilon), r_0\bigr)
$$
where $r_0$ is the value in 
Theorem~\ref{th:tangle_spreading} with $j$ in the theorem taken
to be $\taufund(B)$.

Now it suffices to show that
$$
\widetilde P_1, \widetilde P_2,\ldots,\widetilde P_{\taufund-1}
$$
have vanishing principle parts.

On the contrary, assume that $\widetilde P_j$ does not vanish for 
some $j<\taufund$, and consider the minimum value of such a $j$.
In this case, Lemma~\ref{le:side-step} implies that
$$
\widetilde P_j(k) = (1-d)^k p_{1-d,j} + (d-1)^k p_{d-1,j}
$$
for some constants $p_{1-d,j},p_{d-1,j}$, at least one of which is
positive.

But by
Theorem~\ref{th:tangle_spreading} as applied above, there is an
$\epsilon>0$ such that for sufficiently large $n$ we have the
following: with probability
at least $1-n^{-\taufund(B)}$ we have that a $G\in \cC_n(B)$ either
has a connected component of size less than $r_0<r$ above,
or else we have
$$
|\lambda_i(G)|< d-\epsilon
$$
for all $i\ne 1$ or all $i\ne 1,|V_G|$, according to whether or not
$B$ is bipartite.
If $G$ has a connected component on fewer than $r$ vertices, then
this connected component is a tangle (assuming 
$(d-1)^{1/2}+\varepsilon< d$).  Otherwise we have
$$
|\lambda_i(G)|< d-\epsilon
$$
for appropriate $i$ described above,
with probability asymptotically less than $n^{-j}$.
In this case we have the eigenvalues of $H_B$, excepting one eigenvalue
of $d-1$, and one eigenvalue $1-d$ if $B$ is bipartite.
Hence $p_1=\cdots=p_j=0$ and
$$
\Prob_n[ {\rm AbsoluteException}_n(\varepsilon) ] \leq 
C_\varepsilon n^{-\taufund(B)} .
$$
This 
contradicts the fact that one of
$p_{1-d,j}$ or $p_{d-1,j}$ must be positive.
\end{proof}

In the next subsection we wish to make some remarks on the
fundamental order of a graph, $B$.  The most important of these remarks
is that for a given $B$ there is a finite algorithm to compute
$\taufund(B)$.

\subsection{Computing The Fundamental Order}

We begin by giving a finite algorithm to determine the fundamental
order, $\etafund(B)$, of a graph $B$.

\begin{proposition}
For any connected graph, $B$, with no half-loops, for which
$\ord(B)\ge 1$, there is a finite procedure 
to determine $\etafund(B)$, and $\etafund(B)\ge 1$.
In other words, there is a Turing machine which halts on every input,
and when input the description of a graph, $B$, with no-half loops,
outputs $\etafund(B)$.
\end{proposition}
\begin{proof}
Since $\rho(H_B)>\rhoroot B $, and since
the identity
map $B\to B$ is \'etale, we have that 
$\etafund(B)$ is at most $\ord(B)-1$.

If $L\to B$ is \'etale, consider the {\em type}, $T=T(L)$ of $L$, meaning
(since $L$ is a graph, not a walk) the information consisting of the
graph one obtains by contracting all the beaded paths of
$L$.  It suffices to
consider, for all (the finitely many)
types of underlying graph, $T$, with $\ord(T)<\ord(B)$,
whether there is an \'etale $L\to B$ of type $T$ and with
$\rho(H_L)>\rhoroot B $.
For any $L\to B$ of type $T$, let $\vec k=\vec k(L)$ be, as usual,
be the vector indexed on $E_T$ giving the length of the beaded path
in $L$ corresponding to the edge $e\in E_T$.
Of course, $L$ is isomorphic to the graph $\VLG(T,\vec k)$.

We claim that if there is \'etale $L\to B$ with $L=\VLG(T,\vec k)$
and $\rho(H_L)>\rhoroot B $, then there is another such
graph, $L'=\VLG(T,\vec k)$, with an \'etale map to $B$, for which
each component, $k(e)$, of $\vec k$, satisfies $k(e)\le |V_B|$.
Indeed, let $\pi\from L\to B$ be the \'etale map.
If along some beaded path of $L$, if the image via $\pi$ in $B$ of the
path encounters some vertex, $v$,
twice, we may delete the part of the beaded path between any two
occurrences of $v$.  Hence, by repeated deletions, we obtain an \'etale
$\pi'\from L'\to G$, 
isomorphic to $\VLG(T,\vec k')$ with $\vec k'\le\vec k$, and
with each beaded path having no two occurrences of a $V_B$ vertex under
$\pi'$.  Hence $\vec k'(e)\le |V_B|$, and
$$
\rho(H_{L'}) =
\rho(H_{\VLG(T,\vec k')}) \ge
\rho(H_{\VLG(T,\vec k)}) > \rhoroot B .
$$

Hence it suffices to consider for a finite number of graphs, $T$,
all possible morphisms $\VLG(T,\vec k)\to B$, for a finite number
of vectors, $\vec k$, and to examine which morphisms are \'etale,
and what the values of $\rho(H_{\VLG(T,\vec k)})$.

The only technical point of this algorithm is that we have to be
able to determine whether or not $\rho(H_{\VLG(T,\vec k)})$ is strictly
greater
than $\rhoroot B$.
But 
$$
\rho(H_{\VLG(T,\vec k)}) - \rhoroot B 
$$
is an algebraic integer for which we can obtain bounds on the degree
and coefficients of its minimal equation.  This gives a positive lower
bound on its value if its value is positive, and we can use a
standard algorithm to approximate the Perron-Frobenius eigenvalue
of a matrix with non-negative entries to test the positivity.
\end{proof}

We remark that, in practice, one can often give simpler calculations
to determine $\etafund(B)$, such as $B=W_{d/2}$ for an even positive
integer, $d$.  For example, the arguments
in \cite{friedman_alon} show that for any positive integer
$m\le d/2$,
the smallest value of $\rho(H_L)$ among those graphs of order $m$
that admit an \'etale map to $W_{d/2}$ is attained for $L=W_m$.
(This is easy.)
It follows that $\etafund(W_{d/2})$ is $m-1$ where $m$ is the smallest
positive integer for which
$2m-1>\sqrt{d-1}$.

We note that similar remarks are valid for many models of random covering
graphs of degree $n$ over a fixed base graph, $B$.
In \cite{friedman_alon}, $\etafund(B)$, is computed in a number of
such models,
although the calculation is a bit more involved.  One interesting
remark is that if we use a model where each permutation is restricted
to a cycle of length $n$ (the model called ${\mathcal H}_{n,d}$ 
in \cite{friedman_alon}), then self-loops are impossible, and 
the minimal $\rho(H_L)$ for a graph, $L$, of order $m$ (with $m\le d-1$)
which admits an \'etale map to $W_{d/2}$, is a graph with two vertices
joined by $m+1$ edges.  This gives an $\etafund(B)$ which is roughly
twice that of the Broder-Shamir model.
In particular, two simple and natural models can have
very different $\etafund(B)$.

\subsection{Lower Bound on the Fundamental Order of a $d$-Regular Graph}

It is interesting to note that the fundamental order of a $d$-regular
graph is always greater than roughly $\sqrt{d}$.  This means that
as $d$ gets large, Theorem~\ref{th:main_Alon_Ramanujan} gives progressively
sharper bounds on
the probability of a graph in $\cC_n(B)$ does
not satisfying the bound in the Relativized
Alon Conjecture for any $d$-regular graph, $B$.

\begin{theorem}
\label{th:fund_lower_bound}
Let $B$ be any $d$-regular, connected graph, possibly
with half-loops.  Then
$$
\etafund(B) > \sqrt{d-1};
$$
furthermore, this bound is tight when $B$ is an appropriate
bouquet of half-loops.
\end{theorem}
Our proof follows that in Chapter~6, Section~3 of \cite{friedman_alon}.
We briefly recall the proof.
\begin{proof}
First we show that if 
$\psi$ is a connected graph with at least two vertices, then 
there exists
a graph, $\psi'$, with one vertex fewer than $\psi$ such that
$\psi'$ and $\psi$ have the same order, but
$$
\rho(H_\psi') \ge \rho(H_\psi).
$$
(This is Lemma~6.7 in \cite{friedman_alon}.)
Indeed, let $e$ be a directed edge of $\psi$ whose head and tail are
distinct.
Let $\psi_e$ be the graph with $e$ discarded and the head and tail of
$e$ identified (i.e., discard $t_\psi e$, and redefine the heads and
tails maps in $\psi_e$ so that any edge with head or tail equal to
$t_\psi e$ now has it equal to $h_\psi e$.).
Then $\psi_e$ has the same Euler characteristic as $\psi$.
However, to any strictly non-backtracking closed walk, $c$, in $\psi$,
if we discard all appearances of $e$ we get
a strictly non-backtracking closed walk, $c'$, in $\psi_e$ of the same
length or less;
furthermore this association is injective, since from the 
associated non-backtracking closed walk, $c'$, in $\psi_e$ we can
infer when $e$ was traversed.
Hence the number of strictly non-backtracking closed walks of length at
most $k$
in $\psi$ is at most the same in $\psi_e$, and hence
$$
\rho(H_{\psi_e}) \ge \rho(H_\psi).
$$

It follows that the largest $\rho(H_\psi)$ over all graphs of order $s$
is attained by some graph, $\psi$, with one vertex and $s+1$ edges.
But any such graph has $\rho(H_\psi)=s$.
Hence, if
$$
\rho(H_\psi) > (d-1)^{1/2},
$$
then the order of $\psi$ satisfies
$$
s > (d-1)^{1/2}.
$$
Furthermore, if $s$ is the smallest integer greater than $(d-1)^{1/2}$,
and $\psi$ is the bouquet of $s+1$ half-loops, then this bound is attained.
\end{proof}

\section{Algebraic Models} 
\label{se:p2-algebraic}

In this section we describe a number of variants of the
Broder-Shamir model to which all our theorems.
We shall not try to give an ``all encompassing'' definition
of such models; rather we explain that all these models
have an ``algebraic'' feature which allows us to express
the certified trace (and related traces) in a $(1/n)$-power series
expansion.

At this point we claim that the theorems in Chapter~\ref{ch:p1} hold
when $B$ has half-loops in the Broder-Shamir model of
Definition~\ref{de:broder-shamir}, as well as numerous related 
models.  Rather than characterizing a large class of such models, 
which would probably be rather awkward, we content ourselves to
give some examples that illustrate the diversity of possible models.

We also note that once the Alon conjecture is established for one
model of a random covering of degree $n$ of a graph, $B$, then it
is automatically established for any other model that is {\em contiguous}
with the model for which the conjecture is proven
However, despite a large body of knowledge on contiguity results for
models of a random $d$-regular graph on a large number of vertices
(see, for example, \cite{friedman_alon}, 
the discussion just after Theorem~1.3),
there appears to be much less known about random coverings;
see \cite{greenhill_lifts} for some work in this direction.
We thank Nick Wormald for these remarks and
discussions regarding contiguity of
random covering maps.

\subsection{Theorems~\ref{th:main_Alon} and \ref{th:main_Alon_Ramanujan}
for General Base Graphs}
\label{sb:general_base_graphs}

Here we indicate the modifications needed to prove
Theorems~\ref{th:main_Alon} and \ref{th:main_Alon_Ramanujan} for
$d$-regular graphs, $B$, which may have half-loops.
Since Theorem~\ref{th:tangle_spreading} was proven in the case where
$B$ may have half-loops, it suffices to
prove Theorem~\ref{th:main} for graphs with half-loops.

Let us discuss what theorems can be easily modified for $B$ with
half-loops.  The case where $n$ is even is simplest, where the half-loops
in $B$ yield involutions with no fixed points.  Let us begin with this case.

\subsubsection{$\cC_n(B)$ for even $n$}

If $e\in E_B$ is a half-loop, then in a permutation assignment
$\sigma\from E_B\to\cS_n$, $\sigma(e)$ is an involution with no
fixed points, according to our definition of the Broder-Shamir model.
In this case if $\sigma(e)(i)=j$ then $\sigma(e)(j)=i$; in this case
when we fix values of $\sigma(e)$ we fix $k$ values for $k$ even, and
these $k$ values occur with probability
$$
\frac{(n-k)\oddf}{n\oddf} = (n-1)(n-3)\ldots(n-2k+1),
$$
where $\oddf$ denotes the odd factorial
of \eqref{eq:odd_factorial}.

The odd factorial is the essential modification.  
We remark that since $n$ is even, all the graphs occurring in $\cC_n(B)$
have no half-loops; this slightly simplifies matters.
It is important to note that a non-backtracking walk in a graph with
half-loops is not allowed to traverse a half-loop twice; this is
only relevant to walks in $B$, 
since $G\in\cC_n(B)$ do not have half-loops, and to traverse a 
half-loop twice in $B$ corresponds to taking an edge and its inverse
in $G$ (which means that such a walk in $G$ would not be non-backtracking
in our usual definition of non-backtracking for walks in graphs without
half-loops).
Let us indicate
which parts of the proof of Theorem~\ref{th:main} need modification:

\begin{enumerate}
\item Section~\ref{se:p0-loop}: Theorem~\ref{th:loop_count} goes
through with similar estimates, with
$$
\Bigl( n (n-1) \ldots (n-k'+1) \Bigr)^{-1}
$$
replaced with
$$
\Bigl( (n-1) (n-3) \ldots (n-k'+1) \Bigr)^{-1}
$$
for $k'$ odd.
Coincidences are defined exactly in the same way.
Lemma~\ref{le:coincidences} also holds; the only modification is that
a walk of length $k$
can fix up to $2k$ values of an involution;  hence we want 
$1/(n-2k)$ to be of order $1/n$, which requires us to restrict
$k$ to be, say, at most $n/3$ instead of $n/2$.
\item Subsection~\ref{sb:order_and_pruned} requires the following
modifications: of course, as mentioned in just above
Lemma~\ref{le:usually_spreads_in_fibre}, we set
$$
\chi(B) = |V_B| - |\Edir_B|/2
$$
and we define the degree of a vertex, $v\in V_B$, 
as the number of
directed edges whose heads are $v$; hence a half-loop contributes
one to the degree, and whole-loops contribute two, and the degree
is always an integer.
We still define $B$ to be pruned if each vertex has degree two, and
work only with pruned graphs.
Note that if $B$ has more than one vertex and is connected, 
then any vertex, $v$, with a half-loop is
of degree at least two, since the half-loop contributes one to 
the degree of $v$,
and $v$ must have an edge connecting $v$ to a different vertex of $B$.
\item Section~\ref{se:p1-walk-sums} goes through until
Proposition~\ref{prop:EsymmProd}, where we remark that, as remarked
there, one replaces factorials with odd factorials.
We still get expansion polynomials, albeit different polynomials for
the half-loops, as in Definition~\ref{de:expansion_polynomials}.
Types and forms are defined in the same way; since $n$ is even
the graphs of $\cC_n(B)$ have no half-loops.
\item To Section~\ref{pr:gen_EsymmProd}:
Proposition~\ref{pr:gen_EsymmProd} requires odd factorials for the
half-loops.
\item The side-stepping lemma, Lemma~\ref{le:side-step}, is used
as is, along with the loop calculation, namely 
Theorem~\ref{th:broder_shamir_friedman}, to prove
Theorem~\ref{th:main}.
\end{enumerate}

\subsubsection{$\cC_n(B)$ for odd $n$}

For $n$ odd we have that any half-loop, $e\in E_B$, a permutation
assignment, $\sigma\from E_B\to \cS_n$, for our definition of the
Broder-Shamir model, requires
$\sigma(e)$ to be an involution
with exactly one fixed point.
It is best to view this as two pieces of information, (1) which
$\{1,\ldots,n\}$ is fixed, and (2) the values of such a $\sigma(e)$
on the $n-1$ remaining values.
To give the fixed point of such a $\sigma(e)$ is a probability
$1/n$ event, and to fix $k$ values of the remaining $n-1$ values of 
$\sigma(e)$ can be done, given the fixed point, in
\begin{equation}\label{eq:fix_point_fixed}
(n-2)(n-4)\ldots(n-2k)
\end{equation}
ways.  If we do not condition upon the fixed point of $\sigma(e)$,
then fixing $k$ of the values of $\sigma(e)$ not involving the fixed
point can be done in
\begin{equation}\label{eq:fix_point_not_fixed}
(n-1)(n-3)\ldots(n-2k+1)
\end{equation}
ways.
Therefore a graph in $\cC_n(B)$ has exactly one half-loop for each
half-loop in $E_B$.
Then the type should remember all half-loops it traverses
(so that we only delete vertices of degree two of the type which are
not the starting vertex and are not half-loops).
A non-backtracking walk in $B$ and $G$ cannot traverse a half-loop
twice (however, a non-backtracking walk in $B$ and $G$ can
traverse an edge, $e$, then a half-loop, $e'$, and then $e^{-1}$);
it follows that in our VLG's, the half-loop will always have length
one and always remains unchanged.
For this reason the type should remember the half-loops encountered
in a walk.

It follows that all expansions for types will involve either
\eqref{eq:fix_point_not_fixed} or
\eqref{eq:fix_point_fixed}, according to whether or not the walk
traverses the fixed point of $\sigma(e)$.
The same is true of $\Omega$-types, according to whether or not
the walk traverses the fixed point or the $B$-graph $\Omega$ includes
the fixed point.

\subsection{Some Examples of Algebraic Models}

In this subsection we give examples of other ``algebraic models''
of random coverings of $B$ of degree $n$ to which 
Theorems~\ref{th:main_Alon} and \ref{th:main_Alon_Ramanujan}
hold.
Roughly speaking, this should hold of any model of
random permutation assignments, $\sigma\from V_B\to \cS_n$,
such that the $\sigma(e)$ are independent (modulo the fact
that they respect the involution of $V_B$), and where we 
have power series in $1/n$ to describe the probabilities that
certain values of $\sigma(e)$ are fixed.
We remark that we could allow for some dependence between
the $\sigma(e)$, but only in a fairly simple (and algebraic)
way.  Rather than classify a large number of examples,
which does not seem that important at present,
we will suffice to give a few examples.

One interesting example is when we modify $\cC_n(B)$ so that
edges, $e\in E_B$, that are not half-loops, have $\sigma(e)$
only in the permutations whose cyclic structure is a single
cycle of length $n$.
As shown in \cite{friedman_alon}, at least for the bouquet of
whole-loops, this decreases the fundamental exponent of $B$ by
roughly a factor of two.
The reason is that this model does not allow for half-loops,
and, at least for $B$ being the bouquet of $d/2$ whole loops,
the smallest tangles have two vertices rather than one (in the
case of $\sigma(e)$ being an arbitrary permutation).

A similar modification would be where we specify a given cyclic
structure of $\sigma(e)$ as a finite union of cycles, each of
whose length is either constant or, say, a linear function of $n$.
In this case we consider, for types and $\Omega$-types, each
cycle on its own.  For example, we could insist that certain
of the $\sigma(e)$ would consist of a cycle of length $3$, one
of length $4$, and two cycles of length $(n-7)/2$, assuming that
$n$ is odd.  
We see no reason to be interested in such a model.
We remark, however, that our model of $\cC_n(B)$, for $n$ odd and 
$B$ containing
half-loops, does require special values of $\sigma(e)$, namely
an odd number of fixed points.

One related model that may be of interest is that we could fix
the values of some of the $\sigma(e)$; i.e., we could insist that
$\sigma(e)$ for one $e$ takes $1$ to $5$ and $7$ to $9$, and 
something else for some other other values of $\sigma$;
any such model, provided that we can write power series for the
various types, would work.

We also point out that the $\sigma(e)$ can have dependence,
but only in a way that yields algebraic coefficients for the 
expansion polynomials.
For example, we could take two different edges, $e,e'\in E_B$,
not inverses of each other and not half-loops, and insist
that, say, for one value of $i\in\{1,\ldots,n\}$ we have
$\sigma(e)(i)=\sigma(e')(i)$, and otherwise choose
the rest of $\sigma(e)$ and $\sigma(e')$ independently.  This makes 
these two $\sigma$ values dependent, but only on a mild way.
Again, in the types and $\Omega$-types we would keep track of this
special value of $i$ and $\sigma(e)(i)=\sigma(e')(i)$, provided
that it occurs on the walk or the $B$-graph $\Omega$.
And again, we don't see any particular application of such a model
at present, but this does point out that---strictly speaking---the
values of $\sigma(e)$ for $e$ that are not inverses of one another
can have some weak dependence.

\section{Mod-$S$ Functions}
\label{se:p2-modl}


In this subsection we give a refined notion of polyexponential
functions that, at least in principle, may give more detailed
information about trace methods that could be used, say,
to test conjectures about finer aspects of the trace method.

Recall 
from Example~\ref{ex:not_polyexponential}
the weighted convolution example of
$g_1(k)=g_2(k)=(d-1)^k$, with
$$
(g_1 * g_2 )_{1,2}(k) = 
\left\{ \begin{array}{ll} \bigl( (d-1)^{k+1}-(d-1)^{k/2} \bigr)/(d-2)
& \mbox{if $k$ is even,} \\ 
\bigl( (d-1)^{k+1}-(d-1)^{(k-1)/2} \bigr)/(d-2) 
& \mbox{if $k$ is odd.}  \end{array}
\right.
$$
It follows that this weighted convolution is not exactly a polyexponential.
In Chapter~\ref{ch:p1}, specifically
Theorem~\ref{th:weighted-convolution-is-Ramanujan}, we pointed
out that such convolutions are \glspl{B-Ramanujan function}.
In this section we take an alternate point of view: namely, 
$(g_1 \ast g_2)_{1,2}(k)$ as above has an exact formula, provided that
we are willing to write one formula for $k$ even, and another for
$k$ odd.

In this section we show that,
more generally, a weighted convolution will have an exact formula provided
that we divide its argument by
congruence classes modulo $S$,
where $S$ is the least common multiple of the weights;
we call such functions {\em mod-$S$ polyexponentials}.
Furthermore, we indicate why such formulas may be of interest in
future research.

In the next subsection we prove the basic
facts about mod-$S$ polyexponentials.  
In the subsection thereafter, we explain our interest in such functions
and their formulas.

\subsection{Mod-$S$ Polyexponentials}

\begin{definition}
Let $S$ be a positive integer.
We say
that a function $g=g(k)$ defined on the non-negative
integers is {\em mod-$S$ polyexponential of base $\ell$} if
there are polyexponential functions, $p_0,\ldots,p_{S-1}$ and
an integer $K$ such that for any $i=0,\ldots,S-1$ we have
$$
g(k) = p_i\bigl( (k-i)/S \bigr) \quad
\mbox{if $k\ge K$ and $k\equiv i\pmod{S}$,}
$$
where each $p_i(k)$ is a polynomial in $k$ times $\ell^k$.
By a {\em mod-$S$ polyexponential} we mean any finite sum of
mod-$S$ polyexponentials,
and by the {\em bases} of this sum we mean the set of bases
involved in the sum.
\end{definition}

Clearly for any integer, $s$, any mod-$S$ polyexponential of base $\ell$ 
is a mod-$sS$ polyexponential of base $\ell^s$.

The main goal of this section is to prove the following theorem.

\begin{theorem}
\label{th:polyexp-m-conv}
Let $g_1,\ldots,g_t$ be polyexponential functions with bases $L$,
let $m_1,\ldots,m_t$ be positive integers, and let
$S$ be the least common multiple of $m_1,\ldots,m_t$.
Then
$$
(g_1 \ast g_2 \ast \cdots \ast g_t)_{\vec m}(k)
$$
is a mod-$S$ polyexponential with bases the union over $i=1,\ldots,t$
of
$$
L^{S/m_i} = \{ \ell^{S/m_i} \ | \  \ell\in L \}.
$$
\end{theorem}
\begin{proof}
The proof of this theorem is based on the following fundamental lemma.

\begin{lemma}
For $i=1,2$, let $g_i$ be a mod-$S$ polyexponential of
base $\ell_i$.  
Then
$$
(g_1 \ast g_2)(k)
$$
is a sum of mod-$S$ polyexponentials of
bases $\ell_1$ and $\ell_2$.
\end{lemma}
\begin{proof}
By the linearity of the convolution operator, it suffices to prove this
theorem under the assumption that for some $i_1,i_2$, we have
$$
g_j(k) = 
\left\{ 
\begin{array}{ll}  p_j\bigl(  (k-i_j)/S \bigr)  & 
\mbox{if $k\equiv i_j\pmod{S}$, and} \\
0 & \mbox{otherwise.}
\end{array}\right. 
$$
It is easy to see that 
\begin{enumerate}
\item we may assume that $K=0$, since modifying $g_1$ at a single value
$k_0$
modifies the convolution by a term
of the form $g_2(k-k_0)$, which is a
mod-$S$ polynomial with base $\ell_2$; similarly for any finite
number of modified $g_1$ values, and similarly for modifying any finite
number of $g_2$ values;
\item 
it suffices to compute the convolution for $k$ divisible by $S$
(the other cases of $k$ modulo $S$ are similar).
\item we may assume that $i_1=i_2=0$ (the other cases are similar);
\end{enumerate}
So assume $k$ is divisible by $S$, that $K=0$, and that 
$i_1=i_2=0$.  
Then 
$$
(g_1 \ast g_2)(k)
=
\sum_{t=0}^{k/S} g_1(tS) g_2(k-tS)
=
\sum_{t=0}^{k/S} 
p_1(t) p_2\bigl((k/S)-t\bigr)
$$
$$
=(p_1 \ast p_2)(k/S).
$$
Hence, for $k$ divisible by $S$,
Theorem~\ref{thm:conv_polyexponential} implies that this is a sum
of a polyexponential of $k/S$ in the bases $\ell_1,\ell_2$.
Hence, considering $k$ of any residue class modulo $S$,
$(g_1\ast g_2)(k)$ is a mod-$S$ polyexponential with bases
$\ell_1$ and $\ell_2$.
\end{proof}
As a corollary of the lemma, it follows that the convolution of
any (finite) number of mod-$S$ polyexponentials with bases $L$
is again such a mod-$S$ polyexponential.

Returning to the proof of Theorem~\ref{th:polyexp-m-conv}, consider
polyexponentials $g_1,\ldots,g_t$ with bases $L$,
and $\vec m\ge \vec 1$.
Let $S$ to be the least common multiple of the $m_i$.
For $i=1,\ldots,t$ the functions
$$
\tilde g_i(k) = \left\{ \begin{array}{ll} g_i(k/m_i) & 
\mbox{if $m_i$ divides $k$, and} \\ 0 \mbox{otherwise} \end{array}\right.
$$ 
are mod-$m_i$ polyexponentials with
bases $L$, and hence are mod-$S$ polyexponentials with bases
$L^{S/m_i}$.  Since
$$
(g_1 \ast g_2 \ast \cdots \ast g_s)_{\vec m}(k)
=
(\tilde g_1 \ast \cdots \ast \tilde g_s)(k),
$$
Theorem~\ref{th:polyexp-m-conv} follows.
\end{proof}

\subsection{Strongly Ramanujan Graphs}

\begin{definition} 
Let $B$ be a $d$-regular graph and $m\ge 2$ an integer.  
We say that $B$ is {\em $m$-strongly
Ramanujan} if the Hashimoto eigenvalues of $B$, excepting $d-1$
and the possible eigenvalue $-d+1$ 
all have absolute value strictly less than $(d-1)^{1/m}$.
\end{definition}

For example, $B$ is $2$-strongly Ramanujan if $B$ is Ramanujan and
has no eigenvalue equal to either $\pm 2\sqrt{d-1}$.
For another example, consider any $d$-regular
graph, $B$, on one vertex, i.e., a bouquet of $m_1$ whole-loops and 
$m_2$ half-loops with $2m_1+m_2=d$
(e.g.,
$W_{d/2}$, the bouquet of $d/2$ whole-loops, or $H_d$, the bouquet of
$d$ half-loops).
Then $B$ has eigenvalues $d-1$ and $\pm 1$ (by direct calculation
or by \eqref{eq:Zetahalf}).
Hence any such $B$ is $m$-strongly Ramanujan for any integer $m\ge 2$.

\begin{definition}\label{de:strongly_Ramanujna}
We say that a function $f=f(k)$ is $m$-strongly $B$-Ramanujan if
$f$ is the sum of a mod-$S$ polyexponential functions, for some integer,
$S$, plus an error term which for every $\epsilon>0$ is bounded by
$$
C \bigl( \rho(H_B) + \epsilon \bigr)^{k/m}.
$$
\end{definition}

The same estimates used to prove 
Theorem~\ref{th:certified_trace_expansion_with_tangles} can be modified
to prove the following theorem.

\begin{theorem}
Let $B$ be a $d$-regular graph and $m>0$ an integer for 
which $B$ is $m$-strongly Ramanujan.  Let $t$ be the smallest order
among those feasible $B$-graphs, $\psi$, for which
$$
\rho(H_\psi) > \bigl( \rho(H_B) \bigr)^{1/m} .
$$
Then the coefficients in
\eqref{eq:certified_expansion_without_tangles} are $m$-strongly
$B$-Ramanujan functions.  Furthermore, the coefficients of order
less than $t$ are the same coefficients in the $1/n$-asymptotic expansion
of
$$
\expect{G \in \cC_n(B)}{\Tr(H_B^k)}.
$$
\end{theorem}

For example, if $B$ is any $d$-regular,
then the proof of Theorem~\ref{th:fund_lower_bound}
shows that the above theorem we have
$$
t> (d-1)^{1/m}.
$$
It follows, for example, that if $B$ is a bouquet of self-loops of degree
$d$, then the coefficients of 
$$
\expect{G \in \cC_n(B)}{\Tr(H_B^k)}
$$
to any order of at most $(d-1)^{1/m}$ are $m$-strongly Ramanujan.
In particular, for $m=3$, we see that the coefficients are mod-$2$ 
polyexponential functions plus an $O(d-1)^{(k+\epsilon)/3}$ error term.
The proof of Theorem~\ref{th:main_Alon_Ramanujan} implies that the
base $d-1$ term vanishes for the coefficients of order less than
$\taufund(B)>(d-1)^{1/2}$.
It follows that these coefficients functions given by
$$
(d-1)^{k/2} p_i(k/2) \quad\mbox{or, respectively,}\quad
(d-1)^{(k-1)/2} q_i\bigl( (k-1)/2\bigr)
$$
for $k$, respectively, even or odd.
It may be interesting to determine these coefficients, as these
coefficients contain (at least in principle) information about the 
distribution of $H_B$ eigenvalues.

\section{Remarks for Future Directions}
\label{se:p2-future}

In this section we make remarks for possible future research on 
the Relativized Alon Conjecture.

\subsection{Weighted Hashimoto Matrices}

In this subsection describe some possible generalizations of
Theorem~\ref{th:main} to what we call {\em weighted versions} of
the Hashimoto matrix.

\begin{definition} 
Let $M$ be a matrix with non-negative entries.
By a {\em weighted version} of $M$ we mean any matrix, $\widetilde M$,
of the same dimensions as $M$, and with non-negative entries, such that
for each $i,j$ we have
$$
M_{ij} = 0 \quad\iff\quad
\widetilde M_{ij}=0.
$$
If $\pi\from G\to B$ is a covering map, and
$\widetilde H_B$ is a weighted version of $H_B$, then we define
its {\em weighted pullback},
$\widetilde H_G$, to be the weighted version of $H_G$ with weights
induced from $B$ in the natural way: namely, for each $e,e'\in E_G$ with
$(H_G)_{e,e'}\ne 0$, we set
$$
(\widetilde H_G)_{e,e'} = (H_B)_{\pi(e),\pi(e')}.
$$
We similarly define a weighted pullback of a weighted version of $A_B$.
\end{definition}

It seems likely that Theorem~\ref{th:main}
will generalize to weighted versions, $\widetilde H_B$,
of $H_B$, under some conditions.
As an example of a condition, it seems that the case where
$\rho(\widetilde H_B)<1$ would need some care, for then
$(\widetilde H_B)^k$ does not dominate
$(\widetilde H_B)^{k/2}$, and this may cause problems in our
expansion theorems and/or side-stepping lemmas.

Since, at present, we have no use for the generalization to weighted
Hashimoto matrices, we will not pursue such theorems here.
However, such theorems could give a lot of 
spectral information related to the adjacency matrix and/or weighted
versions of the adjacency matrix.
Perhaps one could use such information to improve our results on the
Relativized Alon Conjecture when the base graph, $B$, is not regular.

\subsection{Direct Adjacency Matrix Traces}

Another way to attack the Relativized Alon Conjecture when the base
graph, $B$, is not regular, would be to try to adapt our methods
to directly estimate the expected values of powers of the
adjacency matrix of a $G\in\cC_n(B)$.
This is the approach taken by Puder in \cite{puder}.
If we look at the graph of a closed walk in $G$ that does not need 
to be non-backtracking, then part of the theory carries over without
much difficulty: namely, we can define coincidences in the same way
as before, and we get can understand the number of coincidences encountered
in terms of the order of the walk.

However, our methods, arising from those of Broder-Shamir \cite{broder},
seem to require some new idea(s) in order to work.  Namely, fix
the graph, $G'=\Graph(w)$, of a walk, $w$, in $G$ that is allowed
to backtrack.
The computation of how many walks are compatible with $G'$ seems difficult,
because the edge multiplicities could differ for each edge.
So it is not clear if there is a ``modified'' adjacency trace for which
we can prove a
$1/n$-asymptotic expansion to arbitrarily large order.
We remark, however, that Puder \cite{puder} obtains enough 
information about such expansions,
for each such $G'$, to get a high probability new adjacency spectral
bound of less than
$1+2(d-1)^{1/2}$;
yet, without modifying adjacency traces, we know that we cannot
obtain the full Relativized Alon Conjecture.
Perhaps by a combination of Puder's methods and ours---and perhaps
some new ideas---one
can find a modified adjacency trace for which one can prove
$1/n$-asymptotic expansions to arbitrary order with coefficients that
are polyexponential with an error term of type
$O(2\sqrt{d-1}+\epsilon)^k$, and thereby establish the 
Relativized Alon Conjecture an arbitrary base graph, $B$.

\appendix

\backmatter



\renewcommand{\glossarypreamble}{Numbers in italic indicate
primary definitions.  Greek letters are alphabetized by their
English spelling.\par\bigskip\par\gdef\glossarypreamble{}}

\glossarystyle{tree}

\renewenvironment{theglossary}{%
  \setlength{\parindent}{0pt}%
     \setlength{\parskip}{5pt plus 5pt}}%

\providecommand{\glossarytoctitle}{\glossaryname}

\printglossaries

\newcommand{\etalchar}[1]{$^{#1}$}
\providecommand{\bysame}{\leavevmode\hbox to3em{\hrulefill}\thinspace}
\providecommand{\MR}{\relax\ifhmode\unskip\space\fi MR }
\providecommand{\MRhref}[2]{%
  \href{http://www.ams.org/mathscinet-getitem?mr=#1}{#2}
}
\providecommand{\href}[2]{#2}


\end{document}
